\documentclass{article}
\usepackage[utf8]{inputenc}
\usepackage{jheppub}

\usepackage{amsmath,amssymb,amsthm,amscd}
\usepackage{mathtools}

\usepackage{xcolor}
\usepackage{url}
\usepackage{listings}
\usepackage{tabularx}
\usepackage{diagbox}
\usepackage{tikz-cd}
\usepackage{nicematrix}
\usepackage{pifont}
\usepackage{booktabs}
\usepackage{multirow}
\usepackage{enumitem}
\usepackage{subcaption}

\usepackage{tikz,textcomp}
\usetikzlibrary{decorations.pathreplacing,calc,fadings,fit,shapes,arrows,positioning,chains,matrix}

\newcommand\quotient[2]{{^{\displaystyle #1}}\Big/{_{\displaystyle #2}}}

\newtheorem{theorem}{Theorem}
\newtheorem{lemma}{Lemma}
\newtheorem{conjecture}{Conjecture}

\newtheorem{definition}{Definition}
\newtheorem{claim}{Claim}
\newtheorem*{questionsnn}{Questions}
\newtheorem*{questionnn}{Question}

\DeclareMathOperator{\ord}{ord}

\newcommand{\IZ}{\mathbb{Z}}
\newcommand{\IC}{\mathbb{C}}
\newcommand{\IP}{\mathbb{P}}

\newcommand{\INnz}{\mathbb{N}^*}

\newcommand{\IQ}{\mathbb{Q}}
\newcommand{\IH}{\mathbb{H}}

\newcommand{\IE}{\mathbb{E}}
\newcommand{\IS}{\mathbb{S}}

\newcommand{\id}{\mathrm{id}}

\newcommand{\cO}{{\cal O}}

\newcommand{\cT}{{\cal T}}

\newcommand{\cL}{{\cal L}}

\newcommand{\bk}{\boldsymbol{k}}
\newcommand{\bl}{\boldsymbol{l}}
\newcommand{\bv}{\boldsymbol{v}}
\newcommand{\bq}{\boldsymbol{q}}
\newcommand{\bQ}{\boldsymbol{Q}}

\newcommand{\KK}{\mathrm{KK}}

\newcommand{\e}{\mathrm e}

\newcommand{\ii}{\mathrm{i}}
\newcommand{\dd}{\mathrm{d}}

\newcommand{\tr}{\mathrm{tr}}
\newcommand{\tors}[1]{#1_{\mathrm{tors}}}
\newcommand{\Tors}[1]{\mathrm{Tors}\,\,#1}
\newcommand{\free}[1]{#1_{\mathrm{free}}}

\DeclareMathOperator{\corank}{corank}
\DeclareMathOperator{\adj}{adj}
\newcommand{\Gr}{\mathrm{Gr}}
\DeclareMathOperator{\codim}{codim}
\newcommand{\Ch}{\mathrm{Ch}}
\newcommand{\chitop}{\chi_{\mathrm{top}}}
\newcommand{\Amp}{\mathrm{Amp}}

\newcommand{\mslash}[1]{\left.\right|_{\raisebox{-0.3ex}{$\scriptstyle #1$}}}

\newcommand{\F}{\mathrm{F}}
\newcommand{\M}{\mathrm{M}}

\newcommand{\U}{\mathrm{U}}
\newcommand{\G}{\mathrm{G}}

\newcommand{\Gsd}{\G_{6\dd}}
\newcommand{\Gfd}{\G_{5\dd}}
\newcommand{\tGsd}{\widetilde{\G}_{6\dd}}
\newcommand{\tGfd}{\widetilde{\G}_{5\dd}}

\newcommand{\tcT}{\widetilde{\cT}}

\newcommand{\Gd}{\G[\cT_d]}
\newcommand{\tGd}{\G[\tcT_d]}
\newcommand{\Gdm}{\G[\cT_{d-1}]}
\newcommand{\tGdm}{\G[\tcT_{d-1}]}

\newcommand{\myvec}[1]{\boldsymbol{#1}}

\newcommand{\Azz}{A_{z_1,z_2}}

\newcommand{\tse}{\gamma}

\newcommand{\tseOne}{{\tse_1}}
\newcommand{\tseTwo}{{\tse_2}}
\newcommand{\tseOnePrime}{{\tse'_1}}
\newcommand{\tseTwoPrime}{{\tse'_2}}
\newcommand{\Jbase}{j}
\newcommand{\gre}{\tse}

\newcommand{\lii}{i}
\newcommand{\lij}{j}
\newcommand{\lik}{k}
\newcommand{\sia}{\alpha}
\newcommand{\sib}{\beta}

\newcommand{\nn}{\nonumber}

\DeclareFontFamily{U}{wncy}{}
\DeclareFontShape{U}{wncy}{m}{n}{<->wncyr10}{}
\DeclareSymbolFont{mcy}{U}{wncy}{m}{n}
\DeclareMathSymbol{\Sh}{\mathord}{mcy}{"58}

\def\SLtwoZ{\text{SL}(2,\mathbb{Z})}

\title{\centering The twisted geometry of 6d F-theory vacua\\with discrete gauge symmetries}

\abstract{
We study the fate of discrete gauge groups and discrete charges of gravitational theories under twisted circle compactification. We then apply our results to six-dimensional F-theory vacua with discrete gauge symmetries and relate them to the geometry of the genus one fibered Calabi-Yau threefolds that underlie the dual M-theory compactifications.
This leads us to introduce a class of geometries, which we call {\it almost generic elliptic/genus one fibered Calabi-Yau threefolds}, and to make detailed conjectures about their properties.
A second twisted circle compactification relates these M-theory vacua to Type IIA compactifications with flat but topologically non-trivial B-fields along the internal geometry. The A-model topological string partition function on such configurations is intimately tied to the twisted-twined elliptic genera of the six-dimensional non-critical strings of the associated F-theory vacuum.
The modular properties of the twisted-twined elliptic genera imply new twisted derived equivalences.
We thus recover and significantly extend earlier results from both the physical and the mathematical literature.
An important outcome of our study is that if the discrete gauge symmetry is not cyclic, then no smooth genus one fibration exists that represents the associated axio-dilaton profile.
}

\author[a]{David Jaramillo Duque}
\author[b]{Amir-Kian Kashani-Poor}
\author[c]{Thorsten Schimannek}

\affiliation[a]{Centre R\&I Talan, 14 rue Pergolèse, 75116 Paris, France}

\affiliation[b]{Laboratoire de Physique de l’\'Ecole normale sup\'erieure,\\
CNRS, PSL Research University and Sorbonne Universit\'es,\\
24 rue Lhomond, 75005 Paris, France}

\affiliation[c]{
Institute for Theoretical Physics \& Department of Mathematics,\\
Utrecht University,\\
3584 CC Utrecht, The Netherlands
}

\emailAdd{amir-kian.kashani-poor@ens.fr}
\emailAdd{thorsten.schimannek@gmail.com}

\DeclareFontFamily{U}{wncy}{}
\DeclareFontShape{U}{wncy}{m}{n}{<->wncyr10}{}
\DeclareSymbolFont{mcy}{U}{wncy}{m}{n}
\DeclareMathSymbol{\Sh}{\mathord}{mcy}{"58} 

\begin{document}

\maketitle

\section{Introduction}
A key lesson of string theory is that geometry encodes physics, and, in turn, physics can elucidate geometry.
Moreover, string theory enhances the ordinary metric structure on the underlying geometry with additional degrees of freedom, such as the B-field and the dilaton, thus enlarging the classical notions of geometry when specifying a string theory background.
Intimately tied to this generalized notion of geometry is the insight that to capture the richness of physical phenomena, we must include string compactifications on singular spaces.
This paper studies the intricate interplay of these notions in the context of F-theory compactifications that lead to gravitational theories with discrete gauge symmetries.

Recall that in the context of a Type IIB string compactification on a space $B$, F-theory encodes a non-trivial axio-dilaton profile -- and the corresponding 7-brane stacks --  in a genus one fibration $\pi:X\rightarrow B$~\cite{Vafa1996}.
We will focus on the situation where $B$ is a projective surface and $X$ a projective Calabi-Yau threefold, such that the corresponding effective theory is a six-dimensional $\mathcal{N}=1$ supergravity (SUGRA)~\cite{Morrison1996a,Morrison1996b}.

The geometric origin of discrete gauge symmetries in F-theory has been studied in~\cite{deBoer:2001wca,Braun:2014oya,Morrison:2014era,Anderson:2014yva,Klevers:2014bqa,Mayrhofer:2014laa,Cvetic:2015moa,Oehlmann:2019ohh,Knapp:2021vkm}, and it is conjectured that the group
\begin{align}
    \Gamma=\pi_0(\Gsd)   
\end{align}
of connected components of the six-dimensional gauge group $\Gsd$ corresponds to the so-called Tate-Shafarevich group (or, more generally, the Weil-Ch\^{a}telet group) that is associated to the fibration.
The definition of this group is somewhat technical and we relegate it to Section~\ref{ss:TSgroup}.
Physically, the elements can be thought of as being associated to geometries that encode the same F-theory vacuum.
However, the geometries can still be rather different and, in general, they lead to different M-theory vacua.

At first glance, this seems to be at odds with the expectation that the compactification of the F-theory vacuum on an additional circle is dual to the corresponding M-theory vacuum.
The resolution of this apparent paradox lies in the possibility to modify the boundary conditions of fields along the circle by the action of a gauge transformation.
If the gauge transformation represents a non-trivial element $\tse\in\Gamma$, we say that the compactification is twisted by $\tse$.
This twist can be thought of as the discrete analogue of a Wilson line and amounts to the choice of a topology of the discrete gauge bundle.
The compactification of the six-dimensional F-theory vacuum on a circle with twist $\tse\in\Gamma$ is then conjecturally dual to the M-theory compactification on the corresponding genus one fibered Calabi-Yau threefold $X^\tse$~\cite{Morrison:2014era,Cvetic:2015moa}.

It will be useful to introduce the following notation.
We denote the effective theories that are associated to the compactifications of F-, M- and Type IIA string theory on a Calabi-Yau threefold respectively by ${\rm T}[X]$ with ${\rm T}\in\{{\rm F},{\rm M},{\rm IIA}\}$.
The choice of fibration structure, which determines the F-theory vacuum ${\rm F}[X]$, will be left implicit.
Moreover, if ${\rm G}$ is the gauge group of ${\rm T}[X]$, we denote the circle compactification with a twist $\tse\in {\rm G}$ along the circle by ${\rm T}[X][S^1_\tse]$.
The duality between F- and M-theory then takes the form
\begin{align}
    {\rm F}[X^0][S^1_\tse]={\rm M}[X^\tse]\,.
    \label{eqn:MFduality}
\end{align}

It turns out that the gauge group $\Gfd^\tse$ of the five-dimensional effective theory of a circle compactification ${\rm F}[X^0][S^1_\tse]$ depends on the choice of twist $\tse\in\Gamma$.
This is perhaps surprising, because unlike continuous gauge symmetries, a discrete gauge symmetry is not associated to a massless gauge boson which could acquire a mass which depends on the choice of the twist.
Nevertheless, we will carefully analyze twisted circle compactifications of gravitational theories with discrete gauge symmetries to determine $\Gfd^\tse$ and find that the group of connected components $\Gamma^\tse=\pi_0(\Gfd^\tse)$ is
\begin{align}
    \Gamma^\tse=\quotient{\Gamma}{\langle\tse\rangle}\,,
\end{align}
where $\langle\tse\rangle\subset\Gamma$ is the subgroup that is generated by $\tse$.
This result will have profound physical and geometrical implications.

If $\Gamma^\tse \neq 0$, we again have the freedom of twisting upon performing a second circle compactification.
For trivial twist, the resulting four-dimensional theory is conjecturally dual to a compactification of Type IIA string theory on the Calabi-Yau $X^\tse$.
On the other hand, if the twist is non-trivial, it has recently been proposed that the twist corresponds to a choice of topology for the B-field background in the dual Type IIA compactification on $X^\tse$~\cite{Schimannek:2021pau,Dierigl:2022zll,Katz:2022lyl,Katz:2023zan,Schimannek:2025cok}.
Denoting the group of topologies of flat B-fields on $X^\tse$ by $B(X^\tse)$, this implies that $B(X^\tse)=\Gamma^\tse$.

Given two elements $\tseOne,\tseTwo\in\Gamma$, we will  denote the geometry $X^\tseOne$ together with the choice of topology for the flat B-field $[\tseTwo]\in B(X^\tseOne)$ by $X^\tseOne_\tseTwo$.\footnote{In the following, we will sometimes -- in the spirit of the above opening remarks -- simply refer to $X^\tseOne_\tseTwo$ as a geometry.}
Together, the dualities between F-theory, M-theory and Type IIA string theory then take the form
\begin{align}
    {\rm F}[X^0][S^1_\tseOne][S^1_\tseTwo]={\rm M}[X^\tseOne][S^1_\tseTwo]={\rm IIA}[X^\tseOne_\tseTwo]\,.
    \label{eqn:MFAduality}
\end{align}

Our goal in this paper is to systematically study twisted compactifications of six-dimensional F-theory vacua and to use the duality with M-theory and Type IIA string theory to relate their physical properties to the geometry of the corresponding Calabi-Yau threefolds and to the modular properties of the topological string partition functions.

We will assume that the F-theory vacuum is almost generic in the following sense:
\begin{definition}
    We call a six-dimensional F-theory vacuum ``almost generic'' if the following holds:
    \begin{enumerate}
        \item The vacuum expectation values (vevs) of the tensor multiplet scalars are generic.
        \item The vevs of the scalar fields in the uncharged hypermultiplets are also generic.
    \end{enumerate}
    \label{def:almostGenericFtheory}
\end{definition}

As our focus is on discrete gauge symmetry, we will for simplicity exclude six-dimensional vacua with massless vector multiplets from our considerations. This implies that the six-dimensional gauge group is finite, such that $\Gamma=\Gsd$.
Moreover, since the Tate-Shafarevich group is always Abelian, it takes the form
\begin{align}
    \Gsd=\mathbb{Z}_{k_1}\times\ldots\times \mathbb{Z}_{k_r}\,,\qquad k_1,\ldots,k_r\in\mathbb{N}\,.
    \label{eqn:gaugegroupFAproduct}
\end{align}
Already in this restricted setup, we will find intricate relationships between the geometry of the genus one fibered Calabi-Yau threefolds and the physics of the twisted circle compactifications.
Even if one further restricts to F-theory vacua without any tensor multiplets, corresponding to the base of the genus one fibration being $B=\mathbb{P}^2$, and assumes that the Calabi-Yau threefold is smooth -- which as we will argue in this paper only happens when $\Gsd=\mathbb{Z}_k$ for some $k\in\mathbb{N}$ -- a very conservative lower bound on the number of the corresponding Calabi-Yau threefolds with $h^{1,1}=2$ is $60$, while the actual number is likely to be significantly larger~\cite{Pioline:wip,Dierigl:wip}.
Moreover, our techniques can be readily applied to study theories with vector multiplets and higher rank gauge groups.

From a mathematical perspective, we will define a class of genus one fibered Calabi-Yau threefolds that we also refer to as \textit{almost generic}.
Compactifications of F-theory on such almost generic fibrations lead to almost generic vacua in the sense above.~\footnote{The question if almost generic F-theory vacua without massless vector multiplets always arise from compactifications on such almost generic fibrations is more difficult and will be discussed in Section~\ref{sec:genericity}.}
Our physical analysis then translates to precise mathematical conjectures about the properties of these fibrations and the geometries that are associated to the various elements of the corresponding Tate-Shafarevich groups.
These conjectures should also make working with the Tate-Shafarevich group of such fibrations more accessible to physicists. They will serve as preparation for our study of F-theory vacua with gauge group $\Gsd=\mathbb{Z}_2\times\mathbb{Z}_2$ in~\cite{wipZ2Z2}.

We will now outline the central themes of our investigation.

\paragraph {Analytic resolutions}
If the twist $\tse\in\Gamma$ generates a proper subgroup $\langle\tse \rangle \subset \Gamma$, the theory ${\rm M}[X^\tse]$ has a non-trivial gauge group $\Gamma^\tse$.
Assuming that the spectrum of the six-dimensional theory is sufficiently generic, the five-dimensional theory will exhibit massless (half-)hypermultiplets that carry a charge only under the discrete part of the gauge group $\Gfd\simeq\ldots \times\Gamma^\tse$.
These states hence remain massless even on a generic point of the Coulomb branch. Yet the fact that they are charged suggests that they too, just like hypermultiplets carrying the charge of a continuous gauge group factor, arise from M2-branes that probe isolated singularities of $X^\tse$.
That these states cannot acquire a Coulomb branch mass geometrically corresponds to the fact that these singularities are $\mathbb{Q}$-factorial terminal and do not admit any K\"ahler small resolution, see e.g.~\cite{Arras:2016evy,Baume:2017hxm}. We therefore expect that $X^\tse$ is smooth if and only if $\Gamma=\langle\tse\rangle$.
Assuming that $\Gamma$ is finite Abelian, this can only happen if $\Gamma$ is cyclic, i.e. $\Gamma=\mathbb{Z}_k$ for some $k\in\mathbb{N}$.
To study six-dimensional F-theory vacua with non-cyclic gauge groups, it seems that one cannot avoid working with singular Calabi-Yau threefolds.

To circumvent the difficulties associated with working with singular geometries, we base our analysis on the proposal~\cite{Katz:2022lyl, Katz:2023zan, Schimannek:2025cok}, summarized in Claims~\ref{claim:mtheoryNodalCYgaugeGroup} and~\ref{claim:mtheoryNodalCYmatter} in Section~\ref{sec:discreteHolonomiesAndGenusOne}, that the gauge group and matter content of the M-theory compactification on the singular Calabi-Yau $X^\gamma$ are encoded in any analytic (thus possibly non-K\"ahler) small resolution $\widehat{X}^\gamma$ of $X^\gamma$.
From our general study of twisted circle compactifications in Section~\ref{sec:kkandhol}, we learn how to determine the discrete gauge symmetry which survives the twisted circle compactification and to identify the Kaluza-Klein towers which retain massless modes. This leads to detailed predictions regarding the topology of the geometries $X^\gamma$ and $\widehat{X}^\gamma$. These are summarized in the Conjectures~\ref{conj:TSandTorsion} and~\ref{conj:I2fibers} in Section~\ref{sec:almostgeneric}, and motivated physically in Sections~\ref{sec:GandLambdaTwisted} and~\ref{ss:twistedAndGenusOne}.

\paragraph{Higgs transitions and smoothing} The singular geometries in the class we consider all admit smooth deformations (this is a consequence of Theorem \ref{thm:NamikawaSteenbrinkDeformation} cited in Appendix \ref{sec:nodalCY3}). Such deformations $\widetilde{X}^\tse$ correspond to Higgs transitions in the M-theory compactifications on the singular geometries $X^\tse$. We study the fibration structure of these deformations in Section \ref{sec:partialHiggsing}. In particular, we conclude that, just like $X^\tse$, $\widetilde{X}^\tse$ exhibits a torus fibration which exhibits a multisection of minimal degree $m_\tse= \ord(\tse)$. By recourse to the interplay between the Chern-Simons terms in the M-theory compactification and anomaly cancellation, we determine topological invariants of $\widetilde{X}^\tse$, notably triple intersection numbers, based on enumerative invariants of $X^0$. This is the content of Conjecture \ref{conj:topologyOfSmoothing}, which summarizes calculations presented in Section \ref{ss:CSterms}.

\paragraph{Modular properties of the topological string partition function}
The compactification from six to four dimensions does not break supersymmetry. The duality~\eqref{eqn:MFAduality} therefore relates BPS excitations in the respective theories.
As a result, the A-model topological string partition function of a smooth elliptically fibered Calabi-Yau threefold encodes the elliptic genera of non-critical strings that exist in the corresponding six-dimensional F-theory compactification~\cite{Klemm:1996hh,Haghighat:2013gba,Haghighat:2014vxa}.
This gives a physical explanation for the modular properties of the topological string partition function that have already been observed in~\cite{Candelas:1994hw} and were further studied in~\cite{Klemm:1996hh,Klemm:2012sx,Alim:2012ss,Haghighat:2013gba,Haghighat:2014vxa,Huang:2015sta,DelZotto:2016pvm,Gu:2017ccq,DelZotto:2017mee,DelZotto:2018tcj,Schimannek:2019ijf,Cota:2019cjx,Duan:2020imo,Knapp:2021vkm,Schimannek:2021pau,Duque:2022tub,Dierigl:2022zll}.

It was found in~\cite{Cota:2019cjx} that the topological string partition function of a smooth genus one fibered Calabi-Yau threefold that exhibits an $N$-section with $N>1$, rather than a section, appears to exhibit modular properties only with respect to the congruence subgroup $\Gamma_1(N)\subset\SLtwoZ$.
In~\cite{Knapp:2021vkm}, it was then observed that the topological string partition functions of two smooth genus one fibrations with {$5$-sections} related by relative homological projective duality transform as vector valued modular forms under an extension of $\Gamma_1(5)$.
Subsequently, in~\cite{Schimannek:2021pau}, it was argued that on a general elliptic or genus one fibered Calabi-Yau threefold, the partition function actually exhibits vector valued modular behavior under the full modular group $\SLtwoZ$ if one takes into account the possibility to choose a non-trivial topology for the B-field.
This was explained in~\cite{Dierigl:2022zll}, focusing on the case $\Gsd=\mathbb{Z}_3$, by proposing the following identification:
\begin{align}
    \begin{array}{c}
    \text{(topological string partition function on $X^\tseOne_\tseTwo$)}\\
    \updownarrow\\
    \text{($\tseOne$-twisted, $\tseTwo$-twined elliptic genera of non-critical strings in ${\rm F}[X^0]$)}
    \end{array}
    \label{eqn:FTopduality}
\end{align}
Our understanding of the gauge groups of twisted circle compactifications allows us to generalize this identification~\eqref{eqn:FTopduality} with the twisted-twined elliptic genera, and therefore also the modular properties of the topological string partition function, for arbitrary finite Abelian groups $\Gamma$ and choices $\tseOne,\tseTwo\in\Gamma$.
We incorporate this identification as the first of three claims in Section~\ref{sec:ellipticGeneraAndTopologicalStrings} on which we base our analysis of the geometries $X^\tseOne_\tseTwo$ \cite{Schimannek:2019ijf, Schimannek:2021pau,Schimannek:2025cok}.

\paragraph{Torsional B-fields, moduli spaces and enumerative invariants}
It has been observed~\cite{Schimannek:2021pau,Katz:2022lyl,Katz:2023zan,Schimannek:2025cok} that, just like for smooth Calabi-Yau threefolds with a trivial B-field topology, the string backgrounds $X^\tseOne_\tseTwo$ admit large volume limits that are related via mirror symmetry to points of maximally unipotent monodromy (MUM-points) in the complex structure moduli spaces of corresponding mirror Calabi-Yau threefolds.
In fact, the different $X^\tseOne_\tseTwo$ that are related to twisted-twined elliptic genera that lie in the same $\SLtwoZ$-orbit are expected to correspond, under mirror symmetry, to different MUM-points in the moduli space of the same Calabi-Yau threefold.\footnote{Considering singular Calabi-Yau threefolds with torsional B-field might seem esoteric. But note that when such an orbit contains the mirror of a smooth Calabi-Yau threefold, taking this mirror as a point of departure leads unavoidably to the consideration of mirrors of such ``esoteric'' configurations.} The category of B-twisted topological branes on $X^\tseOne_\tseTwo$ is conjecturally the derived category of $\tseTwo$-twisted sheaves on any small resolution $\widehat{X}^\tseOne$ of $X^\tseOne$~\cite{caldararuThesis,Caldararu2002}.
Situating different $X^\tseOne_\tseTwo$ in the same moduli space should imply that the corresponding categories can be identified via brane transport.
This implies twisted derived equivalences which are the content of Conjecture~\ref{conj:twisted}, presented in Section \ref{sec:mathresults} and discussed in Section \ref{sec:modularityAndDerivedEquivalences}.

In explicit examples, enumerative evidence supports the identification of MUM points with geometries $X^\tseOne_\tseTwo$ via the notion of torsion-refined Gopakumar-Vafa invariants, introduced in~\cite{Schimannek:2021pau}: when $\tors{H_2(\widehat{X}^\tse,\IZ)}\neq 0$, the conventional Gopakumar-Vafa invariants can be refined by keeping track of the torsion class of curves. In the topological string partition function, the parameter keeping track of this refinement is the torsional B-field, which pairs with the torsion class of the curve to provide a phase. Extracting the torsion-refined invariants requires knowledge of the partition function for different values of the torsional B-field,\footnote{Consider e.g. $\IZ_2$ torsion, and denote the invariants by $n_{\beta,+}$ and $n_\beta,-$, with $\beta$ keeping track of the free part of the curve class.
Then the partition function for vanishing torsional B-field only allows extracting the sum $n_{\beta,+} + n_{\beta,-}$.
On the other hand, a torsional B-field 1 mod 2 introduces a sign in the partition function that depends on the $\mathbb{Z}_2$-charge. Knowledge of both partition functions is required to extract the individual invariants $n_{\beta,+}$ and $n_{\beta,-}$.} providing a check on the overall consistency of our computations and identifications. 

\paragraph{Complete invariants of generic genus one fibered Calabi-Yau threefolds}
A natural question is if a set of invariants exists that uniquely characterizes the topological type of a generic genus one fibered Calabi-Yau threefold.
This question is closely related to the cancellation of global anomalies in the corresponding F-theory vacua, which has been discussed in~\cite{Dierigl:2022zll} and will be further elaborated on in~\cite{Dierigl:wip}.
While we will not discuss questions related to global anomalies in this paper, we have included Conjecture~\ref{conj:datum} that proposes a complete set of topological invariants and is based on the results from~\cite{Dierigl:2022zll,Dierigl:wip}.

\vspace{0.5cm}

\paragraph{Outline}
This paper is organized as follows.
In Section~\ref{sec:kkandhol}, we discuss the physics of twisted circle compactifications and Higgs transitions in gravitational theories that couple to a discrete gauge symmetry.
The results of this section are independent of string theory and rely entirely on field theoretic considerations.
They are illustrated in an example with gauge group $\IZ_4$ in Section~\ref{sec:fieldTheoryExampleZ4}.
The following Section~\ref{sec:almostgenericfibrations} is devoted to the geometry of genus one fibered Calabi-Yau threefolds.
After reviewing the definition of the Tate-Shafarevich and the Weil-Ch\^atelet group, we introduce the class of {\it almost generic genus one fibered Calabi-Yau threefolds} and discuss some of their properties that follow directly from the definition.
Then, in Section~\ref{sec:mathresults}, we summarize the results of our paper in the form of five mathematical conjectures.
We present the physical considerations based on F/M-theory duality that underlie these conjectures in Section~\ref{sec:discreteHolonomiesAndGenusOne}, and those based on M-theory/type IIA duality in Section~\ref{sec:ellipticGeneraAndTopologicalStrings}.
The latter section also provides a more detailed map between the topological string partition function and the twisted-twined elliptic genus of non-critical strings than available to date. Finally, we discuss the conjectures and their implications at the hand of an example with Tate-Shafarevich group $\IZ_4$ in Section~\ref{sec:exampleZ4}.
Numerous technical details are relegated to the appendices.

\subsection{Notation}
Throughout this paper, we will need to introduce a rich set of definitions and notation.
Here, we summarize some of the recurring notation in the hope that this will aid the reader in following the exposition.
We use ``$\stackrel{\rm c}{=}$'' to denote identifications that are conditioned on the conjectures from Section~\ref{sec:mathresults}.
\begin{itemize}
    \item $\e[x] = \exp(2\pi \ii x)$.
    \item $[a]_b$ denotes the equivalence class of $a$ in the additive group $\mathbb{Z}_b=\mathbb{Z}/b\mathbb{Z}$.
    \item $\Gamma$ or $\Gamma^0$ are finite Abelian groups.
    \item $\Gamma^\tse=\Gamma^0/\langle\tse\rangle$ for $\tse\in\Gamma^0$.
    \item $m_\gre$ is the order of a group element $\gre \in \G$.
    \item $m$ is sometimes used for $m_\gre$ when the element $\gre$ is clear from the context.
    \item Given a group $\G$, the corresponding group of characters, or Pontryagin dual, is $\widehat{\G}=\text{Hom}(\G,\U(1))$. 
    \item Given $\tse\in\Gamma$ and $\chi\in\widehat{\Gamma}$, there is a unique $q_\tse(\chi)\in\{0,\ldots,m_\tse-1\}$ such that $\chi(\tse)=\e[q_\tse(\chi)/m_\tse]$.
    \item $\chi_\phi\in\widehat{\G}$ is the charge of a field $\phi$ with respect to a gauge group $\G$.
    \item Geometrically, $\Gamma^0=\Sh_B(X^0)$ will be the Tate-Shafarevich group of an almost generic elliptically fibered Calabi-Yau threefold $\pi_0:X^0\rightarrow B$ in the sense of Definition~\ref{def:almostgeneric}.
    \item Physically, $\Gsd=\Gamma^0$ is the gauge group of the corresponding F-theory vacuum.
    \item $\pi_\tse:X^\tse\rightarrow B$ is the almost generic genus one fibered Calabi-Yau threefold that is associated to the element $\tse\in\Gamma^0=\Sh_B(X^0)\stackrel{\rm c}{=}\tors{H^3(\widehat{X}^0,\IZ)}$.
    \item $\rho_\tse:\widehat{X}^\tse\rightarrow X^\tse$ is an analytic small resolution of $X^\tse$. If $X^\tse$ is smooth, then $\widehat{X}^\tse=X^\tse$.
    \item $\widetilde{X}^\tse$ is a generic smooth deformation of $X^\tse$. If $X^\tse$ is smooth, then $\widetilde{X}^\tse=X^\tse$.
    \item Geometrically, for $\chi\ne 0$, $N_\chi$ is the number of $I_2$-fibers in an almost generic Weierstra{\ss} fibration $X^0$ that are resolved by exceptional curves that represent the class $\chi$ or $-\chi$ in $\widehat{\Gamma}^0\stackrel{\rm c}{=}\tors{H_2(\widehat{X}^0,\IZ)}$.
    \item Physically, $N_\chi$ is the number of half-hypermultiplets with charge $\chi\in\widehat{\Gamma}^0$ in the F-theory compactification on an almost generic Weierstra{\ss} fibration $X^0$.
    \item $X^\tseOne_\tseTwo$ is the Calabi-Yau $X^\tseOne$ together with a choice of B-field topology $\tseTwo\in \Gamma^\tse\stackrel{\rm c}{=}\tors{H^3(\widehat{X}^\tseOne,\IZ)}$.
    \item $\Delta$ is the discriminant locus of a Weierstra{\ss} fibration $X^0$.
    \item $S_\Delta\subset\Delta$ is the set of nodes of the discriminant locus.
    \item $S^\tse\subset X^\tse$ is the set of nodes of $X^\tse$.
    \item $C_{p,\tse}^E$ is the exceptional curve in $\widehat{X}^\tse$ that resolves a node $p\in S^\tse$.
    \item $S_\tse^0=\{\,p\in S^0\,\,\vert\,\,\tse([C_{p,0}^E])=1\,\}$ for $\tse\in \Gamma^0\stackrel{\rm c}{=}\tors{H^3(\widehat{X}^0,\IZ)}=\text{Hom}(\tors{H_2(\widehat{X}^0,\IZ)},{\rm U}(1))$.
    \item $S_{\Delta,\tse}=\pi_0(S^0_{\tse})\subset S_\Delta$.
\end{itemize}

\section{Kaluza-Klein reduction in the presence of discrete holonomies}\label{sec:kkandhol}
In this section, we will study twisted circle compactifications independently from the F-theory/M-theory context which will occupy us for the rest of the paper. We will address the following questions, generalizing a recent discussion from~\cite{Dierigl:2022zll}:

\begin{questionsnn}
Consider a $d$-dimensional quantum field theory $\cT_d$ coupled to gravity and with a finite Abelian gauge group $\Gd=\Gamma_d$.
What is the gauge group $\Gdm$ of the $(d-1)$-dimensional theory $\cT_{d}[S^1_\gre]$ that arises after compactifying on a circle $S^1$ with a twist $\gre\in\Gamma_d$, and how do charges under $\Gd$ map to charges under $\Gdm$?
\end{questionsnn}

When the finite gauge group $\Gamma_d$ arises upon Higgsing a theory $\tcT_d$ with $\U(1)^p$ gauge symmetry via scalar fields of non-primitive charge, a second path to $\cT_{d}[S^1_\gre]$ exists: we can first perform the circle compactification in the presence of a flat $\U(1)$ connection, and then Higgs the resulting theory to arrive at $\cT_{d}[S^1_\gre]$, see Figure \ref{fig:twoPaths}.

\begin{figure}
\centering
\begin{tikzpicture}[
    box/.style={draw, minimum width=3cm, minimum height=2.0cm, align=center, rounded corners=.25cm},
    node distance=3cm and 5cm 
]

\node[box] (upperLeft) {$\tcT_d$ \\ $\tGd = \U(1)^p$};
\node[box] (upperRight) [right=of upperLeft] {$\cT_d$ \\ $\Gd = \Gamma_d$};
\node[box] (lowerLeft) [below=of upperLeft] {$\tcT_{d-1} = \tcT_d[S^1_\gre]$ \\ $\tGdm = \U(1) \times \U(1)^p$};
\node[box] (lowerRight) [below=of upperRight] {$\cT_{d-1} = \cT_d[S^1_\gre]$ \\ $\Gdm = \U(1) \times \Gamma_{d-1}$ \\ $\Gamma_{d-1} = \Gamma_d / \langle \gre \rangle$};

\coordinate (StartHorizontalLeft) at ($(upperLeft.east)+(0.5cm,0)$);=
\coordinate (EndHorizontalLeft) at ($(upperRight.west)-(0.5cm,0)$);

\coordinate (StartHorizontalRight) at ($(lowerLeft.east)+(0.5cm,0)$);
\coordinate (EndHorizontalRight) at ($(lowerRight.west)-(0.5cm,0)$);

\coordinate (StartVerticalLeft) at ($(upperLeft.south)-(0,0.5cm)$);
\coordinate (EndVerticalLeft) at ($(lowerLeft.north)+(0,0.5cm)$);

\coordinate (StartVerticalRight) at ($(upperRight.south)-(0,0.5cm)$);
\coordinate (EndVerticalRight) at ($(lowerRight.north)+(0,0.5cm)$);

\draw[->,thick] (StartHorizontalLeft) -- (EndHorizontalLeft) node[midway, above] {Higgsing via fields} node[midway, below] {of non-primitive charge};
\draw[->, thick] (StartHorizontalRight) -- (EndHorizontalRight) node[midway, above] {Higgsing via fields} node[midway, below] {of non-primitive charge};

\draw[->,thick] (StartVerticalLeft) -- (EndVerticalLeft) node[midway, left, align=center] {Reduction on $S^1$\\w/ holonomy $\gre \in \tGd$} ;
\draw[->,thick] (StartVerticalRight) -- (EndVerticalRight) node[midway, right, align=center] {Reduction on $S^1$\\w/ holonomy $\gre \in \Gamma_d$};

\end{tikzpicture}
\caption{Two paths towards $\cT_d[S^1_\gre]$} \label{fig:twoPaths}
\end{figure}
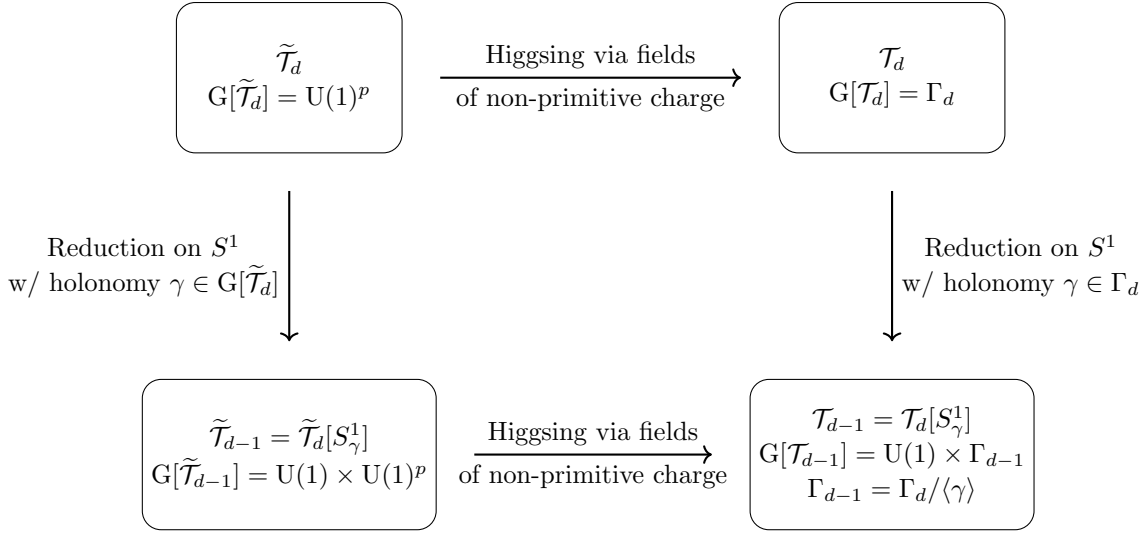

In this section, we will demonstrate, independently via both paths leading to $\cT_{d}[S^1_\gre]$, that the gauge group of the lower dimensional theory is
\begin{align}
    \quotient{\U(1)\times\Gamma_d}{\left(\e\left[-\tfrac{1}{m}\right],\gre\right)}\simeq \U(1)\times\Gamma_{d-1}\,,
\end{align}
where $m$ is the order of $\gre$ in $\Gamma_d$, $\e[x] = \exp(2\pi \ii x)$ and
\begin{equation} \label{eq:GGdm1}
    \Gamma_{d-1} = \Gamma_d/ \langle \gre \rangle \,.
\end{equation}
In addition, we will derive the mapping of charges from $\tGd$ to $\Gdm$ by following either path in Figure \ref{fig:twoPaths} from top left to bottom right.

We emphasize that the derivation of \eqref{eq:GGdm1} starting from $\cT_d$ is intrinsic, i.e. does not depend on embedding the right vertical arrow in Figure \ref{fig:twoPaths} into the full diagram.

\subsection{Some preliminaries on Abelian groups and their characters} \label{sec:mathPrelim}
In the following, we will be interested in matter fields charged under a finite Abelian group $\Gamma$, i.e. that transform in an irreducible representation of $\Gamma$. Such irreducible representations are in 1:1 relation to characters of $\Gamma$, i.e. to group homomorphisms
\begin{equation}
    \chi: \Gamma \rightarrow \U(1) \,.
\end{equation}
We will also use the language of characters when referring to irreducible representations of $\U(1)^p$. A field carrying charge $q$ under $\U(1)$ transforms via the character
\begin{equation}
    \chi_q : \U(1) \rightarrow \U(1) \,, \quad \xi \mapsto \xi^q \,.
\end{equation}
In general, we will refer to the group of characters of a group ${\rm G}$ as
\begin{align}
    \widehat{{\rm G}}=\text{Hom}\left({\rm G},\U(1)\right)\,.
\end{align}
This is also referred to as the Pontryagin dual of the group ${\rm G}$.

Clearly, given a subgroup $\Gamma'\subset\Gamma$, a character $\chi\in\widehat{\Gamma}$ induces a character of the quotient group $\Gamma/\Gamma'$ if and only if $\chi\vert_{\Gamma'}=1$.
We will henceforth use the same symbol $\chi$ also for the induced character. The following fact is somewhat less trivial:
\begin{lemma} \label{lemma:liftingCharacters}
		Let $\Gamma$ be a finite Abelian group and $\Gamma'\subset\Gamma$ a subgroup.
		Any character on $\Gamma'$ can be extended to a character on $\Gamma$.
	\end{lemma}
	\begin{proof}
		See e.g.~\cite[Theorem 3.3]{carthy}.
	\end{proof}

One easily checks the following
\begin{lemma}
        Let $\G$ be a group with center $Z(\G)$ and $\chi:\,\G\rightarrow \U(1)$ a homomorphism. Choose $g\in Z(\G)$. Then the map
        \begin{equation}
            \begin{split}
                f:\,\,\,\quotient{\U(1)\times \G}{\langle \,(\chi(g),g)\,\rangle} &\,\longrightarrow\, \U(1)\times \G/\langle g\rangle \\[0.4cm]
                [\,(\phi,h)\,] &\,\longmapsto\, (\,\phi\chi(h)^{-1},[h]\,)
            \end{split}
        \end{equation}
	is well defined and an isomorphism, with inverse
    \begin{align}
		f^{-1}:\,(\phi,[h])\mapsto [(\phi\chi(h),h)]\,.
	\end{align}
	\label{lem:quotientgroup}
 \end{lemma}
Note that by the fundamental theorem of finite Abelian groups, any finite Abelian group can be expressed as the direct product of cyclic subgroups of prime-power order (also called primary cyclic groups).
The cyclic group $\IZ_r$ of order $r$ can be realized as 
\begin{equation}
    \left(\{ \e\left[\tfrac{k}{r} \right]  \,|\, k = 0, \ldots r-1\}, \, \cdot \right)\,,
\end{equation}
with group composition identified with complex multiplication.
We will also use the additive realization as $\IZ_r=\IZ/r\IZ$.
We will mostly use the multiplicative realization when referring to elements of a gauge group ${\rm G}$ and use the additive notation for the Pontryagin dual group $\widehat{\rm G}$.

For example, given a gauge group ${\rm G}=\IZ_r$ we denote by $\chi_k$ the character that corresponds to $k\in \IZ_r$ and acts on the generator $\e\left[\frac{1}{r} \right]$ of ${\rm G}$ as $\chi_k(\e\left[\frac{1}{r} \right])=\e\left[\frac{k}{r} \right]$.
We then call $k$ the $\IZ_r$ charge of the character and denote it as $q_{\IZ_r}$. 
Note that the additive convention for the character group implies that given $g\in{\rm G}$ we have $(a\chi_k)(g)=\chi_k(g)^a$.

When ${\rm G}= \U(1)^p$ and $g$ is an element of finite order, the quotient in Lemma \ref{lem:quotientgroup} is in fact isomorphic to $\U(1) \times \U(1)^p$. 
\begin{lemma} \label{lem:quotientgroupU1p}
    Let 
    \begin{align*}
    g= \left(\e\left[\frac{k_1}{m}\right]\,,\,\ldots\,,\e\left[\frac{k_p}{m}\right] \right)\in \U(1)^p \,, \quad m,k_1,\ldots,k_p \in \INnz \,.
\end{align*}
    Then the two groups
    \begin{equation*}
        \U(1) \times \U(1)^p \quad \text{and} \quad \quotient{\U(1) \times \U(1)^p}{ \left\langle \left( e\left[-\tfrac{1}{m}\right], g \right) \right\rangle} 
    \end{equation*}
    are isomorphic, and the isomorphism is realized by the map
    \begin{align}
        \begin{split}
        s: \quad \quotient{\U(1) \times \U(1)^p}{ \left\langle \left( e\left[-\tfrac{1}{m}\right], g \right) \right\rangle} \quad &\longrightarrow \quad \U(1) \times \U(1)^p \\
        [(\xi_0,\xi_1,\ldots,\xi_p)] \quad &\longmapsto \quad (\xi_0^m,\xi_0^{k_1} \xi_1,\ldots,\xi_0^{k_p} \xi_p) 
        \end{split}
    \end{align}
    and its inverse
    \begin{align}
        \begin{split}
        s^{-1}: \quad  \U(1) \times \U(1)^p &\longrightarrow \quad \quotient{\U(1) \times \U(1)^p}{ \left\langle \left( e\left[-\tfrac{1}{m}\right], g \right) \right\rangle} \quad\\
        (\eta_0,\eta_1,\ldots,\eta_p) \quad &\longmapsto \quad [(\xi_0,\xi_0^{-k_1} \eta_1,\ldots,\xi_0^{-k_p} \eta_p)] \quad \text{where} \quad \xi_0^m =\eta_0 \,. 
        \end{split}
    \end{align}
\end{lemma}

\subsection{Holonomies of continuous and discrete gauge groups} \label{sec:holonomy}
We will begin by reviewing why twisting a circle compactification with an element of a discrete gauge group $\Gamma$ is the analogue of ``turning on a Wilson line'' around the compactification circle. 

Consider first the continuous gauge group $\U(1)$. Unlike the discrete case, a continuous gauge symmetry requires a gauge field $A$ for its implementation. When the fundamental group of spacetime is non-trivial, flat configurations of the gauge field exist which carry non-trivial holonomy around non-trivial loops. Consider such a holonomy around the compactification circle $S$,
\begin{equation}
    1 \neq \gre = \exp i \oint_{S} A \in \U(1) \,. 
\end{equation}
The background value of a field $\phi$ that carries charge $q$ under $\U(1)$ will satisfy 
\begin{equation}
    D_\mu \phi = (\partial_\mu - i q A_\mu ) \phi = 0 \,,
\end{equation}
ensuring that it is the solution with lowest energy to the equations of motion. Hence, by
\begin{equation}
    \partial_\mu \ln \phi = i q A_\mu  \,\Rightarrow \,
    \frac{\phi(1)}{\phi(0)} = \gre^q \,,
\end{equation}
it will carry twisted boundary conditions around the compactification circle $S$. We have here indicated only the circle coordinate $y \in [0,1)$ as the argument of $\phi$. 

A discrete gauge group $\Gamma$ has no gauge field associated to it. The analogue of gauge holonomy however exists; it is visible at the level of matter fields carrying charge $\chi_\phi$ under $\Gamma$ in the form of non-trivial boundary conditions
\begin{equation}
    \phi(1) = \chi_\phi(\gre) \phi(0)  \,, \quad \gre \in \Gamma\,,
\end{equation}
around the compactification direction.

The Kaluza-Klein ansatz in the presence of a gauge holonomy $\gre \in \G$ is modified, irrespective of whether the gauge group $\G$ is continuous or discrete. For a scalar field $\phi$ transforming in the representation $\chi_\phi$ of $\G$, it becomes
\begin{equation} \label{eq:KKTwisted}
    \phi(y) = \chi_\phi(\gre)^y \sum_{n=-\infty}^\infty \phi_n \e[ n y].
\end{equation}
Let the gauge holonomy $\gre$ be an element of the gauge group of finite order $m$. Keeping in mind the goal of this discussion, to study the commutation of the diagram in Figure \ref{fig:twoPaths}, this is clearly the case of interest also in the case of continuous gauge group. As $\chi_\phi$ is a group homomorphism, it necessarily maps 
\begin{equation} \label{eq:chiPhi}
    \chi_\phi : \, \gre \mapsto \e\left[\frac{q_{\gre}(\chi_\phi)}{m}\right] \,,
\end{equation}
for a unique integer $q_{\gre}(\chi_\phi)\in\{0,\ldots,m-1\}$ associated to the character $\chi_\phi$. The Kaluza-Klein charge of the $n^{\rm th}$ mode $\phi_n$ of the field $\phi$, which governs the phase change of $\phi_n$ under translations along the Kaluza-Klein circle, is thus shifted by the twisting from $n$ to $n + q_{\gre}(\chi_\phi)/m$. Assuming that a scalar field of charge $\chi$ exists in the theory for which $\gcd(q_\gre(\chi),m)=1$, we define the Kaluza-Klein charge $q_{\KK}$ such that the $n^{\rm th}$ mode $\phi_n$ of the field $\phi$ has charge
\begin{equation} \label{eq:defQkk}
    q_\KK = mn + q_{\gre}(\chi_\phi) \,.
\end{equation}
We will henceforth refer to $q_\KK$ as the Kaluza-Klein charge, and to $n$ as the mode or Kaluza-Klein number.

The normalization of the Kaluza-Klein charge is set by the choice of parametrization of the Kaluza-Klein circle; rescaling the parameter such that it takes values in $[0,1/m)$ rather than $[0,1)$ implements the normalization underlying this definition of $q_{\KK}$. This corresponds to a choice of normalization $A \mapsto \tfrac{1}{m} A$ of the Kaluza-Klein gauge field $A$ compared to the conventional normalization. To reiterate, in \eqref{eq:defQkk}, $q_\KK$ is the Kaluza-Klein charge of the mode of the $d$-dimensional scalar field $\phi$ with Kaluza-Klein number $n$, $m$ is the order of the element $\gre$ by which the compactification is twisted, and $q_{\gre}(\chi_\phi)$ depends both on $\gre$ and on the charge $\chi_\phi$ of the scalar field $\phi$.

Upon dimensional reduction, we thus have an action of $\U(1) \times \G$ on the theory, with an element $(e[\theta],h) \in \U(1) \times \G$ acting on the scalar field $\phi_n$ as
\begin{equation} \label{eq:trivialAction}
    \phi_n \,\longmapsto \,\chi_\phi(h)\,\e[(q_{\gre}(\chi_\phi)+m n)\,\theta] \,\phi_n \,.
\end{equation}
To determine the gauge group of the reduced theory, we must mod out by group elements $(e[\theta],h)$ for which this action is trivial on all fields in the theory.

Imposing triviality for all of the modes $\phi_n$ of a given higher dimensional scalar field $\phi$ imposes ${\theta = s/m}$ for $s \in \IZ \mod m$. Since $\theta$ is constant, the gauge field is also invariant under such transformations. The remaining constraint is 
\begin{equation} \label{eq:subConstraint}
    \chi_\phi(h) = \left(\e\left[-\frac{q_{\gre}(\chi_\phi)}{m}\right] \right)^s = \chi_\phi(\gre)^{-s} \,.
\end{equation}
This implies that the action \eqref{eq:trivialAction} can be rendered trivial for all $h \in \langle \gre \rangle$, with the kernel of the action generated by the element $(\e[-\tfrac{1}{m}],\gre)$. Now consider an $h\in \G$ that does not lie in the orbit of $\gre$. Assuming that all possible representations of $\G$ arise in the theory, there exists a field $\phi$ for which $\chi_\phi(\gre) \neq 1$, but $\chi_\phi(h) = 1$. For such a field, the constraint \eqref{eq:subConstraint} can only be satisfied for $s=0$. But as by the same assumption, fields exist
for which $\chi_\phi(h)$ is not trivial and $s=0$ hence does not solve \eqref{eq:subConstraint}, we see that the constraint \eqref{eq:trivialAction} cannot be solved for $h \notin \langle \gre \rangle$. We can thus conclude that the $(d-1)$-dimensional gauge group is given by
\begin{align}
    \quotient{\U(1)\times \G}{\left\langle \left( \e\left[-\tfrac{1}{m}\right], \gre \right) \right\rangle}\,.
    \label{eqn:Gdm1prime}
\end{align}

We can understand this result also from a more conceptual perspective. The definition of $\gamma$-twisting a compactification on a circle $S=S^1$ is that going around $S$ becomes equivalent to acting with $\gamma$. Upon twisting, we therefore need to keep track of the number of times we go around the circle. We can accomplish this by replacing $S$ by a cover $S'$. Going around $S$ $m$-times, with $m$ the order of $\gamma$, is the trivial operation; we therefore choose $S'$ as the $m$-fold cover to avoid redundancy. As a result, moving ${\tiny\frac{1}{m}}$-th along $S'$ is equivalent to acting with $\tse$. We have thus arrived at the statement of~\eqref{eqn:Gdm1prime}, without the need of invoking the action of the gauge group on the fields of the theory.

\subsection{The gauge group of a compactification twisted by a discrete gauge transformation} \label{sec:twistedKK}

Let us first continue the discussion of the previous section under the assumption that the group element $\gre$ with regard to which the compactification is twisted is an element of a discrete gauge group $\Gamma$. In terms of Figure \ref{fig:twoPaths}, we are considering the vertical downward arrow on the right.

By Lemma~\ref{lem:quotientgroup}, we can identify the gauge group of the compactified theory,
\begin{align} \label{eq:QuotientGroupLowerRight}
    \quotient{\U(1)\times \Gamma}{\left\langle \left( e\left[-\tfrac{1}{m}\right], \gre \right) \right\rangle}\,,
\end{align}
with a product group:
introducing the character
\begin{align} \label{eq:deftildeChig}
    \widetilde{\chi}_\gre:\,\langle \gre\rangle \rightarrow \U(1)\,,\quad \gre^a\mapsto \e[-\tfrac{a}{m}]
\end{align}
allows us to write $\langle (\e[-\tfrac{1}{m}],\gre) \rangle = \langle (\widetilde{\chi}_\gre(\gre),\gre) \rangle$. 
By Lemma~\ref{lemma:liftingCharacters}, we can find a character $\chi_\gre \in \widehat{\Gamma}$ that lifts $\widetilde{\chi}_\gre$.
Then Lemma~\ref{lem:quotientgroup} applies and we have a non-canonical isomorphism 
\begin{align} \label{eq:isoQuotientProduct}
        u:\quad \quotient{\U(1)\times \Gamma}{\left\langle \left( e\left[-\tfrac{1}{m}\right], \gre \right) \right\rangle} \quad&\longrightarrow \quad \U(1) \times \quotient{\Gamma\!\!}{\!\!\langle \gre \rangle} \\[10pt] 
    [(e[\sigma],h)] \,&\longmapsto \,(e[\sigma] \chi_\gre(h)^{-1} , [h])\,, \nn
\end{align}
with inverse 
\begin{equation} \label{eq:finv}
    u^{-1}:\,(e[\sigma],[h]) \longmapsto [(e[\sigma]\chi_\gre(h), h)] \,,
\end{equation}
which depends on the choice of lift.
We denote the charge lattices that correspond respectively to the left- and right-hand side of~\eqref{eq:isoQuotientProduct} by
\begin{align} \label{eq:5dchargeLattice}
    \begin{split}
        \Lambda=&\text{ker}\left( \e\left[-\tfrac{1}{m}\right], \gre \right)\subset \mathbb{Z}\times \widehat{\Gamma}\,,\quad \Lambda'=\mathbb{Z}\times \widehat{\Gamma/\langle \gre\rangle}\,.
    \end{split}
\end{align}
The pullback isomorphism that is associated to the map~\eqref{eq:isoQuotientProduct} then acts as
\begin{align}
    \begin{split}
        u^* : \Lambda'\rightarrow \Lambda\,,\quad (q,\chi)  \mapsto (q,\chi-q\chi_\gre)\,.
        \end{split}
        \label{eqn:upullback}
\end{align}
Note that the image indeed defines a character of the quotient group \eqref{eq:QuotientGroupLowerRight}, as
\begin{equation}
    (q,\chi-q\chi_\gre)(\e[-\tfrac{1}{m}] ,\gre ) = \e[-\tfrac{q}{m}]\,\chi_\gre(\gre)^{-q}\,\chi(\gre) =1 \,.
\end{equation}
To verify the last equality, recall that $\chi_\gre$ is a lift of $\widetilde{\chi}_\gre$ defined in \eqref{eq:deftildeChig}, and that $\chi \in \widehat{\Gamma/\langle \gre \rangle}$ maps $\gre$ to 1.

We can render this discussion more explicit in the case that $\Gamma$ is a finite Abelian group.
By the fundamental theorem of finite Abelian groups, we can then write
\begin{align}
    \Gamma=\mathbb{Z}_{q_1}\times\ldots\times \mathbb{Z}_{q_p}\,,
    \label{eqn:Zproduct}
\end{align}
such that
\begin{align} \label{eq:twistingG}
    \gre= \left(\e\left[-\frac{n_1}{q_1}\right]\,,\,\ldots\,,\e\left[-\frac{n_p}{q_p}\right] \right)\,,
\end{align}
for some $q_1,\ldots,q_p\in\mathbb{N}$ and $-n_1,\ldots,-n_p\in\mathbb{N}$.
We choose this slightly odd sign convention for the $n_i$ in order to be able to directly identify them with the Kaluza-Klein numbers of the Higgs modes with charge $q_i$ in Section~\ref{sec:higgsing}. Canceling common factors to simplify each fraction to its lowest terms, then expressing each factor in terms of the lowest common denominator of all $p$ factors yields 
\begin{align} \label{eq:twistingGm}
    \gre= \left(\e\left[\frac{k_1}{m}\right]\,,\,\ldots\,,\e\left[\frac{k_p}{m}\right] \right)\,,
\end{align}
where in terms of $r_i=\text{gcd}(n_i,q_i)$, $\mathfrak{n}_i=n_i/r_i$ and $\mathfrak{q}_i=q_i/r_i$ for $i=1,\ldots, p$, 
\begin{align}
	m=\text{lcm}(\mathfrak{q}_1,\ldots,\mathfrak{q}_p)\,,\quad k_i\equiv -m\mathfrak{n}_i/\mathfrak{q}_i\text{ mod }m\,,
\end{align}
and we define $\bk = (k_1, \ldots, k_p)$ for future convenience.
Note that $m$ thus defined computes the order of the group element $\gre$, matching our notation from above.

To obtain an explicit form of the isomorphism \eqref{eq:isoQuotientProduct}, we need to lift the character $\widetilde{\chi}_\gre$ from $\langle \gre \rangle$ to $\Gamma$.
We have $\widehat{\Gamma}=\mathbb{Z}_{q_1}\times\ldots\times \mathbb{Z}_{q_p}$, where the character corresponding to $\bl = (l_1, \ldots, l_p)\in \mathbb{Z}_{q_1}\times\ldots\times\mathbb{Z}_{q_p}$ (in additive notation) acts as\footnote{\label{footnote:chilVSchiphi}We use $\chi$ indexed by $\bl \in \widehat{\Gamma}$ to indicate the character associated to $\bl$, and indexed by $\phi$ to indicate the charge of $\phi$ under $\Gamma$. We trust that this will not give rise to  confusion.}
\begin{align}
    \chi_{\bl}:\,\,\Gamma\rightarrow \U(1)\,,\quad (\xi_1,\ldots,\xi_p)\mapsto \xi_1^{l_1}\ldots\xi_p^{l_p}\,.
\end{align}
The condition $\chi_{\bl}(\gre)=\e[-\tfrac{1}{m}]$ is then equivalent to
\begin{align}
 	1 + \sum\limits_{i=1}^p l_i k_i=1 -\sum\limits_{i=1}^pm l_i\frac{\mathfrak{n}_i}{\mathfrak{q}_i}\equiv 0\text{ mod }m\,.
    \label{eqn:liftingCondition}
 \end{align}
Given such a choice for $\bl$, which we call $\bl_\gre$ (such that $\chi_\gre$ as introduced above can be written as $\chi_{\bl_\gre}$), the explicit form of the map \eqref{eq:isoQuotientProduct} is
\begin{align}
    u:\quad \left[(\xi_0,\ldots,\xi_p) \right] \,\longmapsto\, (\xi_0\xi_1^{-(l_\gre)_1}\cdot\ldots\cdot \xi_p^{-(l_\gre)_p},\left[\xi_1,\ldots,\xi_p\right])\,.
\end{align}
The corresponding action on the charge lattice is given by
\begin{align}
     u^*:\quad \bv \,\longmapsto\, U\bv\,,
     \label{eqn:Upullback}
\end{align}
in terms of the matrix
\begin{align}
	U=\left(\begin{array}{cccccc}
		1&0&0&\ldots&0&0\\
		-(l_g)_1&1&0&\ldots&0&0\\
		\vdots&&&\ddots&&\\
		-(l_g)_p&0&0&\ldots&0&1
	\end{array}\right)\,.
\end{align}

\subsection{Compactifying in the presence of Wilson lines and Higgsing}
\label{sec:higgsing}
We now shift our focus to the upper left box in Figure \ref{fig:twoPaths}: we consider a $d$-dimensional theory with gauge group $\U(1)^{p}$, and assume that it contains massless scalar fields that, after acquiring a non-zero vacuum expectation value, break the gauge group to a finite subgroup $\Gamma\subset \U(1)^p$ (the upper horizontal arrow in Figure \ref{fig:twoPaths}).
Assuming that the number of Higgs fields is $l$, we will consider the corresponding charges as the columns of a $p\times l$ matrix $Q$.

The physics of the Higgs transition only depends on the sublattice of the charge lattice of the unbroken group generated by the charges of the Higgs fields. We are free to consider instead the matrix $QW$ for any $W\in \text{SL}(l,\mathbb{Z})$. This amounts to labeling the Higgs sector by different linear combinations of the Higgs fields.
On the other hand, there is no canonical ordering of the $p$ $\U(1)$ factors in $\U(1)^p$. Changing this identification amounts to replacing the charge matrix by a matrix $VQ$ for some $V\in\text{SL}(p,\mathbb{Z})$.
Taken together, this allows us to assume that $Q$ is in Smith normal form\footnote{Recall that any $p \times l$ matrix $Q$ can be put in the form $VQW=\mathrm{diag}(q_1, \ldots, q_r, 0, \ldots, 0)$, with invertible integer matrices $V$ and $W$ implementing elementary row and column operations. See \cite{bookSmithNormalForm} for an introduction to such matters.}
Since we assume that the unbroken gauge group after the Higgs transition is finite, we can restrict to the case $p=l$, such that
\begin{align}
    Q=\left(\begin{array}{cccc}
        q_1& 0 &\ldots&0\\
         0 &q_2&\ldots&0\\
           &   &\ddots& \\
         0 & 0 &\ldots&q_p
    \end{array}
    \right)\,.
\end{align}
We will call $\phi_i$ the scalar field of charge $(0,\ldots,q_i,\ldots, 0)$. For future reference, we also introduce the vector $\bq = (q_1, \ldots, q_p)$.

Considering the $d$-dimensional spacetime to contain a circle factor permits us to consider a background for the gauge group of the $d$-dimensional theory $\widetilde{\cT}_d$ with holonomy around this circle. To induce the twisting \eqref{eq:twistingG} upon Higgsing (the upper horizontal arrow of Figure \ref{fig:twoPaths}), we must choose
\begin{align} \label{eq:defgU1}
    \gre= \left(\e\left[-\frac{n_1}{q_1}\right]\,,\,\ldots\,,\e\left[-\frac{n_p}{q_p}\right] \right) \in \U(1)^p\,.
\end{align}
Following the left vertical arrow in Figure \ref{fig:twoPaths}, we now wish to first
compactify on this circle (i.e. describe the theory from a $(d-1)$-dimensional vantage point), before considering the Higgsing.
From \eqref{eq:KKTwisted}, we can read off the mass of the $n^{\mathrm{th}}$ Kaluza-Klein mode of the field $\phi_i$; it is, up to normalization,
\begin{equation}
    m_n \propto \left\vert -\frac{n_i}{q_i} q_i + n \right\vert \,.
\end{equation}
As promised below equation \eqref{eq:twistingG}, the $n_i^{\mathrm{th}}$ mode of $\phi_i$, $n=n_i$, is massless, and can thus acquire a vacuum expectation value.

Following the discussion in Section \ref{sec:holonomy}, the gauge group of the $(d-1)$-dimensional theory $\widetilde{\cT}_{d-1}$ can be written as
\begin{align} \label{eq:TwiddleQuotientGaugeGroup}
    \quotient{\U(1)\times\U(1)^p}{\langle \,(\e[-\tfrac{1}{m}],\gre)\,\rangle}\,.
\end{align}
By Lemma \ref{lem:quotientgroupU1p}, this is isomorphic to the group $\U(1)\times\U(1)^p$, with the isomorphism realized by the maps $s$ and $s^{-1}$.
We denote the charge lattices that correspond respectively to~\eqref{eq:TwiddleQuotientGaugeGroup} and $\U(1)\times\U(1)^p$ by
\begin{align}
    \widetilde{\Lambda}=\{\,(v_0,\ldots,v_p)\in\mathbb{Z}^{p+1}\,\vert\,(1,k_1,\ldots,k_p)\cdot \bv \equiv 0\text{ mod }m\,\}\,,\quad \widetilde{\Lambda}'=\mathbb{Z}\times\mathbb{Z}^p\,.
\end{align}
The constraint on $\bv \in \IZ^{p+1}$ in the definition of $\widetilde{\Lambda}$ ensures that $\langle(\e[-\tfrac{1}{m}],\gre)\rangle$ is mapped to the identity by the corresponding character.
The isomorphism between the charge lattices induced by $s$ takes the form
\begin{align}
\begin{split}
        s^*:\, \widetilde{\Lambda}'&\,\longrightarrow\,  \widetilde{\Lambda}\\
    \bv \,&\longmapsto\, S\,\bv\,,
    \end{split}
        \label{eqn:Spullback}
\end{align}
where the matrix $S$ is defined as
\begin{align}
    S=\left(\begin{array}{cccccc}
    m&k_1&k_2&\ldots&k_{p-1}&k_p\\
    0&1&0&\ldots&0&0\\
     & & &\ddots& & \\
    0&0&0&\ldots&0&1
    \end{array}\right)\,.
\end{align}
Concretely, $s^*:(n,\bl) \mapsto (mn + \bl \cdot \bk, \bl)$.

Calculating $\chi_{\bl}(\gre)=\e[\tfrac{\bl \cdot \bk}{m}]$ and comparing to \eqref{eq:chiPhi} (recall that $\chi_\phi$ in that equation refers to the charge carried by a field $\phi$ -- here, we have explicitly labeled the possible charges via $\bl$, see footnote \ref{footnote:chilVSchiphi}) allows us to equate $q_{\gre}(\chi_{\bl}) = \bl \cdot \bk$, and thus to identify the charge with regard to the first factor of the gauge group in the form $\U(1) \times \U(1)^p$ with the Kaluza-Klein number, and to identify the first component of the image charge with $q_{\KK}$ as defined in \eqref{eq:defQkk}. The latter determines the mass of the Kaluza-Klein modes in $\widetilde{\cT}_{d-1}$. We defined the integers $n_i$ in equation \eqref{eq:twistingG} such that the corresponding modes of the $d$-dimensional fields are massless and can serve as Higgs fields. As a consequence, the $(d-1)$-dimensional Higgs fields have vanishing $q_\KK$ charge,
\begin{align} \label{eq:chargeMapS}
    S\,\left(\begin{array}{cccc}
        n_1&n_2&\ldots&n_p\\
        q_1& 0 &\ldots&0\\
         0 &q_2&\ldots&0\\
           &   &\ddots& \\
         0 & 0 &\ldots&q_p
    \end{array}\right)=\left(\begin{array}{cccc}
         0&0&\ldots&0\\
        q_1& 0 &\ldots&0\\
         0 &q_2&\ldots&0\\
           &   &\ddots& \\
         0 & 0 &\ldots&q_p
    \end{array}\right)\,.
\end{align}

To determine the gauge group upon Higgsing (the lower horizontal arrow of Figure \ref{fig:twoPaths}), we need to compute the kernel of the character under which the Higgs fields transform. The form of the charge matrix $\bQ$ on the RHS of equation \eqref{eq:chargeMapS} yields this kernel immediately, as a subgroup of \eqref{eq:TwiddleQuotientGaugeGroup}. It is
\begin{align}
    \ker(\chi_{\bQ}) \,\,= \,\,\quotient{\U(1)\times \Gamma}{\langle \,(\e[-\tfrac{1}{m}],\gre)\,\rangle} \,\lhook\joinrel\xlongrightarrow{\iota} \,\quotient{\U(1)\times\U(1)^p}{\langle \,(\e[-\tfrac{1}{m}],\gre)\,\rangle}\,,
\end{align}
with
\begin{align}
    \Gamma= \ker(\chi_{\bq}) =\mathbb{Z}_{q_1}\times\ldots\times \mathbb{Z}_{q_p}\,,
\end{align}
and where we retain the name $\gre$ for the element \eqref{eq:defgU1} also when considered as an element of $\Gamma$. As discussed in Section~\ref{sec:twistedKK}, the quotient constituting $\ker(\chi_{\bQ})$ can be identified with
\begin{align} \label{eq:productDm1}
    \U(1)\times \quotient{\Gamma\!\!}{\!\!\langle\, \gre\,\rangle}\,,
\end{align}
with corresponding charge lattice $\Lambda'=\mathbb{Z}\times\widehat{\Gamma/\langle \gre\rangle}$.

We now have all the ingredients in place to work out the image of a charge of the gauge group of $\widetilde{\cT}_{d-1}$ in the presentation $\U(1) \times \U(1)^p$ in the character group of the gauge group of $\cT_{d-1}$ in the presentation \eqref{eq:productDm1}.
Combining these elements, we see that the map between groups is given by
\begin{equation}
    s|_{\ker{\chi_{\bQ}}} \circ \iota \circ u^{-1} :\quad \U(1)\times \quotient{\Gamma\!\!}{\!\!\langle\, \gre\,\rangle}\,\,\longrightarrow \,\,\U(1) \times \U(1)^p  \,,
\end{equation}
which in turn induces the map
\begin{align} \label{eq:chargeMapu1upu1quotientGamma}
    \begin{split}
    (u^{-1})^* \circ \iota^* \circ s^* : \quad\widetilde{\Lambda}' \,\, &\longrightarrow \,\, \Lambda' \\
    \bv \,\, & \longmapsto \,\, U^{-1}S \,\bv \,.
    \end{split}
\end{align}
Note that to identify the charges properly, we need to take into account how $\Gamma/\langle \gre \rangle$ is embedded in $\Gamma$ as a subgroup, as the map \eqref{eq:chargeMapu1upu1quotientGamma} yields these charges as embedded in $\widehat{\Gamma}$.

\subsection{General groups}
As discussed at the end of Section~\ref{sec:holonomy}, one way to phrase the result~\eqref{eqn:Gdm1prime} is that in order to obtain the gauge group of the twisted compactification, one first has to go to the $m$-fold covering space of the circle over which the boundary conditions on the fields become trivial.
Then, a $1/m$-shift along that covering circle can be identified with the gauge transformation associated to the twist $\gre$.
This allows us to generalize the result for gauge groups that are not finite Abelian.

Consider as gauge group a topological group ${\rm G}$ and denote the connected component of the identity by ${\rm G}^0$.
If ${\rm G}$ is locally path connected, then the zeroth fundamental group is ${\rm G}/{\rm G}^0\simeq \pi_0({\rm G})$, such that we have a map $p:\,{\rm G}\rightarrow \pi_0({\rm G})$.

After compactifying on a circle with twist $\gre\in {\rm G}$ of finite order $m$, the gauge group is
\begin{align}
    {\rm G}_{d-1}=\quotient{\U(1)\times C_{\rm G}(\gre)}{\left\langle \left({\rm e}\left[-\tfrac{1}{m}\right],\gre\right)\right\rangle}\,,
\end{align}
where $C_{\rm G}(\gre)$ is the centralizer of $\gre$ in ${\rm G}$.
The group ${\rm G}$ is broken to the centralizer because elements $h\in {\rm G}$ that do not commute with $g$ change the boundary conditions and are therefore not part of the gauge group of the twisted compactification.
The zeroth fundamental group of ${\rm G}_{d-1}$ is
\begin{align}
    \pi_0({\rm G}_{d-1})=\pi_0\left(\quotient{C_{\rm G}(\gre)}{\langle \gre\rangle}\right)=\quotient{\pi_0(C_{\rm G}(\gre))}{\langle p(\gre)\rangle}\,.
\end{align}
The first equality follows because taking the quotient of ${\rm G}_{d-1}$ by $\U(1)$ does not change $\pi_0$.

Note that here we have only considered twists by inner automorphisms of ${\rm G}$.
Twists by outer automorphisms have been considered for example in~\cite{Anderson:2023wkr,Anderson:2023tfy,Ahmed:2024wve} but are beyond the scope of our discussion.

\paragraph{Example}
Consider
\begin{align}
    {\rm G}=\frac{{\rm SU}(2)\times\mathbb{Z}_4}{\mathbb{Z}_2}\,,
\end{align}
where the $\mathbb{Z}_2$ is generated by $(-1,-1)$.
Then $\pi_0({\rm G})=\mathbb{Z}_2$ with the two components respectively represented by $(1,1)$ and $(1,i)$.
Considering the twist $\gre=(1,i)$ of order $4$, the centralizer is $C_{\rm G}(\gre)={\rm G}$ and
\begin{align}
    {\rm G}_{d-1}=\quotient{\U(1)\times {\rm G}}{\langle\left({\rm e}\left[-\tfrac{1}{m}\right],\gre\right)\rangle}=\quotient{\U(1)\times{\rm SU}(2)}{\langle(-1,-1)\rangle}\,,
\end{align}
with $\pi_0({\rm G}_{d-1})=0$.

\subsection{Example $\Gamma=\IZ_4$}
\label{sec:fieldTheoryExampleZ4}
To illustrate the discussion from Section~\ref{sec:twistedKK} at the hand of a concrete example, we will now consider twisted compactifications of gravity coupled to a $\Gamma=\IZ_4$ gauge theory.
In line with the conventions announced above, we will represent the gauge group $\IZ_4$ multiplicatively in terms of roots of unity $\e\left[\tfrac{k}{4}\right]$, $k=0,\ldots,3$, while we use additive notation for the dual character group $\widehat{\Gamma}=\IZ_4\simeq \IZ/4\IZ$.
We denote by $N_k$ the number of charged fields in the $\dd$-dimensional theory $\cT_d$ of charge $k\,\mathrm{mod}\, 4$ under $\IZ_4$.
The twist is again denoted by $\gre\in\Gamma$.

\begin{table}[ht!]
        \centering
        \subfloat[$\gre=1\,,\quad \Lambda=\mathbb{Z}\times \mathbb{Z}_4$\phantom{$\begin{array}{c}x\\[-.1em]x\end{array}$}\label{tab:exZ4Ns0}]{
        \begin{tabular}{|ccccccccc|}\hline
        \multicolumn{3}{|c|}{\diagbox{$q_{\IZ_4}$}{$q_{\KK}$}} & 0 & 1 & 2 & 3 & 4 & $\cdots$ \\ \hline
        \multicolumn{3}{|c|}{0} & $N_0$ & $N_0$ & $N_0$ & $N_0$ & $N_0$ & $\cdots$ \\
        \multicolumn{3}{|c|}{1} & $N_1$ & $N_1$ & $N_1$ & $N_1$ & $N_1$ & $\cdots$ \\
        \multicolumn{3}{|c|}{2} & $N_2$ & $N_2$ & $N_2$ & $N_2$ & $N_2$ & $\cdots$ \\
        \multicolumn{3}{|c|}{3} & $N_3$ & $N_3$ & $N_3$ & $N_3$ & $N_3$ & $\cdots$ \\\hline
        \end{tabular}}\hspace{.2cm}
        \subfloat[$\gre={\e\left[\tfrac{1}{4}\right]}\,,\quad \Lambda=\text{ker}\left({\e\left[-\tfrac{1}{4}\right]},{\e\left[\tfrac{1}{4}\right]}\right)\subset \mathbb{Z}\times \IZ_4$\phantom{$\begin{array}{c}x\\[-.1em]x\end{array}$}\label{tab:exZ4Ns1}]{
        \begin{tabular}{|ccccccccc|}\hline
        \multicolumn{3}{|c|}{\diagbox{$q_{\IZ_4}$}{$q_{\KK}$}} & 0 & 1 & 2 & 3 & 4 & $\cdots$ \\ \hline
        \multicolumn{3}{|c|}{0} & $N_0$ &  0  &  0  &  0  & $N_0$ & $\cdots$ \\
        \multicolumn{3}{|c|}{1} &  0  & $N_1$ &  0  &  0  &  0  & $\cdots$ \\
        \multicolumn{3}{|c|}{2} &  0  &  0  & $N_2$ &  0  &  0  & $\cdots$ \\
        \multicolumn{3}{|c|}{3} &  0  &  0  &  0  & $N_3$ &  0  & $\cdots$ \\\hline
        \end{tabular}}\\[.5cm]
        \subfloat[$\gre={\e\left[\tfrac{2}{4}\right]}\,,\quad \Lambda=\text{ker}\left({\e\left[\tfrac{2}{4}\right]},{\e\left[\tfrac{2}{4}\right]}\right)\subset \mathbb{Z}\times\IZ_4$\phantom{$\begin{array}{c}x\\[-.1em]x\end{array}$}\label{tab:exZ4Ns2}]{
        \begin{tabular}{|ccccccccc|}\hline
        \multicolumn{3}{|c|}{\diagbox{$q_{\IZ_4}$}{$q_{\KK}$}} & 0 & 1 & 2 & 3 & 4 & $\cdots$ \\ \hline
        \multicolumn{3}{|c|}{0} & $N_0$ &  0  & $N_0$  &  0  & $N_0$ & $\cdots$ \\
        \multicolumn{3}{|c|}{1} &  0  & $N_1$ &  0  &  $N_1$  &  0  & $\cdots$ \\
        \multicolumn{3}{|c|}{2} & $N_2$ &  0  & $N_2$ &  0  &  $N_2$  & $\cdots$ \\
        \multicolumn{3}{|c|}{3} &  0  &  $N_3$  &  0  & $N_3$ &  0  & $\cdots$ \\\hline
        \end{tabular}}\hspace{.2cm}
        \subfloat[$\gre={\e\left[\tfrac{3}{4}\right]}\,,\quad \Lambda=\text{ker}\left({\e\left[-\tfrac{1}{4}\right]},{\e\left[-\tfrac{1}{4}\right]}\right)\subset \mathbb{Z}\times \IZ_4$\phantom{$\begin{array}{c}x\\[-.1em]x\end{array}$}\label{tab:exZ4Ns3}]{
        \begin{tabular}{|ccccccccc|}\hline
        \multicolumn{3}{|c|}{\diagbox{$q_{\IZ_4}$}{$q_{\KK}$}} & 0 & 1 & 2 & 3 & 4 & $\cdots$ \\ \hline
        \multicolumn{3}{|c|}{0} & $N_0$ &  0  &  0  &  0  & $N_0$ & $\cdots$ \\
        \multicolumn{3}{|c|}{1} &  0  & 0 &  0  &  $N_1$  &  0  & $\cdots$ \\
        \multicolumn{3}{|c|}{2} &  0  &  0  & $N_2$ &  0  &  0  & $\cdots$ \\
        \multicolumn{3}{|c|}{3} &  0  &  $N_3$  &  0  & 0 &  0  & $\cdots$ \\\hline
        \end{tabular}}
        \caption{The multiplicities of states with a given charge $(q_{\IZ_4},q_{\KK})\in\Lambda$ for the different $\gre$-twisted compactifications of a $d$-dimensional theory with gauge group $\Gamma=\IZ_4$.}
        \label{tab:exZ4Ns}
\end{table}

\begin{itemize}
    \item \textbf{No twist, $\gre=1$:} This is the trivial case. The order of $\gre$ is $m=1$ and the gauge group of the compactified theory is $\U(1)\times \IZ_4$ with charge lattice $\Lambda'=\Lambda=\mathbb{Z}\times\IZ_4$.
    The multiplicities of states with a given charge in $\Lambda$ are listed in Table~\ref{tab:exZ4Ns0}.
    
    \item \textbf{Twist $\gre=\e\left[\pm\tfrac{1}{4}\right]$:} Let us now twist by an element $\gre \in\IZ_4$ that generates the group.
    The order of $\gre$ is $m=4$ and the gauge group of the compactified theory takes the form
    \begin{align}
        \quotient{\U(1)\times\IZ_4}{(\e\left[-\tfrac{1}{4}\right],\e\left[\pm\tfrac{1}{4}\right])}\simeq \U(1)\,,
    \end{align}
    while the corresponding isomorphic expressions for the charge lattice are
    \begin{align}
        \Lambda=\text{ker}(\e\left[-\tfrac{1}{4}\right],\e\left[\pm\tfrac{1}{4}\right])\subset\IZ\times\IZ_4\,,\quad \Lambda'=\IZ\,.
    \end{align}
    The multiplicities of states with a given charge in $\Lambda$ are listed in Tables~\ref{tab:exZ4Ns1} and~\ref{tab:exZ4Ns3}.
    In terms of the notation introduced below~\eqref{eqn:Zproduct}, we have
    \begin{align}
        n=\mathfrak{n}=\mp1\,,\quad q=\mathfrak{q}=4\,,\quad r=1\,,\quad k=\pm 1\,,
    \end{align}
    where we omit the subscript since the number of factors in~\eqref{eqn:Zproduct} is $p=1$.

    In order to find the possible lifts $\chi_{l}$ of the character $\tilde{\chi}_g$ in~\eqref{eq:deftildeChig}, we solve~\eqref{eqn:liftingCondition}, which here takes the form
    \begin{align}
        1\pm l\equiv 0\text{ mod }4\,.
    \end{align}
    We therefore have a unique lift $\chi_l$ with $l=\mp 1$ (note that the sign here is still correlated with the one in the choice of twist $\gre=\e\left[\pm\tfrac{1}{4}\right]$).
    The corresponding pullback map~\eqref{eqn:upullback} acts as
    \begin{align}
        u^*:\,\Lambda'\rightarrow \Lambda\,,\quad q_{\KK}\mapsto (q_{\KK},\,\pm q_{\KK})\,.
    \end{align}
    Using Table~\ref{tab:exZ4Ns1} and~\ref{tab:exZ4Ns3}, we obtain the multiplicities of charged states in the basis $\Lambda'$:
    \begin{align}
        \text{        \begin{tabular}{|c|cccccc|}\hline
        {$q_{\KK}$} & 0 & 1 & 2 & 3 & 4 & $\cdots$ \\ \hline
          & $N_0$ &  $N_{\pm 1}$ & $N_2$ & $N_{\mp 1}$ & $N_0$ & $\cdots$ \\\hline
        \end{tabular}}
    \end{align}
    
    \item \textbf{Twist $\gre=-1$:} This is the last case to consider. The order of $\gre$ is now $m=2$.
     The gauge group of the compactified theory is
     \begin{align}
        \quotient{\U(1)\times\IZ_4}{(-1,-1)}\simeq \U(1)\times\IZ_2\,,
     \end{align}
     and the charge lattice takes the corresponding isomorphic forms
    \begin{align}
        \Lambda=\text{ker}(-1,-1)\subset\IZ\times\IZ_4\,,\quad \Lambda'=\IZ\times \IZ_2\,.
    \end{align}
    The multiplicities of states with a given charge in $\Lambda$ are listed in Table~\ref{tab:exZ4Ns2}.

    Using again the notation introduced below~\eqref{eqn:Zproduct}, we have
    \begin{align}
        n=-2\,,\quad q=4\,,\quad r=2\,,\quad \mathfrak{n}=-1\,,\quad \mathfrak{q}=2\,,\quad k=1\,.
    \end{align}
    The different lifts of the character~\eqref{eq:deftildeChig} are given by $\chi_l$ with $l$ solving the equation
    \begin{align}
        1+ l\equiv 0\text{ mod }2\,.
    \end{align}
    We therefore find two inequivalent lifts correponding to $l\equiv \pm 1\text{ mod }4$.
    
    The corresponding pullback maps~\eqref{eqn:upullback} act as
    \begin{align}
        u^*:\,\Lambda'\rightarrow \Lambda\,,\quad (q_{\KK},q_{\IZ_2})\mapsto (q_{\KK},\,2q_{\IZ_2}\mp q_{\KK})\,.
    \end{align}
    Using Table~\ref{tab:exZ4Ns2}, we obtain the multiplicities of charged states in the basis $\Lambda'$:
    \begin{align}
        \text{\begin{tabular}{|c|cccccc|}\hline
        \diagbox{$q_{\IZ_2}$}{$q_{\KK}$} & 0 & 1 & 2 & 3 & 4 & $\cdots$ \\ \hline
            0 & $N_0$ &  $N_{\mp 1}$ & $N_2$ & $N_{\pm 1}$ & $N_0$ & $\cdots$ \\
            1 & $N_2$ &  $N_{\pm 1}$ & $N_0$ & $N_{\mp 1}$ & $N_2$ & $\cdots$ \\\hline
        \end{tabular}}
    \end{align}
\end{itemize}

\section{Almost generic genus one fibered Calabi-Yau threefolds}
\label{sec:almostgenericfibrations}
Before applying the discussion of the previous section in the context of F-theory and M-theory, we will introduce a class of geometries that we refer to as {\it almost generic genus one fibered/elliptic Calabi-Yau threefolds}.
We discuss some of their properties and then summarize our physical results, which will be derived in Sections~\ref{sec:discreteHolonomiesAndGenusOne} and~\ref{sec:modularityTopString}, in terms of a set of precise mathematical conjectures.

\subsection{Elliptic curves and the Tate-Shafarevich group} \label{ss:TSgroup}
It is instructive to briefly recall the situation for genus one curves before moving on to fibrations.
This material is nicely reviewed in~\cite{aglecture}.

An \textit{elliptic curve} is a projective curve $E$ of genus one over a field $k$ together with a $k$-rational point $p\in E$.
There exists a unique isomorphism $E\simeq \mathbb{C}/(\mathbb{Z}+\mathbb{Z}\tau)$, under which $p$ maps to $0\in\mathbb{C}$ and $\tau$ lies in some fixed fundamental domain inside $\overline{\mathbb{H}}\cup\{\infty\}$.
The addition of points in $\mathbb{C}$ equips $E$ with the structure of an Abelian group.

On the other hand, given any smooth projective curve $C$ of genus one over a field $k$, the group of line bundles of degree zero can be geometrically represented by the so-called Jacobian $J(C)=\text{Pic}(C)^0$, which is again a curve of genus one.
While $C$ does not necessarily have any $k$-rational points, the Jacobian $J(C)$ always has at least one $k$-rational point, corresponding to the trivial line bundle $\mathcal{O}_C\in\text{Pic}(C)^0$. It thus has the structure of an elliptic curve, with marked point $\mathcal{O}_C$. The curve $C$ itself can be identified with the space of line bundles of degree one. The tensor product of line bundles induces an action
\begin{align}
    \sigma:\,J(C)\times C\rightarrow C\,,
\end{align}
that is free and transitive.
This turns $C$ into a principal homogeneous space $(C,\sigma)$ for $J(C)$, i.e. a $J(C)$-torsor.
Two $J$-torsors can be considered equivalent if there is a $k$-isomorphism between them that is compatible with the $J$-action.
In particular, if $C$ itself has a $k$-rational point, then it is equivalent to $J=J(C)$ as a $J$-torsor.

The set of inequivalent $J$-torsors $(C,\sigma)$ can itself be equipped with an Abelian group law.\footnote{We follow here the description from~\cite[Appendix A]{Cvetic:2015moa}.}
The resulting group is called the \textit{Weil-Ch{\^a}telet group} $\text{WC}_k(J)$.
Given two $J$-torsors $(C,\sigma)$ and $(C',\sigma')$, the product is the curve
\begin{align}
    C''=\quotient{C\times_k C'}{\{\,\,(x,x')\sim (y,y')\,\,\Leftrightarrow\,\, \exists g\in J\,\text{ s.t. }(x,x')=\left(\sigma(g,y),\sigma'(-g,y')\right)\,\,\}}\,,
\end{align}
together with the action
\begin{align}
    \sigma'':\,J\times C''\rightarrow C''\,,\quad (g,(x,x'))\mapsto (\sigma(g,x),x')\,.
\end{align}
It is easy to see that the group law is commutative.

Let us denote by $(J,\sigma_J)$ the Jacobian $J$ together with the natural action
\begin{align}
    \sigma_J:J\times J\rightarrow J\,,\quad (g,h)\mapsto g+h\,.
\end{align}
We then have $(C',\sigma')=(J,\sigma_J)\cdot (C,\sigma)$, with the curve $C'$ being
\begin{align}
    C'=\quotient{J\times_k C}{\{\,\,(h,x)\sim (h',x')\,\,\Leftrightarrow\,\, \exists g\in J\,\text{ s.t. }(h,x)=(g+h',\sigma(-g,x')\,\,\,\}}\,.
\end{align}
Since $(h,x)\sim\left(0,\sigma(h,x)\right)$, every point of $C'$ can be uniquely represented by $(0,x)$ for some $x\in C$, such that $C'\simeq C$.
One can also check that $\sigma'$ acts as
\begin{align}
    \sigma':\,J\times C'\rightarrow C'\,,\quad \left(g,(0,x)\right)\mapsto (g,x)\sim \left(0,\sigma(g,x)\right)\,,
\end{align}
and therefore $(C',\sigma')\simeq (C,\sigma)$.
Thus $(J,\sigma_J)$ represents the identity element of the group.

On the other hand, if we let $(C',\sigma')=(C,\sigma^{-1})\cdot (C,\sigma)$, we can choose any $x_0\in C$ and find that every point of $C'$ has a unique representative $\left(x_0,\sigma(g,x_0)\right)\in C\times_kC$ for some $g\in J$ and therefore $C'\simeq J$.
Calculating
\begin{align}
    \sigma':\,J\times C'\,,\quad\left(h,(x_0,\sigma(g,x_0)\right)\mapsto \left(\sigma(h,x_0),\sigma(g,x_0)\right)\,\sim\left(x_0,\sigma(g+h,x_0)\right)\,,
\end{align}
we can also see that $(C',\sigma')\simeq(J,\sigma_J)$ and therefore $(C,\sigma^{-1})$ represents the inverse of $(C,\sigma)$.

The \textit{Tate-Shafarevich group} $\Sh_k(J)\subseteq \text{WC}_k(J)$ consists of the $J$-torsors that have a rational point ``locally'', that is over every completion $k_p$ of $k$.
While we do not have a good intuition for this subgroup in the case of curves, there is a geometric interpretation after applying the construction to fibrations.

A genus one fibration $\pi:X\rightarrow B$ can be interpreted as a genus one curve over the function field of $B$. The previous constructions can then be carried out relatively over $B$.
This sentence of course hides a significant (!) amount of technical complexity and for details, we refer the reader to~\cite{dg92} as well as~\cite{caldararuThesis} and~\cite{Caldararu2002}.
In the following, we will refer to a genus one fibration $\pi:J\rightarrow B$ as an \textit{elliptic fibration} if it has a section, that is a subvariety $D\subset J$ such that $\pi\vert_D: D\rightarrow B$ is a birational morphism.
The Tate-Shafarevich group $\Sh_B(J)$ then consists of genus one fibrations that have relative Jacobian fibrations isomorphic to $J$ and that locally have a section.

\subsection{Almost generic genus one fibered Calabi-Yau threefolds}
\label{sec:almostgeneric}

Consider a projective threefold $X^0$ that exhibits an elliptic fibration $\pi:X^0\rightarrow B$ with a section over a surface $B$.
We will assume that $X^0$ is in Weierstra{\ss} form\footnote{By~\cite[Proposition 2.4]{dg92}, every elliptic threefold is birationally equivalent to a Weierstra{\ss} model.}
\begin{align}
    \{\,y^2=x^3+fxz^4+gz^6\,\}\subset \mathbb{P}_{231}(\mathcal{L}^2\oplus\mathcal{L}^3\oplus\mathcal{O}_B)\,,
    \label{eqn:weierstrass}
\end{align}
for some line bundle $\mathcal{L}$ and sections $f\in\Gamma(B,\mathcal{L}^4)$ and $g\in\Gamma(B,\mathcal{L}^6)$.
The threefold $X^0$ is Calabi-Yau if  $\mathcal{L}$ is the anti-canonical line bundle $\omega_B^{-1}$ on $B$. The \textit{discriminant locus} $\Delta$, over which the fiber develops singularities, is given by the zero locus of the discriminant polynomial $d$,
\begin{align} \label{eq:Delta}
    d = 4 f^3 + 27g^2\,,\quad \Delta=\{\,d=0\,\}\subset B\,.
\end{align}

\begin{table}[ht!]
    \centering
    \begin{tabular}{|c|c|c|}\hline
    Type&$\text{ord}(f,g,d)$&Description\\\hline
    $I_0$&$(0,0,0)$&Generic point of $B$\\\hline
    $I_1$&$(0,0,1)$&Smooth point of $\Delta$\\\hline
    $I_2$&$(0,0,2)$&Node of $\Delta$\\\hline
    $II$&$(\ge1,1,2)$& Cusp of $\Delta$\\\hline
    \end{tabular}
    \caption{The Kodaira types of fibers that appear in almost generic elliptic Calabi-Yau threefolds.}
    \label{tab:Kodaira}
\end{table}

The singular fibers of the fibration can be classified according to the vanishing orders of $f$, $g$ and $d$.
We will only be interested in fibrations with fibers of Kodaira type $I_0$, $I_1$, $I_2$ and $II$, with the corresponding vanishing orders listed in Table~\ref{tab:Kodaira}.

\begin{definition}
We will call a projective Calabi-Yau threefold that is a Weierstra{\ss} elliptic fibration $\pi:X^0\rightarrow B$ over a projective surface $B$ an \textit{almost generic} elliptic Calabi-Yau threefold if the following properties are satisfied:
\begin{enumerate}[label=AG{\arabic*}]
    \item \label{condge:smoothqfac} The base $B$ is smooth and $X^0$ is $\mathbb{Q}$-factorial and has at most isolated nodes as singularities. 
    \item \label{condge:disc} The discriminant locus $\Delta$
    \begin{enumerate}
        \item \label{condge:redirr} is a reduced and irreducible curve
        \item \label{condge:discsing} and has at most isolated cusps and nodes as singularities.
    \end{enumerate}
    \item \label{condge:TSWC} The Tate-Shafarevich group and the Weil-Ch{\^a}telet group of the fibration coincide,
    \begin{align}
        \Sh_B(X^0)=\text{\normalfont WC}_B(X^0)\,.
    \end{align}
\end{enumerate}

More generally, we will call a projective genus one fibered Calabi-Yau threefold an \textit{almost generic} genus one fibered Calabi-Yau threefold if the Weierstra{\ss} model of the corresponding relative Jacobian fibration is an almost generic elliptic Calabi-Yau threefold.
We will denote the genus one fibration that corresponds to an element $\tse\in\Sh_B(X^0)$ by $X^{\tse}$.
If an almost generic elliptic (genus one fibered) Calabi-Yau threefold is smooth, we refer to it as a \textit{generic} elliptic (genus one fibered) Calabi-Yau threefold.
\label{def:almostgeneric}
\end{definition}

Given an almost generic elliptic Calabi-Yau threefold $X^0$, we denote the set of nodes of $\Delta$ by $S_{\Delta}\subset\Delta$, the number of nodes by $n_{I_2}=\# S_{\Delta}$ and the number of cusps of $\Delta$ by $n_{II}$.
The set of nodes in $X^\tse$, for $\tse\in\Sh(X^0)$, will be denoted by $S^\tse\subset X^\tse$.
In the following, we collect a series of remarks with regard to Definition~\ref{def:almostgeneric} and its implications.
\begin{enumerate}
    \item The Kodaira types of the fibers in the Weierstra{\ss} fibration $X^0$ are as described in Table~\ref{tab:Kodaira}.
    Moreover, $X^0$ has an isolated nodal singularity in every $I_2$ fiber and is smooth everywhere else.
    In particular, $\# S^0=\# S_{\Delta}=n_{I_2}$.

    This follows from an explicit analysis of the behavior of the Weierstra{\ss} fibration over nodes and cusps of the discriminant that we relegate to Appendix~\ref{sec:singularFibers}.

    \item
    The number of cusps is
    \begin{align} \label{eq:numberCusps}
        n_{II}=24 K_B \cdot K_B\,,
    \end{align}
    where $K_B$ is the canonical divisor of $B$. 
    
    Recall that the discriminant locus has a cusp at points where $g$ itself has a simple zero and $f$ also vanishes.
     The expression~\eqref{eq:numberCusps} then follows from the fact that $f$ and $g$ are respectively sections of $\mathcal{L}^4$ and $\mathcal{L}^6$ with $\mathcal{L}$ the anti-canonical bundle on $B$.

    \item
    The relative Jacobian of a generic genus one fibered Calabi-Yau threefold is only generic if the fibration itself already has a section, in which case $X$ and $J(X)$ represent the same element of the Tate-Shafarevich group and are birationally equivalent.
    
    To see this, first note that if $X$ lacks a section, it necessarily exhibits at least one $I_2$ fiber.
    The reason is that the unimodularity of the intersection pairing requires the existence of a fiber component that is intersected only once by the multisection.
    Since $I_2$ fibers correspond to nodes of the Weierstra{\ss} fibration, the latter can only be smooth if $X$ itself already has a section.

    \item \label{comment:nonKaehlerResolution}
    If $n_{I_2}>0$, there exist $2^{n_{I_2}}$ analytic small resolutions $\rho:\widehat{X}^0\rightarrow X^0$, none of which are K\"ahler.
    Given any such resolution, the homology class of the exceptional curve $[\rho^{-1}(p)]\in H_2(\widehat{X}^0,\mathbb{Z})$ over a node $p\in S^0$ will be torsion or trivial.

    As we will explain in Appendix~\ref{sec:nodalCY3}, this follows from the assumption~\ref{condge:smoothqfac} that $X^0$ is $\IQ$-factorial and has at most isolated nodes as singularities.
    
    \item The elements $\tse$ and $\tse^{-1}$ of $\Sh_B(X^0)$ are represented by the same geometry, i.e. $X^\gamma\simeq X^{\gamma^{-1}}$.

    This is analogous to the situation for genus one curves discussed in Section~\ref{ss:TSgroup} and follows from~\cite[Remark 4.5.3]{caldararuThesis}.

    \item The Mordell-Weil group of $X^0$ is trivial.
    
    This follows from the assumption of $\IQ$-factoriality of the Weierstra{\ss} fibration, as imposed in \ref{condge:smoothqfac}, together with~\cite[Corollary 2.7]{Kloosterman}.
    If we also assume \ref{condge:discsing}, this can be argued for rather simply: The Weierstra{\ss} model of an elliptic fibration blows down all fibral curves which do not meet the zero section, leaving the generic fiber as the only fibral curve class. Hence, $b_2(X^0) = b_2(B) + 1$. The $\IQ$-factoriality of $X^0$, given the assumption that it only exhibits isolated nodes, implies $b_2(X^0) = b_4(X^0)$, see Appendix \ref{sec:nodalCY3}. Hence, the homology classes of any two sections on $X^0$ must differ by a vertical divisor.
    
    \item The Calabi-Yau threefold $X^\tse$ is $\IQ$-factorial for all $\tse\in \Sh_B(X^0)$ and
    \begin{align}
        b_2(X^\tse)=b_4(X^\tse)=b_2(B)+1\,.
    \end{align}

    To see this, we use the Shioda-Tate-Wazir theorem extended to genus one fibrations, as was argued for in \cite[Section 8]{Braun:2014oya}.
    Together with the triviality of the Mordell-Weil group of $X^0$, this implies that $b_4(X^\tse)=b_2(B)+1$. The RHS of this relation is a lower bound on $b_2(X^\tse)$.
    By Theorem~\ref{thm:NamikawaSteenbrinkDefect}, we have $b_2(X^\tse) = b_4(X^\tse)-\sigma(X^\tse)$, in terms of the defect $\sigma(X^\tse)$ defined in~\eqref{eqn:defect}.
    Since the defect cannot be negative, it follows that $\sigma(X^\tse)=0$ and therefore $X^\tse$ is $\IQ$-factorial.
    
    \item The base $B$ is a weak del Pezzo surface, i.e. $\mathbb{P}^1\times \mathbb{P}^1$, the Hirzebruch surface $\mathbb{F}_2$ or the result of a sequence $S_i$ of blow-ups of $S_0 = \mathbb{P}^2$ at up to $8$ points, with $S_i$ a blow-up of $S_{i-1}$ at a point not lying on a curve of self-intersection $-2$.
    
    As we review in Section \ref{sec:genusonefibrations}, the smooth base of a smooth genus one fibered Calabi-Yau threefold is either a rational surface or birational to an Enriques surface. The arguments there also go through for a Calabi-Yau threefold with $\IQ$-factorial singularities (see~\cite[Theorem 8.2]{KawamataKodairaDimension} for the statement that the Kodaira dimension of such Calabi-Yau threefolds still vanishes).
    The conditions we impose in \ref{condge:disc} allow us to restrict the possible base surfaces further. As discussed in~\cite{Morrison:2012np}, the existence of a rational curve $C\subset B$ with $C\cdot C<-2$ in the base $B$ implies that the discriminant divisor $\Delta$ has a one-dimensional component over which the vanishing order of $4f^3+27g^2$ is greater than one.
    In order to satisfy condition~\ref{condge:redirr}, we must therefore restrict to base surfaces such that  every rational curve $C\subset B$ satisfies $C\cdot C\ge -2$.
    Furthermore, the base cannot be an Enriques surface, as the canonical class of an Enriques surface squares to zero; hence the Weierstra{\ss} coefficients $f,g \in \Gamma(B,\cO)$ are constant and therefore $\Delta$ is not a curve.
    Therefore, $B$ must be a weak del Pezzo surface.
    Recall that an ordinary del Pezzo surface, defined as a smooth projective surface with ample anti-canonical bundle, is either $\mathbb{P}^1\times\mathbb{P}^1$ or the blow-up $\text{dP}_{9-k}$ of $\mathbb{P}^2$ in $0\le k\le 8$ generic points.
    To also allow for $-2$ curves, the condition on the anti-canonical bundle is weakened to merely being nef and big in the definition of weak (or generalized) del Pezzo surface, see for example~\cite{Derenthal2008}.
    
    \item The integer homology $H_*(B,\mathbb{Z})$ is torsion free and both the Brauer group $\text{Br}(B)$ and the fundamental group $\pi_1(B)$ are trivial.

    This follows from $B$ being a weak del Pezzo surface and the fact that both the Brauer group and the fundamental group are a birational invariant for smooth projective varieties.

    \item We can then also show the following:
\begin{lemma}
    The fundamental group of an almost generic genus one fibered Calabi-Yau threefold $\pi:X^\tse\rightarrow B$ is trivial.
\end{lemma}
\begin{proof}
    As the base $B$ of the fibration is simply connected, we can deform a given loop ${r:[0,1] \rightarrow X^\tse}$ such that $r([0,1])$ lies in a single smooth fiber $\Sigma_p=\pi^{-1}(p)$ for some $p\in X^\tse$.  Choose a generic point $p' \in \Delta$ and a path $s:[0,1] \rightarrow X^\tse$ from $p$ to $p'$ which intersects $\Delta$ only in $p'$. As one cycle of the generic fiber pinches as we approach the $I_1$ fiber lying over $p'$, we have $H_1\left(\pi^{-1}(s([0,1]),\IZ\right) = \IZ$. Hence, one of the 1-cycles of $\Sigma_p$ is homotopically trivial in $X^\tse$. On the other hand, the fact that the fibration has fibers of type $I_1$ and $II$ implies that the group acting on the $1$-cycles of $\Sigma_p$ that is generated by the monodromies around singular fibers\footnote{Kodaira worked out the monodromies associated to degenerations of the elliptic fiber in the case of elliptic surfaces~\cite{KodairaII}. To apply this discussion to an isolated singularity of a genus one fibered threefold lying above a point $p\in B$, we consider a disk $D$ in a small neighborhood of $p$ which intersects $\Delta$ only in $p$, see also~\cite[Section 4]{Grassi:2001xu}. Note that if $p$ is a node or a cusp of $\Delta$, deforming $D$ away from $p$ while keep its intersection with a small $S^3$ surrounding $p$ fixed will result in multiple intersections of the deformed disk with regular points of $\Delta$. This permits relating the monodromy around nodes or cusps to the product of (conjugated) monodromies associated to $I_1$ fibers.} is the full modular group $\SLtwoZ$. Hence, both 1-cycles of $\Sigma_p$ are homotopically on the same footing, and $r([0,1])$ therefore contractible to a point.
\end{proof}

    \item
    There exist smooth deformations $\widetilde{X}^\tse$ of $X^\tse$ with a compatible genus one fibration $\tilde{\pi}_\gamma:\widetilde{X}^\tse\rightarrow B$.

    The existence of the smoothing follows from the fact the $X^0$ is $\mathbb{Q}$-factorial with at most isolated nodal singularities together with Theorem~\ref{thm:NamikawaSteenbrinkDeformation}.
    The existence of the compatible genus one fibration follows from~\cite[Theorem 31]{Kollar:2012pv}.

    \item The condition~\ref{condge:TSWC}, namely that the Tate-Shafarevich group and the Weil-Ch{\^a}telet group coincide, is certainly the most non-trivial. However, based on the physical discussion that we will come to in Section~\ref{sec:discreteHolonomiesAndGenusOne}, we are tempted to expect that it is automatically satisfied if the first two conditions are met.

    \item Our definition of almost generic genus one fibered Calabi-Yau threefolds is inspired by the definition of a generic elliptic threefold in~\cite{Caldararu2002}.
    We expect that a generic genus one fibered Calabi-Yau threefold according to Definition~\ref{def:almostgeneric} is also a generic elliptic threefold according to the definition from~\cite{Caldararu2002}.
    It would be interesting to understand the precise relationship between the two definitions.

\end{enumerate}

\subsection{Smooth genus one fibered Calabi-Yau threefolds}
\label{sec:genusonefibrations}

We will now summarize some known properties of smooth genus one fibered Calabi-Yau threefolds $\pi:X\rightarrow B$ with general fiber a curve $\Sigma$ of genus one.

Note first that the base of the fibration must satisfy 
\begin{equation} \label{eq:baseHodge}
    h^{i,0}(B)=0\,, \quad i=1,2 \,;
\end{equation} 
otherwise, the pullback along the fibration would lead to non-vanishing $h^{i,0}(X)$, in contradiction to $X$ being Calabi-Yau.
Using Dolbeault's theorem together with the Hirzebruch-Riemann-Roch theorem, it follows that the holomorphic Euler characteristic of the base is $\chi(\cO_B) = 1$.
Noether's formula $12\chi(\cO_B) = K_B^2 + c_2(T_B)$ thus implies
\begin{equation} \label{eq:Noether}
    K_B \cdot K_B = 10 - b_2(B) \,.
\end{equation}

The conditions \eqref{eq:baseHodge} have further implications, see e.g.~\cite[Section 6]{Grassi:2011hq}.
By the Enriques-Kodaira classification of surfaces, $b_1(B) = 0$ restricts the base of the fibration, up to birational equivalence, to be one of the following: a rational surface, an Enriques surface, a K3 surface, or a minimal surface of general type, see e.g.~\cite[Table 10]{BarthComplexSurfaces}.
A K3 surface as base is excluded by $h^{2,0}(\mathrm{K3})=1$.
Surfaces of positive Kodaira dimension, namely surfaces of general type and elliptic surfaces, are excluded by the subadditivity of the Kodaira dimensions, ${\kappa(X) \ge \kappa(E) + \kappa(B)}$ for $\dim(E) = 1$ \cite{ViehwegSubadditivity}, as $\kappa(X)=\kappa(E)=0$. We can conclude that $B$ is either a rational surface or birational to an Enriques surface. If $(X,\pi)$ is generic, in the sense of Section~\ref{sec:almostgeneric}, we have already noted that $B$ has to be a weak del Pezzo surface.

In the following, we will restrict the discussion to the case
\begin{equation} \label{eq:generic}
    b_2(X)=b_2(B)+1
\end{equation}
of relevance in this paper. This in particular implies that the fibration is flat, i.e. all of the fibers are one dimensional. We will also assume that there are no multiple fibers, which would correspond to points $p\in B$ over which the fiber $\pi^{-1}(p)$ is non-reduced.

Denote a basis of divisors on $B$ by $\Jbase_\alpha\in \text{Pic}(B),\,\alpha=1,\ldots,b_2(B)$.
The corresponding divisors $J_\alpha=\pi^*\Jbase_\alpha$ on $X$ are referred to as vertical divisors.
On the other hand, a divisor is referred to as an $m$-section, with the integer $m\ge 1$, if the degree of the restriction to the generic fiber $\Sigma$ is $m$.

The extension of the Shioda-Tate-Wazir theorem~\cite{SHIODA1972,tate1965algebraic,SB_1964-1966__9__415_0,Wazir2004} to genus one fibrations that are not elliptic was discussed in~\cite[Section 8]{Braun:2014oya} and implies that
\begin{align} \label{eq:STWbasis}
    \text{Pic}(X)=\langle J_0,J_1,\ldots,J_{b_2(B)}\rangle\,,    
\end{align}
with $J_0$ an $m$-section for some minimal $m$.

We denote the intersection numbers by
\begin{align}
    c_{ijk}=J_i\cdot J_j\cdot J_k\,,\quad \Omega_{\alpha\beta}=j_\alpha\cdot j_\beta\,,
    \label{eqn:topInvcandb}
\end{align}
and denote by $\Omega^{\alpha\beta}$ the inverse of the matrix $\Omega_{\alpha\beta}$.
Since $J_\alpha\cdot J_\beta=\Omega_{\alpha\beta}\Sigma$, we have
\begin{align}
     c_{i\alpha\beta}=m\delta_{i,0}\Omega_{\alpha\beta}\,.
\end{align}
Fiberwise integration provides a pushforward map from the cohomology of a fiber bundle to its base.
The pushforward map is needed to define the height-pairing of the $m$-section, which is given by
\begin{align}
    \pi_*(J_0\cdot J_0)=c_{00\alpha}\Omega^{\alpha\beta}j_\beta\,.
\end{align}
The relation follows from the fact that the pushforward is adjoint to the pullback, such that
\begin{align}
    j_\alpha\cdot \pi_*(J_0\cdot J_0)=J_\alpha\cdot J_0\cdot J_0\,.
\end{align}
If $J_0$ is a section, one also has
\begin{align}
    \pi_*(J_0\cdot J_0)=K_B\,.
    \label{eqn:selfintersectionSection}
\end{align}

The relation
\begin{align}
    J_\alpha \cdot c_2(TX)=-12 \Jbase_\alpha\cdot K_B\,,
    \label{eqn:intVerticalSecondChern}
\end{align}
was proven for special cases in~\cite[Appendix D]{Grimm:2013oga} and~\cite[Appendix B]{Cota:2019cjx}. More generally, the identity
\begin{equation} \label{eq:jalphaC2}
    J_\alpha \cdot c_2(TX) = \chitop(J_\alpha) \,,
\end{equation}
follows from adjunction~\cite{Grimm:2013oga}.\footnote{We can choose a smooth representative for $J_\alpha$ to which we apply the adjunction formula as the pullback of a base point free linear system is base point free.}
As $J_\alpha$ is a smooth elliptic surface, only singular fibers contribute to its Euler characteristic.
If all but $I_1$ fibers (which contribution +1 to $\chitop$) can be avoided, and all $I_1$ fibers lie over the discriminant $\Delta$ of the Jacobian fibration, then \eqref{eqn:intVerticalSecondChern} follows from \eqref{eq:jalphaC2} by $[\Delta] = - 12 K_B$, see \eqref{eq:Delta}.
In fact, the analysis of the Chern-Simons terms in Section~\ref{ss:CSterms} can be seen as a physical derivation of~\eqref{eqn:intVerticalSecondChern}.

Finally, given any Calabi-Yau threefold $X$, there is a numerical criterion to determine whether $X$ exhibits a genus one fibration.
The results from~\cite{Wilson1989} imply that if there is a divisor $D$ on a Calabi-Yau threefold $X$ such that
\begin{align}
    D^3=0\,,\quad D^2\ne 0\,,\quad D\cdot c_2(TX)>0\,,
    \label{eqn:fibcondition}
\end{align}
then $X$ exhibits a genus one fibration structure with respect to which $D$ is a vertical divisor, see also~\cite{OGUISO1993,Wilson1994,Kollar:2012pv}.

\subsection{Five mathematical conjectures}
\label{sec:mathresults}
Our physical analysis in Sections~\ref{sec:discreteHolonomiesAndGenusOne} and~\ref{sec:ellipticGeneraAndTopologicalStrings} will lead us to conjecture properties of almost generic genus one fibered Calabi-Yau threefolds that we will now summarize.
They can be seen as natural generalizations of results from~\cite{dg92,Caldararu2002,caldararuThesis,Donagi:2003av}.
We will discuss the conjectures at the hand of a concrete example in Section~\ref{sec:exampleZ4}.

Let $\pi_0:X^0\rightarrow B$ be an almost generic elliptically fibered Calabi-Yau threefold and $\rho_0:\widehat{X}^0\rightarrow X^0$ any small resolution.
Given $\tse\in\Sh_B(X^0)$, we denote the corresponding almost generic genus one fibered Calabi-Yau threefold by $\pi_{\tse}:X^{\tse}\rightarrow B$ and a choice of small resolution by $\rho_{\tse}:\widehat{X}^{\tse}\rightarrow X^{\tse}$.
We also introduce
\begin{align} \label{eq:defGamma}
    \Gamma^{\tse}=\tors{H^3(\widehat{X}^{\tse},\mathbb{Z})}\,,\qquad \widehat{\Gamma}^\tse=\text{Hom}(\Gamma^\tse,\U(1))=\tors{H_2(\widehat{X}^\tse,\IZ)}\,,
\end{align}
where the last equality follows from the universal coefficient theorem.

Recall that $S_\Delta\subset \Delta$ and $S^\tse \subset X^\tse$ are respectively the sets of nodes of $\Delta$ and $S^\tse$.
We will also use the following definitions:
\begin{itemize}
    \item $m_{\tse}=\text{ord}(\tse)$ is the order of $\tse$ in $\Sh_B(X^0)$.
    \item Anticipating conjecture \ref{conj:TSandTorsion}\ref{en:Sh}, $q_\tse:\widehat{\Gamma}^0\rightarrow \mathbb{N}$ assigns the unique integer $q_\tse(\chi)\in\{0,\ldots,m_\tse - 1\}$ to the torsion class $\chi$ such that $\tse(\chi)={\rm e}[q_\tse(\chi)/m_\tse]$.
    \item $C_{p,\tse}^E=\rho_\tse^{-1}(p)$ for $p\in S^\tse$ is the exceptional curve in $\widehat{X}^\tse$ that resolves the node $p$.
\end{itemize}

The conclusions from our physical analysis can then be summarized in four conjectures.
\paragraph{Tate-Shafarevich group and torsion curves}
The first conjecture describes the relation between the Tate-Shafarevich group of $X^0$ and the torsion homology of the fibrations $\widehat{X}^\tse$ for $\tse\in\Sh_B(X^0)$.
\begin{conjecture}
We make the following claims:
\begin{enumerate}[label=\alph*)]
    \item \label{en:Sh} $\Sh_B(X^0)=\Gamma^0$
    \item \label{en:Gamma} $\Gamma^{\tse}=\Gamma^0/\langle\tse\rangle$
\end{enumerate}
\label{conj:TSandTorsion}
\end{conjecture}
Note that Conjecture~\ref{conj:TSandTorsion}\ref{en:Gamma} can be equivalently stated as $\widehat{\Gamma}^\tse=\text{\normalfont ker}\,\tse$, and allows us to embed the torsion of curves in $\widehat{X}^\tse$ into $\widehat{\Gamma}^0$.

\paragraph{Singularities and reducible fibers}
Let us now assume that Conjecture~\ref{conj:TSandTorsion} holds.
The following Conjecture~\ref{conj:I2fibers} can be seen as a refinement of Conjecture~\ref{conj:TSandTorsion}\ref{en:Gamma} and describes the homology classes of the components of the different $I_2$-fibers in $\widehat{X}^\tse$.

We denote by $H_2(\widehat{X}^\tse,\mathbb{Z})_{\rm F}$ the homology classes of curves that are orthogonal, in terms of the intersection product, to all vertical divisors.
This subgroup is generated by the class of the generic fiber $[\Sigma_\tse]$ and by the classes $[C_{p,\tse}^A],[C_{p,\tse}^B]$ of the rational components of the $I_2$-fibers in $\widehat{X}^\tse$,
\begin{align}
    C_{p,\tse}=C_{p,\tse}^A\cup C_{p,\tse}^B=(\pi_\tse\circ\rho_\tse)^{-1}(p)\,,\quad p\in S_\Delta\,.
\end{align}
These classes are not all independent; in particular, $[\Sigma_\tse]=[C_{p,\tse}^A]+[C_{p,\tse}^B]$.
We will refer to these curves as \textit{fibral curves}.

Note that even though $H_2(\widehat{X}^\tse,\mathbb{Z})_{\rm F}\simeq \mathbb{Z}\times\widehat{\Gamma}^\tse$, there is in general no canonical choice for this splitting: for any isomorphism
\begin{align}
    \iota: \,H_2(\widehat{X}^0,\mathbb{Z})_{\rm F}\rightarrow \mathbb{Z}\times \widehat{\Gamma}^0\,,\quad C \mapsto \left(\iota(C)_1, \iota(C)_2\right) \,,
\end{align}
the map $\iota_\chi(C) = (\iota(C)_1, \iota(C)_1\chi+ \iota(C)_2)$ for any $\chi \in \widehat{\Gamma}^0$ is an equally valid choice. However, for $\tse=0$, we can fix this ambiguity by requiring that the isomorphism satisfies $\iota_0([\Sigma_\tse])=(1,\text{id})$; $\iota_0$ is then unique determined up to automorphisms of $\widehat{\Gamma}^0$.\footnote{For $\iota([\Sigma_\tse]) = n > 1$, requiring $\iota([\Sigma_\tse]) = (n, \id)$ fixes the ambiguity up to an order $n$ element of $\widehat{\Gamma}^0$.}
For each $p\in S^0$, we can then choose the labels $A,B$ such that
\begin{align}
    C_{p,0}^B=C_{p,0}^E\,,\quad [C_{p,0}^A]=(1,-\chi_p)\,,\quad [C_{p,0}^B]=(0,\chi_p)\,,
\end{align}
for some $\chi_p\in\widehat{\Gamma}^0$.
Note that $\chi_p$ depends on the choice of $\widehat{X}^0$ among the $2^{\# S_\Delta}$ small resolutions of $X^0$; flopping the exceptional curve over $p$ exchanges $\chi_p$ with $-\chi_p$.

The following Conjecture~\ref{conj:I2fibers} generalizes this to $\widehat{X}^\gamma$ with $\gamma\ne 0$:

\begin{conjecture} 
Let $S_\tse^0=\{\,p\in S^0\,\,\vert\,\,\tse([C_{p,0}^E])=1\,\}$ and let $S_{\Delta,\tse}=\pi_0(S^0_{\tse})\subset S_\Delta$ denote the corresponding subset of the set of nodes of the discriminant locus.
We claim that
\begin{align}
    \text{\normalfont Hom}\left(H_2(\widehat{X}^\tse,\mathbb{Z})_{\rm F},\U(1)\right)=\quotient{\U(1)\times\Gamma^0}{\langle \,(\e[-\tfrac{1}{m_\tse}],\tse)\,\rangle}\,,
\end{align}
in terms of the Pontryagin dual $\Gamma^0=\text{\normalfont Hom}(\widehat{\Gamma}^0,\U(1))$, such that pullback gives an isomorphism
\begin{align}
    \iota_\tse:\,H_2(\widehat{X}^\tse,\mathbb{Z})_{\rm F}\rightarrow\text{\normalfont ker}\,(\e[-\tfrac{1}{m_\tse}],\tse)\subset \mathbb{Z}\times\widehat{\Gamma}^0\,,
    \label{eqn:conj2iso}
\end{align}
with $\iota_\tse([\Sigma_\tse])=(m_\tse,\text{\normalfont id})$.
Moreover, for $p\in S_{\Delta}\backslash S_{\Delta,\gamma}$ we can choose the labels $A,B$ such that
\begin{align}
    \begin{split}
        \iota_\tse([C_{p,\tse}^{A}])=&\left(m_\tse-q_\tse(\chi_p),\,-\chi_p\right)\,,\quad \iota_\tse([C_{p,\tse}^{B}])=\left(q_\tse(\chi_p),\,\chi_p\right)\,.
    \end{split}
    \label{eqn:curvesXgammaI2}
\end{align}
Note that $q_\tse(-\chi)=m_\tse-q_\tse(\chi)$ and therefore~\eqref{eqn:curvesXgammaI2} is independent of the choice of small resolution $\widehat{X}^0$ of $X^0$ up to exchange of the labels $A$ and $B$.
For $p\in S_{\Delta,\gamma}$ the curves represent the classes
\begin{align}
    \begin{split}
    \iota_\tse([C_{p,\tse}^A])=&\left(m_\tse,\,\mp\chi_p\right)\,,\quad \iota_\tse([C_{p,\tse}^B])=\left(0,\,\pm\chi_p\right)\,,
    \end{split}
\end{align}
where the sign of $\pm\chi_p$ depends on the choice of the small resolution $\widehat{X}^\gamma$.
This also implies that
\begin{align}
    \pi_{\tse}(S^{\tse})=S_{\Delta,\tse}\,.
\end{align}
There exist $2^{\# S_{\Delta,\tse}}$ small resolutions of $X^\tse$, none of which are K\"ahler if $S_{\Delta,\tse}\ne \emptyset$.
\label{conj:I2fibers}
\end{conjecture}
Note that by Poincaré duality, $\iota_\tse([\Sigma_\tse])=(m_\tse,\text{\normalfont id})$ implies that $\widehat{X}^\tse$ exhibits an $m_\tse$-section, but no $m$-section with $m<m_\tse$.
The image of this $m_\tse$-section under $\rho_\tse:\widehat{X}^\tse\rightarrow X^\tse$ also gives an $m_\tse$-section on $X^\tse$.

\paragraph{Topology of smoothing}
In Section~\ref{ss:CSterms}, we will compare the Chern-Simons terms of the twisted F-theory compactification with the M-theory compactification on $X^\gamma$. This will allow us to predict the topological invariants of a smooth deformation $\widetilde{X}^\gamma$ of $X^\gamma$.

The results are summarized in Conjecture~\ref{conj:topologyOfSmoothing}:
\begin{conjecture}
Let $\widetilde{X}^{\tse}$ be a smooth deformation of $X^{\tse}$ with a compatible fibration $\tilde{\pi}_\tse:\widetilde{X}^{\tse}\rightarrow B$.
We claim that $\widetilde{X}^\tse$ has an $m_\tse$-section $J_0\in\text{\normalfont Pic}(\widetilde{X}^\tse)$ with respect to this fibration and no $m'$-section with $m'<m_\tse$.
We can choose a basis of vertical divisors $J_\sia=\tilde{\pi}_\tse^* \Jbase_\sia$ with $\Jbase_\sia\in\text{\normalfont Pic}(B)$ for $\sia\in\{1,\ldots,b_2(B)\}$ and introduce the basis
\begin{align}
    D_0=J_0-\frac{1}{2m_\tse}\tilde{\pi}_\tse^*\circ\tilde{\pi}_{\tse*}\left(J_0\cdot J_0\right)\,,\quad D_\sia=J_\sia\,.
    \label{eqn:conj3changeOfBasis}
\end{align}
We define the intersection numbers $\Omega_{\sia\sib}=\Jbase_\sia\cdot \Jbase_\sib$ as well as 
\begin{align}
    k_{\lii\lij\lik}=D_\lii\cdot D_\lij\cdot D_\lik\,,\quad \kappa_\lii=D_\lii\cdot c_2(T\widetilde{X}^{\tse})\,,
    \label{eqn:nonTrivialIntersections}
\end{align}
where $\sia,\sib\in\{1,\ldots,b_2(B)\}$ and $\lii,\lij,\lik\in\{0,\ldots,b_2(B)\}$.
For non-trivial $\chi\in\widehat{\Gamma}^0$, we also introduce
\begin{align}
    N_\chi=N_{\chi}'+N_{-\chi}'\,,\quad N_\chi'=\#\{\,p\in S^0\,\,\vert\,\,[C_{p,0}^E]=\chi\,\}\,.
    \label{eqn:Nchi}
\end{align}
We then further claim that the Euler characteristic of $\widetilde{X}^{\tse}$ is
\begin{align}
    \chi(\widetilde{X}^{\tse})=60\big(b_2(B)-10\big)+2\cdot \#(S_\Delta\backslash S_{\Delta,\gamma})\,, 
    \label{eqn:conj3euler}
\end{align}
while the intersection numbers~\eqref{eqn:nonTrivialIntersections} are given by
\begin{align}k_{\lii\sia\sib}=\delta_{\lii,0}m_{\tse}\Omega_{\sia\sib}\,,\quad k_{00\sia}=0\,,\quad \kappa_\sia=-12\Jbase_\sia\cdot K_B\,,
    \label{eqn:conj3ints}
\end{align}
as well as
\begin{align}
\begin{split}
	k_{000}=&\frac{m_{\tse}^3(10-b_2(B))}{4}-\frac{1}{8m_{\tse}}\sum\limits_{\chi\in\widehat{\Gamma}^0}q_\tse(\chi)^2\left(m_{\tse}-q_\tse(\chi)\right)^2N_\chi\,,\\
        \kappa_0=&4m_{\tse}(13-b_2(B))-\frac{1}{2m_{\tse}}\sum\limits_{\chi\in\widehat{\Gamma}^0}q_\tse(\chi)\left(m_{\tse}-q_\tse(\chi)\right)N_\chi\,.
\end{split}
	\label{eqn:twistedKM}
\end{align}
\label{conj:topologyOfSmoothing}
\end{conjecture}

Note that the intersection numbers $k_{\lii\sia\sib}$ and $k_{00\sia}$ in~\eqref{eqn:conj3ints} are just a consequence of $J_0$ being an $m_\tse$-section and $D_\alpha$ being vertical divisors, while the relation $\kappa_\sia=-12\Jbase_\sia\cdot K_B$ was discussed below equation~\eqref{eqn:intVerticalSecondChern} above.
In our F-theory discussion, we will also use the number of uncharged half-hypermultiplets $N_0$, which can be expressed as
\begin{align} \label{eq:uncharged}
    N_0=604-58b_2(B)-2\cdot\# S^0\,.
\end{align}

\paragraph{Twisted derived equivalences}
In Section~\ref{sec:modularityAndDerivedEquivalences}, we will use the modular properties of the topological string partition function to derive the following Conjecture~\ref{conj:twisted}.

The conjecture pertains to the derived categories of twisted sheaves on the analytic small resolutions $\widehat{X}^\tse$, where the twist is given by an element of the (cohomological) Brauer group.
For a brief introduction to the Brauer group and twisted sheaves, see e.g.~\cite[Section 4.4]{Katz:2022lyl} (geared towards physicists) or~\cite{caldararuThesis}.
To improve readability, we will write for the moment $X=\widehat{X}^\tse$.
The cohomological Brauer group of $X$ is defined as $\text{Br}'(X)=\tors{H^2(X,\mathcal{O}_X^*)}$.
Since the canonical class of $X$ is trivial and therefore $h^{0,2}(X)=0$, Dolbeault's theorem implies that $H^2(X,\mathcal{O}_X)=0$.
The long exact sequence associated to the exponential sheaf sequence $0\rightarrow \IZ\rightarrow \mathcal{O}_X\rightarrow \mathcal{O}_X^*\rightarrow 0$
then implies that
\begin{align}
    \tors{H^2(X,\mathcal{O}_X^*)}=\tors{H^3(X,\IZ)}\,.
\end{align}
Note that $\mathrm{Br}'(\widehat{X}^\tse)$ thus coincides with $\Gamma^\tse$ as defined in \eqref{eq:defGamma}.
The ``actual'' Brauer group $\text{Br}(X)$ is defined in terms of equivalence classes of sheaves of Azumaya algebras or, equivalently, of projective bundles on $X$.
For details, we again refer to~\cite[Section 4.4]{Katz:2022lyl} or~\cite{caldararuThesis}.
The equivalence of the Brauer group and the cohomological Brauer group is true for quasi-projective varieties~\cite{JongBrauer}. For analytic varieties, it has been proven only in special cases, see e.g.~\cite{Schrer2005}.
For our purposes, it will be sufficient to work with the cohomological Brauer group.

Let $\tseOne,\tseTwo\in \Sh_B(X^0)$.
After performing a partial smoothing if necessary, we can assume that $S^0_{\tseOne}\cap S^0_{\tseTwo}=\emptyset$.
Denote by $[\tseTwo]_{\tseOne}$ the equivalence class of $\tseTwo$ in $\Gamma^\tseOne=\Gamma^0/\langle\tseOne\rangle$.
We then denote the derived category of $[\tseTwo]_{\tseOne}$-twisted sheaves on $\widehat{X}^\tseOne$ by $D^b(\widehat{X}^{\tseOne},[\tseTwo]_{\tseOne})$ and make the following conjecture:

\begin{conjecture}
    Given any $g={\tiny\left(\begin{array}{cc}a&b\\c&d\end{array}\right)}\in{\normalfont\SLtwoZ}$, there exists a twisted derived equivalence
        \begin{align} \label{eq:derivedModular}
            D^b(\widehat{X}^{\tseOne},[\tseTwo]_{\tseOne})\simeq D^b(\widehat{X}^{\tseOnePrime},[\tseTwoPrime]_{\tseOnePrime})\,,\quad 
            \left(\begin{array}{c}\tseOnePrime\\\tseTwoPrime\end{array}\right)=\left(\begin{array}{c}\tseOne^a\tseTwo^c\\\tseOne^b\tseTwo^d\end{array}\right)\,.
        \end{align}
    \label{conj:twisted}
\end{conjecture}

When $\tseOne$ generates $\Gamma^0$, $\tseTwo=\tseOnePrime=1$ and $\tseTwoPrime=\tseOne$, the claim reduces to~\cite[Theorem 5.1]{Caldararu2002}.
On the other hand, if $\tseOne=\tseOnePrime=1$, $\tseTwo$ generates $\Gamma^0$ and $\tseTwoPrime=\tseTwo^k$ for some $k$ with $\text{\normalfont gcd}(k,\text{\normalfont ord}(\tseTwo))=1$, the claim reduces to~\cite[Theorem 6.1]{Caldararu2002}.
Otherwise, it implies that the results from~\cite{Caldararu2002,Donagi:2003av} generalize to genus one fibered threefolds without a section that have themselves isolated nodes.

\paragraph{Complete invariants of generic fibrations}
Let us now introduce an invariant that characterizes a generic genus one fibration:
\begin{definition}
Let $B$ be a weak del Pezzo surface, $m\in\INnz$ and $\mathfrak{l}\in H^2(B,\IZ)/mH^2(B,\IZ)$.
Let $S$ be a function
\begin{align}
    S:\{1,\ldots,m-1\}\rightarrow \mathbb{N}\,,
\end{align}
with $S(q)=S(m-q)$.
We say that $(B,m,\mathfrak{l},S\,)$ is a generic Calabi-Yau threefold genus one fibration datum.
\end{definition}
Based on the observations from~\cite{Dierigl:2022zll,Dierigl:wip}, we make the following conjecture:
\begin{conjecture}
    Let $(B,m,\mathfrak{l},S\,)$ be a generic Calabi-Yau threefold genus one fibration datum.
    There exists at most one family of smooth Calabi-Yau threefolds $X$ that exhibits a genus one fibration structure $\pi:X\rightarrow B$ such that the following holds:
    \begin{itemize}
        \item The relative Jacobian $X^0$ of $X$ over $B$ satisfies $\Sh_B(X^0)=\IZ_m$.
        \item There exists an $m$-section $J_0$ on $X$ such that
        \begin{align}
            \mathfrak{l}=\frac12\left(m^2\,c_1(B)+\pi_*\left(J_0\cdot J_0\right)\right)\,.
        \end{align}
        \item With $N_\chi$ as in Conjecture~\ref{conj:topologyOfSmoothing} and $\tse$ the element in $\Sh_B(X^0)$ corresponding to $X$, one has
        \begin{align}
            N_\chi=S\left(q_{\tse}(\chi)\right)\,,\quad \chi\in\IZ_m\,.
        \end{align}
    \end{itemize}
    \label{conj:datum}
\end{conjecture}
Physically, this corresponds to the claim that a six-dimensional supergravity with a finite Abelian gauge group is, up to choices of vacuum expectation values, uniquely specified by the string charge lattice, the spectrum of charged hypermultiplets and the discrete Green-Schwarz anomaly coefficient that has been introduced in~\cite{Dierigl:2022zll}.

We say that a generic Calabi-Yau threefold genus one fibration datum is \textit{admissible} if there exists a fibration $X$ as in Conjecture~\ref{conj:datum}.
In~\cite{Dierigl:2022zll} it has been observed that for $B=\mathbb{P}^2$ and $m=3$ a necessary condition for the datum $(\mathbb{P}^2,3,\mathfrak{l},S)$ to be admissible is that
$S(1)=0\text{ mod }9$ and $\mathfrak{l}=0$ or $S(1)=6\text{ mod }9$ and $\mathfrak{l}\in\{1,2\}$.
Determining the necessary and sufficient conditions for a generic Calabi-Yau threefold genus one fibration datum to be admissible is an open problem. One of the motivations for this paper is to lay the groundwork for such a classification.

\section{Discrete holonomies in F-theory and genus one fibrations}
\label{sec:discreteHolonomiesAndGenusOne}
We will now apply our discussion from Section~\ref{sec:kkandhol} to circle compactifications of six-dimensional F-theory vacua with discrete Wilson lines and the duality with M-theory.
We assume some familiarity with F-theory on the side of the reader and refer to~\cite{Weigand:2018rez} for a review.

Let us first recall some notation.
We will denote by ${\rm F}[X]$ the effective theory that is associated to a genus one fibered Calabi-Yau threefold $X$, leaving the particular choice of fibration implicit.
Similarly, we use ${\rm M}[X]$ to refer to the effective theory that is associated to the compactification of M-theory on $X$.
Given an effective theory ${\rm F}[X]$ with finite Abelian gauge group $G=\Gamma$, we denote by ${\rm F}[X][S^1_\tse]$ the compactification of ${\rm F}[X]$ on a circle with twist $\tse\in \Gamma$.
We will focus on six-dimensional vacua that are almost generic in the sense of Definition~\ref{def:almostGenericFtheory}.

Standard properties of F-theory imply that if $X^0$ is an almost generic elliptic Calabi-Yau threefold, in the sense of Definition~\ref{def:almostgeneric}, then ${\rm F}[X^0]$ is an almost generic F-theory vacuum without vector multiplets.
The gauge group of the effective theory is then expected to be the Tate-Shafarevich group of the fibration~\cite{Braun:2014oya},\footnote{More generally, as was pointed out in~\cite[Footnote 15]{Bhardwaj:2015oru}, the F-theory gauge group should be identified with the Weil-Ch\^{a}telet group of the fibration.
The geometries that represent elements of the Weil-Ch\^{a}telet group but not of the Tate-Shafarevich group have so-called multiple fibers and/or non-flat fibers.
As we will further discuss in Section~\ref{sec:genericity}, we expect those geometries to only arise for F-theory vacua that are at special points of the tensor branch.
In the almost generic case, the corresponding fibrations should be such that the Weil-Ch\^{a}telet group and the Tate-Shafarevich group coincide.
}
\begin{align}
    \Gsd=\Sh_B(X^0)\,,
\end{align}
and the twisted circle compactification ${\rm F}[X^0][S^1_\tse]$ is expected to be dual to ${\rm M}[X^\tse]$, where $X^\tse$ is the genus one fibration that represents the element $\gamma\in\Sh_B(X^\tse)$.

The question that we want to answer is the following:
\begin{questionnn}
Consider a six-dimensional F-theory vacuum ${\rm F}[X^0]$ with finite Abelian gauge group
\begin{align*}
    \Gsd=\Gamma^0\,,
\end{align*}
and the twisted circle compactification ${\rm F}[X^0][S^1_\tse]$ for some element $\tse\in\Gamma^0$.
Assuming that there exists an M-theory compactification ${\rm M}[X^\tse]$ on a Calabi-Yau threefold $X^\tse$ such that ${\rm M}[X^\tse]$ is dual to ${\rm F}[X^0][S^1_\tse]$, what are the properties of $X^\tse$?
\end{questionnn}
We will carefully study the properties of the twisted circle compactification ${\rm F}[X^0][S^1_\tse]$ in order to answer this question and to physically argue for the Conjectures~\ref{conj:TSandTorsion},~\ref{conj:I2fibers} and~\ref{conj:topologyOfSmoothing} from Section~\ref{sec:mathresults}.

As an important part of our argument, we will adopt the following general assumption about the behavior of M-theory on nodal Calabi-Yau threefolds from~\cite{Katz:2022lyl,Katz:2023zan,Schimannek:2025cok}:
\begin{claim}
    The gauge group $\Gfd$ of the effective theory from M-theory compactified on a projective Calabi-Yau threefold $X$ with at most isolated nodal singularities is given by
    \begin{align}
        \Gfd=\text{\normalfont Hom}\left(H_2(\widehat{X},\mathbb{Z}),{\rm U}(1)\right)\,,
    \end{align}
    where $\widehat{X}$ is any, possibly non-K\"ahler, analytic small resolution of $X$.
    \label{claim:mtheoryNodalCYgaugeGroup}
\end{claim}
\begin{claim}
    With $X$ and $\widehat{X}$ as in Claim \ref{claim:mtheoryNodalCYgaugeGroup}, the holomorphic curves $C \subset \widehat{X}$ give rise to half-hypermultiplets of charges $\pm [C] \in H_2(\widehat{X},\IZ)$ in $\M[X]$.
    \label{claim:mtheoryNodalCYmatter}
\end{claim}
In compactifications on smooth Calabi-Yau threefolds, both Claim \ref{claim:mtheoryNodalCYgaugeGroup} and Claim \ref{claim:mtheoryNodalCYmatter} are well-understood. In particular, the origin of the hypermultiplets in Claim \ref{claim:mtheoryNodalCYmatter} are M2 and anti-M2 branes wrapping the holomorphic curves $C$. The subtlety here is that $\widehat{X}$ is a non-projective resolution of the singular Calabi-Yau variety $X$; its status as a compactification manifold hence needs to be clarified. 

We will see that Claim~\ref{claim:mtheoryNodalCYgaugeGroup} implies that $X^\tse$ is in general singular, namely whenever $\langle \tse\rangle$ is a proper subgroup of $\Gamma^0$.
Since many relevant geometric and physical results are only known to hold in the smooth case, some of our arguments will rely on the partially Higgsed theory ${\rm F}[J(\widetilde{X}^\tse)]$, where $\widetilde{X}^\tse$ is a generic smooth deformation of $X^\tse$.
This corresponds to the sublocus of the Higgs branch where only the subgroup $\langle\tse\rangle\subset\Gsd$ remains unbroken.
The relationship between the theories is summarized in Figure~\ref{fig:twoPathsFM}.

In Section~\ref{sec:genericity} we will discuss under which circumstances an almost generic F-theory vacuum without massless vector multiplets can be obtained as a compactification ${\rm F}[X^0]$ on  an almost generic elliptic Calabi-Yau threefold $X^0$.

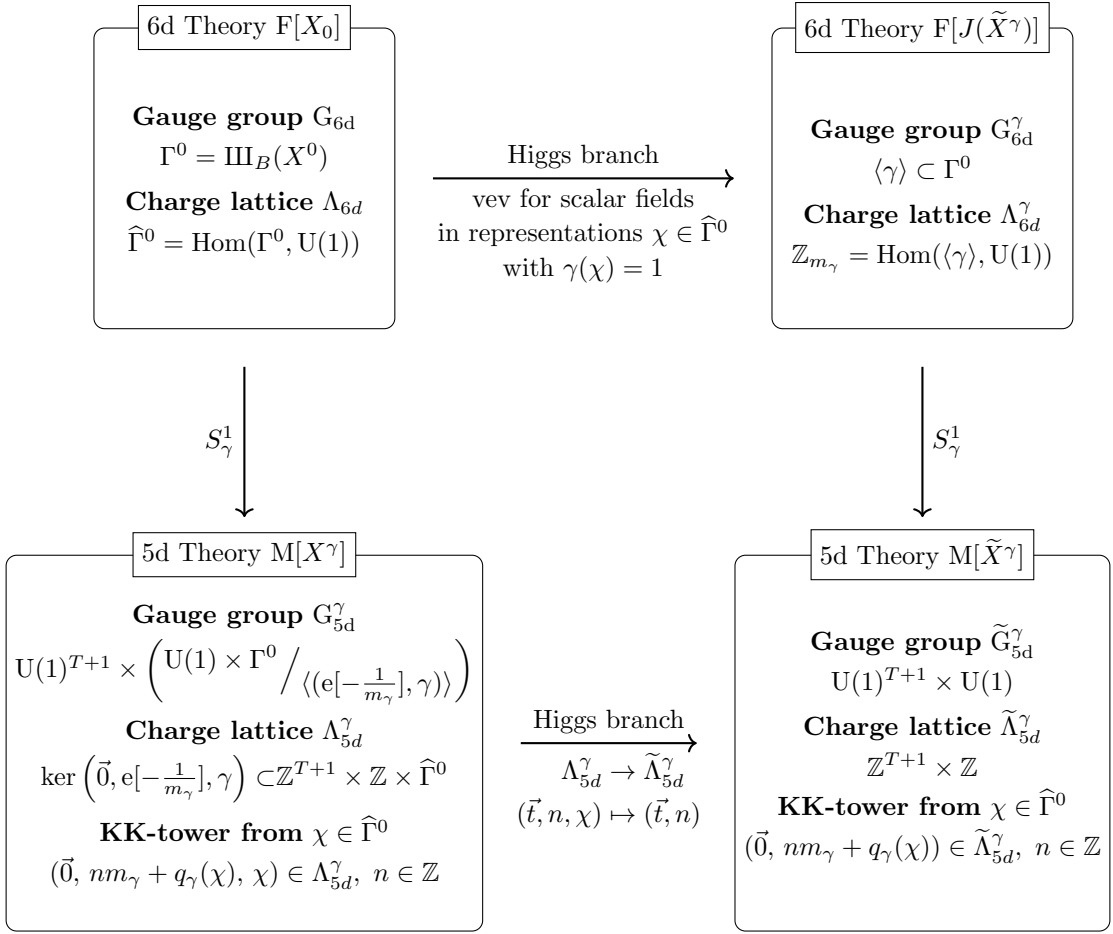
\begin{figure}
\centering
\begin{tikzpicture}[
    boxHigh/.style={rectangle, rounded corners, text centered, draw, minimum width=4cm, align=center, minimum height=4cm},
    boxLow/.style={rectangle, rounded corners, text centered, draw, minimum width=4cm, align=center, minimum height=5cm},
    header/.style={label={[rectangle, fill=white, draw, anchor=center, minimum width=2cm, name=\tikzlastnode-header]north:{#1}}},
    node distance=3cm and 5cm 
]

\node[boxHigh, header = {6d Theory ${\rm F}[X_0]$}] (upperLeft) {
    \phantom{block space}\\
    \textbf{Gauge group} $\Gsd$\\[.2em]
    $\Gamma^0=\Sh_B(X^0)$\\[.4em]
    \textbf{Charge lattice} $\Lambda_{6d}$\\[.2em]
    $\begin{aligned}
    \widehat{\Gamma}^0=\text{Hom}(\Gamma^0,\U(1))
    \end{aligned}$\\
};
\node[boxHigh, header = {6d Theory ${\rm F}[J(\widetilde{X}^\tse)]$}] (upperRight) [right=of upperLeft] {
    \phantom{block space}\\
    \textbf{Gauge group} $\Gsd^\tse$\\[.2em]
    $\langle \tse\rangle\subset\Gamma^0$\\[.4em]
    \textbf{Charge lattice} $\Lambda_{6d}^\tse$\\[.2em]
    $\mathbb{Z}_{m_\tse}=\text{Hom}(\langle\tse\rangle,\U(1))$
};
\node[boxLow, header = {5d Theory ${\rm M}[X^\tse]$}] (lowerLeft) [below=of upperLeft] {
    \phantom{block space}\\ 
    \textbf{Gauge group} $\Gfd^\tse$\\[.2em]
    $\U(1)^{T+1}\times\left(\quotient{\U(1)\times \Gamma^0}{\langle(\e[-\tfrac{1}{m_\tse}],\tse)\rangle}\right)$ \\[.4em]
    \textbf{Charge lattice} $\Lambda_{5d}^\tse$\\[.2em]
    $\begin{aligned}
        \text{ker}\left(\vec{0},\e[-\tfrac{1}{m_\tse}],\tse\right)\subset&\mathbb{Z}^{T+1}\times \mathbb{Z}\times \widehat{\Gamma}^0\\
    \end{aligned}$\\[.4em]
    \textbf{KK-tower from} $\chi\in\widehat{\Gamma}^0$\\[.2em]
    $(\vec{0},\,nm_\tse+q_\tse(\chi),\,\chi)\in \Lambda_{5d}^\tse,\,\,n\in\mathbb{Z}$\\
};
\node[boxLow, header = {5d Theory ${\rm M}[\widetilde{X}^\tse]$}] (lowerRight) [below=of upperRight] {
    \phantom{block space}\\
    \textbf{Gauge group} $\tGfd^{\tse}$\\[.2em]
    $\U(1)^{T+1}\times \U(1)$\\[.4em]
    \textbf{Charge lattice} $\widetilde{\Lambda}_{5d}^\tse$\\[.2em]
    $\begin{aligned}
        \mathbb{Z}^{T+1}\times\mathbb{Z}
    \end{aligned}$\\
    \textbf{KK-tower from} $\chi\in\widehat{\Gamma}^0$\\[.2em]
    $(\vec{0},\,nm_\tse+q_\tse(\chi))\in \widetilde{\Lambda}_{5d}^\tse,\,\,n\in\mathbb{Z}$\\
};

\coordinate (StartHorizontalLeft) at ($(upperLeft.east)+(0.5cm,0)$);
\coordinate (EndHorizontalLeft) at ($(upperRight.west)-(0.5cm,0)$);

\coordinate (StartHorizontalRight) at ($(lowerLeft.east)+(0.5cm,0)$);
\coordinate (EndHorizontalRight) at ($(lowerRight.west)-(0.5cm,0)$);

\coordinate (StartVerticalLeft) at ($(upperLeft.south)-(0,0.5cm)$);
\coordinate (EndVerticalLeft) at ($(lowerLeft.north)+(0,0.5cm)$);

\coordinate (StartVerticalRight) at ($(upperRight.south)-(0,0.5cm)$);
\coordinate (EndVerticalRight) at ($(lowerRight.north)+(0,0.5cm)$);

\coordinate (text6d) at ($(upperLeft.west)-(0.8cm,0)$);
\coordinate (text5d) at ($(lowerLeft.west)-(0.8cm,0)$);

\draw[->,thick] (StartHorizontalLeft) -- (EndHorizontalLeft) node[midway, above] {Higgs branch} node[midway, below, align=center] {
    vev for scalar fields\\
    in representations $\chi\in\widehat{\Gamma}^0$\\
    with $\tse(\chi)=1$
};
\draw[->, thick] (StartHorizontalRight) -- (EndHorizontalRight) node[midway, above] {Higgs branch} node[midway, below] {
$\begin{aligned}
    \Lambda_{5d}^\tse&\rightarrow \widetilde{\Lambda}_{5d}^\tse\\
    (\vec{t},n,\chi)&\mapsto(\vec{t},n)
\end{aligned}$
};

\draw[->,thick] (StartVerticalLeft) -- (EndVerticalLeft) node[midway, left, align=center] {$S^1_\tse$} ;
\draw[->,thick] (StartVerticalRight) -- (EndVerticalRight) node[midway, right, align=center] {$S^1_\tse$};

\end{tikzpicture}
\caption{The relationship between the F- and M-theory compactifications that are associated to an almost generic elliptic Calabi-Yau threefold $X^0$.} \label{fig:twoPathsFM}
\end{figure}

\subsection{Six-dimensional F-theory vacua}
Let us first summarize some general properties of six-dimensional almost generic F-theory vacua and the corresponding $\mathcal{N}=1$ effective supergravities in order to fix some of our notation.

The massless field content consists of one gravity multiplet, $T$ tensor multiplets, $V$ vector multiplets and some number of hypermultiplets.
We will focus on theories with $V=0$.
Each hypermultiplet can be thought of as being composed of two charge conjugate half-hypermultiplets and therefore is associated to an element of the set $[\chi]\in\widehat{\Gamma}/(\chi\sim-\chi)$.
It will be convenient to introduce $N_\chi$ to denote the number of half-hypermultiplets that transform in the representation $\chi\in\widehat{\Gamma}$, such that $N_\chi=N_{-\chi}$. The number $H_\chi$ of hypermultiplets of charge $\pm\chi$ is then
\begin{align}
    H_\chi=\left\{\begin{array}{cl}
    \frac12N_\chi&\text{ if }\chi=-\chi\\
    N_\chi&\text{ else}
    \end{array}\right.\,,
\end{align}
and the total number of hypermultiplets is
\begin{align} \label{eq:HthroughNchi}
	H=\frac12\sum_{\chi\in\widehat{\Gamma}}N_\chi\,.
\end{align}

The theory contains $T+1$ tensor fields: one self-dual tensor in the gravity multiplet (which will be assigned the index $0$) and one anti-self dual tensor in each of the $T$ tensor multiplets (indexed from $1$ through $T$ in the following).
They couple to (anti-)self-dual tensionless strings and the $SO(1,T)$-invariant intersection form on the string charge lattice is denoted by $\Omega^{\sia\sib}$ with inverse $\Omega_{\sia\sib}$, where $\alpha,\beta=0,\ldots,T$.

The six-dimensional effective action contains a Green-Schwarz term~\cite{Green:1984bx,GREEN1984117,Sagnotti:1992qw}
\begin{align}
	S^{\text{GS}}_{6\rm d}=-\frac14 a^\sia \Omega_{\sia\sib}B^\sib\text{tr}\,R^2\,,
\end{align}
coupling the tensor fields $B^\beta$ to the Ricci tensor $R$. This cancels the gravitational anomaly if the coefficients $a^\alpha$ satisfy $a^\sia \Omega_{\sia\sib}a^\sib=9-T$ and
\begin{align}
	H+29T=273\,.
    \label{eqn:gravitationalAnomaly}
\end{align}
In general, the action also contains topological terms that are necessary to cancel global anomalies.
For the case $\Gsd=\mathbb{Z}_3$, these have recently been studied in~\cite{Dierigl:2022zll}, see also~\cite{Monnier:2018nfs,Dierigl:2025rfn,Dierigl:wip}, but the details will not be relevant for our analysis.

\subsection{Gauge groups and charge lattices of the twisted compactifications}
\label{sec:GandLambdaTwisted}

We now want to study the gauge group and the charge lattice of the twisted circle compactification of an almost generic F-theory vacuum with $V=0$ and gauge group $\Gsd$.

First note that the gauge group of the untwisted circle compactification ${\rm F}[X^0][S^1_0]$ is
\begin{align}
    \Gfd^0=\U(1)^{T+1}\times\U(1)\times \Gsd\,.
\end{align}
The factor $\U(1)^{T+1}$ arises from dualizing the (anti-)self-dual tensor fields and the additional $\U(1)$ is the Kaluza-Klein gauge symmetry that one obtains after reducing the metric along the circle.

More generally, from our discussion in Section~\ref{sec:kkandhol}, we know that the gauge group of the twisted circle compactification ${\rm F}[X^0][S^1_\tse]$ is given by
\begin{align}
    \Gfd^\tse=\U(1)^{T+1}\times\left(\quotient{\U(1)\times \Gsd}{\langle(\e[-\tfrac{1}{m_\tse}],\tse)\rangle}\right)\,,
    \label{eqn:G5dgamma}
\end{align}
where $m_\tse=\text{ord}(\tse)$ denotes the order of $\tse\in\Gsd$.
The charge lattice of ${\rm F}[X^0][S^1_\tse]$ takes the form
\begin{align}
    \Lambda_{5\dd}^\tse=\text{ker}\left(\vec{0},\e[-\tfrac{1}{m_\tse}],\tse\right)\subset\mathbb{Z}^{T+1}\times \mathbb{Z}\times \widehat{\Gsd}\,,
    \label{eqn:5dchargelatticeMth}
\end{align}
in terms of the six-dimensional charge lattice $\widehat{\Gsd}=\text{Hom}(\Gsd,\U(1))$.
The Kaluza-Klein tower that is associated to a six-dimensional half-hypermultiplet with charge $\chi\in\widehat{\Gsd}$ consists of states with charges
\begin{align}
    \left(\vec{0},n m_\tse + q_\tse(\chi),\chi\right)\in\Lambda_{5\dd}^\tse\,,\quad n\in\mathbb{Z}\,.
    \label{eqn:5dKKtowerMth}
\end{align}
The mass of these states satisfies
\begin{align}
    M_{n,\chi}^\tse \propto \vert n m_\tse + q_\tse(\chi)\vert\,.
    \label{eqn:5dKKtowerMthMass}
\end{align}
We have introduced the notation $N_\chi$ for the number of half-hypermultiplets for each charge~\eqref{eqn:5dKKtowerMth}.

Recall from Section~\ref{sec:twistedKK} that the group~\eqref{eqn:G5dgamma} is isomorphic to
\begin{align} \label{eq:G5dgammaiso}
    \Gfd^\tse\simeq \U(1)^{T+1}\times \U(1)\times \Gsd/\langle\tse\rangle\,,
\end{align}
via a non-canonical isomorphism that depends on the choice of a lift of a certain character. However, as can be seen from~\eqref{eq:isoQuotientProduct}, the projection from $\Gfd^\tse$ to $\Gsd/\langle\tse\rangle$ is canonical and does not depend on this choice.
To isolate the discrete component of the gauge group, we can consider its zeroth fundamental group,
\begin{align} \label{eq:finite5d}
    \pi_0(\Gfd^\tse)=\Gsd/\langle\tse\rangle\,.
\end{align}
The charge lattice associated to this component,
\begin{align}
    \text{Hom}\left(\Gsd/\langle\tse\rangle,\U(1)\right)=\text{ker}\,\tse\subset\widehat{\Gsd}\,,
    \label{eqn:torsionCharges5dMth1}
\end{align}
coincides, via the isomorphism \eqref{eqn:5dchargelatticeMth}, with the torsion sublattice of the charge lattice,
\begin{align} 
    \Tors{\Lambda_{5\dd}}^{\tse}= \text{Hom}\left(\Gsd/\langle\tse\rangle,\U(1)\right)\,.
    \label{eqn:torsionCharges5dMth2}
\end{align}

\subsection{Twisted compactifications and genus one fibered Calabi-Yau threefolds} \label{ss:twistedAndGenusOne}
Let us now assume that the projective, elliptically fibered Calabi-Yau threefold $\pi_0:X^0\rightarrow B$, is almost generic in the sense of Definition~\ref{def:almostgeneric}.
Our goal in this section is to connect the discussion from Section~\ref{sec:GandLambdaTwisted} to the geometry of the fibrations $X^\tse$ for $\tse\in\Sh_B(X^0)$ by using the duality
\begin{align}
    {\rm F}[X^0][S^1_\tse]={\rm M}[X^\tse]\,.
    \label{eqn:dualityFM}
\end{align}
Note that six-dimensional hypermultiplets contain charge conjugate pairs of half-hypermultiplets. As a result, the compactifications ${\rm F}[X^0][S^1_\gamma]$ and ${\rm F}[X^0][S^1_{\gamma^{-1}}]$ are physically equivalent.
This corresponds to the fact that the elements $\gamma$ and $\gamma^{-1}$ can be represented by the same geometry $X^{\gamma}\simeq X^{\gamma^{-1}}$, see the discussion in Section \ref{ss:TSgroup}.
We will now use~\eqref{eqn:dualityFM} and Claims~\ref{claim:mtheoryNodalCYgaugeGroup} and \ref{claim:mtheoryNodalCYmatter} to arrive at Conjectures~\ref{conj:TSandTorsion} and~\ref{conj:I2fibers}.

As discussed in Section~\ref{sec:almostgenericfibrations}, we can choose  analytic small resolutions $\rho_\tse:\widehat{X}^\tse\rightarrow X^\tse$.
Then Claim~\ref{claim:mtheoryNodalCYgaugeGroup} states that the gauge group~\eqref{eqn:G5dgamma} is equal to
\begin{align}
    \Gfd^\tse=\text{Hom}(H_2(\widehat{X}^\tse,\mathbb{Z}),\U(1))\,,
\end{align}
such that the 5d charge lattice can be identified with
\begin{align} \label{eq:lambda5dH2}
    \Lambda_{5\dd}^\tse=H_2(\widehat{X}^\tse,\mathbb{Z})\,.
\end{align}
By \eqref{eqn:torsionCharges5dMth2}, the assumption that $\Gsd=\Sh_B(X^0)$
yields
\begin{align} 
    \text{Tors}\,\Lambda_{5\dd}^0= \text{Hom}\left(\Sh_B(X^0),\U(1)\right)\,,
\end{align}
from which an expression for $\Sh_B(X^0)$ immediately follows by use of the Pontryagin duality theorem, 
\begin{align}
    \Sh_B(X^0)=\text{Hom}\left(\tors{H_2(\widehat{X}^0,\mathbb{Z})},\U(1)\right)\,.
\end{align}
With the identification \eqref{eq:lambda5dH2}, the equations \eqref{eqn:torsionCharges5dMth1} and \eqref{eqn:torsionCharges5dMth2} imply 
\begin{equation}\tors{H_2(\widehat{X}^\tse,\mathbb{Z})}=\text{ker}\,\tse  \subset \tors{H_2(\widehat{X}^0,\mathbb{Z})} \,.
\end{equation}
This leads to parts~\ref{en:Sh} and~\ref{en:Gamma} of Conjecture~\ref{conj:TSandTorsion}. 

To justify Conjecture~\ref{conj:I2fibers}, we rely on Claim \ref{claim:mtheoryNodalCYmatter}, which specifies the localized charged spectrum of $M[X^\tse]$.

We will use the notation from Section~\ref{sec:mathresults}. In particular, we set 
\begin{align}
    \widehat{\Gamma}^{\tse}=\tors{H_2(\widehat{X}^{\tse},\mathbb{Z})}\,,\quad \Gamma^\tse=\text{Hom}(\widehat{\Gamma}^\tse,\U(1))\,,
\end{align}
such that by the above reasoning
\begin{equation}
    \Gsd = \Gamma^0 \,, \quad \pi_0(\Gfd^\tse) = \Gamma^\tse\,, \quad \mathrm{Tors}\, \Lambda_{5\dd}^\gamma = \widehat{\Gamma}^\gamma\,.
\end{equation}

From the perspective of the M-theory compactification ${\rm M}[X^\tse]$, the six-dimensional F-theory limit in which one recovers ${\rm F}[X^0]$ corresponds to the limit in the K\"ahler moduli space of $X^\tse$ where the volume of all fibers shrinks to zero volume.
The charged states that become massless in the six-dimensional decompactification limit therefore arise from M2-branes that wrap curves which are supported on individual fibers of the fibration (M2-branes wrapping curves with support in different fibers cannot lead to bound states). Our description of the 5d charge lattice in~\eqref{eqn:5dchargelatticeMth} allows us to determine the lattice of the homology classes of fibral curves. It is given by
\begin{align}
    H_2(\widehat{X}^\tse,\mathbb{Z})_{\rm F}=\text{ker}\left(\e[-\tfrac{1}{m_\tse}],\tse\right)\subset \mathbb{Z}\times \widehat{\Gamma}^0\,.
    \label{eqn:fibralCurvesChargeLattice}
\end{align}
Recall that the fibral curves are those that are orthogonal to all vertical divisors. This implies that the mass of the corresponding states does not depend on the vacuum expectation values of the tensor multiplet scalars.

Since the Kaluza-Klein tower is generated by adding M2-branes that wrap the generic fiber $\Sigma_\tse$ of the fibration $\widehat{X}^\tse$, and since all fields in the tower carry the same discrete charge, we find that
\begin{align} \label{eq:classOFgenericFiber}
    [\Sigma_\tse]=(m_\tse,\text{id})\in H_2(\widehat{X}^\tse,\mathbb{Z})_{\rm F}\,.
\end{align}
As pointed out below Conjecture \ref{conj:I2fibers}, this implies by Poincar{\' e} duality that there exists an $m_\tse$-section $E_\tse$ on $\widehat{X}^\tse$.
The only other source of fibral curves are components of the $I_2$-fibers in $\widehat{X}^\tse$ over the nodes $S_\Delta$ of the discriminant locus.
Recall that over each such node, the fiber $C_{p,\tse}=(\pi_\tse\circ\rho_\tse)^{-1}(p)$ takes the form $C_{p,\tse}=C_{p,\tse}^A\cup C_{p,\tse}^B$, where $C_{p,\tse}^A,C_{p,\tse}^B$ are two rational curves that intersect transversely in two points.
The corresponding homology classes satisfy
\begin{align}
    [C_{p,\tse}^A]+[C_{p,\tse}^B]=[\Sigma_\tse]\,.
    \label{eqn:homologyABfibralCurves}
\end{align}
If $X^\tse$ has a node over $p$, then we can identify $C_{\tse,p}^B$ with the exceptional curve in $\widehat{X}^\tse$ that resolves this node.

Consider a fixed hypermultiplet in the six-dimensional theory ${\rm F}[X^0]$ that contains two half-hypermultiplets in the non-trivial representations $\pm\chi\in\widehat{\Gamma}^0$.
The states in the corresponding Kaluza-Klein tower in $\F[X^0][S^1_\tse] = \M[X^\gamma]$ must be localized over one of the points $p\in S_{\Delta}$.

If $\tse(\chi)=1$, then, by the definition of $q_\tse(\chi)$ in the introduction to Section \ref{sec:mathresults}, $q_\tse(\chi)=0$, hence, by \eqref{eqn:5dKKtowerMthMass}, the Kaluza-Klein tower in $\M[X^\gamma]$ contains massless modes.
Therefore, $X^\tse$ must have a nodal singularity over $p$ and the homology class of the corresponding exceptional curve in $\widehat{X}^\tse$ is pure torsion.
As we have argued that the class of the generic fiber is given by \eqref{eq:classOFgenericFiber}, its decomposition~\eqref{eqn:homologyABfibralCurves} over a node $p$ implies that
\begin{align}
    [C_{p,\tse}^A]=(m_\tse,\mp\chi)\,,\quad [C_{p,\tse}^B]=(0,\pm\chi)\,,
    \label{eqn:FTHfibralCurveClassesUntwisted}
\end{align}
where the sign of $\chi$ depends on the particular choice of small resolution $\rho_\tse:\widehat{X}^\tse\rightarrow X^\tse$.
We define $\chi_p\in \widehat{\Gamma}^0$ for $p\in S_{\Delta}$ to be the torsion of the exceptional curve in $\widehat{X}^0$ that resolves the node over $p$, i.e. $[C_{p,0}^B]=(0,\chi_p)$.

On the other hand, if $\tse(\chi)\ne 1$, then none of the singularities of $X^\tse$ lie on the fibers over the nodes $p\in S_{\Delta}$ that support the corresponding states; the map $\rho_\tse:\widehat{X}^\tse\rightarrow X^\tse$ is hence an isomorphism in a neighborhood of $p$.
In order to be compatible with~\eqref{eqn:fibralCurvesChargeLattice} and~\eqref{eqn:homologyABfibralCurves} and to produce the correct Kaluza-Klein tower, the homology classes of the components of the fiber must take the form
\begin{align}
    [C_{p,\tse}^{A}]=\left((n+1)m_\tse-q_\tse(\chi),\,-\chi\right)\,,\quad [C_{p,\tse}^{B}]=\left(q_\tse(\chi)-n m_\tse,\,\chi\right)\,,
    \label{eqn:FTHfibralCurveClassesTwisted}
\end{align}
for some $n\in\mathbb{Z}$ and a suitable choice for the labels $A,B$.

In fact, we can show that $n=0$.
Otherwise, by $q_\gamma(\chi) \in \{0, \ldots, m_\gamma-1\}$, the images of the classes of the two components of the fibers in the free part of $H_2(\widehat{X}^\tse,\mathbb{Z})_{\rm F}$ would have opposite sign.
As will be further discussed below in Section~\ref{sec:partialHiggsing} around equation \eqref{eq:I2classesSmoothing}, this would be incompatible with the fact that the smoothing $\widetilde{X}^\tse$ of $X^\tse$ is a smooth projective variety.
Therefore,~\eqref{eqn:FTHfibralCurveClassesTwisted} with $n=0$ is the only possible choice. The equalities~\eqref{eqn:FTHfibralCurveClassesUntwisted} and ~\eqref{eqn:FTHfibralCurveClassesTwisted} conclude the physical derivation of Conjecture~\ref{conj:I2fibers}.

\subsection{Partial Higgsings and smooth deformations}
\label{sec:partialHiggsing}
Consider the theory ${\rm M}[X^\tse]$ for a given $\tse\in\Gamma^0$. As pointed out above equation \eqref{eqn:FTHfibralCurveClassesUntwisted}, each Kaluza-Klein tower of half-hypermultiplets in a representation $\chi\in\widehat{\Gamma}^0$ satisfying $\tse(\chi) = 1$ contains a massless element that has charge 0 under the Kaluza-Klein gauge symmetry.
In theories with 8 supercharges, the possibility to acquire vacuum expectation values is unobstructed also to higher order for the scalars in massless hypermultiplets that only carry a discrete gauge charge.\footnote{This is in contrast with hypermultiplets that are charged under a continuous gauge symmetry. In this case, $\mathcal{N}=2$ supersymmetry gives rise to a potential, see e.g.~\cite{Andrianopoli:1996cm}.}
We will now show that going to a generic point of the five-dimensional Higgs branch in fact breaks the discrete part of the gauge group $\Gfd$ completely and the resulting theory does not contain any charged massless hypermultiplets. Geometrically, this amounts to a complex structure deformation from $X^\tse$ to a smooth projective variety $\widetilde{X}^\tse$.
Since the relevant vacuum expectation values directly descend from those of scalar fields in the six-dimensional theory ${\rm F}[X^0]$, we expect that $\widetilde{X}^\tse$ exhibits a compatible genus one fibration $\tilde{\pi}_\tse:\widetilde{X}^\tse\rightarrow B$ which makes the following diagram commute:
\begin{center}
\begin{tikzpicture}[>=stealth, node distance=3cm, auto]
    \tikzstyle{block} = [text centered]

    \node[block] (top-left)                {$\F[X^0]$};
    \node[block] (top-right) [right of=top-left]   {$\F[J(\widetilde{X}^\gamma)]$};
    \node[block] (bottom-left) [below of=top-left]  {$\M[X^\gamma]$};
    \node[block] (bottom-right) [below of=top-right]{$\M[\widetilde{X}^\gamma]$};

    \draw[->] (top-left) -- (top-right)
    node[midway, above] {\tiny Higgsing} ;
    \draw[->] (bottom-left) -- (bottom-right)
    node[midway, above] {\tiny Higgsing} ;
    \draw[->] (top-left) -- (bottom-left)
    node[midway, left, align=center] {$S^1_{\tse}$} ;
    \draw[->] (top-right) -- (bottom-right)
    node[midway, right, align=center] {$S^1_{\tse}$} ;
    \path (top-left) -- (bottom-right) coordinate[midway] (center);
    \node at (center) {
  \begin{tikzpicture}[scale=1.5, >={Stealth[length=6pt, width=6pt]}]
    \draw[->, line width=0.6pt] (0,0) 
      arc[start angle=50, end angle=390, radius=0.5];
  \end{tikzpicture}
};
  \end{tikzpicture}
\end{center}
This can be seen as a physical version of~\cite[Theorem 31]{Kollar:2012pv}.
Note that the varieties $\widetilde{X}^\tse$ for $\gamma \in \Gamma^0$ no longer belong to the same Tate-Shafarevich group.\footnote{A more consistent notation would therefore be $\widetilde{X^\tse}$.} In particular, $\widetilde{X}^0$ is the generic Weierstrass fibration over a weak del Pezzo surface, hence does not exhibit any torsion in homology; by Conjecture \ref{conj:TSandTorsion}\ref{en:Sh}, its Tate-Shafarevich group is thus trivial.
From the upper arrow in the diagram, we read off that the relative Jacobian $J(\widetilde{X}^\tse)$ of $\widetilde{X}^\tse$ is a partial smoothing of $X^0$.
The geometry $\widetilde{X}^0$ is thus a smooth deformation of $J(\widetilde{X}^\tse)$.

For the untwisted compactification, $\gamma = \mathrm{id}$, hence $\chi(\gamma)=1$ for all $\chi$. The gauge group of ${\rm M}[\widetilde{X}^0]$ is thus given by
\begin{align}
    \tGfd^0=\bigcap\limits_{\substack{\chi\in\widehat{\Gamma}^0\\N_\chi\ne 0}}\ker\,(\vec{0},0,\chi)\subset \U(1)^{T+1}\times\U(1)\times \Gamma^0\,,
    \label{eqn:GTwistedtilde5d}
\end{align}
where $N_\chi$ was introduced above equation \eqref{eq:HthroughNchi}.
However, since the generic Weierstra{\ss} fibration over a weak del Pezzo surface does not have any torsion in homology, it does not give rise to a discrete gauge symmetry, i.e. $\tGfd^0\simeq \U(1)^{T+1}\times\U(1)$. Hence,
\begin{align}\bigcap\limits_{\substack{\chi\in\widehat{\Gamma}^0\\N_\chi\ne 0}}\ker\,\chi = \{1\}\quad \Rightarrow \quad \langle \,\chi\in\widehat{\Gamma}^0\,\vert\,N_\chi\ne 0\,\rangle = \widehat{\Gamma}^0\,.
    \label{eqn:subsetGamma0}
\end{align}
For arbitrary $\tse \in \Sh_B(X^0)$, the commutation of the above diagram requires that on a generic point of the Higgs branch of the six-dimensional theory ${\rm F}[J(\widetilde{X}^\tse)]$, all hyperscalars acquire a VEV whose five-dimensional Kaluza-Klein towers exhibit a massless mode. This is the case when $\tse(\chi)=1$. Hence, the unbroken gauge group of ${\rm F}[J(\widetilde{X}^\tse)]$ is given by
\begin{align}
    \tGsd^\tse=\bigcap\limits_{\substack{\chi\in\widehat{\Gamma}^0\\\tse(\chi)=1}}\ker\,\chi=\langle\tse\rangle\subset \Gamma^0\,.
\end{align}
Upon twisted circle compactification, according to the discussion leading to equation \eqref{eqn:Gdm1prime}, the gauge group descends to
\begin{align}
    \tGfd^\tse= \U(1)^{T+1}\times\left(\quotient{\U(1)\times \langle \tse\rangle}{\langle(\e[-\tfrac{1}{m_\tse}],\tse)\rangle}\right) \simeq \U(1)^{T+1}\times \U(1)\,.
    \label{eqn:GTwistedtilde5d}
\end{align}
This isomorphism (the first factor merely goes along for the ride) was established in Lemma \ref{lem:quotientgroup}; it acts as
\begin{align}
    u: \,(\vec{\xi},\,\phi,\,\gamma^k)\mapsto(\vec{\xi},\,\phi\,\e[\tfrac{k}{m_\tse}])\,.
\end{align}
The dual charge lattice takes the form
\begin{align}
    \widetilde{\Lambda}':=(u^{-1})^*\left({\widetilde{\Lambda}^{\tse}_{5\dd}}\right)=\mathbb{Z}^{T+1}\times\mathbb{Z}\,.
    \label{eqn:twistedHiggsedCharges}
\end{align}
By \eqref{eqn:upullback}, followed by the Higgsing which completely breaks the discrete component of the gauge group, we have a map
\begin{align}
    \Lambda^\tse_{5\dd}\rightarrow \widetilde{\Lambda}'\,,\quad \left(\vec{t},n,\chi\right)\mapsto \left(\vec{t},n\right)
    \label{eqn:HiggsedCharges}
\end{align}
that maps the charge of a state in ${\rm M}[X^\tse]$ to the charge of the image state after Higgs transition to ${\rm M}[\widetilde{X}^\tse]$. Combining~\eqref{eqn:5dKKtowerMth} with~\eqref{eqn:HiggsedCharges}, we see that the gauge charges of a Kaluza-Klein tower in ${\rm M}[\widetilde{X}^\tse]$ associated to a $6\dd$ hypermultiplet whose pre-image before Higgs transition had discrete charge $\chi\in\widehat{\Gamma}^0$ in ${\rm F}[X^0]$ take the form
\begin{align}
    \left(\vec{0},\,nm_\tse+q_\tse(\chi)\right)\in \widetilde{\Lambda}'\,,\quad n\in\mathbb{Z}\,.
    \label{eqn:higgsedKKtower}
\end{align}

Recall from the discussion leading to \eqref{eq:defQkk} that this form of the Kaluza-Klein charge, with $m_\tse$ the order of $\tse \in \widehat{\Gamma}^0$, holds under the hypothesis that the higher dimensional theory exhibits a field with charge $\chi$ satisfying $\gcd(q_{\tse}(\chi),m_\tse)=1$. This is indeed true in the case of $\F[J(\widetilde{X}^\gamma)]$. To see this, recall that the characters carried by fields in $\F[X^0]$ span $\widehat{\Gamma}^0$ by \eqref{eqn:subsetGamma0}. To reach a contradiction, assume that their restriction to $\langle \gamma \rangle$ does not span $\mathrm{Hom}(\langle\tse\rangle,\U(1))$. Any character $\langle \gamma \rangle$ by which this set needs to be augmented lifts, by Lemma \ref{lemma:liftingCharacters}, to a character of $\Gamma^0$ and is required to span $\widehat{\Gamma}^0$. However, the original set already spanned. We can hence conclude that the set
\begin{align}
    \{\,\chi_p\,\,\vert\,\,p\in S_{\Delta}\,,\,\, \tse(\chi_p)\ne 1\,\}
\end{align}
generates $\mathrm{Hom}(\langle\tse\rangle,\U(1))$. In other words, 
\begin{align} \label{eq:gcdmgammachis}
    \text{gcd}\left(m_\tse, \{\,q_\tse(\chi_p)\,\,\vert\,\,p\in S_{\Delta}\,,\,\, \tse(\chi_p)\ne 1\,\,\}\right)=1\,,
\end{align}
as was to be shown.

This discussion has implications for the structure of the genus one fibration of $\widetilde{X}^\tse$. As explained in Appendix~\ref{sec:nodalCY3}, the Mayer-Vietoris sequence gives an isomorphism
\begin{align}
    \quotient{H_2(\widehat{X}^\tse,\mathbb{Z})}{\langle\{\,[\rho_\tse^{-1}(p)]\,\vert\,p\in S^\tse\,\}\rangle}=H_2(\widetilde{X}^\tse,\mathbb{Z})\,,
\end{align}
where $S^\tse$ is the set of nodes in $X^\tse$. This quotient directly corresponds to the map~\eqref{eqn:HiggsedCharges}.
The last $\mathbb{Z}$ factor in~\eqref{eqn:twistedHiggsedCharges} can be identified with the lattice of homology classes of fibral curves
\begin{align}
    H_2(\widetilde{X}^\tse,\mathbb{Z})_{\rm F}=\mathbb{Z}\,.
\end{align}
Since the Kaluza-Klein tower should again be generated by multiply wrapping the generic fiber $\widetilde{\Sigma}_\tse$, \eqref{eqn:higgsedKKtower} implies that the divisor $\widetilde{E}_\tse$ associated to the Kaluza-Klein $\U(1)$ symmetry satisfies $\widetilde{E}_\tse \cdot \widetilde{\Sigma}_\tse = m_\gamma$. In particular, \eqref{eqn:higgsedKKtower} implies that $\widetilde{X}^\tse$ must exhibit an $m_\tse$-section $\widetilde{E}_\tse$. By \eqref{eq:gcdmgammachis}, this $m_\tse$ must be minimal, i.e. no $k$-section with $k<m_\tse$ may occur.

Let $p\in S_{\Delta}$ be one of the nodes that correspond to a representation $\chi_p\in\widehat{\Gamma}^0$ with $\tse(\chi_p)\ne 1$.
Then $\widetilde{X}^\tse$ has an $I_2$-fiber over $p$ that takes the form $\widetilde{C}_{p,\tse}=\widetilde{C}_{p,\tse}^A\cup \widetilde{C}_{p,\tse}^B$. By $[\widetilde{C}_{p,\tse}] = [\widetilde{\Sigma}_\tse]$ and~\eqref{eqn:higgsedKKtower}, we know that the homology classes of $\widetilde{C}_{p,\tse}^A$ and $\widetilde{C}_{p,\tse}^B$ must be of the form~\eqref{eqn:FTHfibralCurveClassesTwisted}, absent the discrete component. For any value of $n$ determining these classes other than $n=0$, the homology classes of the two holomorphic curves would have opposite signs in $H_2(\widetilde{X}^\tse,\mathbb{Z})_{\rm F}$, in contradiction with the fact that $\widetilde{X}^\tse$ is projective and therefore K\"ahler. We thus can identify the homology classes as
\begin{align} \label{eq:I2classesSmoothing}
    [\widetilde{C}_{p,\tse}^{A}]=m_\tse-q_\tse(\chi)\,,\quad [\widetilde{C}_{p,\tse}^{B}]=q_\tse(\chi)\,.
\end{align}
By~\eqref{eqn:HiggsedCharges}, we can conclude, as announced above, that also in~\eqref{eqn:FTHfibralCurveClassesTwisted}, we must set $n=0$. 

\subsection{Intersection numbers from Chern-Simons theory}
\label{ss:CSterms}
In this subsection, we will match the Chern-Simons terms in the effective actions of ${\rm F}[J(\widetilde{X}^\tse)][S^1_\tse]$ and ${\rm M}[\widetilde{X}^\tse]$ in order to arrive at Conjecture~\ref{conj:topologyOfSmoothing}.
Other references in which the F-theory and M-theory Chern-Simons terms have been compared are for example~\cite{Bonetti:2011mw,Bonetti:2013ela,Grimm:2013oga,Anderson:2014yva,Bhardwaj:2019fzv}.

Recall that 6d (1,0) theories require a Green-Schwarz term\footnote{The hypermultiplets charged under a discrete gauge group will in general also lead to global anomalies that have to be cancelled by a topological Green-Schwarz mechanism, see e.g.~\cite{Dierigl:2022zll}.}
\begin{equation} \label{eq:6dGS}
    S^{\mathrm{GS}}_{6\dd} = -\int \frac{1}{2} \Omega_{\alpha \beta} B^\alpha X_4^\beta
 \end{equation}
to cancel gravitational, gauge and mixed anomalies. Here, $\alpha$ and $\beta$ index the $T+1$ tensor fields in the theory. The symmetric matrix $\Omega_{\alpha \beta}$ together with two sets of coefficients $a^\alpha$ and $b^\beta_a$, with $a$ indexing simple gauge group factors, parametrize the anomaly polynomial $I_8$.\footnote{As we will re-derive below, this matrix coincides with its namesake introduced in Section~\ref{sec:genusonefibrations} when $S^{\mathrm{GS}}_{6\dd}$ arises from compactification of F-theory on a Calabi-Yau threefold. Until that point, $\Omega_{\alpha \beta}$ signifies the field theory quantity in this section.} 
For the local anomalies to allow cancellation via a Green-Schwarz term, $I_8$ must factorize as
\begin{equation}
    I_8 = \frac{1}{2} \Omega_{\alpha \beta} X_4^\alpha X_4^\beta \,,
\end{equation}
with 
\begin{equation}
    X_4^\alpha = \frac{1}{2} a^\alpha \tr R^2 + 2 \sum_a b_a^\alpha \tr F_a^2 \,.
\end{equation}
Crucially, the constants $\Omega_{\alpha \beta}$, $a^\alpha$ and $b^\alpha_a$ only depend, up to basis change, on the charged field content of the theory. In the absence of non-Abelian gauge symmetry, the sole constraint that the matrix $\Omega_{\alpha \beta}$ and the constants $a^\alpha$ must satisfy is
\begin{equation} \label{eq:OmegaA}
    a^\alpha \Omega_{\alpha \beta} a^\beta = 9-T \,.
\end{equation}

The (untwisted) circle reduction of 6d (1,0) theories obtained from F-theory via compactification on a smooth elliptically fibered Calabi-Yau threefold has been worked out in~\cite{Ferrara:1996wv,Bonetti:2011mw}. As we have argued for in Section~\ref{sec:partialHiggsing}, $\widetilde{X}^0$ is such a geometry. We will first review the discussion of~\cite{Bonetti:2011mw} pertaining to this case, focusing on the 5d Chern-Simons terms and in the absence of non-Abelian gauge symmetry, before generalizing to the twisted case.

The five-dimensional Chern-Simons terms are of the form
\begin{align}
    S_{5\dd}^{\text{CS}}=\int\left[-\frac12\Omega_{\alpha\beta}A^0F^{\alpha}F^{\beta}-\frac18 a^\alpha\Omega_{\alpha\beta}A^\beta\text{tr}\,R^2+k_0A^0F^0F^0+\lambda_0A^0\text{tr}\,R^2\right]\,,
	\label{eqn:ftheff}
\end{align}
where the gauge fields $A^\alpha$ arise from dualizing the tensor fields and $A^0$ denotes the Kaluza-Klein gauge field. The first two terms in \eqref{eqn:ftheff} originate from the six-dimensional kinetic terms of the tensor fields and from the Green-Schwarz term \eqref{eq:6dGS}.
The last two terms, on the other hand, are generated by one loop contributions after integrating out massive fields.

In~\cite{Bonetti:2013ela}, the coefficients $k_0$ and $\lambda_0$ are calculated by summing up the one loop contributions of the relevant Kaluza-Klein towers. This leads to the expressions
\begin{align}
\begin{split}
    k_0 =&\frac{1}{12} \sum_{k=1}^\infty k^3\big[ 2(V-H-T)* \underline{1} + 2*\underline{5} + (1-T)*(\underline{-4}) \big]\\
    =& \sum_{k=1}^\infty \frac{k^3 \big[T+3-H]}{6} =-\frac{9-T}{24}\,,\\
    \lambda_0=&\frac{1}{96}\sum_{k=1}^\infty k\big[ 2(V-H-T)*\underline{1} + 2*(\underline{-19}) + (1-T)*\underline{8} \big]\\
    =& - \sum_{k=1}^\infty \frac{k \big[10(T+3)+2H \big]}{96} = \frac{12-T}{24}\,.
\end{split}
\label{eqn:oneloopCS}
\end{align}
The sum over $k$ in both expressions ranges over the elements of the Kaluza-Klein towers, while the $k$-factors reflect the Kaluza-Klein $\U(1)$ charge of the $k^{\mathrm{th}}$ modes. The three underlined integers in the expression for $k_0$ and for $\lambda_0$ are determined by the representations of the states contributing to the Chern-Simons term under the little group in 5d, $\mathrm{SU}(2)\times \mathrm{SU}(2)$: they are associated to the representations $(\frac{1}{2},0)$, $(\tfrac{1}{2},1)$ and $(1,0)$ respectively. A sign multiplies the contribution of hyper- and tensor multiplets as the spinors/tensors in these multiplets transform in the conjugate representations $(0,\frac{1}{2})$ and $(0,1)$ respectively. 

To obtain the final expression for $k_0$ and $\lambda_0$,
the gravitational anomaly cancellation condition~\eqref{eqn:gravitationalAnomaly} has been implemented and zeta-function regularization $\sum k^s\rightarrow \zeta(-s)$ performed, with $\zeta(-3)=\tfrac{1}{120}$ and $\zeta(-1) = -\tfrac{1}{12}$. The intermediate expressions are obtained by setting $V=0$.
We have included them because in the twisted setting, the hyper and tensor multiplets exhibit different Kaluza-Klein charges and thus enter with different $k$-dependent coefficients in the sum.

The sought after relation to intersection numbers of the compactification manifold emerges upon considering the effective Lagrangian of the M-theory compactification on a smooth Calabi-Yau threefold.
The Chern-Simons terms arising from this compactification take the form~\cite{Cadavid:1995bk,Antoniadis:1997eg}
\begin{align}
    S_{5\dd}^{\text{CS}}=\int\left[-\frac16k_{ijk} A^i F^j F^k+\frac{1}{96}\kappa_i A^i\text{tr}\,R^2\right]\,,
	\label{eqn:mtheff}
\end{align}
where $k_{ijk}$ are the triple intersection numbers of an appropriate basis of divisors, and $\kappa_i$ their intersections with the second Chern class of the compactification manifold.

Comparing \eqref{eqn:ftheff} and \eqref{eqn:mtheff} leads to the identification
\begin{align}
	k_{000}=\frac{9-T}{4}\,,\quad k_{i\alpha\beta}=\delta_{i,0}\Omega_{\alpha\beta}\,,\quad k_{00\beta}=0\,,
        \quad \kappa_0=4(12-T)\,,\quad  \kappa_\beta=-12\Omega_{\alpha\beta}a^\beta\,,\label{eqn:untwistedKijkKappa}
\end{align}
while the condition \eqref{eq:OmegaA}, ensuring factorization of the anomaly polynomial, expressed in terms of intersection numbers becomes 
\begin{equation} \label{eq:factorizationIntersection}
    \kappa_\alpha\Omega^{\alpha\beta}\kappa_\beta=144(9-T) \,,
\end{equation}
with $\Omega^{\alpha \beta}$ denoting the inverse of the matrix $\Omega_{\alpha \beta}$.

Based on the Shioda-Tate-Wazir theorem generalized to smooth genus one fibered Calabi-Yau threefolds, we have introduced a basis $\{J_0, J_\alpha \}$ of divisors in \eqref{eq:STWbasis}. In the case of $\widetilde{X}^0$, we require only the standard version of this theorem, for which $J_0$ is a section and $\{J_\alpha\}$ are $b_2(B)$ vertical divisors  of the smooth elliptic Calabi-Yau threefold. This leads to the identification
\begin{align}
    T=b_2(B)-1 \,.
\end{align}
We impose the constraint $k_{00\beta} = 0$ by shifting the section $J_0$ by appropriate vertical divisors. We thus arrive at the new basis
\begin{align} \label{eq:shiftJtoDuntwisted}
    D_0=J_0-\frac12\tilde{\pi}_0^{*}\circ\tilde{\pi}_{0*}(J_0\cdot J_0)=J_0-\frac12\tilde{\pi}_0^{*}K_B\,,\quad D_\alpha=J_\alpha\,.
\end{align}
Upon the identification
\begin{align} \label{eq:MmathID}
    \Omega_{\alpha\beta}=\Jbase_\alpha\cdot \Jbase_\beta\,,\quad a^\alpha=\Omega^{\alpha\beta}\Jbase_\beta\cdot K_B\,,
\end{align}
this basis satisfies all constraints in~\eqref{eqn:untwistedKijkKappa}, as can be verified via the relations reviewed in Section~\ref{sec:genusonefibrations}.
Note that the second equality in~\eqref{eq:MmathID} implies ${a^\alpha = c^\alpha}$ in the notation of~\eqref{eq:expansionKB}.
Recalling that the constraint~\eqref{eq:factorizationIntersection} is a reformulation of~\eqref{eq:OmegaA}, we see that it follows immediately from~\eqref{eq:cOc} and the expression for $K_B \cdot K_B$ in \eqref{eq:Noether} following from Noether's formula.

Note that at least for theories with $0\le T<9$, the existence of a divisor $D$ on $\widetilde{X}^0$ that satisfies the conditions~\eqref{eqn:fibcondition} is guaranteed by~\eqref{eqn:untwistedKijkKappa} and therefore we could already deduce from the Chern-Simons terms that $\widetilde{X}_0$ has to exhibit an elliptic fibration.

\bigskip

We will now adapt this analysis to the twisted compactifications from 6d to 5d which arise upon compactifying M-theory on the resolved geometry $\widetilde{X}^\tse$.

In the presence of the discrete Wilson line along the circle, the values of $k_0$ and $\lambda_0$ change due to the shift of the Kaluza-Klein charges of the fields that propagate in the loop, but the structure of the five-dimensional Chern-Simons terms~\eqref{eqn:ftheff} remains the same.

The Kaluza-Klein tower associated to a state that in the six-dimensional theory ${\F}[X^0]$ before the Higgs transition carried charge $\chi\in\widehat{\Gamma}^0$ is given in~\eqref{eqn:higgsedKKtower}. Recall that these charges are with regard to a rescaled Kaluza-Klein gauge field $\tilde{A}^0=\frac{1}{m_\tse}A^0$, see the discussion below equation \eqref{eq:defQkk}. With this normalization, the Chern-Simons terms take the form
\begin{align}
        S_{5\dd}^{\text{CS}}=\int\left[-\frac{m_\tse}{2}\Omega_{\alpha\beta}\tilde{A}^0F^{\alpha}F^{\beta}-\frac18 a^\alpha\Omega_{\alpha\beta}A^\beta\text{tr}\,R^2+\tilde{k}_0\tilde{A}^0\tilde{F}^0\tilde{F}^0+\tilde{\lambda}_0\tilde{A}^0\text{tr}\,R^2\right]\,.
	\label{eqn:ftheffTwisted}
\end{align}
Taking into account the shifted KK-charges, the expressions in~\eqref{eqn:oneloopCS} are  modified to
\begin{align}
	\begin{split}
        \tilde{k}_0=&\frac16\sum\limits_{k=0}^\infty \left[(km_\tse)^3\left(T+3\right)-\frac12\sum\limits_{\chi\in\widehat{\Gamma}^0}\left(k m_\tse +q_\tse(\chi)\right)^3N_\chi\right]\,,\\
		\tilde{\lambda}_0=&-\frac{1}{96}\sum\limits_{k=0}^\infty \left[10km_\tse (T+3)+\sum\limits_{\chi\in\widehat{\Gamma}^0}\left(k m_\tse+q_\tse(\chi)\right)N_\chi\right]\,.
	\end{split}
	\label{eqn:oneloopCStwisted}
\end{align}
In order to regularize the sums, we will use the Hurwitz zeta function
\begin{align}
	\zeta(s,a)=\sum\limits_{n=0}^\infty\frac{1}{(n+a)^s}\,.
\end{align}
For positive $n$, $\zeta(-n,a)=-B_{n+1}(a)/(n+1)$ in terms of the Bernoulli polynomials $B_n(x)$. These satisfy the useful identities
\begin{align}
	\begin{split}
	B_2\left(1-x\right)+B_2\left(x\right)=\frac13+2x(x-1)\,,\quad B_4\left(1-x\right)+B_4\left(x\right)=-\frac{1}{15}+2x^2(x-1)^2\,.
	\end{split}
\end{align}
We then find
\begin{align}
	\begin{split}
		\tilde{k}_0=&\frac{m_\tse^3}{6}\left[\frac{T+3}{120}-\frac14\sum\limits_{\chi\in\widehat{\Gamma}^0}\left(\zeta\left(-3,\frac{q_\tse(\chi)}{m_\tse}\right)+\zeta\left(-3,1-\frac{q_\tse(\chi)}{m_\tse}\right)\right)N_\chi\right]\\
		=&\frac{m_\tse^3}{6}\left[\frac{T+3}{120}-\frac{1}{16}\sum\limits_{\chi\in\widehat{\Gamma}^0}\left(\frac{1}{15}-\frac{2}{m_\tse^4}q_\tse(\chi)^2(m_\tse-q_\tse(\chi))^2\right)N_\chi\right]\\
		=&\frac{m_\tse^3}{6}\left[\frac{T-H+3}{120}+\frac{1}{8m_\tse^4}\sum\limits_{\chi\in\widehat{\Gamma}^0}q_\tse(\chi)^2(m_\tse-q_\tse(\chi))^2N_\chi\right]\\
		=&m_\tse^3\left[\frac{T-9}{24}+\frac{1}{48m_\tse^4}\sum\limits_{\chi\in\widehat{\Gamma}^0}q_\tse(\chi)^2(m_\tse-q_\tse(\chi))^2N_\chi\right]\,,
	\end{split}
\end{align}
where we have again used the gravitational anomaly cancellation $H=273-29T$,
and similarly
\begin{align}
	\begin{split}
		\tilde{\lambda}_0=&m_\tse\left[\frac{12-T}{24}-\frac{1}{192m_\tse^2}\sum\limits_{\chi\in\widehat{\Gamma}^0}q_\tse(\chi)\left(m_\tse-q_\tse(\chi)\right)N_\chi\right]\,.
	\end{split}
	\label{eqn:oneloopCStwisted2}
\end{align}
In order to match~\eqref{eqn:mtheff} and~\eqref{eqn:ftheffTwisted}, we must have
\begin{align}
	k_{000}=\frac{m_\tse^3(9-T)}{4}-\frac{1}{8m_\tse}\sum\limits_{\chi\in\widehat{\Gamma}^0}q_\tse(\chi)^2\left(m_\tse-q_\tse(\chi)\right)^2N_\chi\,,\quad k_{i\alpha\beta}=\delta_{i,0}m_\tse\Omega_{\alpha\beta}\,,\quad k_{00\beta}=0\,,
	\label{eqn:twistedK}
\end{align}
as well as
\begin{align}
	\kappa_0=4m_\tse(12-T)-\frac{1}{2m_\tse}\sum\limits_{\chi\in\widehat{\Gamma}^0}q_\tse(\chi)\left(m_\tse-q_\tse(\chi)\right)N_\chi\,,\quad \kappa_\beta=-12\Omega_{\alpha\beta}a^\beta\,.
	\label{eqn:twistedKappa}
\end{align}
The condition \eqref{eq:factorizationIntersection} expressed in terms of intersection numbers is not modified; it remains $\kappa_\alpha\Omega^{\alpha\beta}\kappa_\beta=144(9-T)$, as in the untwisted case.

We have argued in Section~\ref{sec:partialHiggsing} that $\widetilde{X}^\tse$ exhibits a genus one fibration $\tilde{\pi}_\tse:\widetilde{X}^\tse\rightarrow B$ with an $m_\tse$-section $J_0$.
We now require the Shioda-Tate-Wazir theorem, as generalized to genus one fibrations in~\cite[Section 8]{Braun:2014oya}, to again conclude, based on \eqref{eq:STWbasis}, that $T = b_2(B) -1$. We perform the analogous shift to \eqref{eq:shiftJtoDuntwisted} to satisfy the constraint $k_{00\beta}=0$, and arrive at the basis of divisors
\begin{align} \label{eq:basisDiv}
    D_0=J_0-\frac{1}{2m_\tse}\tilde{\pi}_\tse^{*}\circ\tilde{\pi}_{\tse*}(J_0\cdot J_0)\,,\quad D_\alpha=J_\alpha\,.
\end{align}
The identifications
\begin{align}
   \Omega_{\alpha\beta}=\Jbase_\alpha\cdot \Jbase_\beta\,,\quad a^\alpha=\Omega^{\alpha\beta}\Jbase_\beta\cdot K_B
\end{align}
yield $k_{i \alpha \beta}$ as given in~\eqref{eqn:twistedK}, as well as  
 $\kappa_\beta$ as given in~\eqref{eqn:twistedKappa} for those cases where the proofs of~\eqref{eqn:intVerticalSecondChern} from~\cite[Appendix D]{Grimm:2013oga},~\cite[Appendix B]{Cota:2019cjx} or the argument given in Section~\ref{sec:genusonefibrations} apply.
On the other hand, the values of $k_{000}$, $\kappa_0$ and for $\kappa_\beta$ in full generality are physical predictions for the topological invariants of $\widetilde{X}^\tse$.

To determine the Euler characteristic of $\widetilde{X}^\tse$, we can consider the gravitational anomaly cancellation condition~\eqref{eqn:gravitationalAnomaly} and write it as
\begin{align}
    H_U=273-29T-H_C\,,
\end{align}
where $H_C$ and $H_U$ respectively denote the number of charged und uncharged hypermultiplets in ${\rm F}[J(\widetilde{X}^\tse)]$, such that $H=H_U+H_C$.
The number of charged hypermultiplets is given by $H_C=\widetilde{N}_\tse/2$, with
\begin{align}
    \widetilde{N}_\tse=\sum\limits_{\substack{\chi\in\widehat{\Gamma}^0\\\tse(\chi)\ne 1}}N_\chi\,.
\end{align}
On the other hand, the number of uncharged hypermultiplets is related to the complex structure deformations of $\widetilde{X}^\tse$ via $H_U=h^{2,1}(\widetilde{X}^\tse)+1$.
Using also $T=b_2(B)-1=b_2(\widetilde{X}^\tse)-2$, we obtain
\begin{align}
    \chi(\widetilde{X}^\tse)=60\big(b_2(B)-10\big)+\widetilde{N}_\tse\,.
\end{align}
Together with the predictions~\eqref{eqn:twistedK} and~\eqref{eqn:twistedKappa} based on comparing the Chern-Simons terms, we have thus arrived at Conjecture~\ref{conj:topologyOfSmoothing}.

Let us again point out that, at least for $0\le T<9$, the topological invariants~\eqref{eqn:twistedK} and~\eqref{eqn:twistedKappa} already imply the existence of a genus one fibration on $\widetilde{X}^\tse$ via the results of~\cite{Wilson1989}.
The existence of an $m_\tse$-section can then also be deduced from $k_{0\alpha\beta}=m_\tse\Omega_{\alpha\beta}$ .

\subsection{Do all almost generic vacua arise from almost generic fibrations?}
\label{sec:genericity}
Up to this point, we have motivated the first three of the conjectures formulated in Section \ref{sec:mathresults}, concerning the geometric properties of the almost generic genus one fibered Calabi-Yau threefolds $X^\tse$, by studying the properties of the corresponding F- and M-theory vacua.
In this section, we will discuss under which circumstances a six-dimensional almost generic F-theory vacuum, in the sense of Definition~\ref{def:almostGenericFtheory}, with the number of vector multiplets $V=0$, can take the form ${\rm F}[X^0]$ for some almost generic elliptic Calabi-Yau threefold $X^0$.
While we will not arrive at a complete set of necessary and sufficient conditions, we will collect relevant mathematical and physical considerations and highlight the potential loopholes that might obstruct a one-to-one correspondence between such F-theory vacua and almost generic elliptic Calabi-Yau threefolds.

As our starting point, let $X^0$ be any projective Calabi-Yau threefold with an elliptic fibration $\pi_0:X^0\rightarrow B$ leading to an almost generic F-theory vacuum $\F[X^0]$ without massless vector multiplets.
Since ${\rm F}[X^0]$ depends only on the axio-dilaton profile encoded in $X^0$, we can assume that the fibration is in Weierstra{\ss} form.
As we have already discussed in the introduction of this section, it is expected that the gauge group of ${\rm F}[X^0]$ is the Weil-Ch{\^a}telet group $\Gamma^0=\text{WC}_B(X^0)$.
This is because the elements of $\text{WC}_B(X^0)$ essentially correspond to the inequivalent genus one fibrations that encode the same axio-dilaton profile.
The twisted circle compactification ${\rm F}[X^0][S^1_\tse]$ is then expected to be dual to ${\rm M}[X^\tse]$, where $X^\tse$ is the genus one fibration that represents the element $\gamma\in\text{WC}_B(X^0)$.

\paragraph{Condition~\ref{condge:smoothqfac}}
Since we assume that the theory ${\rm F}[X^0]$ is at a sufficiently generic point of the tensor branch, we can assume that the base $B$ of the fibration is smooth.
Furthermore, $X^0$ must be $\IQ$-factorial. Else, it admits a projective small resolution~\cite[Corallary 4.5]{Kawamata1988}. The physical interpretation of such resolutions, however, is going to a generic point of the Coulomb branch of a continuous five-dimensional gauge symmetry.
Since the base of the fibration is already smooth, this gauge symmetry would have to descend from six-dimensions, in contradiction to the assumption that the six-dimensional gauge group is finite.

The question whether $X^0$ contains at most isolated nodes as singularities is closely related to Condition~\ref{condge:disc} that we will discuss below.
Assuming that Condition~\ref{condge:disc} holds, the potential singularities of $X^0$ are discussed in Appendix~\ref{sec:singularFibers}.
The upshot is that $X^0$ will have at most isolated nodes in the $I_2$-fibers over nodes of the discriminant, but it could potentially also have isolated $A_2$-singularities over some of the cusps of the discriminant. These locally are of the form
\begin{align} \label{eq:A2}
    \{\,z_1^2+z_2^2+z_3^2+z_4^n=0\,\}\subset\mathbb{C}^4 \,, \quad n=3 \,.
\end{align}
As $n$ is odd, such $A_2$-singularities do not admit any small resolution \cite[Corollary 1.16]{ReidMinMod}. Their presence implies the existence of localized uncharged matter \cite{Arras:2016evy}.\footnote{In \cite{Arras:2016evy}, the existence of such hypermultiplets is demonstrated explicitly in examples where such singularities arise as enhancements of type $I_1$ and type $II$ singularities.}
Giving a vacuum expectation value to the corresponding scalars removes the singularities. It follows that if $X^0$ satisfies Condition~\ref{condge:disc}, it must be smooth over each of the cusps of the discriminant.

\paragraph{Condition~\ref{condge:disc}}
Let us now discuss the singular fibers of $X^0$.
Due to the absence of any continuous gauge symmetry in ${\rm F}[X^0]$, we can assume that the fibers over the generic points of any one-dimensional component of the discriminant locus $\Delta$ are of type $I_1$ or $II$, see e.g.~\cite[Table 4.1]{Weigand:2018rez}.
Fibrations with a one-dimensional component of the discriminant locus over which the generic fiber is of type $II$ were discussed, for example, in~\cite{Arras:2016evy,Grassi:2018rva}.

Let us write $\Delta=\Delta_{I_1}\cup \Delta_{II}$, where $\Delta_{I_1}$ and $\Delta_{II}$ are the one-dimensional components of $\Delta$ over which the generic fibers are respectively of type $I_1$ and $II$.
We want to argue that $\Delta_{II}$ is empty, such that the only type $II$ fibers of the fibration that exist are over isolated points of $\Delta_{I_1}$.

First note that a non-trivial component $\Delta_{II}$ can only exist if the Weierstra{\ss} coefficients factorize as $f=h f'$, $g=h g'$ and neither $f',g'$ nor $g',h$ have a factor in common. 

Let us first discuss the situation where $f \not\equiv 0$. Then $d \sim h^2$ and ${\Delta_{II}=\{\,h=0\,\}}$. Also, $g'$ cannot be constant, else $h \in \Gamma(B,\cL^6)$ would force $f' \in \Gamma(B,\cL^{-2})$, in contradiction to $\cL = \omega_B^{-1} \neq \cO_B$ for $B$ a weak del Pezzo surface.
Thus, one necessarily has $\Delta_{I_1}\ne \emptyset$.
In examples, e.g. those discussed in~\cite[Section 4.2]{Arras:2016evy} and~\cite[Appendix C.1]{Arras:2016evy}, the component $\Delta_{II}$ then supports isolated fibers of type $III$ and/or $IV$.
These arise in particular over the points $\Delta_{I_1}\cap \Delta_{II}$.
Either of these fiber types corresponds to a $\mathbb{Q}$-factorial terminal singularity of the form \eqref{eq:A2} in the Weierstra{\ss} fibration, implying the presence of localized uncharged hypermultiplets~\cite{Arras:2016evy}.
Giving a vacuum expectation value to the corresponding scalar fields corresponds to a partial smoothing of the fibration that removes these singularities. We therefore expect, if the vacuum expectation values of the uncharged hypermultiplet scalars are sufficiently generic, that $\Delta_{I_1}\cap\Delta_{II}=\emptyset$.
We have already argued that $\Delta_{I_1}\ne \emptyset$.
It would be interesting to either exclude the situation where $\Delta_{II}$ is a non-trivial curve that does not intersect $\Delta_{I_1}$ or to study the physical implications of such a configuration.

We are left with the possibility that $f$ vanishes identically.
This would lead to an isotrivial fibration with $\Delta=\Delta_{II}$.
If $g$ is a generic section of $\mathcal{L}^{\otimes 6}$, then $\{\,g=0\,\}$ is a smooth curve, every singular fiber of $X^0$ is of type $II$ and $X^0$ itself is smooth.
In particular, the F-theory vacuum in this case exhibits only uncharged hypermultiplets; a generic choice of the corresponding vacuum expectation values then deforms the geometry to a generic Weierstra{\ss} fibration over $B$.
However, if $g$ is not generic, the situation is more complicated and we will not try to analyze it in this paper.

Even assuming that $\Delta=\Delta_{I_1}$, we still have to argue that, without affecting the charged spectrum of the corresponding F-theory compactification, one we always deform the complex structure such that $\Delta_{I_1}$ is irreducible and has only isolated nodes and cusps as singularities.
Again, we will leave a more careful analysis of this question to future work.

\paragraph{Condition~\ref{condge:TSWC}}
Let us now discuss why we expect that
\begin{align}
    \Sh_B(X^0)=\text{WC}_B(X^0)\,.
    \label{eqn:ShEqualsWC}
\end{align}
The obstruction to this equality would be the existence of genus one fibrations $X^\tse$ that represent elements of $\text{WC}_B(X^0)$ and contain fibers that are either non-reduced -- so-called multiple fibers -- or non-flat, that is of dimension greater than one.

We will first argue that the fibrations $X^\gamma$ do not have any non-flat fibers.
The physics associated with non-flat fibers has been discussed, for example in~\cite{Lawrie:2012gg,Braun:2013nqa,Borchmann:2013hta,Buchmuller:2017wpe,Apruzzi:2018nre,Dierigl:2018nlv,Oehlmann:2019ohh,Apruzzi:2019opn,Apruzzi:2019enx}.
Non-flat fibers in a genus one fibered Calabi-Yau threefold $\pi:X\rightarrow B$ appear when the dimension of the fiber over a point $p\in B$ jumps to $\text{dim}\,\pi^{-1}(p)=2$.
Since in the F-theory limit of ${\rm M}[X]$, the surface $\pi^{-1}(p)$ shrinks to a point, the M5-branes wrapped on that surface correspond to tensionless strings in the six-dimensional theory.
In six dimensions, tension can be generated for the strings by giving a vacuum expectation value to the scalar in the corresponding tensor multiplet that couples to the string.
Geometrically, this amounts to blowing up the base of the fibration.
Since we have restricted our analysis to six-dimensional F-theory vacua ${\rm F}[X^0]$ with a generic choice for the tensor multiplet scalars, they do not contain any tensionless strings.
We can therefore assume that the corresponding torus fibrations $X^\gamma$ are flat.

We also expect that the fibrations do not contain any multiple fibers. Such fibers  have been discussed in the context of F-theory in~\cite{deBoer:2001wca,Bhardwaj:2015oru,Anderson:2018heq,Oehlmann:2019ohh,Anderson:2023wkr,Anderson:2023tfy,Ahmed:2024wve}.
One type of multiple fibers arises in fibrations over a base $B$ that exhibits orbifold singularities, see e.g.~\cite{Bhardwaj:2015oru,Anderson:2018heq,Anderson:2019kmx,Kohl:2021rxy}.
Physically, this leads to an enhancement of $\pi_0(\Gsd)$ that arises only at a special sublocus of the tensor branch.
We have again excluded this possibility by assuming that we are on a sufficiently generic point of the tensor branch.

Another origin of multiple fibers are compactifications which involve an outer automorphism of the \textit{affine} Dynkin diagram of a non-Abelian gauge symmetry.
Such examples have been discussed, for example, in~\cite[Appendix B]{Oehlmann:2019ohh} and~\cite{Bhardwaj:2019fzv,Kim:2019dqn,Braun:2021lzt,Kim:2021cua,Bhardwaj:2022ekc,Lee:2022uiq,Anderson:2023wkr}.
Since we assume that the gauge group of ${\rm F}[X^0]$ is finite, we expect that this second type of multiple fibers is again absent in the geometries $X^\gamma$.

In~\cite[Appendix B]{Oehlmann:2019ohh}, a class of genus one fibrations was discussed that exhibits multiple fibers over isolated points on a 1-dimensional component of the discriminant locus over which the fiber is generically of type $I_2$.
The fibrations also exhibit an additional fibral divisor and the corresponding six-dimensional gauge symmetry was argued to be $\Gsd={\rm SU}(2)\times\mathbb{Z}_4$.
After performing a Higgs transition that breaks the gauge group to $\mathbb{Z}_4$, the multiple fibers disappear.
This leads us to believe that this type of multiple fiber can only appear in the presence of a non-trivial continuous gauge symmetry.

We are not aware of any other situation in which multiple fibers can arise.
Without any non-flat or multiple fibers in any of the fibrations $X^\gamma$ for $\tse\in\text{WC}_B(X^0)$, we expect that all of these fibrations are locally elliptic and therefore we expect~\eqref{eqn:ShEqualsWC} to hold.

\section{Discrete holonomies in M-theory and twisted-twined elliptic genera}
\label{sec:ellipticGeneraAndTopologicalStrings}
Let $X^0$ be an almost generic elliptic Calabi-Yau threefold in the sense of Definition \ref{def:almostgeneric}.
The F-theory vacuum ${\rm F}[X^0]$ then exhibits a finite Abelian gauge group $\Gsd=\Gamma^0$, where $\Gamma^0=\Sh_B(X^0)$.
In the previous section, we studied the twisted circle compactifications ${\rm F}[X^0][S^1_\tse]$ for $\tse\in\Gamma^0$.
The duality with M-theory incorporates the twisting into the compactification threefold,
\begin{align}
    {\rm F}[X^0][S^1_\tse]={\rm M}[X^\tse]\,,
\end{align}
where $X^\tse$ is the genus one fibered Calabi-Yau threefold that is associated to the element $\tse\in\Sh_B(X^0)$.

If we compactify on another circle, the resulting four-dimensional theory is dual to Type IIA string theory on $X^\tse$, i.e.
\begin{align}
    {\rm F}[X^0][S^1_\tse][S^1]={\rm M}[X^\tse][S^1]={\rm IIA}[X^\tse]\,.
    \label{eqn:FtheoryIIAuntwisted}
\end{align}
This naturally leads to the question of what happens on the Type IIA side of the duality if we twist the compactification along both of the circles.

It was argued in~\cite{Schimannek:2021pau,Dierigl:2022zll,Katz:2022lyl,Katz:2023zan} that given two elements $\tseOne,\tseTwo\in\Sh_B(X^0)$, the generalization of~\eqref{eqn:FtheoryIIAuntwisted} takes the form
\begin{align} 
    {\rm F}[X^0][S^1_{\tseOne}][S^1_{\tseTwo}]= {\rm M}[X^{\tseOne}][S^1_{\tseTwo}]={\rm IIA}[X^{\tseOne}_{\tseTwo}]\,,
    \label{eqn:FtheoryIIAtwisted}
\end{align}
where $X^{\tseOne}_{\tseTwo}$ refers to the Calabi-Yau threefold $X^\tseOne$ together with a B-field topology on $X^\tseOne$ that corresponds to the element $\tseTwo$.
We will review in Section~\ref{sec:Bfields} that the topology of a flat B-field on $X^\tseOne$ is conjecturally classified by elements of $\Gamma^\tseOne=\tors{H^3(\widehat{X}^\tseOne,\mathbb{Z})}$.
By invoking Conjecture~\ref{conj:TSandTorsion}\ref{en:Gamma}, we also have $\Gamma^\tseOne=\Gamma^0/\langle\tseTwo\rangle$ and can indeed associate to $\tseTwo$ a B-field topology $[\tseTwo]\in\Gamma^\tseOne$.
Recall from the discussion above Conjecture~\ref{conj:twisted} that $\tors{H^3(\widehat{X}^\tseOne,\mathbb{Z})}$, and therefore $\Gamma^\tseOne$, can also be identified with the cohomological Brauer group $\text{Br}(\widehat{X}^\tseOne)$.

In this section, we will argue for the following three claims, which can be seen as consequences of the duality \eqref{eqn:FtheoryIIAtwisted} and generalize earlier results from~\cite{Cota:2019cjx,Schimannek:2021pau,Dierigl:2022zll,Katz:2022lyl}:
\begin{claim}\label{claim:ellipticGtopZ}
    The A-model topological string partition function at non-trivial base degree $\beta$
    on $X^\tseOne_\tseTwo$ is related by a change of coordinates to the twisted-twined elliptic genus
    $\IE{\tiny\left[\begin{array}{c}\tseOne\\\tseTwo\end{array}\right]}_\beta(\tau,\lambda)$.
\end{claim}
\begin{claim}\label{claim:monodromy}
    In the large base limit of the stringy K\"ahler moduli space of $X^\tseOne_\tseTwo$, the monodromy group is given by the subgroup $\Gamma^{\tseOne}_{\tseTwo}$ of the modular group that leaves the twisting and twining elements of the elliptic genus invariant.
\end{claim}
\begin{claim}\label{claim:cusps}
    The cusps of the modular curve $\SLtwoZ/\Gamma^{\tseOne}_{\tseTwo}$ correspond to MUM points associated to the geometries $X^\tseOnePrime_\tseTwoPrime$ in the $\SLtwoZ$ orbit of $(\tseOne,\tseTwo)
    $.
\end{claim}

We will first discuss the twisted-twined elliptic genera and their modular properties in Section~\ref{sec:twistedmodularity}.
Then, in Section~\ref{sec:Bfields}, we will review flat but topologically non-trivial B-fields on Calabi-Yau threefolds with at most isolated nodal singularities and discuss how the notion of a complexified K\"ahler class extends to these singular spaces.
We will substantiate Claim~\ref{claim:ellipticGtopZ} in Section~\ref{sec:modularityTopString} by working out the detailed map between the parameters of the elliptic genus and the topological string partition function, filling in some details that have hitherto not appeared in the literature.
Claim~\ref{claim:cusps} and its relation to Conjecture~\ref{conj:twisted} will be discussed in Section~\ref{sec:modularityAndDerivedEquivalences}.

Note that for Claim \ref{claim:cusps} to be consistent with Claim \ref{claim:monodromy}, we must have $\Gamma^{\tseOne}_\tseTwo \cong \Gamma^\tseOnePrime_{\tseTwoPrime}$ whenever $(\tseOne,\tseTwo)$ and $(\tseOnePrime,\tseTwoPrime)$ lie in the same $\SLtwoZ$ orbit. This follows from the fact that conjugate subgroups are isomorphic.

\subsection{Twisted boundary conditions and the modular group}
\label{sec:twistedmodularity}

From a Type IIB perspective, the non-critical strings in ${\rm F}[X^0]$ arise from D3-branes that wrap holomorphic curves in the base $B$ of the Weierstra{\ss} fibration $X^0$. The worldsheet theory of the resulting effective string is a two-dimensional $(0,4)$-SCFT. The string charge lattice can be identified with $H_2(B,\mathbb{Z})$, the possible classes of curves wrapped by the D3-brane. For each $\beta\in H_2(B,\mathbb{Z})$, we can compute the elliptic genus
\begin{align}
    {\mathbb E}_\beta(\tau,\lambda)=\text{Tr}_{\mathcal{H}^\beta}(-1)^Fq^{H_0}\bar{q}^{\tilde{H}_0}y^{J_3^L}
\end{align}
associated to the string of this charge wrapping a torus $T^2$ of complex structure $\tau$. Here, $q=e^{2\pi {\rm i}\tau}$, $y=e^{2\pi {\rm i} \lambda}$ and $J_3^L$ is the Cartan generator associated to a global ${\rm SU}(2)_L$ symmetry.

The gauge symmetry of ${\rm F}[X^0]$ induces a corresponding global symmetry $\Gamma^0$ of the string worldsheet theory which can be used to twist the theory.
Let $\tseOne, \tseTwo \in \Gamma^0$ and denote by $U_\tseTwo$ the unitary operator representing $\tseTwo$ on the Hilbert space.
We can define the $\tseOne$-twisted and $\tseTwo$-twined elliptic genus
\begin{align}
    {\mathbb E}{\tiny\left[\begin{array}{c}\tseOne\\\tseTwo\end{array}\right]}_\beta(\tau,\lambda)=\text{Tr}_{\mathcal{H}_{\tseOne}^\beta}(-1)^Fq^{H_0}\bar{q}^{\tilde{H}_0}y^{J_3^L}U_{\tseTwo}\,,
\end{align}
where $\mathcal{H}_{\tseOne}^\beta$ is the Hilbert space of $\tseOne$-twisted states of charge $\beta$.
Note that the twining by $\tseTwo$ can alternatively be interpreted in terms of twisted boundary conditions in the time direction. We will therefore also refer to $ {\mathbb E}{\tiny\left[\begin{array}{c}\tseOne\\\tseTwo\end{array}\right]}_\beta$ as the $(\tseOne,\tseTwo)$-twisted elliptic genus.

Two bases of the integral homology $H_2(T^2,\IZ)$ of the worldsheet torus are related by an $\mathrm{SL}(2,\IZ)$ transformation. We will write this transformation as a left action on the basis vector $(\Sigma_1, \Sigma_2)^T$. This leads to the familiar $\mathrm{SL}(2,\IZ)$ action on the complex structure $\tau$ of the torus,
\begin{align}
    g=\left(\begin{array}{cc}
    a&b\\
    c&d
    \end{array}\right)\in \text{SL}(2,\mathbb{Z})\,,\quad g:\,\tau\mapsto\frac{a\tau+b}{c\tau+d}\,.
\end{align}
We write the twisting-twining group action associated to the two basis cycles as 
\begin{equation}
    (\gamma_1, \gamma_2) \cdot \begin{pmatrix} \Sigma_1 \\ \Sigma_2 \end{pmatrix} = \gamma_1 \Sigma_1 + \gamma_2 \Sigma_2 \,.
\end{equation}
A transformation of the basis vector must then be accompanied by the inverse transformation acting from the right on the vector $(\gamma_1, \gamma_2)$,\footnote{The technical necessity for this transformation stems from the behavior of the key quantities
\begin{align}
    Q_{k}(\mu, \lambda, \tau) = \frac{1}{(2\pi \ii)^k}\sum_{(m_1,m_n) \neq (0,0)} \frac{\lambda^{-m_1} \mu^{m_2}}{(m_1 \tau + m_2)^k} 
\end{align}
that arise when considering chiral algebras with automorphism on a torus \cite[Theorem 4.8]{Dong:1997ea}; here, $(\mu, \lambda)$ are eigenvalues of $(\gamma_1, \gamma_2)$. Formal manipulations of the infinite sum are justified for $k \ge 3$. It is then easy to see \cite[Theorem 4.6]{Dong:1997ea} that
\begin{align}
    Q_k (\mu, \lambda, g \tau) = (c \tau + d)^k Q_k\left( (\mu, \lambda) g, \tau\right) \,.
\end{align}}
such that for all $g={\tiny\left(\begin{array}{cc}a&b\\c&d\end{array}\right)}\in\SLtwoZ$, one has
\begin{align}
\begin{split}
    {\mathbb E}{\tiny \left[ \left(\tseOne,\tseTwo\right)g\right]}_\beta(\tau,\lambda)={\mathbb E}{\tiny\left[\begin{array}{c}\tseOne^a\tseTwo^c\\\tseOne^b\tseTwo^d\end{array}\right]}_\beta(\tau,\lambda)=&\rho^{\tseOne}_{\tseTwo}(g,\lambda) {\mathbb E}{\tiny\left[\begin{array}{c}\tseOne\\\tseTwo\end{array}\right]}_\beta \left(\frac{a\tau+b}{c\tau+d}, \frac{\lambda}{c\tau+d}\right)\,.
    \end{split}
    \label{eqn:ttEllGenMod}
\end{align}
Here, $\rho^{\tseOne}_{\tseTwo}(g,\lambda)$ is a phase that measures the 't~Hooft anomalies of the global symmetries of the worldsheet theory, see~\cite{Dierigl:2022zll}.

Due to the transformation of $(\gamma_1, \gamma_2)$, the partition function is only invariant, potentially up to an anomalous phase, under modular transformations that leave the boundary conditions invariant, as was already noted in~\cite{ZUBER1986127} for general two-dimensional CFTs. This leads us to the following definition:
\begin{definition}\label{def:stabilizerGroup}
Given a finite Abelian group $G$ and elements $\tseOne,\tseTwo\in G$, we define
\begin{align}
    \Gamma^{\tseOne}_{\tseTwo}=\left\{\,\,\left(\begin{array}{cc}a&b\\c&d\end{array}\right)\in{\normalfont\SLtwoZ}\,\,\big\vert\,\,\left(\begin{array}{c}\tseOne\\\tseTwo\end{array}\right)=\left(\begin{array}{c}\tseOne^a\tseTwo^c\\\tseOne^b\tseTwo^d\end{array}\right)\,\,\right\}\,.
\end{align}
\end{definition}
In view of Claim \ref{claim:monodromy} and Claim \ref{claim:cusps}, the group $\Gamma^\tseOne_\tseTwo$ for $\gamma_1, \gamma_2 \in \Gamma^0$ will play an important role in the following. To determine it, we first note that $\langle \tseOne,\tseTwo\rangle\subset\Gamma^0$ is closed under the $\SLtwoZ$ action; we can hence restrict our considerations to this subgroup.
The fundamental theorem of finite Abelian groups implies the isomorphism
\begin{align}
     \langle \tseOne,\tseTwo\rangle=\mathbb{Z}_{k_1}\times\mathbb{Z}_{k_2}\,,
\end{align}
such that $k_1$ divides $k_2$. 
It will be convenient to use additive notation. We will denote the images of $\tseOne,\tseTwo$ under the isomorphism as $(r_1,r_2)$ and $(s_1,s_2)$ respectively.
Setting $m=k_2/k_1$, we can equivalently consider the stabilizer of $(r_1',r_2)=(mr_1,r_2)$ and $(s_1',s_2)=(ms_1,s_2)$ as elements of $\mathbb{Z}_{k_2}\times\mathbb{Z}_{k_2}$. Defining
\begin{align}
    \text{Stab}_{k_2}(t) = \left\{ g \in \mathrm{SL}(2,\IZ_{k_2}) \, | \, t = tg \right\}\,,
\end{align}
where $t$ denotes the tuple $t=\left(\begin{array}{c}
        (r_1',r_2),
        (s_1',s_2)
    \end{array}\right)$,
we obtain the group $\Gamma^{\gamma_1}_{\gamma_2}$ as the preimage
\begin{align}
    \Gamma^{\gamma_1}_{\gamma_2}=\pi_{k_2}^{-1}\text{Stab}_{k_2}(t)\
\end{align}
under the projection $\pi_{k_2}:\,\text{SL}(2,\mathbb{Z})\rightarrow \text{SL}(2,\mathbb{Z}_{k_2})$.

The possible twists and corresponding symmetry groups for the simplest case  $\Gamma^0=\mathbb{Z}_2$ are shown in Table~\ref{tab:twistSymZ2}. The case $\Gamma^0=\mathbb{Z}_4$, which will be needed in Section \ref{sec:exampleZ4}, is shown in Table~\ref{tab:twistSymZ4}. Finally, we also consider the case $\Gamma^0=\mathbb{Z}_2\times\mathbb{Z}_2$, which will be relevant in the forthcoming \cite{wipZ2Z2}, in Table~\ref{tab:twistSymZ2Z2}.
A discussion of the relevant congruence subgroups can be found in Appendix~\ref{sec:congruenceSubgroups}.

\begin{table}[ht!]
\centering
\subfloat[$\Gamma^0=\mathbb{Z}_2$\label{tab:twistSymZ2}]{
\begin{tabular}{|c|c|}\hline
    $r,\,s\in\mathbb{Z}_2$&$\Gamma^r_s$\\\hline
    $(r,s)=(0,0)$&$\text{SL}(2,\mathbb{Z})$\\\hline
    $(r,s)=(1,0)$&$\Gamma^1(2)$\\\hline
    $(r,s)=(0,1)$&$\Gamma_1(2)$\\\hline
    $(r,s)=(1,1)$&$\Lambda_2$\\\hline
\end{tabular}
}
\subfloat[$\Gamma^0=\mathbb{Z}_4$\label{tab:twistSymZ4}]{
\begin{tabular}{|c|c|}\hline
    $r,\,s\in\mathbb{Z}_4$&$\Gamma^{r}_{s}$\\\hline
    $(r,s)=(0,0)$&$\text{SL}(2,\mathbb{Z})$\\\hline
    $(r,s)=(\pm 1,0)$&$\Gamma^1(4)$\\\hline
    $(r,s)=(2,0)$&$\Gamma^1(2)$\\\hline
    $(r,s)=(2,\pm 1)$&$\Gamma_1(4)\cap\Gamma(2)$\\\hline
    $(r,s)=(0,\pm 1)$&$\Gamma_1(4)$\\\hline
    $(r,s)=(0,2)$&$\Gamma_1(2)$\\\hline
\end{tabular}
}\\
\subfloat[$\Gamma^0=\mathbb{Z}_2\times\mathbb{Z}_2$\label{tab:twistSymZ2Z2}]{
\begin{tabular}{|c|c|}\hline
    $(r_1,r_2),\,(s_1,s_2)\in\mathbb{Z}_2\times\mathbb{Z}_2$&$\Gamma^{(r_1,r_2)}_{(s_1,s_2)}$\\\hline
    $(r_1,r_2)=(s_1,s_2)=(0,0)$&$\text{SL}(2,\mathbb{Z})$\\\hline
    $(r_1,r_2)\ne (0,0)\,,\,\,(s_1,s_2)=(0,0)$&$\Gamma^1(2)$\\\hline
    $(r_1,r_2)= (0,0)\,,\,\,(s_1,s_2)\ne (0,0)$&$\Gamma_1(2)$\\\hline
    $(r_1,r_2)=(s_1,s_2)\ne (0,0)$&$\Lambda_2$\\\hline
    $(r_1,r_2)\ne(0,0)\,,\,\,(s_1,s_2)\ne(0,0)\,,\,\,(r_1,r_2)\ne (s_1,s_2)$&$\Gamma(2)$\\\hline
\end{tabular}
}
\label{tab:twistSym}
\caption{The stabilizer groups $\Gamma^\tseOne_\tseTwo$ for three examples.}
\end{table}

The fact that the half-hypermultiplets in the F-theory vacuum ${\rm F}[X^0]$ always contain charge conjugate pairs entails the identification ${\rm F}[X^0][S^1_{\tseOne}]\simeq {\rm F}[X^0][S^1_{\tseOne^{-1}}]$ and therefore $X^\tseOne\simeq X^{\tseOne^{-1}}$.
Similarly, we expect that
\begin{align}
    {\rm F}[X^0][S^1_{\tseOne}][S^1_{\tseTwo}]\simeq {\rm F}[X^0][S^1_{\tseOne^{-1}}][S^1_{\tseTwo^{-1}}]\,.
\end{align}
This implies that the twisted-twined elliptic genera must be invariant under simultaneous inversion of the twist and the twining,
\begin{align}
    {\mathbb E}{\tiny\left[\begin{array}{c}\tseOne\\\tseTwo\end{array}\right]}_\beta(\tau,\lambda)= {\mathbb E}{\tiny\left[\begin{array}{c}\tseOne^{-1}\\\tseTwo^{-1}\end{array}\right]}_\beta(\tau,\lambda)\,.
    \label{eqn:ellGenChargeInversion}
\end{align}
In general, this imposes further conditions, in addition to the modularity.
We will explicitly verify this in the example in Section~\ref{sec:exampleZ4}, using the results in Appendix~\ref{app:mformsG14capG2}.

To briefly summarize the relevant point, we can take $\Gamma^0=\mathbb{Z}_4$ and consider ${\mathbb E}{\tiny\left[\begin{array}{c}1\\0\end{array}\right]}_\beta(\tau,\lambda)$ and  ${\mathbb E}{\tiny\left[\begin{array}{c}2\\1\end{array}\right]}_\beta(\tau,\lambda)$ that, by Table~\ref{tab:twistSymZ4}, should respectively transform as meromorphic Jacobi forms for the groups
\begin{align}
    \Gamma^1_0=\Gamma^1(4)\,,\quad \Gamma^2_1=\Gamma_1(4)\cap\Gamma(2)\,.
\end{align}
The two twisted-twined elliptic genera are related by a modular transformation, but, as we show in Appendix~\ref{app:mformsG14capG2}, the rings of modular forms for the two groups $\Gamma^1_0$ and $\Gamma^2_1$ are not isomorphic.
The reason is that~\eqref{eqn:ellGenChargeInversion} together with~\eqref{eqn:ttEllGenMod} implies that
\begin{align}
    {\mathbb E}{\tiny\left[\begin{array}{c}2\\1\end{array}\right]}_\beta(\tau,\lambda)= {\mathbb E}{\tiny\left[\begin{array}{c}2\\3\end{array}\right]}_\beta(\tau,\lambda)\propto {\mathbb E}{\tiny\left[\begin{array}{c}2\\1\end{array}\right]}_\beta(\tau+1,\lambda)\,,
\end{align}
which is only satisfied by a subset of $\Gamma^2_1$ modular or Jacobi forms. We show in Appendix~\ref{app:mformsG14capG2} that $M(\Gamma^1(4))$ is actually isomorphic to the subring $\oplus_{k\ge 0}\Upsilon_k\subset M(\Gamma_1(4)\cap\Gamma(2))$, where
\begin{align}
    \Upsilon_k=\left\{\,\, \phi(\tau)\in M_k\left(\Gamma_1(4)\cap\Gamma(2)\right)\,\,\big\vert\,\, \phi\left(\tau+1\right)=(-1)^k\phi(\tau) \,\,\right\}\,.
\end{align}

\subsection{B-field topology and complexified K\"ahler classes on nodal Calabi-Yau threefolds}
\label{sec:Bfields}
Generalizing the works \cite{Vafa:1986wx,Vafa:1994rv}, the authors of~\cite{Aspinwall:1995rb} propose that the A-model on a smooth Calabi-Yau threefold $X$ depends not merely on the conventional complexified K\"ahler class, but also on a discrete datum, an element of $\tors{H^3(X,\IZ)}$. The two together are captured by the cohomology group $H^2(X,\IC^*)$ \cite{Aspinwall:1995rb}. It is natural to associate the additional degree of freedom to the flat B-field; this then takes values in $H^2(X,\U(1))$, as follows from the following simple argument: Assume that $X$ is simply connected, such that ${H^1(X,\U(1)) = 0}$. The short exact sequence
\begin{align}
    0\rightarrow \mathbb{Z}\rightarrow\mathbb{R}\rightarrow\U(1)\rightarrow 0
\end{align}
then leads to the exact sequence in cohomology
\begin{align}
    0\rightarrow H^2(X,\mathbb{Z})\rightarrow H^2(X,\mathbb{R})\rightarrow H^2(X,\U(1))\rightarrow \tors{H^3(X,\mathbb{Z})}\rightarrow 0\,.
\end{align}
This implies the sought for relation
\begin{align} 
    H^2(X,\U(1))\simeq \quotient{H^2(X,\mathbb{R})}{H^2(X,\IZ)}\times\tors{H^3(X,\mathbb{Z})}\,,
    \label{eqn:H2iso}
\end{align}
with the first factor associated to the conventional (periodic) B-field, and the second to the torsional contribution.\footnote{The torsional contribution can be intuitively attributed to a torsional field strength H.}
As discussed above Conjecture~\ref{conj:twisted}, the torsional factor in~\eqref{eqn:H2iso} can be identified with the cohomological Brauer group of $X$,
\begin{align}
    \text{Br}'(X)= \tors{H^3(X,\mathbb{Z})}\,,
\end{align}
which for smooth projective varieties is isomorphic to the Brauer group $\text{Br}(X)$~\cite{GabberDJ}.
Note that the isomorphism \eqref{eqn:H2iso} is not canonical.
The universal coefficient theorem gives
\begin{align}
    H^2(X,\U(1))=\text{Hom}\left(H_2(X,\mathbb{Z}),\U(1)\right)\,,\quad \tors{H^3(X,\IZ)}=\text{Hom}\left(\tors{H_2(X,\mathbb{Z})},\U(1)\right)\,,
\end{align}
and therefore a choice of isomorphism~\eqref{eqn:H2iso} induces a dual isomorphism
\begin{align}
    H_2(X,\mathbb{Z})\simeq\free{H_2(X,\mathbb{Z})}\times \tors{H_2(X,\mathbb{Z})}\,.
    \label{eqn:H2liso}
\end{align}

Moving beyond the smooth case, we now turn our attention to $\IQ$-factorial Calabi-Yau threefolds $X$ whose only singularities are isolated nodes.
We recall from~\cite{Katz:2022lyl,Katz:2023zan} the following claim, which is in the same spirit as Claims \ref{claim:mtheoryNodalCYgaugeGroup} and \ref{claim:mtheoryNodalCYmatter} above:
\begin{claim}\label{claim:Bfield}
    Given a Calabi-Yau threefold $X$ with isolated nodes, the cohomology class of a flat B-field takes values in $H^2(\widehat{X},\U(1))$, where $\widehat{X}$ is any small resolution of $X$.\footnote{Recall from Comment \ref{comment:nonKaehlerResolution} in Section \ref{sec:almostgeneric} that when $X$ is $\IQ$-factorial, none of the resolutions $\widehat{X}$ can be K\"ahler.}
\end{claim}
\noindent 

In the context of $\IQ$-factorial Calabi-Yau threefolds with isolated nodes, Claim \ref{claim:Bfield} specifies how to generalize the real part of the complexified K\"ahler class beyond the smooth case. For smooth Calabi-Yau varieties (which are in particular K\"ahler manifolds), the imaginary part is a choice of K\"ahler class. With $X$ not a manifold and $\widehat{X}$ not K\"ahler, we need a substitute for this notion.

On $X$, we adopt the conventional substitute of replacing the K\"ahler cone by the ample cone $\Amp(X)$: The first Chern class $c_1$ provides the map
\begin{align}
    c_1: \, \Amp(X) \, \rightarrow H^{1,1}(X) \cap H^2(X, \IZ) \,,
\end{align}
and by the Nakai-Moishezon-Kleiman criterion (see e.g. \cite[Theorem 1.2.23]{LazarsfeldPositivityI}, the image of an ample line bundle $L$ on a projective variety $X$ satisfies the inequality\footnote{The integral is to be taken over the smooth locus of $X$.}
\begin{align}
    \int_V c_1(L)^{\dim(V)} > 0
\end{align}
for every positive-dimensional irreducible subvariety $V \subset X$, just as we would require for a form representing a K\"ahler class.

As Claim \ref{claim:Bfield} identifies the class of the flat B-field with a class on the small resolution $\widehat{X}$ of $X$, we may wish to map $c_1(L)$ to a class on $\widehat{X}$ as well. This is easy thanks to~\cite[Theorem 5.16]{Grassi:2018rva}, which applies as $X$ is assumed $\IQ$-factorial and isolated nodes are terminal. This theorem states that $X$ satisfies rational Poincar\'e duality, i.e. we have a canonical isomorphism $H^2(X,\mathbb{Q}) \simeq H_4(X,\mathbb{Q})$. As $\widehat{X}$ is a small resolution of $X$, we further have $H_4(X,\IQ) = H_4(\widehat{X}, \IQ)$. Finally, by Poincaré duality on $\widehat{X}$, we see that $c_1(L)$ determines a class in $H^2(\widehat{X}, \IQ)$.

Given a divisor $J\in\text{Pic}(X)\otimes\mathbb{Q}$, we can hence use the map $\text{Pic}(X)\rightarrow H^2(X,\mathbb{Z})$ composed with the isomorphism just discussed to obtain an element of $H^2(\widehat{X},\mathbb{Q})$ that we will also denote by $J$.
Given divisors $\omega_{\rm B},\omega_{\rm K}\in\text{Pic}(X)\otimes\mathbb{Q}$, we say that
\begin{align}
    \omega_{\mathbb{C}}=\omega_{\rm B}+{\rm i}\omega_{\rm K}\
\end{align}
is a \textit{complexified K\"ahler class} on $X$ if $\omega_{\rm K} \in \Amp(X)$.~\footnote{The divisor $\omega_{\rm B}$ should not be confused with the canonical bundle $\omega_B$ on $B$.}

Combining a complexified K\"ahler class with an element $\alpha\in\tors{H^3(\widehat{X},\mathbb{Z})}\simeq \text{Br}'(\widehat{X})$, we can assign a {\it complexified volume} to a curve class on $\widehat{X}$ via the map
\begin{align}
    \text{Vol}_{\IC}:\,\free{H_2(\widehat{X},\mathbb{Z})}\times\tors{H_2(\widehat{X},\mathbb{Z})}\rightarrow\IC\,,\quad (\beta,\chi)\mapsto\frac{q_\alpha(\chi)}{m_\alpha}+\int_\beta \omega_{\IC}\,,
    \label{eqn:Cvolume}
\end{align}
where we have used the same notation as in \eqref{eq:chiPhi}.
A singularity is stabilized in the sense of \cite{Aspinwall:1995rb} when $\text{Vol}_{\IC}([C_p]) \neq 0$ for every exceptional curve $C_p \in \widehat{X}$ over a node $p \in X$.

It is natural to ask in what sense the imaginary part of this complexified volume is the volume that is measured by an actual metric on $X$.
Since $X$ is $\mathbb{Q}$-factorial, we can choose a smoothing $\widetilde{X}$ of $X$ which is a K\"ahler Calabi-Yau threefold. To justify the terminology we have introduced, one should show that $\omega_{\IC}$ on $X$, together with an associated degenerate metric, can be obtained as a limit from a (conventional) K\"ahler form on $\widetilde{X}$. Away from the singular locus $S$ of $X$, the metric should be smooth and assign a positive volume to the holomorphic curves in $X-S$. We leave this interesting task for future work.

\subsection{Twisted-twined elliptic genera and topological strings}
\label{sec:modularityTopString}

To relate the twisted-twined elliptic genera to the topological string partition function on $X^{\tseOne}_{\tseTwo}$, we need to introduce a basis of divisors on $X^{\tseOne}$ and a suitable parametrization of the complexified K\"ahler class.

We know from the discussion in Section~\ref{sec:discreteHolonomiesAndGenusOne} that $X^{\tseOne}$ has an $m_\tseOne$-section, with $m_\tseOne$ the order of $\tseOne$ in $\Gamma^0$. We denote a choice of such an $m_\tseOne$-section by $J_0$. We take a basis of divisors $\Jbase_\sia\in \text{Pic}(B)$, $\sia\in\{1,\ldots,b_2(B)\}$, that generate a simplicial sub-cone of the K\"ahler cone, with corresponding vertical divisors $J_\sia=\pi_\tseOne^*(\Jbase_\sia)$.
We can then expand the complexified K\"ahler class on $X^{\tseOne}$ as
\begin{align}\omega_{\mathbb{C}}=\tilde{\tau}\,J_0+\tilde{t}^\sia\, J_\sia
\end{align}
such that $\text{Im}\,\tilde{\tau}$ is $1/m_\tseOne$ times the volume of the generic fiber. 
The complexified volume~\eqref{eqn:Cvolume} of the generic fiber also depends on the choice of topology for the flat B-field.
We choose the torsional part of the cohomology class of the flat B-field to be
\begin{align}
    [\tseTwo]\in \tors{H^3\left(\widehat{X}^{\tseOne},\mathbb{Z}\right)}\simeq \quotient{\Gamma^0}{\langle\tseOne\rangle}\,.
\end{align}
There is a compatible choice of isomorphism
\begin{align}
    H_2(X^\tseOne,\mathbb{Z})\simeq \mathbb{Z}^{T+1}\times\mathbb{Z}\times\widehat{\Gamma}^{\tseOne}\,,
    \quad H_2(X^\tseOne,\mathbb{Z})_F\simeq \mathbb{Z}\times\widehat{\Gamma}^{\tseOne}\,,
    \label{eqn:H2ismorphism}
\end{align}
such that the complexified volume of a curve in the class $(\vec{v},n,\chi)$ is given by
\begin{align}
    \text{Vol}_{\mathbb{C}}(\vec{v},q,\chi)=\frac{q_\tseTwo(\chi)}{m_\tseTwo}+n\tilde{\tau}+\tilde{t}^\alpha\vec{v}_\alpha\,.
\end{align}Following~\cite{Cota:2019cjx,Schimannek:2019ijf}, we then change to the same basis which already made an appearance in~\eqref{eq:basisDiv},
\begin{align}
    D_0=J_0-\frac{1}{2m_\tseOne}\pi_\tseOne^*\circ\pi_{\tseOne*}(J_0\cdot J_0)\,,\quad D_\sia=J_\sia\,,
\end{align}
and introduce $t^\sia$ such that $\omega_{\mathbb{C}}=\tilde{\tau}\,D_0+t^\sia\, D_\sia$.
Note that $D_0$ is invariant under shifts of $J_0$ by vertical divisors and is therefore independent of the particular choice of $m_\tseOne$-section. We also introduce the exponentiated complexified K\"ahler parameters $Q^\beta=e^{2\pi {\rm i} t^\sia \Jbase_\sia\cdot\beta}$ for $\beta\in H_2(B,\mathbb{Z})$.

We denote the A-model topological string partition function on $X^{\tseOne}_{\tseTwo}$ by $Z[X^{\tseOne}_{\tseTwo}](t,\tilde{\tau},\lambda)$, where $\lambda$ is the topological string coupling, and expand it as
\begin{align}
    Z[X^{\tseOne}_{\tseTwo}](t,\tilde{\tau},\lambda)=Z[X^{\tseOne}_{\tseTwo}]_0(\tilde{\tau},\lambda)\left(1+\sum\limits_{\beta\in H_2(B,\mathbb{Z})}Z[X^{\tseOne}_{\tseTwo}]_\beta(\tilde{\tau},\lambda)Q^\beta\right)\,.
\end{align}
We refer to $Z[X^{\tseOne}_{\tseTwo}]_\beta(\tilde{\tau},\lambda)$ as the base degree $\beta$ topological string partition function.

Generalizing the observations from~\cite{Schimannek:2021pau,Dierigl:2022zll}, this is related to the twisted-twined elliptic genera of the non-critical strings in ${\rm F}[X^0]$ by identifying the elliptic argument $\tau$ with the complexified volume of the generic fiber in $X^{\tseOne}$.
Recall from \eqref{eqn:5dchargelatticeMth} that  the charge lattice of the twisted F-theory compactification ${\rm F}[X^0][S^1_\tseOne]$ can be written as
\begin{align}
\Lambda_{5\dd}^{\tseOne}=\ker\left(\vec{0},\e[-\tfrac{1}{m_{\tseOne}}],\tseOne\right)\subset \mathbb{Z}^{T+1}\times\mathbb{Z}\times\widehat{\Gamma}^0\,.
\end{align}
In these conventions, the Kaluza-Klein towers associated to the half-hypermultiplets are generated by the charge $(\vec{0},m_{\tseOne},0)$.\footnote{Recall that we are using additive notation for characters. We therefore denote the unit character as 0.}
On the other hand, the isomorphism~\eqref{eqn:upullback} for a choice of lift $\chi_{\tseOne}$ of the character on $\langle \gamma_1 \rangle$ induced by $\gamma_1 \mapsto \e[-\tfrac{1}{m_{\tseOne}}]$ maps $(\vec{0},m_{\tseOne},0)$ to the element
\begin{align}
    (\vec{0},m_{\tseOne},{m_\tseOne}\chi_{\tseOne})\in \mathbb{Z}^{T+1}\times\mathbb{Z}\times\widehat{\quotient{\Gamma^0}{\langle\tseOne\rangle}}.
\end{align}
This isomorphism must be chosen to be consistent with~\eqref{eqn:H2ismorphism}.
With~\eqref{eq:classOFgenericFiber}, this implies that the generic fiber of the fibration represents the class $(\vec{0},m_{\tseOne},m_\tseOne\chi_{\tseOne})$. The complexified volume of the generic fiber is then
\begin{align}
    \tau:=\text{Vol}_{\mathbb{C}}(\vec{0},m_{\tseOne},m_\tseOne\chi_{\tseOne})=\frac{q_\tseTwo(m_\tseOne\chi_{\tseOne})}{m_\tseTwo}+m_{\tseOne}\tilde{\tau}\,.
    \label{eqn:CvolumeFibre}
\end{align}
We are now in a position to make precise the parameter identification underlying Claim \ref{claim:ellipticGtopZ}: the topological string partition function should depend on a parametrization of the complexified K\"ahler cone, while the modular parameter of the elliptic genus should coincide with the volume of the general elliptic fiber. Consequently,
\begin{align}
    Z[X^{\tseOne}_{\tseTwo}]_\beta(\tilde{\tau},\lambda)={\mathbb E}{\tiny\left[\begin{array}{c}\tseOne\\\tseTwo\end{array}\right]}_\beta\left(m_\tseOne\tilde{\tau}+\frac{q_\tseTwo({m_\tseOne}\chi_{\tseOne})}{m_\tseTwo},\lambda\right)\,.
    \label{eqn:EgenusTopString}
\end{align}

The shift of the argument on the right-hand side of~\eqref{eqn:EgenusTopString} by $q_\tseTwo({m_\tseOne}\chi_{\tseOne})/m_\tseTwo$ has been observed in an example in~\cite[Section 7.4]{Schimannek:2021pau}.
We now understand its physical and geometrical origin.
Its non-trivial compatibility with the modular properties discussed in Section~\ref{sec:twistedmodularity} will be further demonstrated in the example that we discuss in Section~\ref{sec:exampleZ4}.

Note that~\eqref{eqn:CvolumeFibre} depends on the choice of the lift $\chi_\tseOne$.
This corresponds to the fact that given any element $\chi'\in \widehat{\Gamma^0/\langle\tseOne\rangle}$, we get an isomorphism
\begin{align}
    H_2(X^\tseOne,\mathbb{Z})_F\rightarrow H_2(X^\tseOne,\mathbb{Z})_F\,,\quad (n,\chi)\mapsto (n,\chi+n\chi')\,,
\end{align}
which is compensated by the action on $\tilde{\tau}$ given by
\begin{align}
    \tilde{\tau}\mapsto \tilde{\tau}+\frac{q_\tseTwo({m_\tseOne}\chi_{\tseOne})-q_\tseTwo\left({m_\tseOne}(\chi_{\tseOne}+\chi')\right)}{m_\tseOne m_\tseTwo}\,.
\end{align}
Let us therefore stress again that the choice of isomorphism~\eqref{eqn:H2ismorphism} is not independent of the choice of lift $\chi_{\tseOne}$: both have to be made consistently in order for~\eqref{eqn:EgenusTopString} to hold.

The charge $\beta$ twisted-twined elliptic genus ${\mathbb E}{\tiny\left[\begin{array}{c}\tseOne\\\tseTwo\end{array}\right]}_\beta(\tau,\lambda)$ takes the form
\begin{align}
    {\mathbb E}{\tiny\left[\begin{array}{c}\tseOne\\\tseTwo\end{array}\right]}_\beta(\tau,\lambda)=\frac{\Delta(\tau)^{l_{\beta}}}{\eta(\tau)^{12c_1(B)\cdot\beta}}\frac{\phi_\beta(\tau,\lambda)}{\prod_{\sia=1}^{b_2(B)}\prod_{s=1}^{\beta_\sia}\phi_{-2,1}(\tau,s\lambda)}\,,
\end{align}
where $\eta(\tau)$ is the Dedekind eta function, $\phi_{-2,1}(\tau,\lambda)$ is a weak Jacobi form of weight $-2$ and index $1$ and we use $\beta_\sia=\Jbase_\alpha\cdot\beta$.
We elaborate on $\Delta(\tau)$ below.
The numerator $\phi_\beta(\tau,\lambda)$ is a $\Gamma^{\tseOne}_{\tseTwo}$ weak Jacobi form that depends on $X^{\tseOne}_{\tseTwo}$.
The weight and index of $\phi_\beta(\tau,\lambda)$ are such that ${\mathbb E}{\tiny\left[\begin{array}{c}\tseOne\\\tseTwo\end{array}\right]}_\beta(\tau,\lambda)$ is a meromorphic Jacobi form of weight 0 and index $\beta(\beta-c_1(B))/2$.
Some basic properties of the relevant modular and Jacobi forms are collected in Appendix~\ref{app:modularJacobi}.

We denote by $\Delta(\tau)$ a certain cusp form with a multiplier system for $\Gamma^{\tseOne}_{\tseTwo}$ of weight $2$ that is introduced in Appendix~\ref{sec:cuspForms}.
It depends on $\tseOne,\tseTwo$ and a choice of this isomorphism
\begin{align}
    \Gamma^0\simeq \mathbb{Z}_{k_1}\times\ldots\mathbb{Z}_{k_n}\,.
\end{align}
The exponent of $l_\beta$ of $\Delta(\tau)$ is linear in $\beta$ and also depends on the choice of isomorphism.
The multiplier system captures the global anomaly of the discrete gauge symmetry of the six-dimensional F-theory vacuum~\cite{Dierigl:2022zll}.
If $\tseTwo$ is trivial and~\eqref{eqn:cuspFormGammaIso} is chosen such that $\tseOne$ corresponds to the element $1\text{ mod }m_\tseOne$ in $\langle\tseOne\rangle\simeq \IZ_{m_\tseOne}$ then
\begin{align}
    \Delta(\tau)=\Delta[\begin{array}{c|cc}m_\tseOne&1&0\end{array}](\tau)\,,
\end{align}
in terms of~\eqref{eqn:cuspForm1}, and the exponent is fixed  modulo $m_\tseOne$ by the congruency condition~\cite[Section 4.5]{Cota:2019cjx}
\begin{align}
    l_{\beta}=\frac12\left[m_\tseOne^2c_1(B)+\pi_{\tseOne*}(J_0\cdot J_0)\right]\cdot \beta\,\,\,\text{mod}\,\,\,m_\tseOne\,.
\end{align}
We will come back to the question of how this exponent can be determined in general in~\cite{wipZ2Z2}.

\subsection{Modularity and twisted derived equivalences}
\label{sec:modularityAndDerivedEquivalences}

As discussed in Section~\ref{sec:twistedmodularity}, the twisted-twined elliptic genera associated to the non-critical strings in ${\rm F}[X^0]$ transform like vector-valued Jacobi forms under the action of the modular group.
Given an element $g\in\SLtwoZ$ and $\tseOne,\tseTwo\in\Gamma^0$, let us denote
\begin{align}
 g=\left(\begin{array}{cc}a&b\\c&d\end{array}\right)\in\SLtwoZ\,,\quad \left(\begin{array}{c}\tseOnePrime\\\tseTwoPrime\end{array}\right)=\left(\begin{array}{c}\tseOne^a\tseTwo^c\\\tseOne^b\tseTwo^d\end{array}\right)\,.
 \label{eqn:SL2zactionDerivedEquivalence}
\end{align}
We denote the image of the cusp ${\rm i}\infty$ under $g$ by $\tau_\infty(g)\in\mathbb{Q}\cup\{{\rm i}\infty\}$, such that
\begin{align}
    \tau_\infty(g)=\left\{\begin{array}{cl}
        \frac{a}{c}&\text{ for }c\ne 0\\
        {\rm i}\infty&\text{ else}
    \end{array}\right.\,.
\end{align}

Using the identification~\eqref{eqn:EgenusTopString} with the topological string partition function, the vector valued modularity~\eqref{eqn:ttEllGenMod} is the manifestation of Claims \ref{claim:monodromy} and \ref{claim:cusps}~\cite{Knapp:2021vkm,Schimannek:2021pau}: the limit
\begin{align}
    \tau\rightarrow \tau_\infty(g)\in\mathbb{Q}\cup\{{\rm i}\infty\}\,,\quad t_\sia\rightarrow {\rm i}\infty\,,
\end{align}
in the stringy K\"ahler moduli space of $X^\tseOne_\tseTwo$ is a boundary point of maximally unipotent monodromy (MUM) that can be interpreted as the large volume limit associated to $X^{\tseOnePrime}_{\tseTwoPrime}$.
In particular, if two pairs $(\tseOne, \tseTwo)$ and $(\tseOnePrime,\tseTwoPrime)$ lie on the same $\SLtwoZ$ orbit, then MUM points associated to $X^{\tseOne}_{\tseTwo}$ and $X^{\tseOnePrime}_{\tseTwoPrime}$ lie on the same moduli space.

From the open string perspective, the stringy K\"ahler moduli space can be interpreted as a slice in the space of Bridgeland stability conditions on the category of topological B-branes.
The fact that the large volume limits associated to $X^\tseOne_\tseTwo$ and $X^{\tseOnePrime}_{\tseTwoPrime}$ exist in the same moduli space therefore implies that they share the same category of topological B-branes.

For a smooth projective Calabi-Yau $X$ with a flat B-field that is topologically trivial, the latter is given by the bounded derived category of coherent sheaves $D^b(X)$~\cite{Witten:1998cd,Sharpe:1999qz,Douglas:2000gi,Aspinwall:2004jr}.
However, in the presence of a B-field that is flat but topologically non-trivial, and therefore represents a non-trivial (cohomological) Brauer class $\alpha\in\text{Br}'(X)$, the correct category of topological B-branes is the $\alpha$\textit{-twisted derived category} $D^b(X,\alpha)$~\cite{Kapustin:2000aa}.

The identification of categories is, from the physical perspective, using the open string worldsheet theory to transport the branes between different points in the moduli space~\cite{Herbst:2008jq}.
In order for this theory to be non-singular, we will assume the equivalent conditions
\begin{align}
    S^0_{\tseOne}\cap S^0_{\tseTwo}=\emptyset\quad\Leftrightarrow\quad \Gamma^0=\langle \tseOne,\tseTwo\rangle\,.
    \label{eqn:fullyStabilized}
\end{align}
This ensures that the exceptional curves over all of the nodes in $X^\tseOne$ measure a non-trivial B-field holonomy and are therefore ``stabilized'' in the sense of~\cite{Aspinwall:1995rb}.
As a consequence, all of the states that arise from 2-branes have a non-zero mass in the effective theory.
Note that if the condition~\eqref{eqn:fullyStabilized} should not be satisfied, we can always deform away the ``unstable'' nodes.

It is easy to see that condition~\eqref{eqn:fullyStabilized} is preserved under the $\SLtwoZ$ action~\eqref{eqn:SL2zactionDerivedEquivalence}.
The action $\tau\rightarrow-1/\tau$ just exchanges $S^0_{\tseOne}$ and $S^0_{\tseTwo}$.
On the other hand, the action $\tau\rightarrow\tau+1$ maps $S^0_{\tseTwo}$ to $S^0_{\tseTwo\tseOne}$, but after intersecting with $S^0_{\tseOne}$ we restrict to nodes that correspond to representations on which $\tseOne$ acts trivially and therefore $S^0_\tseOne\cap S^0_\tseTwo=S^0_\tseOne\cap S^0_{\tseOne\tseTwo}$.

We then expect that for any $\SLtwoZ$ element~\eqref{eqn:SL2zactionDerivedEquivalence}, there is a twisted derived equivalence that relates the categories of topological B-branes on $\widehat{X}^\tseOne_\tseTwo$ and $\widehat{X}^{\tseOnePrime}_{\tseTwoPrime}$, leading to Conjecture~\ref{conj:twisted}.

\section{Example with $\Gamma^0=\mathbb{Z}_4$}
\label{sec:exampleZ4}
We will now illustrate the discussion in this paper by way of the example of a six-dimensional F-theory vacuum with gauge group
\begin{align}
    \Gamma^0=\IZ_4\,,
\end{align}
and where the numbers $N_q$ of half-hypermultiplets with $\mathbb{Z}_4$-charge $q\text{ mod }4$ are
\begin{align}
    N_0=126\,,\quad N_1=144\,,\quad N_2=132\,,\quad N_3=144\,.
\end{align}

By our discussion in Section~\ref{sec:fieldTheoryExampleZ4}, depending on the twist $\tse\in\Gamma^0$, a circle compactification will lead to a 5d theory with the discrete gauge symmetry intact, reduced to $\IZ_2$ or completely absent. As discussed in Section~\ref{sec:almostgenericfibrations}, we expect the associated torus-fibered geometries to exhaust the geometries underlying the elements of the Tate-Shafarevich group $\Sh(X) = \IZ_4$.
We will write $X^k$ to denote the geometry that is associated to the twist $\tse=\e[\tfrac{k}{4}] \in \IZ_4$, with $k=0,\ldots,3$, and have $X^1=X^3$.

By Conjectures \ref{conj:TSandTorsion} and \ref{conj:I2fibers}, the variety $X^1$ should be smooth and exhibit a 4-section. We define $X^1$ as a complete intersection in a $\IP^3$ bundle over base $\IP^2$, with generic fiber a degree 4 curve $C$ in $\IP^3$. Intersection with the relative hyperplane divisor yields the 4-section of this fibration.
We will construct the singular 2-section geometry $X^2$ as a double cover of a $\IP^1$ bundle over $\IP^2$, following a discussion in~\cite{An2001}: in this reference, an intermediate step in the construction of the Jacobian of a degree 4 curve $C$ in $\IP^3$ is a rational map to a double cover $C'$ of $\IP^1$, which will serve as the generic fiber of $X^2$. 
The singular elliptic fibration $X^0$ is the relative Jacobian of both $X^1$ and $X^2$.

\subsection{The elements of the Tate-Shafarevich group} \label{sec:exampleGeom}
We will now discuss the three varieties associated to the twisted circle compactifications of this example.
\begin{itemize}
    \item We construct the smooth Calabi-Yau threefold $X^1=X^3$ as a complete intersection in the projective bundle $\pi_V:V\rightarrow\mathbb{P}^2$,
        \begin{align}
            V=\mathbb{P}\left(\mathcal{O}_{\mathbb{P}^2}(-1)\oplus\mathcal{O}_{\mathbb{P}^2}^{\oplus 3} \right)\,.
            \label{eqn:ex2V}
        \end{align}
    As we show in Appendix~\ref{sec:appGeom}, the vanishing locus of a generic section of the rank 2 vector bundle
        \begin{align}
            E=\zeta^{\otimes 2}\otimes\pi_V^*\left(\mathcal{O}_{\mathbb{P}^2}(1)\oplus\mathcal{O}_{\mathbb{P}^2}(1)\right)\,,
        \end{align}
    with $\zeta$ denoting the relative hyperplane bundle on $V$, is a Calabi-Yau threefold whose generic fiber is a degree 4 curve in $\IP^3$.
    In terms of the corresponding homogeneous coordinates $x_1$ of weight $(1,-1)$, $x_{2,3,4}$ of weight $(1,0)$ and $y_{1,2,3}$ of weight $(0,1)$, see Table \ref{tab:toricV}, we can write
        \begin{align}
            X^{1}=\{\,p_1(x_{1,\ldots,4},y_{1,\ldots,3})=p_2(x_{1,\ldots,4},y_{1,\ldots,3})=0\,\}\subset V\,,
            \label{eqn:discZ4ex}
        \end{align}
    with $p_1,p_2$ being polynomials that are both weighted homogeneous of degree $(2,1)$.  The fiber curve $C$ over each base point $[y_1:y_2:y_3]$ in the base of $X^1$ is hence described by the intersection of two quaternary quadrics.
    This allows us to introduce two $4 \times 4$ symmetric matrices $A_1(y_{1,\ldots,3})$, $A_2(y_{1,\ldots,3})$ such that
    \begin{equation} \label{eq:quaternaryQuadrics}
        p_i = \myvec{x}^T A_i(y_{1,\ldots,3}) \myvec{x} \,, \quad i =1,2 \,. 
    \end{equation}  
    Since \eqref{eqn:discZ4ex} is a complete intersection of two ample hypersurfaces, the Lefschetz hyperplane theorem implies that
        \begin{align}\label{eq:h2X1}
            H_2(X^1,\mathbb{Z})=H_2(V,\mathbb{Z})=\mathbb{Z}\times\mathbb{Z}\,.
        \end{align}
    
     \item We can construct the variety $X^2$ as a singular double cover of the trivial $\mathbb{P}^1$-bundle 
    \begin{align} \label{eq:W}
       \pi_W \,:\, W=\mathbb{P}^1\times\mathbb{P}^2 \rightarrow \IP^2\,.
    \end{align}
    To this end, we introduce the pencil of quadrics\footnote{In simpler terms, this is a two dimensional linear family of symmetric matrices.} 
    \begin{equation} \label{eq:pencilquad}
        A_{z_1,z_2}\left(y_{1,\ldots,3}\right) = z_1 A_1(y_{1,\ldots,3}) + z_2 A_2(y_{1,\ldots,3}) 
    \end{equation}
    based on the matrices $A_{1,2}$ introduced in \eqref{eq:quaternaryQuadrics}, and define the double cover of $W$ by
    \begin{equation}
        y^2 = \det A_{z_1,z_2}\left(y_{1,\ldots,3}\right) \,,
        \label{eqn:X2eq}
    \end{equation}
     with $z_{1,2}$ homogeneous coordinates on the first factor of $W$ and $y_{1,\ldots, 3}$ homogeneous coordinates on the second.
    It follows from~\cite[Theorem 3.1]{An2001} that over each point of the base $\IP^2$, the fiber curve $C'$ is a double cover of $\IP^1$ that has the same Jacobian as the corresponding fiber $C$ of $X^1$.

    To see that the variety $X^2$ thus defined is a Calabi-Yau threefold, note that the degree assignment of the coordinates $x_1$ and $x_{2,3,4}$ given above equation \eqref{eqn:discZ4ex} implies that
    \begin{align}
        \deg_y(A_{z_1,z_2})_{1,1}=3 \,,\quad \deg_y(A_{z_1,z_2})_{1,i\neq1}=2 \,, \quad
        \deg_y(A_{z_1,z_2})_{i\neq1,j\neq 1}=1 \,. \quad
    \end{align}
    Hence, $\det A_{z_1,z_2}$ is a bidegree (4,6) polynomial in the variables $z_{1,2}$, $y_{1,\ldots,3}$. It follows that
    \begin{equation}
        \det A_{z_1,z_2} \in \Gamma(\pi_1^* \cO_{\IP^1}(2)^{\otimes 2} \otimes \pi_2^* \cO_{\IP^2}(3)^{\otimes 2})\,, 
    \end{equation}
    where $\pi_i$ for $i=1,2$ is respectively the projection from $W$ to $\mathbb{P}^1$ and $\mathbb{P}^2$.
    The determinant is thus a section of the anti-canonical bundle of $W$, which is exactly what is needed for the double cover to be Calabi-Yau (see e.g. \cite[Section 2.1]{MirandaSmooth}). 
    
    We can see the 2-section of $X^2$ explicitly by introducing polynomials $e_i(y_1,y_2,y_3)$, $i=0,\ldots,4$, such that
    \begin{align} \label{eq:coeffDoubleCoverQuartic}
    y^2 = \det A_{z_1,z_2}= e_0 \,z_1^4 + 4e_1 \,z_1^3 z_2 + 6e_2\, z_1^2 z_2^2 + 4e_3\, z_1 z_2^3 + e_4\, z_2^4 \,.
    \end{align}
    Two linearly equivalent 2-sections of the variety are then given by $\{\,z_1=0\,\}$ and $\{\,z_2=0\,\}$.
    
    \item The variety $X^0$ can be obtained by replacing the fibers of $X^2$ with the corresponding Weierstra{\ss} curves, using the results from~\cite{Weil1983} as cited in~\cite{An2001}.
\end{itemize}

In order to verify our conjectures in this example, we have to proceed in a slightly non-linear fashion.
We first focus on verifying that $X^0$ satisfies the conditions~\ref{condge:smoothqfac} and~\ref{condge:disc} while at the same time showing that Conjectures~\ref{conj:TSandTorsion} and~\ref{conj:twisted} hold for the different choices of $\tse,\tseOne,\tseTwo\in\mathbb{Z}_4$.

\subsection{Almost genericity and Conjectures~\ref{conj:TSandTorsion} and~\ref{conj:twisted}}
Using the explicit expression for the discriminant locus $\Delta\subset \mathbb{P}^2$, which one obtains by constructing $X^0$, one can check that it is reduced and irreducible.
After making a sufficiently generic choice for the coefficients in~\eqref{eqn:discZ4ex}, one can determine by explicit calculation that it has 426 isolated singularities, of which 216 are cusps and 210 are nodes.\footnote{We will obtain these numbers more elegantly below.}
Therefore, $X^0$ satisfies condition~\ref{condge:disc}.

The base of $X^0$ is $\IP^2$, hence smooth.
It thus remains to show $\IQ$-factoriality of $X^0$ to obtain~\ref{condge:smoothqfac}.
By a theorem of Namikawa and Steenbrink, quoted in Appendix~\ref{sec:nodalCY3} as Theorem~\ref{thm:NamikawaSteenbrinkDeformation}, the $\IQ$-factoriality of $X^0$ is equivalent to $b_2(\widehat{X}^0)=b_2(X^0)$.
Since $X^0$ is a Weierstra{\ss} fibration, we have $b_2(X^0)=2$.
On the other hand, $b_2(\widehat{X}^0)>2$ would imply that $\widehat{X}^0$, and therefore also $X^0$, has a non-zero Mordell-Weil rank.
However, since the group of sections of $X^0$ acts on the $4$-sections of $X^1$, this contradicts $b_2(X^1)=2$.
Therefore $b_2(\widehat{X}^0)=2$ and $X^0$ is $\mathbb{Q}$-factorial.

Verifying Conjecture~\ref{conj:TSandTorsion} requires knowing $\Sh_{\IP^2}(X^0)$ as well as knowledge of the torsion in the 3-cohomology of the analytic small resolutions of all elements of $\Sh_{\IP^2}(X^0)$. We compute $\Sh_{\IP^2}(X^0)$ by considering the smooth element $X^1$ of this group. The fibration of $X^1$ over $\mathbb{P}^2$ does not exhibit any non-flat or multiple fibers and it has a 4-section that is induced by the relative hyperplane divisor in the ambient $\mathbb{P}^3$-bundle. Applying a discussion analogous to that of~\cite[Example 1.18]{dg92} to $X^1$ shows that
\begin{align}
    \Sh_{\mathbb{P}^2}(X^0)=\mathbb{Z}_4\,.
    \label{eqn:ex4secTS}
\end{align}
To determine the cohomological torsion of $\widehat{X}^0$ and thus verify Conjecture~\ref{conj:TSandTorsion}\ref{en:Sh}, we use the twisted derived equivalence 
\begin{align} \label{eq:derEquivX0}
    D^b(X^1)\simeq D^b(\widehat{X}^0,\alpha)\,,
\end{align}
where $\alpha\in\text{Br}(\widehat{X}^0)$ is an element of order four.
This equivalence follows from an application of~\cite[Theorem 5.1]{Caldararu2002}.
An argument analogous to that in~\cite[Appendix A]{Katz:2022lyl}, which is in turn based on~\cite{Addington2020}, then shows that $H^3(\widehat{X}^0,\mathbb{Z})/\langle\alpha\rangle$ has to be torsion free and therefore $\tors{H_3(X^1,\IZ)}$ is indeed generated by $\alpha$ and equal to $\mathbb{Z}_4$, in agreement with~\ref{conj:TSandTorsion}\ref{en:Sh}.\footnote{In fact, this argument falls slightly short of a proof, because the relevant~\cite[Corollary 1.2]{Moulinos2019} requires the geometry to be a scheme, while $\widehat{X}^0$ is in general only a Moishezon manifold. The problem can be circumvented by following~\cite[Footnote 1]{Addington2020} and proving a compatibility between pushforwards on algebraic and topological K-theory. This compatibility is expected to hold in the case at hand, but we are not aware of any proof in the literature. We thank Sheldon Katz and Nicolas Addington for helpful discussions on this point.}

Since $X^1$ is smooth, we have $\widehat{X}^1 \cong X^1$ and $\tors{H_3(X^1,\IZ)}$ vanishes by~\eqref{eq:h2X1}. Conjecture~\ref{conj:TSandTorsion}\ref{en:Gamma} thus holds for $\gamma = \e[\frac{1}{4}]$.
An argument analogous to the proof of Proposition 3.3 in~\cite{Katz:2023zan} can be used to prove that\footnote{This will be elaborated on in~\cite{wipZ2Z2}.}
\begin{align}
    \tors{H^3(\widehat{X}^2,\mathbb{Z})}\simeq \tors{H_2(\widehat{X}^2,\mathbb{Z})}\simeq \mathbb{Z}_2\,.
\end{align}
Conjecture~\ref{conj:TSandTorsion}\ref{en:Gamma} therefore also holds for $\gamma = \e[\frac{1}{2}]$.

Let us now turn to Conjecture~\ref{conj:twisted}.
It suffices to check the claim for the generators of $\SLtwoZ$,
\begin{align}
    T=\left(\begin{array}{cc}1&1\\0&1\end{array}\right)\,,\quad S=\left(\begin{array}{cc}0&1\\-1&0\end{array}\right)\,.
\end{align}
For $T$, the claim is trivial, since $[\tseTwo]_\tseOne=[\tseTwo \tseOne]_\tseOne$.
On the other hand, the twisted derived equivalence~\eqref{eq:derEquivX0} precisely corresponds to the $S$-transformation with $\tseOne= \e[\frac{1}{4}]$ and $\tseTwo=1$.
The $S$-transformation for the case $\tseOne= \e[\frac{1}{4}]$ and $\tseTwo=\e[\frac{1}{2}]$ corresponds to the twisted derived equivalence
\begin{align} \label{eq:derEquivX2}
    D^b(X^1)\simeq D^b(\widehat{X}^2,[\tseTwo]_\tseOne)\,,
\end{align}
This was shown in~\cite{Calabrese2015}, using the results from~\cite{Kuznetsov2013}.\footnote{See Appendix~\ref{ss:X1inblowup} for the identification of $X^1$ with the variety called $X$ in \cite{Calabrese2015}.}
In the case $\tseOne= \e[\frac{1}{2}]$ and $\tseTwo=1$, in order to apply Conjecture~\ref{conj:twisted}, we first have to remove the unstable nodes as discussed in Section~\ref{sec:modularityAndDerivedEquivalences}.
This amounts to replacing $X^2$ by a generic smooth deformation $\widetilde{X}^2$.
The claim of Conjecture~\ref{conj:twisted} then follows from an application of~\cite[Theorem 5.1]{Caldararu2002}.
The $S$-transformation for all other cases is either trivial or equivalent to one of these three.

\subsection{Homology of exceptional curves and Conjecture~\ref{conj:I2fibers}}
All singular points of a Jacobian fibration lie above singular points of the discriminant divisor $\Delta=0$ on the base \cite{MirandaSmooth}. The discriminant divisor of $X^0$ has 426 singular points on the base $\IP^2$. Of these, $216 = 24\times9$ are cusps, as follows from \eqref{eq:numberCusps}; above these lie cuspidal type $II$ fibers whose points are smooth inside $X^0$. Conjecture~\ref{conj:I2fibers} is concerned with the remaining 210 points, which are nodes. We have called this set $S_{\Delta}$ in Section \ref{sec:almostgeneric}. Above these lie nodal $I_2$ fibers, each of which contains a singular point of $X^0$.

All singular points of $X^1$ and of $X^2$ likewise lie above singular points of the discriminant divisor -- this and the further claims made about these two varieties in this subsection are justified in Appendix \ref{app:discSing}. The fibers over the 216 cusps are also cuspidal in both $X^1$ and $X^2$, and contain only smooth points of the geometry.
The set of nodes of the discriminant can be decomposed as $S_\Delta=S_\Delta^{(1)}\cup S_\Delta^{(2)}$.
On $S_\Delta^{(1)}$, which contains 144 points, the four zeros of $\det A_{z_1,z_2}$ -- recall equation \eqref{eq:pencilquad} -- coincide in pairs, while at the 66 points of $S_\Delta^{(2)}$, only two zeros coincide, but the rank of $A_{z_1,z_2}$ drops by 2. The number 144 of resolved $I_2$ fibers follows from \cite[Equation 2.23]{Morrison:2014era}: it is given by
\begin{equation} \label{eq:numberI2quartic}
    4(-K_B+L)\cdot(-3K_B -L) = 4(-2K_B)(-2K_B) = 4\times6\times6 = 144 \,.
\end{equation}
Here, $2L$ designates the line bundle of which the coefficient of $z_1^4$ in \eqref{eq:coeffDoubleCoverQuartic} is a section.
Since the coefficient is a polynomial of degree six in the homogeneous coordinates on the base $B=\mathbb{P}^2$, we have $L=-K_B$.
The number 66 follows from~\cite[Theorem 1]{HarrisTu84}: the cohomology class of the degeneracy locus over which a symmetric bundle map $\phi: E \rightarrow E^*$ from a smooth complex vector bundle $E$ of rank~4 to its dual drops rank by at least 2 is given by
\begin{equation}
    -4 (c_1(E)c_2(E) - c_3(E) ) \,.
\end{equation}
The matrix $A_{z_1,z_2}$ can be interpreted as a bundle map 
\begin{align}
    A_{z_1,z_2} \,:\, \cO_{\IP^2 \times \IP^1}(-1,0) \oplus \cO_{\IP^2 \times \IP^1}(0,0)^{\oplus 3} \,\rightarrow\, \cO_{\IP^2 \times \IP^1}(2,1) \oplus \cO_{\IP^2 \times \IP^1}(1,1)^{\oplus 3} \,.
\end{align}
Using~\cite[Theorem 10]{HarrisTu84}, we can formally twist by $\cO_{\IP^2\times\IP^1}(-\tfrac{1}{2},-\tfrac{1}{2})$ and then apply~\cite[Theorem 1]{HarrisTu84}.

\begin{itemize}
    \item \textbf{Twist $\tse=\e[\tfrac{1}{4}]$:} The fibration $X^1$ exhibits resolved $I_2$-fibers over all points $p\in S_\Delta$.
    Moreover, the two rational components of an $I_2$-fiber over a point $p\in S_{\Delta}^{(k)}$, $k=1,2$, transversely intersect the 4-section respectively $4-k$ and $k$ times (see also~\cite{Oehlmann:2019ohh}).

    The fiber homology of $X^{1}$ takes the form
    \begin{align}
        H_2(X^{1},\mathbb{Z})_{\rm F}\simeq \Lambda'= \mathbb{Z}\,,
    \end{align}
    with the notation $\Lambda'$ (and $\Lambda$ which will arise momentarily) following \eqref{eq:5dchargeLattice}.
    We can choose the labels $A,B$ such that the curve classes $[C_{p,1}^A],[C_{p,1}^B] \in \Lambda'$ satisfy
    \begin{align}
        \left([C_{p,1}^A],[C_{p,1}^B]\right)=\left\{\begin{array}{cl}
            (3,1)&\text{ for }p\in S_{\Delta}^{(1)}\\[.3em]
            (2,2)&\text{ for }p\in S_{\Delta}^{(2)}
        \end{array}\right.\,.
    \end{align}
    Using the isomorphism
    \begin{align}
        u^*:\,\Lambda'&\rightarrow\Lambda=\text{ker}\left(\e[-\tfrac{1}{4}],\e[\tfrac{1}{4}]\right)\subset\mathbb{Z}\times\mathbb{Z}_4\\ q_{\KK}&\mapsto (q_{\KK},[q_{\KK}]_4)\,,
    \end{align}
    we find that this is compatible with Conjecture~\ref{conj:I2fibers} if we identify
    \begin{align}
        \chi_p=\left\{\begin{array}{cl}
            [1]_4&\text{ for }p\in S_{\Delta}^{(1)}\\[.3em]
            {[2]_4}&\text{ for }p\in S_{\Delta}^{(2)}
        \end{array}\right.\,.
        \label{eqn:exZ4chiAssumption}
    \end{align}
    \item \textbf{Twist $\tse=\e[\tfrac{2}{4}]$:} 
    The fibration $X^2$ exhibits a nodal singularity over each point $p\in S_\Delta^{(2)}$, but a resolved $I_2$-fiber over each point $p\in S_\Delta^{(1)}$.
    Each of the rational components of the $I_2$-fibers is transversely intersected once by the 2-section on $X^2$.
    On the other hand, in a small resolution $\rho_2:\widehat{X}^2\rightarrow X^2$ each of the exceptional curves over the node in $X^2$ over a point $p\in S_\Delta^{(2)}$ represents the non-trivial 2-torsion class.

    The fiber homology of $\widehat{X}^{2}$ takes the form
    \begin{align}
        H_2(\widehat{X}^{2},\mathbb{Z})_{\rm F}\simeq \Lambda'= \mathbb{Z}\times\mathbb{Z}_2\,,
    \end{align}
    and we can choose the labels $A,B$ such that
    \begin{align}
        \Lambda'\ni\left([C_{p,2}^A],[C_{p,2}^B]\right)=\left\{\begin{array}{cl}
            \left((1,[0]_2),\,(1,[1]_2)\right)&\text{ for }p\in S_{\Delta}^{(1)}\\[.3em]
            \left((2,[0]_2),\,(0,[1]_2)\right)&\text{ for }p\in S_{\Delta}^{(2)}
        \end{array}\right.\,.
    \end{align}
    Using the isomorphism
    \begin{align}
        u^*:\,\Lambda'\longrightarrow\Lambda=\text{ker}\left(\e[-\tfrac{2}{4}],\e[\tfrac{2}{4}]\right)\subset\mathbb{Z}\times\mathbb{Z}_4\,,
    \end{align}
    acting as
    \begin{align}
        \quad \left(q_{\KK},[q_{\IZ_2}]_2\right)\longmapsto \left(q_{\KK},[2q_{\IZ_2}-q_{\KK}]_4\right)\,,
    \end{align}
    we find that
    \begin{align}
        \Lambda\ni\left([C_{p,2}^A],[C_{p,2}^B]\right)=\left\{\begin{array}{cl}
            \left((1,[3]_4),\,(1,[1]_4)\right)&\text{ for }p\in S_{\Delta}^{(1)}\\[.3em]
            \left((2,[2]_4),\,(0,[2]_4)\right)&\text{ for }p\in S_{\Delta}^{(2)}
        \end{array}\right.\,.
    \end{align}
    Assuming again~\eqref{eqn:exZ4chiAssumption}, this confirms that Conjecture~\ref{conj:I2fibers} also holds for $X^2$.

    \item \textbf{Twist $\tse=\e[\tfrac{0}{4}]$:} The fibration $X^0$ exhibits a nodal singularity over each point of both $S_{\Delta}^{(1)}$ and $S_{\Delta}^{(2)}$. We have already established that
    \begin{align}
        H_2(\widehat{X}^{0},\mathbb{Z})_{\rm F}=\mathbb{Z}\times\mathbb{Z}_4\,.
    \end{align}
    Together with~\eqref{eqn:exZ4chiAssumption}, Conjecture~\ref{conj:I2fibers} now predicts that
    \begin{align} \label{eq:ExCurveClassesX0}
        \Lambda\ni\left([C_{p,0}^A],[C_{p,0}^B]\right)=\left\{\begin{array}{cl}
            \left((1,[\pm 3]_4),\,(0,[\mp 1]_4)\right)&\text{ for }p\in S_{\Delta}^{(1)}\\[.3em]
            \left((1,[2]_4),\,(0,[2]_4)\right)&\text{ for }p\in S_{\Delta}^{(2)}
        \end{array}\right.\,,
    \end{align}
    where the sign depends on the particular choice of $\widehat{X}^0$ and changes under flop transitions.
\end{itemize}

\subsection{Intersection numbers and Conjecture~\ref{conj:topologyOfSmoothing}}

Let us now verify Conjecture~\ref{conj:topologyOfSmoothing} for the different choices of twist $\tse\in\mathbb{Z}_4$. We have determined
\begin{equation}
    |S_{\Delta}^{(1)}| = 144 \,, \quad |S_{\Delta}^{(2)}| = 66
\end{equation}
above. Together with the prediction \eqref{eq:ExCurveClassesX0} for the  classes of the exceptional curves of $X^0$ and with equation \eqref{eq:uncharged} which computes the number of uncharged hypermultiplets, this determines the values of the constants $N_\chi$ introduced in \eqref{eqn:Nchi} to be
\begin{equation}
    N_0 = 126 \,, \quad N_1 = N_3 = 144 \,, \quad N_2 = 132 \,,
\end{equation}
where we have written $N_i := N_{[i]_4} $.

We choose a basis of the K\"ahler cone $J_0,J_1$ such that $J_1$ is the vertical divisor and $J_0$ is the class of an $m$-section, potentially shifted by multiples of $J_1$. After changing to the basis
\begin{align}
    D_0=J_0-\frac{1}{2m}\tilde{\pi}^*\circ\tilde{\pi}_*\left(J_0\cdot J_0\right)=J_0-\frac{c_{001}}{2m}J_1\,,\quad D_1=J_1\,,
\end{align}
as in~\eqref{eqn:conj3changeOfBasis}, we can then compare
\begin{align}
    k_{\lii\lij\lik}=D_\lii\cdot D_\lij\cdot D_\lik\,,\quad \kappa_\lii=D_\lii\cdot c_2(T\widetilde{X})\,,
\end{align}
with the results~\eqref{eqn:twistedK} and~\eqref{eqn:twistedKappa}.
By construction we have $\kappa_{001}=\kappa_{111}=0$ and $\kappa_{011}=m$.
Since $X$ is genus one fibered over $\mathbb{P}^2$ we also have $b_1=\kappa_1=36$.
The only non-trivial invariants to compare are therefore
\begin{align}
    k_{000}=c_{000}-\frac{3}{4m}c_{001}^2\,,\quad \kappa_0=b_0-36\frac{c_{001}}{2m}\,.
    \label{eqn:exz4p2changebasis}
\end{align}

\begin{itemize}
    \item The geometry $X^1$ is already smooth and therefore $\widetilde{X}^1=X^1$.
        A basis of the K\"ahler cone is given by the divisors $J_0=\{\,x_2=0\,\}$ and $J_1=\{\,y_1=0\,\}$ and the corresponding topological invariants~\eqref{eqn:topInvcandb} are
        \begin{align}
            (c_{000},c_{001},c_{011},c_{111})=(9,\,8,\,4,\,0)\,,\quad (b_0,b_1)=(54,\,36)\,,
            \label{eqn:exz4topinv}
        \end{align}
        while the Euler characteristic is $\chi=-120$.
    Since $S_{\Delta,1}=\emptyset$ and $\# S_\Delta=210$, we find from~\eqref{eqn:conj3euler} that indeed $\chi(\widetilde{X}^1)=-120$.
    On the other hand, using~\eqref{eqn:exz4p2changebasis} with~\eqref{eqn:exz4topinv} and $m=4$ we find agreement with the prediction
     \begin{align}
        k_{000}=-3\,,\quad \kappa_0=18\,,
    \end{align}
    from~\eqref{eqn:twistedKM}, using $N_1=N_3=144$ and $N_2=132$.

    \item The smooth deformation $\widetilde{X}^2$ of $X^2$ is a generic Calabi-Yau double cover of $W=\mathbb{P}^1\times\mathbb{P}^2$.
    Denoting the projection by $\tilde{\pi}_{\rm dc}:\widetilde{X}^2\rightarrow W$ and using again the homogeneous coordinates $z_1,z_2,y_{1,2,3}$ on $W$ we have a basis of the K\"ahler cone
    \begin{align}
        J_0=\tilde{\pi}_{\rm dc}^*(\{\,z_1=0\,\})\,,\quad J_1=\tilde{\pi}_2^*(\{\,y_1=0\,\})\,,
    \end{align}
    where $\tilde{\pi}_2=\pi_W\circ\tilde{\pi}_{\rm dc}$ and $J_0$ is a 2-section on $\widetilde{X}^2$.
    The Euler characteristic of $\widetilde{X}^2$ is $\chi=-252$ and the other topological invariants~\eqref{eqn:topInvcandb} are
    \begin{align}
        (c_{000},c_{001},c_{011},c_{111})=(0,\,0,\,2,\,0)\,,\quad (b_0,b_1)=(24,\,36)\,.
    \end{align}

    We also have $S_{\Delta,2}=66$ and therefore~\eqref{eqn:conj3euler} correctly reproduces the Euler characteristic of $\widetilde{X}^2$.
    The change of basis~\eqref{eqn:exz4p2changebasis} is trivial and we can immediately verify that
    \begin{align}
        k_{000}=0\,,\quad \kappa_0=24\,,
    \end{align}
    matches with the prediction from~\eqref{eqn:twistedKM}.

    \item The smooth deformation of $X^0$ is the generic Weierstra{\ss} fibration $\widetilde{X}^0$ over $\mathbb{P}^2$.
    The latter can be constructed as a resolution of a degree 18 hypersurface in $\mathbb{P}^4_{11169}$ and has been studied for example in~\cite{Candelas:1994hw,Alim:2012ss,Huang:2015sta}.
    The Euler characteristic is $\chi=-540$ and in terms of a suitable basis $J_0,J_1$ for the K\"ahler cone one finds the topological invariants~\eqref{eqn:topInvcandb} to be
    \begin{align}
        (c_{000},c_{001},c_{011},c_{111})=(9,\,3,\,1,\,0)\,,\quad (b_0,b_1)=(102,\,36)\,.
    \end{align}
    Using that $S_{\Delta}=S_{\Delta,0}$, we find that the Euler characteristic is correctly reproduced by~\eqref{eqn:conj3euler}.
    Moreover, after the change of basis~\eqref{eqn:exz4p2changebasis} we find the invariants
    \begin{align}
        k_{000}=\frac94\,,\quad \kappa_0=48\,,
    \end{align}
    that match with the prediction from~\eqref{eqn:twistedKM}.
\end{itemize}

\subsection{Modularity}
After having verified the conjectures made in Section~\ref{sec:mathresults}, we will now illustrate the modular properties of the topological string partition function that have been discussed in Section~\ref{sec:ellipticGeneraAndTopologicalStrings}.
The technical details of the underlying calculations can be found in~\cite[Section 7]{Schimannek:2021pau}.\footnote{We use a slightly different 4-section fibration in our example in order to be able to use the Lefschetz hyperplane theorem, but the relevant calculations to obtain the base degree one topological string partition functions are essentially the same.}

Recall that the modular groups $\Gamma^{\tseOne}_{\tseTwo}$ for the different twistings and twinings $\tseOne,\tseTwo\in\mathbb{Z}_4$ are listed in Table~\ref{tab:twistSymZ4}, and the corresponding rings of modular forms discussed in Appendix~\ref{app:modularJacobi}.
The partition functions fall into three different $\SLtwoZ$ orbits\footnote{The orbit of an integral vector under $\SLtwoZ$ is uniquely determined by the gcd of the two entries.} that correspond to three different modular curves and stringy K\"ahler moduli spaces.
The three stringy K\"ahler moduli spaces are not connected by a continuous deformation, at least not directly,\footnote{They can be connected by sequences of conifold transitions that are not compatible with the torus fibration structure.} but by a discrete change of B-field topology that is possible at certain points in the moduli space.
The modular curves arise as slices in the large base limit of the moduli spaces~\cite{Knapp:2021vkm,Schimannek:2021pau}.
The relationship of the different moduli spaces and the relevant large volume limits is illustrated in Figure~\ref{fig:exz4geos} that has been adapted from~\cite[Figure 9]{Schimannek:2021pau}.

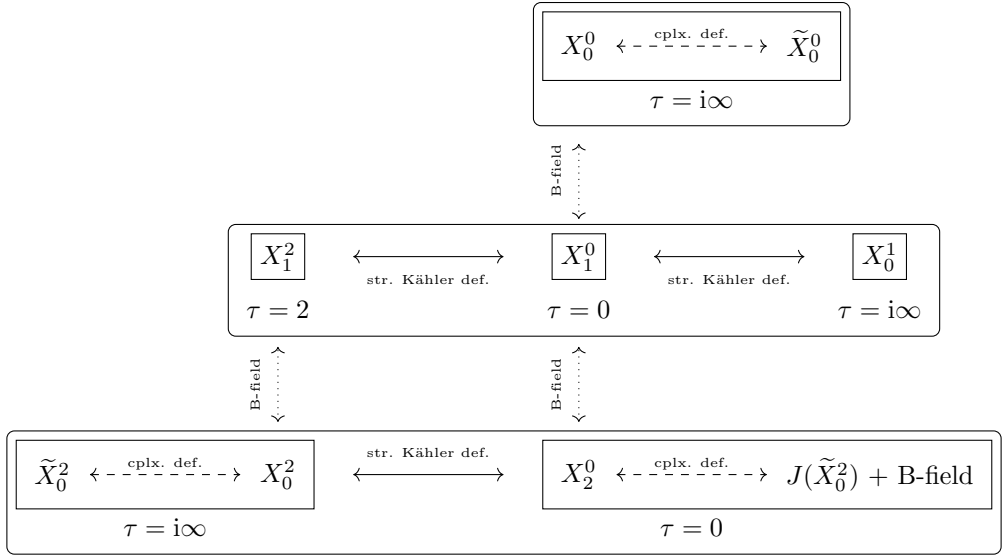
\begin{figure}[h!]
	\begin{tikzpicture}[remember picture,node distance=4mm]
            \begin{scope}[shift={(0,0)}]
		      \node[align=center] (x00) at (5,5) {$X^0_0$};
		      \node[align=center] (xt00) at (8,5) {$\widetilde{X}^0_0$};
                \node[align=center] (l00) at (6.5,4.3) {$\tau={\rm i}\infty$};
		      \node[align=center] at (6.5,5.15) {\tiny cplx. def.};
                \node[draw,fit=(x00) (xt00)] (ob00) {};
                \node[draw,fit=(ob00) (l00),rounded corners=0.1cm] {};
		      \draw [<->,dashed] (5.5,5) to (7.5,5);
            \end{scope}
            \begin{scope}[shift={(0,-.65)}]
                \node[align=center,rotate=90] at (4.7,4) {\tiny B-field};
            \end{scope}
            \begin{scope}[shift={(0,-.8)}]
                \node[draw,align=center] (x21) at (1,3) {$X^2_1$};
                \node[align=center] (l21) at (1,2.3) {$\tau=2$};
		      \node[draw,align=center] (x01) at (5,3) {$X^0_1$};
                \node[align=center] (l01) at (5,2.3) {$\tau=0$};
		      \node[draw,align=center] (x10) at (9,3) {$X^1_0$};
                \node[align=center] (l10) at (9,2.3) {$\tau={\rm i}\infty$};
		      \node[align=center] at (3,2.7) {\tiny str. K\"ahler def.};
		      \node[align=center] at (7,2.7) {\tiny str. K\"ahler def.};
		      \draw [<->] (2,3) to (4,3);
		      \draw [<->] (6,3) to (8,3);
		      \draw [<->,dotted] (5,3.5) to (5,4.5);
                \node[draw,fit=(x21) (x01) (x10) (l21) (l01) (l10),rounded corners=0.1cm,minimum height=1cm] {};
            \end{scope}
            \begin{scope}[shift={(0,-1.5)}]
		      \draw [<->,dotted] (1,1.5) to (1,2.5);
		      \node[align=center,rotate=90] at (.7,2) {\tiny B-field};
		      \draw [<->,dotted] (5,1.5) to (5,2.5);
		      \node[align=center,rotate=90] at (4.7,2) {\tiny B-field};
            \end{scope}
            \begin{scope}[shift={(0,-1.7)}]
		      \node[align=center] (x20) at (1,1) {$X^2_0$};
		      \node[align=center] (x02) at (5,1) {$X^0_2$};
		      \node[align=center] (xt20) at (-2,1) {$\widetilde{X}^2_0$};
		      \node[align=center] (xt02) at (9,1) {$J(\widetilde{X}^2_0)$ + B-field};
		      \draw [<->,dashed] (-1.5,1) to (0.5,1);
		      \draw [<->,dashed] (5.5,1) to (7.5,1);
		      \node[align=center,rotate=0] at (-0.5,1.15) {\tiny cplx. def.};
		      \node[align=center,rotate=0] at (6.5,1.15) {\tiny cplx. def.};
		      \draw [<->] (2,1) to (4,1);
		      \node[align=center] at (3,1.3) {\tiny str. K\"ahler def.};
                \node[align=center] (l20) at (-0.5,0.3) {$\tau={\rm i}\infty$};
                \node[align=center] (l02) at (6.5,0.3) {$\tau=0$};
                \node[draw,fit=(x20) (xt20)] (ob20) {};
                \node[draw,fit=(x02) (xt02)] (ob02) {};
                \node[draw,fit=(ob20) (ob02) (l20) (l02),rounded corners=0.1cm] {};
            \end{scope}
            
	\end{tikzpicture}
	\centering
	\caption{Relationship between the three stringy K\"ahler moduli spaces and the relevant large volume limits in the example with $\Gamma^0=\mathbb{Z}_4$.
            Geometries that correspond to the same large volume limit are surrounded by square boxes while large volume limits that are contained in the same stringy K\"ahler moduli space are surrounded by rounded boxes.
            In each case, the large base limit $t\rightarrow {\rm i}\infty$ is left implicit.
        }
	\label{fig:exz4geos}
\end{figure}

We obtain the base degree one topological string partition functions
\begin{align} \label{eq:ZtopEZ4}
    \begin{split}
    \begin{array}{rlcrl}   Z\left[X^0_0\right]_1(\tilde{\tau},\lambda)=&\mathbb{E}{\tiny\left[\begin{array}{c}0\\0\end{array}\right]}_1\left(\tilde{\tau},\lambda\right)\,,&\phantom{xx}&
    Z\left[X^1_0\right]_1(\tilde{\tau},\lambda)=&\mathbb{E}{\tiny\left[\begin{array}{c}1\\0\end{array}\right]}_1\left(4\tilde{\tau},\lambda\right)\,,\\
    Z\left[X^0_1\right]_1(\tilde{\tau},\lambda)=&\mathbb{E}{\tiny\left[\begin{array}{c}0\\1\end{array}\right]}_1\left(\tilde{\tau},\lambda\right)\,,&&
    Z\left[X^2_0\right]_1(\tilde{\tau},\lambda)=&\mathbb{E}{\tiny\left[\begin{array}{c}2\\0\end{array}\right]}_1\left(2\tilde{\tau},\lambda\right)\,,\\
    Z\left[X^0_2\right]_1(\tilde{\tau},\lambda)=&\mathbb{E}{\tiny\left[\begin{array}{c}0\\2\end{array}\right]}_1\left(\tilde{\tau},\lambda\right)\,,&&
    Z\left[X^2_1\right]_1(\tilde{\tau},\lambda)=&\mathbb{E}{\tiny\left[\begin{array}{c}2\\1\end{array}\right]}_1\left(2\tilde{\tau}+\frac12,\lambda\right)\,,
    \end{array}
    \end{split}
\end{align}
in terms of the twisted-twined elliptic genera
\begin{align}
    \begin{split}
    \begin{array}{rlcrl}
    \mathbb{E}{\tiny\left[\begin{array}{c}0\\0\end{array}\right]}_1(\tau,\lambda)=&-\frac{E_4(31E_4^3+113E_6^2)}{48\eta(\tau)^{36}\phi_{-2,1}(\tau,\lambda)}\,,&\phantom{xx}&
    \mathbb{E}{\tiny\left[\begin{array}{c}1\\0\end{array}\right]}_1\left(\tau,\lambda\right)=&\frac{1}{2^{13}}\frac{3\tilde{\epsilon}_1^4(\tilde{\epsilon}_2-\tilde{\epsilon}_1^2)^5(7\tilde{\epsilon}_1^2-6\tilde{\epsilon}_2)}{\eta(\tau)^{36}\phi_{-2,1}(\tau,\lambda)}\,,\\
    \mathbb{E}{\tiny\left[\begin{array}{c}0\\1\end{array}\right]}_1\left(\tau,\lambda\right)=&\frac{1}{2^{13}}\frac{3\epsilon_1^4(\epsilon_2-\epsilon_1^2)^5(7\epsilon_1^2-6\epsilon_2)}{\eta(\tau)^{36}\phi_{-2,1}(\tau,\lambda)}\,,&&
    \mathbb{E}{\tiny\left[\begin{array}{c}2\\0\end{array}\right]}_1(\tau,\lambda)=&\frac{3\Delta_{2,3}^3(\tilde{e}_4-17\tilde{e}_2^2)}{16\eta(\tau)^{36}\phi_{-2,1}(\tau,\lambda)}\,,\\
    \mathbb{E}{\tiny\left[\begin{array}{c}0\\2\end{array}\right]}_1(\tau,\lambda)=&\frac{3\Delta_{3,4}^3(e_4-17e_2^2)}{16\eta(\tau)^{36}\phi_{-2,1}(\tau,\lambda)}\,,&&
    \mathbb{E}{\tiny\left[\begin{array}{c}2\\1\end{array}\right]}_1\left(\tau,\lambda\right)=&\frac{\rm i}{2^{13}}\frac{3\hat{\epsilon}_1^4(\hat{\epsilon}_2-\hat{\epsilon}_1^2)^5(7\hat{\epsilon}_1^2-6\hat{\epsilon}_2)}{\eta(\tau)^{36}\phi_{-2,1}(\tau,\lambda)}\,,
    \end{array}
    \end{split}
\end{align}
and the modular and Jacobi forms defined in Appendix~\ref{app:modularJacobi}. Note that in accord with~\eqref{eqn:EgenusTopString} and Conjecture~\ref{conj:I2fibers}, the coefficient $m$ of $\tilde{\tau}$ on the RHS of the equations~\eqref{eq:ZtopEZ4} indicates that $X^\tseOne_\tseTwo$ exhibits an $m$-section, and no $m'$-section for $m'<m$. The shift of the argument in the final equality of~\eqref{eq:ZtopEZ4} comes about as 
\begin{equation}
    \frac{q_{\e[\tfrac{1}{4}]}\left(2\chi_{\e[\tfrac{1}{2}]}\right)}{4} = \frac{1}{2} \chi_{\e[\tfrac{1}{2}]}(\e[\tfrac{1}{4}]) = \frac{1}{2} \,.
\end{equation}
All of the cases not explicitly listed in~\eqref{eq:ZtopEZ4} can be easily obtained by shifting $\tau\rightarrow \tau+1$ and using~\eqref{eqn:ttEllGenMod}, as well as the invariance under charge conjugation~\eqref{eqn:ellGenChargeInversion}.

Comparing with Table~\ref{tab:twistSymZ4}, we find that all of the elliptic genera transform as meromorphic Jacobi forms under the correct modular group.
Using the transformation behavior discussed in Appendix~\ref{app:modularJacobi}, one can also check that
\begin{align}
    \begin{split}
    \mathbb{E}{\tiny\left[\begin{array}{c}2\\1\end{array}\right]}_1\left(\tau,\lambda\right)\propto \mathbb{E}{\tiny\left[\begin{array}{c}1\\0\end{array}\right]}_1\left(\frac{2\tau+1}{\tau+1},\frac{\lambda}{\tau+1}\right)\,,\quad 
    \mathbb{E}{\tiny\left[\begin{array}{c}0\\1\end{array}\right]}_1\left(\tau,\lambda\right)\propto \mathbb{E}{\tiny\left[\begin{array}{c}1\\0\end{array}\right]}_1\left(-\frac{1}{\tau},\frac{\lambda}{\tau}\right)\,,
    \end{split}
\end{align}
as well as
\begin{align}
    \mathbb{E}{\tiny\left[\begin{array}{c}0\\2\end{array}\right]}_1\left(\tau,\lambda\right)\propto \mathbb{E}{\tiny\left[\begin{array}{c}2\\0\end{array}\right]}_1\left(-\frac{1}{\tau},\frac{\lambda}{\tau}\right)\,,
\end{align}
as is expected from~\eqref{eqn:ttEllGenMod}.
Some of the Gopakumar-Vafa invariants for $X^1$ and the $\mathbb{Z}_2$-refined Gopakumar-Vafa invariants for $X^2$ can be found in~\cite[Appendix D]{Katz:2022lyl}.

\section{Outlook and open questions}
\label{sec:questions}

In this paper, we have analyzed twisted circle compactifications of gravitational theories exhibiting discrete gauge symmetries and explored their implications for F-theory compactifications on a class of geometries that we call almost generic genus one fibered Calabi-Yau threefolds.

 Our analysis directly applies to compactifications with six-dimensional gauge groups that are arbitrary products of cyclic groups. The example we studied in detail in Section \ref{sec:exampleZ4} corresponds to the gauge group $\Gsd=\IZ_4$. No examples of six-dimensional F-theory compactifications with non-cyclic discrete gauge groups have been studied so far in the literature.
Part of our motivation for this paper is to lay the groundwork for the upcoming~\cite{wipZ2Z2}, which aims to fill this gap by constructing explicit examples of almost generic genus one fibered Calabi-Yau threefolds that lead to F-theory vacua with gauge group $\Gsd=\IZ_2\times\IZ_2$.

While we have obtained detailed results about the geometry of the Calabi-Yau threefolds that correspond to F-theory vacua with only discrete gauge symmetries, it is natural to ask which discrete groups can actually appear as gauge groups of F-theory compactifications on Calabi-Yau threefolds.
The fact that the set of possible groups is finite follows from the finiteness of the set of elliptically fibered Calabi-Yau threefolds~\cite{Gross:1993fd,Filipazzi:2021dcw}.
Explicit constructions exist in the literature for $\Gsd=\IZ_2$~\cite{Morrison:2014era,Klevers:2014bqa,Mayrhofer:2014laa,Pioline:wip}, $\mathbb{Z}_3$~\cite{Klevers:2014bqa,Cvetic:2015moa,Dierigl:2022zll,Pioline:wip}, $\IZ_4$~\cite{Braun:2014qka,Oehlmann:2019ohh,Pioline:wip}, $\IZ_5$~\cite{Knapp:2021vkm,Pioline:wip} and soon for $\IZ_2\times\IZ_2$~\cite{wipZ2Z2}.
As we discussed in Section~\ref{sec:ellipticGeneraAndTopologicalStrings}, the stringy K\"ahler moduli space of the $\tseOne=\e[1/N]$ twisted and $\tseTwo=1$ twined compactification of a six-dimensional F-theory vacuum with gauge group $\Gsd=\ldots\times \IZ_N$ should correspond to a Calabi-Yau threefold that has a stringy K\"ahler moduli space which in the large base limit contains the modular curve $\IH/\Gamma_1(N)$.
However, as was also pointed out in~\cite[Section 7]{Knapp:2021vkm}, this modular curve is only simply connected for $N\le 10$ and $N=12$.
The conjecture that moduli spaces in quantum gravity have to be simply connected~\cite{Ooguri:2006in} then suggests that these might be the only values for $N$ such that $\IZ_N$ can appear as a factor of $\Gsd$.
However, we cannot exclude the possibility that a non-trivial 1-cycle on the modular curve is contracted in the interior of the moduli space.
The question of the possible discrete gauge symmetries is therefore still wide open.
The bounds on six-dimensional supergravities and F-theory vacua, without explicitly taking into account the possible discrete gauge symmetries, are under intense investigation, see e.g.~\cite{Kumar2009,Kumar2010,Kumar2011,Seiberg2011,Kim2019,Tarazi:2021duw,Grassi:2023aks,Hamada2024a,Hamada2024b,Kim:2024hxe,Hamada:2025vga,Birkar:2025rcg,Birkar:2025gvs}.
We hope that our results provide a starting point to also address the bounds on the discrete gauge sector.

In Section~\ref{sec:genericity}, we have discussed the relationship between almost generic F-theory vacua without massless vector multiplets and almost generic elliptically fibered Calabi-Yau threefolds.
While every almost generic elliptically fibered Calabi-Yau threefold corresponds to such an F-theory vacuum, we have encountered various subtleties that could potentially lead to counterexamples for the converse statement.
We believe that the almost generic F-theory vacua with finite gauge groups form a rich yet tractable class of theories that allow to explore question related to string universality, global anomalies and the landscape of Calabi-Yau threefolds.
It would therefore be useful to further understand the class of the corresponding geometries, expanding our analysis from Section~\ref{sec:genericity}.

Given the already significant complexity of this case, we have restricted ourselves to F-theory vacua that do not exhibit any massless vector multiplets and are at a sufficiently generic point in the space of vacuum expectation values of both the hypermultiplet and tensor multiplet scalars.
Our techniques can easily be adapted to include massless vector multiplets and therefore study theories where the six-dimensional gauge group is not finite.
On the other hand, at special points in the moduli space of the tensor multiplet scalars, the corresponding genus one fibered Calabi-Yau threefolds are expected to exhibit non-flat and/or non-reduced fibers.
As one of the consequences, the Tate-Shafarevich group is then only a subgroup of the Weil-Ch\^atelet group.
Examples of such theories have been discussed for example in~\cite{deBoer:2001wca,Lawrie:2012gg,Braun:2013nqa,Borchmann:2013hta,Bhardwaj:2015oru,Buchmuller:2017wpe,Anderson:2018heq,Apruzzi:2018nre,Dierigl:2018nlv,Oehlmann:2019ohh,Apruzzi:2019opn,Apruzzi:2019enx,Oehlmann:2019ohh,Anderson:2023wkr,Anderson:2023tfy,Ahmed:2024wve}, but still pose many open questions.
It would be very interesting to try to expand our analysis to such cases and thus obtain a more systematic understanding of the physical role of the Weil-Ch\^atelet group and its relationship to the Tate-Shafarevich group.
Again, we leave this to future work.

\section*{Acknowledgments}

\noindent
We thank Markus Dierigl, Antonella Grassi, Thomas Grimm, Sheldon Katz, Albrecht Klemm, Johanna Knapp, Paul Oehlmann, Boris Pioline, Emanuel Scheidegger, Eric Sharpe, Washington Taylor and Timo Weigand for helpful discussions.
We also thank Paul Oehlmann for comments on the draft.
The work of AKKP is supported under ANR grant ANR-21-CE31-0021.
The work of TS was also supported during part of the work on this project by the same grant.

\appendix

\section{Calabi-Yau threefolds with isolated nodes}
\label{sec:nodalCY3}

In this appendix, we will review some of the properties of projective Calabi-Yau threefolds with isolated nodal singularities.

We first introduce some relevant notions.
Given a projective variety $X$ with singular locus $S\subset X$, a birational morphism $\pi:\widehat{X}\rightarrow X$ is called a resolution if $\widehat{X}$ is smooth and an isomorphism over $X\backslash S$. We will often refer to $\widehat{X}$ itself as the resolution with the morphism being left implicit.
Any component of $\pi^{-1}(S)$ that is of codimension one in $\widehat{X}$ is called an exceptional divisor of the resolution.
If the resolution does not have any exceptional divisors, it is called small, otherwise, it is called large.
We call a smooth variety $\widetilde{X}$ a smoothing or a smooth deformation of $X$ if the two varieties are related by a deformation of the complex structure, i.e. of the defining equations.
For a nice overview of the general properties of singular Calabi-Yau threefolds, we refer to~\cite[Appendix A]{Cox1999-nj} and~\cite{Arras:2016evy}.

We will now review the local geometry of isolated nodal singularities on threefolds. These are also called threefold ordinary double points or conifold singularities. They are the simplest isolated surface singularities.
Locally, the geometry around a node is just the well known conifold
\begin{align}
    V=\{\,uv-zw=0\,\}\subset\mathbb{C}^4\,,
    \label{eqn:conifold}
\end{align}
with the singularity being at the origin $u=v=z=w=0$. One way to projectively resolve this space is to blow-up the singular point. This however introduces an exceptional divisor that renders the canonical class non-trivial.
The resulting space of this large resolution is therefore not Calabi-Yau, see e.g.~\cite[Section 4.1]{Davies:2009ub}.
There also exist two small resolutions of the conifold: these are obtained by blowing up along divisors inside $V$, namely
\begin{align}
    \widehat{V}=\text{Bl}_{\{u=z=0\}}V\,,\quad \widehat{V}'=\text{Bl}_{\{u=w=0\}}V\,.
\end{align}
The two small resolutions are related by the Atiyah-flop and are both isomorphic to the total space of $\mathcal{O}_{\mathbb{P}^1}(-1)\oplus\mathcal{O}_{\mathbb{P}^1}(-1)$.
In particular, both of the small resolutions are projective.

We now consider the global situation. Let $X$ be a projective threefold $X$ with $n$ isolated conifold singularities.
When considered as an analytic variety, one can choose a neighborhood for each node on which $X$ is described (upon choice of appropriate holomorphic coordinates) by the equation~\eqref{eqn:conifold}.
As a result, one can always define $2^n$ analytic small resolutions $\pi:\widehat{X}\rightarrow X$ that are related by flopping the $n$ individual exceptional curves.
However, in general, none of these analytic small resolutions is a projective variety.

Note that every non-singular complex projective variety admits a K\"ahler metric that is induced by the Fubini–Study metric after embedding into projective space.
Conversely, a compact K\"ahler variety of dimension $d\ge 3$ with holonomy ${\rm SU}(d)$ is projective~\cite[Proposition 1]{Beauville1983}.
In the context of compact Calabi-Yau threefolds, we can therefore use the terms projective small resolution and K\"ahler small resolution interchangeably.

Consider for example the nodal quintic
\begin{align}
    X=\{\,U^{(4)}(x_1,\ldots,x_5)V^{(1)}(x_1,\ldots,x_5)-Z^{(4)}(x_1,\ldots,x_5)W^{(1)}(x_1,\ldots,x_5)=0\,\}\subset\mathbb{P}^4\,,
\end{align}
where $U^{(4)},Z^{(4)}$ are generic homogeneous polynomials of degree 4 and $V^{(1)},W^{(1)}$ are generic homogeneous polynomials of degree 1 in the homogeneous coordinates on $\mathbb{P}^4$.
By comparing with~\eqref{eqn:conifold} and invoking Bezout's theorem, it is easy to see that $X$ has 16 nodal singularities at
\begin{align}
    S=\{\,U^{(4)}=V^{(1)}=Z^{(4)}=W^{(1)}=0\,\}\subset X\,.
\end{align}
As discussed in~\cite[Section 1.1]{Candelas:1989ug}, one can construct a projective small resolution $\widehat{X}$ of $X$ as a complete intersection in $\mathbb{P}^4\times\mathbb{P}^1$.
One can easily calculate $h^{1,1}(\widehat{X})=2$ and check that all of the exceptional curves are in the same homology class in $\widehat{X}$.
The geometry $\widehat{X}'$ that one obtains by flopping all of the exceptional curves in $\widehat{X}$ simultaneously is also a projective Calabi-Yau threefold and, as discussed in~\cite[Section 4.5]{Doran:2024kcb}, can be obtained as a complete intersection in a $\mathbb{P}^1$-bundle on $\mathbb{P}^4$.

Intuitively, the flop changes the sign of the homology class of the curve.
Since in going from $\widehat{X}$ to $\widehat{X}'$, the class of all of the exceptional curves is mapped to its inverse, changing the sign of the contribution of the dual to this class to the K\"ahler class renders the volume of all exceptional curves of the flopped geometry positive.
However, if one considers any of the $2(2^{15}-1)$ small resolutions of $X$ that are related to $\widehat{X}$ by flopping only a subset of the exceptional curves, there is no possibility to choose a K\"ahler class on the resulting geometry such that the volume of all of the exceptional curves remains positive.

A general criterion for when a projective threefold $X$ with a set of isolated nodes $S$ admits a projective small resolution was given in~\cite[p.98]{WernerThesis} (see~\cite[Theorem 11.2]{werner2022smallresolutionsspecialthreedimensional} for a recent translation):
\begin{theorem}[Werner'87]
Let $X$ be a complex projective threefold with isolated nodes and $\pi:\widehat{X}\rightarrow X$ any analytic small resolution.
There exists a projective small resolution of $X$ if and only if the class of all of the exceptional curves in $\widehat{X}$ is non-trivial in $H_2(\widehat{X},\mathbb{Q})$.
\label{thm:Werner87}
\end{theorem}
In other words, a projective small resolution of $X$ exists if and only if none of the exceptional curves in any analytic small resolution $\widehat{X}$ is trivial or torsion in $H_2(\widehat{X},\mathbb{Z})$.

The homology groups of a threefold $X$ with $n$ isolated nodes $S\subset X$ and a small resolution $\rho:\widehat{X}\rightarrow X$ can be compared using the Mayer-Vietoris sequence, following~\cite{Clemens1983}.
When treated as a topological space, we can glue $3$-discs $D^3$ with $\partial D^3=C$ into each of the exceptional curves $C$ inside $\widehat{X}$ to get a topological space $X'$ that is homotopic to $X$.
The Mayer-Vietoris sequence then gives a long exact sequence in singular homology
\begin{align}
    0\rightarrow H_3(\widehat{X})\rightarrow H_3(X)\rightarrow \mathbb{Z}^n\rightarrow H_2(\widehat{X})\rightarrow H_2(X)\rightarrow 0\,.
\end{align}
This implies that
\begin{align}
    H_2(X)=H_2(\widehat{X})/\langle\{\,[\rho^{-1}(p)]\in H_2(\widehat{X})\,,\,\, p\in S\,\}\rangle\,,
\end{align}
and taking the Euler characteristic of the sequence one also obtains
\begin{align}
        b_3(X)-b_2(X)=b_3(\widehat{X})-b_2(\widehat{X})+n\,.
\end{align}
Another part of the same Mayer-Vietoris sequence also gives
\begin{align}
    0\rightarrow H_4(\widehat{X})\rightarrow H_4(X)\rightarrow 0\,,
\end{align}
such that $b_4(\widehat{X})=b_4(X)$, as should be the case for small resolutions of threefolds.

One can also use the Mayer-Vietoris sequence to compare the homology of $X$ and a smooth deformation $\widetilde{X}$, again following~\cite{Clemens1983}.
To this end, one glues $4$-cells into each of the 3-cycles that is contracted when deforming $\widetilde{X}$ back to $X$.
This gives the exact sequence in singular homology
\begin{align}
        0\rightarrow H_2(\widetilde{X})\rightarrow H_2(X)\rightarrow 0\,,
\end{align}
and shows that $b_2(X)=b_2(\widetilde{X})$.

Let us now briefly discuss the notions of terminal singularities and $\mathbb{Q}$-factoriality.
A normal variety $X$ is said to have only \textit{terminal singularities} if the canonical divisor $K_X$ is $\mathbb{Q}$-Cartier and there exists a projective resolution $\rho:Y\rightarrow X$ such that
\begin{align}
    K_Y=\rho^*K_X+\sum_{i}a_i E_i\,,
\end{align}
where the sum runs over all of the exceptional divisors $E_i$ and all of the coefficients $a_i$ are strictly positive~\cite[Section 3.1]{Matsuki2002}.
This is a local property. Isolated nodes in a threefold are always terminal.

On the other hand, a normal variety $X$ is said to be \textit{$\mathbb{Q}$-factorial} if all (algebraic) Weil divisors on $X$ are $\mathbb{Q}$-Cartier.
This is not a local property, it depends on the global structure of the geometry.
The variety being $\mathbb{Q}$-factorial is equivalent to the vanishing of the so-called \textit{defect}
\begin{align}
\sigma(X)=\text{rk}\left(\text{Weil}(X)/\text{Cart}(X)\right)\,.
    \label{eqn:defect}
\end{align}
The following theorem, which is part of~\cite[Theorem 3.2]{Namikawa1995}, will allow us to equate $b_2(X)$ and $b_4(X)$ when $X$ is $\IQ$-factorial:
\begin{theorem}[Namikawa, Steenbrink, `95]\label{thm:NamikawaSteenbrinkDefect}
Let X be a normal projective threefold with only isolated hypersurface singularities such that $H^2(X, \cO_X)=0$. Let $b_i(X)$ denote the $i$-th Betti number of the singular cohomology of $X$. Then 
\begin{align}
    \sigma(X)=b_4(X)-b_2(X)\,.
\end{align}
\end{theorem}

Note that if a partial K\"ahler small resolution $\widehat{X}$ of $X$ exists, then $b_2(\widehat{X})>b_2(X)$. But since $b_2(\widehat{X})=b_4(\widehat{X})=b_4(X)$, this automatically implies that $\sigma(X)>0$.
The converse fact, that $\sigma(X)>0$ implies the existence of a partial K\"ahler small resolution, follows from~\cite[Corollary 4.5]{Kawamata1988}.

Note on the other hand that if $\sigma(X)=0$, then given any small resolution $\rho:\widehat{X}\rightarrow X$, one has $b_2(\widehat{X})=b_2(X)$.
For a projective threefold $X$ with isolated nodes, being $\mathbb{Q}$-factorial is therefore equivalent to the fact that the homology classes of \textit{all} of the exceptional curves are trivial or torsion in $H_2(\widehat{X},\mathbb{Z})$.
The nodes of $X$ are then all $\mathbb{Q}$-factorial terminal singularities.
Furthermore, we have the following~\cite[Theorem 1.3]{Namikawa1995}
\begin{theorem}[Namikawa, Steenbrink, `95] \label{thm:NamikawaSteenbrinkDeformation}
    Let $X$ be a $\IQ$-factorial Calabi-Yau threefold which admits only isolated rational hypersurface singularities. Then $X$ can be deformed to a smooth Calabi-Yau threefold $\widetilde{X}$.
\end{theorem}

Physically, as was discussed for example in~\cite{Arras:2016evy,Grassi:2018rva}, the presence of $\mathbb{Q}$-factorial terminal singularities corresponds to the fact that M-theory compactified on such geometries contains localized massless matter that is uncharged or only charged under a finite part of the gauge group.
Therefore, these states remain massless on a generic point of the Coulomb branch. The Calabi-Yau threefold is hence singular and does not admit any K\"ahler small resolution.

We can therefore say, under the hypotheses of Theorem~\ref{thm:NamikawaSteenbrinkDefect}, that we are on a generic point of a chamber of the Coulomb branch if ${\sigma(X)=0}$.
The fact that from such points, $X$ can be deformed to a smooth threefold is the mathematical counterpart of the observation that going to a generic point of the Higgs branch by giving a vacuum expectation value to all of the massless scalars breaks part of the gauge symmetry such that after the transition/Higgsing, all massless states are completely uncharged.
The deformation also transforms uncharged localized matter, which arises from M2-branes localized at singularities, into unlocalized uncharged matter, which arises for example from expanding the M-theory 3-form along harmonic 3-forms on the Calabi-Yau threefold.

\section{Singular fibers of the Weierstra{\ss} fibration}
\label{sec:singularFibers}

In this section, we will discuss some properties of Weierstra{\ss} fibrations with isolated fibers of Kodaira types $I_2$ and $II$.

\paragraph{Type $I_2$ fibers over nodes of the discriminant}
We first show that over an isolated node of the discriminant, the Weierstra{\ss} fibration has a fiber of type $I_2$ and itself exhibits an isolated nodal singularity.
Let us assume that $u,v$ are local coordinates on $B$ such that the discriminant $\Delta=\{\,4f^3+27g^2=0\,\}$ has an isolated nodal singularity at $u=v=0$.
Writing $d(u,v)=4f^3+27g^2$, this means that $d$, $\partial_u d$ and $\partial_v d$ vanish at $u=v=0$,
while at the same time the determinant of the Hessian matrix is non-vanishing, i.e. $h(0,0)\ne 0$ for
\begin{align}
    h(u,v)=\det\left(\begin{array}{cc}
        \partial_u\partial_ud&\partial_u\partial_vd\\
        \partial_v\partial_ud&\partial_v\partial_vd
    \end{array}\right)\,.
\end{align}
One can check that $h(0,0)\ne 0$ implies that neither $f$ nor $g$ is allowed to vanish at the origin.
The vanishing orders of $(f,g,\Delta)$ are therefore $(0,0,2)$ and the fibration has an $I_2$-fiber over the origin.

We can make an ansatz for the Weierstra{\ss} coefficients
\begin{align} \label{eq:fandg}
    \begin{split}
        f=&a_0+a_{1,0}u+a_{0,1}v+a_{2,0}u^2+a_{1,1}uv+a_{0,2}v^2+\ldots\,,\\
        g=&b_0+b_{1,0}u+b_{0,1}v+b_{2,0}u^2+b_{1,1}uv+b_{0,2}v^2+\ldots\,,
    \end{split}
\end{align}
and use the freedom to rescale
\begin{align} 
    f\rightarrow \lambda^2f\,,\quad g\rightarrow \lambda^3g\,,\quad \lambda\in\mathbb{C}^*\,,
\end{align}
to set $a_0=-3$.
Imposing $d=\partial_u d=\partial_v d=0$ at $u=v=0$ then fixes $b_0=2$ as well as $a_{1,0}=-b_{1,0}$ and $a_{0,1}=-b_{0,1}$.

One can then check that the Weierstra{\ss} fibration $p(x,y,u,v)=0$, with
\begin{align}
    p(x,y,u,v)=y^2-\left(x^3+fx+g\right)\,,
    \label{eqn:WSappendix}
\end{align}
has a singularity at $(x,y,u,v)=(1,0,0,0)$ and that this is the only singularity over $u=v=0$.
Moreover, the determinant of the Hessian matrix of $p(x,y,u,v)$ at $(1,0,0,0)$ is proportional, as a polynomial in the coefficients of $u$ and $v$ in~\eqref{eq:fandg}, to $h(0,0)$.
The Weierstra{\ss} fibration itself therefore has an ordinary double point at $(1,0,0,0)$.

\paragraph{Type $I_2$ fibers over other singularities}
Although an isolated node of the discriminant always corresponds to an $I_2$-fiber of the Weierstra{\ss} fibration, the converse is not true in general.
To see this, let us consider again the general ansatz for the Weierstra{\ss} coefficients in a local neighborhood.

Consider again coordinates $u,v$ on $\mathbb{C}^2$ and take the Weierstra{\ss} coefficients
\begin{align}
    f=-3(1+2v)\,,\quad g=2+3v(2+v)+u^2(1-3v)\,.
\end{align}
Then the vanishing orders of $(f,g,\Delta)$ are $(0,0,2)$ at the origin.
However, the determinant of the Hessian matrix of $4f^3+27g^2$ vanishes at the origin and the discriminant has a cuspidal singularity
\begin{align}
    \Delta=\{\,0=u^2+v^3+\ldots\}\,.
\end{align}
In fact, after changing coordinates
\begin{align}
    x\rightarrow z_4+1\,,\quad u\rightarrow z_2\left(1+\frac{\sqrt{3}}{2}(z_1+\sqrt{3}z_4)\right)\,,\quad v\rightarrow z_4+\frac{1}{\sqrt{3}}z_1\,,\quad y\rightarrow {\rm i}z_3\,,
\end{align}
the Weierstra{\ss} equation
\begin{align}
    y^2=x^3-3(1+2v)x+2+3v(2+v)+u^2(1-3v)
\end{align}
takes the form
\begin{align}
    0=z_1^2+z_2^2+z_3^2+z_4^3+\mathcal{O}(z^4)\,.
\end{align}
The fibration therefore has an $A_2$ singularity at $(x,y,u,v)=(1,0,0,0)$.
This does not admit \textit{any} small resolution, not even a non-K\"ahler one.

We have seen above that for the general ansatz of an isolated $I_2$-fiber, the Hessian matrix of the Weierstra{\ss} fibration at the origin is proportional to that of the discriminant.
$I_2$-fibers over non-nodal singularities of the discriminant can therefore only occur if the Weierstra{\ss} fibration itself has a singularity that is not an isolated node.

\paragraph{Type $II$ fibers over cusps of the discriminant}
It is easy to check that isolated fibers of type $II$ correspond to cuspidal singularities of the discriminant but the Weierstra{\ss} fibration remains smooth:
consider a point where the vanishing order of $(f,g,\Delta)$ is $(\ge1,1,2)$.
We can make an ansatz
\begin{align}
    \begin{split}
        f=&a_{1,0}u+a_{0,1}v+a_{2,0}u^2+a_{1,1}uv+a_{0,2}v^2+\ldots\,,\\
        g=&b_{1,0}u+b_{0,1}v+b_{2,0}u^2+b_{1,1}uv+b_{0,2}v^2+\ldots\,,
    \end{split}
\end{align}
and without loss of generality we can assume that $b_{0,1}=0$ and $b_{1,0}\ne 0$.
Then the discriminant indeed has a cusp at $u=v=0$.
However, we can check that the Weierstra{\ss} fibration $p(x,y,u,v)=0$, with $p(x,y,u,v)$ as in~\eqref{eqn:WSappendix}, would have a singularity over $u=v=0$ if and only if $b_{1,0}=b_{0,1}=0$.

\section{Further details on the geometries of the example of Section~\ref{sec:exampleGeom}} \label{sec:appGeom}
In this appendix, we will elaborate on the construction and geometric properties of the Calabi-Yau threefolds $X^k$, $k=0,1,2$, that we introduced in Section~\ref{sec:exampleGeom}.

\subsection{Toric data associated to $X^1$ and $X^2$} \label{ss:toric}
\paragraph{$X^1$ torically} The toric data for the ambient space
\begin{align} \label{eq:defV}
    V=\mathbb{P}\left(\mathcal{O}_{\mathbb{P}^2}(-1)\oplus \mathcal{O}_{\mathbb{P}^2}^{\oplus 3}\right)
\end{align}
of the variety $X^1$ introduced in Section~\ref{sec:exampleGeom} is given in Table~\ref{tab:toricV}.

\begin{table}
\begin{equation} \nonumber
    \left[\begin{NiceArray}{c|cccccc|c}
        \text{hom. coord.}    & \multicolumn{6}{|c|}{\text{1-cone gens.}} & \text{divisor class} \\ \hline
        x_1 & 1 & 0 & 0 & 0 & 0 & 0 & -H\\
        x_2 & 0 & 1 & 0 & 0 & 0 & 0 & D_{2}\\
        x_3 & 0 & 0 & 1 & 0 & 0 & 0 & D_{3}\\
        x_4 & 0 & 0 & 0 & 1 & 0 & 0 & D_{4}\\        
        y_1 & 1 & 0 & 0 & 0 & 1 & 0 & H\\
        y_2 & 0 & 0 & 0 & 0 & 0 & 1 & H\\
        y_3 & 0 & 0 & 0 & 0 & -1 & -1 & H\\
    \end{NiceArray}\right]
    \quad \xmapsto{\, \IP \,} \quad
    \left[\begin{NiceArray}{c|ccccc|c}
        \text{hom. coord.}    & \multicolumn{5}{|c|}{\text{1-cone gens.}} & \text{divisor class} \\ \hline
        x_1 & 1 & 0 & 0 & 0 & 0 & D_\zeta - H\\
        x_2 & 0 & 1 & 0 & 0 & 0 & D_\zeta\\
        x_3 & 0 & 0 & 1  & 0 & 0 & D_\zeta\\
        x_4 & -1 & -1 & -1  & 0 & 0 & D_\zeta\\        
        y_1 & 1 & 0 & 0 & 1 & 0 & H\\
        y_2 & 0 & 0 & 0 & 0 & 1 & H\\
        y_3 & 0 & 0 & 0 & -1 & -1 & H\\
    \end{NiceArray}\right] 
\end{equation} 
\caption{Toric data for $\cO_{\IP^2}(-1)\oplus \cO_{\IP^2}^{\oplus 3} $ (LHS) and its projectivization $V$ (RHS).} \label{tab:toricV} 
\end{table}

From the table, we can read off the anti-canonical class $-K_V$ of $V$ as the sum of all torically invariant divisors.
This gives
\begin{equation}
    -K_V = 4D_\zeta + 2H \,,
\end{equation}
where $D_\zeta$ is the relative hyperplane class of the projective bundle $\pi_V:V\rightarrow \mathbb{P}^2$ and $H$ is the pullback of the hyperplane class on $\mathbb{P}^2$ to $V$.
The Calabi-Yau threefold $X^1$ is defined as the complete intersection associated to the nef partition 
\begin{equation}
    K_V = D+D \,, \quad -D = 2D_\zeta + H \,.
\end{equation}
In other words, $X^1$ can be obtained as the intersection of the vanishing loci of two generic sections of
\begin{equation}
    \cO_V(-D) = \zeta^{\otimes 2} \otimes \pi_V^* \cO_{\IP^2}(1) \,,
\end{equation}
where $\zeta$ is the relative hyperplane bundle on $V$.
As the complete intersection in $V$ is associated to the nef partition of the anti-canonical bundle of $V$, the variety $X^1$ is Calabi-Yau. 

Note that the divisor $D$ is ample.\footnote{This distinguishes it from the $\IZ_4$ example studied in~\cite[Section 7.2]{Schimannek:2021pau}. However, we expect that the corresponding 4-section fibration $X_1^{(4)}$ in that reference also exhibits $H_2(X_1^{(4)},\IZ)=\mathbb{Z}^2$, i.e. does not exhibit any torsion in homology.} The Lefshetz hyperplane theorem thus applies to $X^1$ and identifies its homology below degree $\dim V - 2=3$ with that of $V$.

\paragraph{$X^2$ torically} The threefold $X^2$ can be described as the anti-canonical divisor of a line bundle over the base
\begin{equation}
    W = \IP^1 \times \IP^2 \,, \quad \pi_W: W \rightarrow \IP^2 = B \,.
\end{equation}
The toric data of the total space of this line bundle is given in Table \ref{tab:toricX2}.
\begin{table}
    \centering
    \begin{minipage}{0.45\textwidth}
        \centering
        \begin{equation} \nonumber
            \left[\begin{NiceArray}{c|cccc|c}
            \text{hom. coord.}    & \multicolumn{4}{|c|}{\text{1-cone gens.}} & \text{divisor class} \\ \hline
           y   & 1  &  0 & 0 & 0 & 3 H_1 + 2 H_2\\
            z_1 & -1 &  1 & 0 & 0 & H_{2}\\
            z_2 & -1 & -1 & 0 & 0 & H_{2}\\
            y_1 & -1 &  0 & 1 & 0 & H_1\\
            y_2 & -1 &  0 & 0 & 1 & H_1\\
            y_3 & -1 &  0 & -1 & -1 & H_1\\
        \end{NiceArray}\right]
        \end{equation} 
        \caption{Toric data for the ambient space of $X^2$.} \label{tab:toricX2} 
    \end{minipage}\hfill
    \begin{minipage}{0.45\textwidth}
        \centering
        \begin{equation} \nonumber
        \left[\begin{NiceArray}{c|ccccc|c}
        \text{hom. coord.}    & \multicolumn{5}{|c|}{\text{1-cone gens.}} & \text{divisor class} \\ \hline
        \lambda   & 0  &  0 & 1 & 1 & 1 & E\\
        z_1 & 1 &  0 & 0 & 0 & 0 & H_{5}\\
        z_2 & 0 &  1 & 0 & 0 & 0 &H_{5}\\
        z_3 & 0 &  0 & 1 & 0 & 0 & H_5 - E\\
        z_4 & 0 &  0 & 0 & 1 & 0 & H_5 - E\\
        z_5 & 0 &  0 & 0 & 0 & 1 & H_5 - E\\
        z_6 & -1 &  -1 & -1 & -1 & -1 & H_5\\
        \end{NiceArray}\right]
        \end{equation} 
         \caption{Toric data for the blowup of $\IP^5$ along a plane.} \label{tab:toricblowupP5} 
    \end{minipage}
\end{table}
The anti-canonical class of the geometry is given by 
\begin{equation}
    2[y] = 4 H_2 + 6 H_1 \,.
\end{equation}
We see that the defining equation~\eqref{eqn:X2eq} of $X^2$ is indeed a section of the corresponding anti-canonical line bundle.

\subsection{$X^1$ as a complete intersection in a blowup of $\IP^5$} \label{ss:X1inblowup}
The ambient space $V$ defined as the projectivization of a vector bundle over $\IP^2$ in \eqref{eq:defV} can also be obtained by blowing up $\IP^5$ along a plane. The map from $V$ to $\IP^5 \times \IP^2$ is given by
\begin{equation}
    ([y_1:y_2:y_3],[x_1:x_2:x_3:x_4] \mapsto ( [x_2:x_3:x_4:x_1 y_1:x_1 y_2:x_1 y_3],[y_1:y_2:y_3]) \,,
\end{equation}
whose image is manifestly the blowup of $\IP^5$ along the vanishing locus of the three homogeneous coordinates $z_3,z_4,z_5$ in the nomenclature introduced in Table \ref{tab:toricblowupP5}. It is not difficult to arrive at the identification of divisors
\begin{equation}
    H_5 = D_\zeta \,, \quad E = D_\zeta - H \,.
\end{equation}
We conclude that $X^1$ corresponds to the intersection $(3H_5-E,3H_5-E)$ in the blowup.

\subsection{Behavior of $X^{1,2} \in \Sh(X^0) = \IZ_4$ above singularities of the discriminant divisor}  \label{app:discSing}

\paragraph{Double covers of quartics and $X^2$} \label{ss:X2}
We have defined $X^2$ as the double cover of the trivial $\IP^1$ bundle $W= \IP^1 \times \IP^2$ that is the vanishing locus of the polynomial
\begin{equation} \label{eq:quartic}
    F = y^2 - P_4 \,,
\end{equation}
where $P_4$ is the quartic polynomial 
\begin{equation} \label{eq:P4det}
    P_4 = \det \Azz(y_{1,\ldots,3}) \,,
\end{equation}
in $z_1$ and $z_2$.
The $4\times 4$ matrix $\Azz$ was introduced in~\eqref{eq:pencilquad}.
Until we state otherwise below, the discussion will not be specific to $P_4$ being a determinant.
We will follow the notation in~\eqref{eq:coeffDoubleCoverQuartic} for its coefficients. 

The equation~\eqref{eq:quartic} can be mapped to Weierstra{\ss} form, following e.g.~\cite{An2001}.
In terms of the $e_0,\ldots,e_4$ introduced in~\eqref{eq:coeffDoubleCoverQuartic}, the two invariants of the quartic equation are given by\footnote{Invariants are polynomials in the coefficients of the quartic that upon a $\mathrm{GL}(2)$-action on its variables $(z_1,z_2)$ are invariant up to a power of the determinant of the $\mathrm{GL}(2)$ matrix. This power is 4 for $i$ and 6 for $j$.}
\begin{align} \label{eq:iAndjForQuartic}
    \begin{split}
    f &= 4\left(4 e_1 e_3-e_0 e_4 -3a_2^2\right)\,, \\
    g &= 4\left(e_0 e_2 e_4 + 2e_1 e_2 e_3 - e_0 e_3^2 - e_4 e_1^2 -e_2^3\right) \,.
    \end{split}
\end{align}
These invariants determine the Weierstra{\ss} equation and its discriminant:
\begin{equation}
    y^2 = x^3 +f\, x +g\,, \quad \Delta = 4f^3 +27 g^2 \,.
\end{equation}

The singularities of $X^2$ lie at points on $X^2$ at which $\nabla F = 0$. The conditions $\partial_y F = \partial_{z_i} F = 0$ imply that the fibers of $X^2$ over the discriminant divisor of $P_4$ are singular. As the discriminant of the quartic and its Weierstra{\ss} form coincide, $X^2$ has singular fibers over the same points of the base variety $B$ as $X^0$.

We will next argue that, just as for $X^0$, singular points of  $X^2$ lie over singular points of the discriminant divisor; however, again in parallel to $X^0$, not all of these points have singularities of the threefold above them. Those that do are a subset of those that do for $X^0$, as we will show.

To see that a singular point of $X^2$ must necessarily lie over a singular point of the discriminant divisor, introduce an affine parameter $z$ on the $\IP^1$ coordinatized by $(z_1:z_2)$ and let $s$ be a local parameter on the base. Consider a point above the discriminant divisor at which $\partial_s F = 0$. Shift $z$ and $s$ if necessary so that they vanish at this point, and keep any other local parameters fixed. The Taylor expansion of $P_4$ starts at quadratic order in both $z$ and $s$:
\begin{equation}
    P_4 = \frac{1}{2} \left( P_{4,zz} z^2 + P_{4,zs} z s + P_{4,ss} s^2 \right) + \ldots \,.
\end{equation}
If $P_{4,zz} \neq 0$, we can write
\begin{equation}
    P_4 = \frac{1}{2} P_{4,zz} (z- \alpha_1 s)(z - \alpha_2 s)  + \ldots \,.
\end{equation}
We see that in this case, $P_4$ has zeros $z_{1,2} = \alpha_{1,2} s + \cO(s^2)$ that move linearly in $s$ close to $s=0$. Hence,
\begin{equation}
    \Delta = s^2 \Delta_{\mathrm{red}}(s) \,, \quad \Delta_{\mathrm{red}}(0) = \mathrm{const.} \,,
\end{equation}
i.e. the derivative of the discriminant vanishes at this point with regard to the local parameter $s$. If $P_{4,zz} = 0$, then $P_4$ at $s=0$ has at least a triple zero at $z=0$. In this case, explicit calculation using \eqref{eq:iAndjForQuartic} shows that $f$ and $g$ vanish, $\Delta = 4f^3 + 27 g^2$ thus vanishes to at least second order in $s$. Repeating the argument for the second local parameter on the base, the claim that singularities of $X^2$ must lie over singularities of the discriminant divisor follows. 

The behavior over the singular points of the discriminant depends on the vanishing pattern of its four zeros: as we are considering a 2d base, the patterns that generically occur are $(2,1,1)$ (only two zeros coincide), $(3,1)$ (three zeros coincide), and $(2,2)$ (two pairs of two zeros coincide). We will discuss these cases in turn.

\begin{itemize}
    \item $(3,1)$: $f$, $g$ vanish for this case. These are hence the points over which $X^0$ has a cusp. Just as in the case of $X^0$, no singularities of $X^2$ lie above these points: $\partial_s \Delta=0$ without the need to impose a constraint on the $s$-derivative of the coefficients of $P_4$. Hence, $\partial_s F \neq 0$, generically.
    \item $(2,2)$: $f$ and $g$ are non-zero in this case. $X^0$ hence exhibits nodes above these points of the base. However, unlike for $X^0$, no singularity lies above these points in the case of $X^2$. To see this, express the coefficients $e_0, \ldots, e_4$ in \eqref{eq:coeffDoubleCoverQuartic} in terms of the two distinct zeros $x_1, x_2$ of the quartic; explicit computation then shows that the coefficients of each formal derivative $\partial_s e_i$, $i=0,\ldots,4$ in $\partial_s \Delta$ vanish. Hence, just as for the $(3,1)$ pattern, $\partial_s \Delta$ vanishes without the need to impose conditions on the $s$-derivative of the coefficients of $P_4$. Hence, $\partial_s F \neq 0$, generically. The geometry of the fiber over such points is easy to visualize \cite[Section 2.2]{Morrison:2014era}: it is a resolved $I_2$ fiber consisting of two rational degree 1 curves intersecting in two points.
    \item $(2,1,1)$: $f$ and $g$ are again non-zero.
    The Jacobian fibration $X^0$ therefore has a nodal singularity over each of these points.
    Unlike the case $(2,2)$, $\partial_s \Delta =0$ does not follow simply from the vanishing pattern; the vanishing of the derivative implies a constraint on the derivatives $\partial_s a_i$, $i=0,\ldots, 4$. This constraint renders $\partial_s F = 0$ as well. This can easily be seen by direct calculation: express the coefficients $e_i$ in terms of the three distinct zeros $x_1, x_2, x_3$ of $\Delta$, take the formal $s$ derivative of the discriminant and solve for e.g. $\partial_s e_0$. Substituting back into $\partial_s F$, this is seen to vanish at the double zero in $x$. 
\end{itemize}
The discussion so far has been for general $P_4$. For $P_4$ in the form \eqref{eq:P4det}, we will see in the following discussion that the cases $(3,1)$ and $(2,2)$ coincide with the matrix $A_{z_1,z_2}$ having corank 1, while it drops rank by 2 when the singularity has vanishing pattern $(2,1,1)$.

\paragraph{Intersection of two quadrics and $X^1$}
Consider two quadrics $U$ and $V$ in $\IP^n$ and the associated $(n+1)\times(n+1)$ symmetric matrices $Q_U$ and $Q_V$.
Their intersection up to projective equivalence is governed by the so-called Segre symbol $\sigma$ of the pencil of quadrics associated to $Q_U + \lambda Q_V$. 
The definition of the Segre symbol can be found e.g. in~\cite[Section 8.6.1]{DolgachevBook} or~\cite{FevolaQuadrics}.
It takes the form $\sigma = [\sigma_1, \ldots, \sigma_r]$ in terms of tuples of positive integers
\begin{equation}
    \sigma_i = (\sigma_i^{(1)},\ldots,\sigma_i^{(s_i)}) \,, \quad \sigma_i^{(1)} \ge \ldots \ge \sigma_i^{(s_i)} \,,
\end{equation}
with $r \le n+1$, and $\sum_{i=1}^r \sum_{k=1}^{s_i} \sigma_i^{(k)} = n+1 $. It implies that an invertible matrix $C$ exists such that $C^T Q_U C$ and $C^T Q_V C$ are in Jordan normal form with $\sigma_i^{(k)} \times \sigma_i^{(k)}$ blocks of the respective form
\begin{equation}
    J_{U,i}^{(k)} = 
        \begin{pmatrix}
            0 & \cdots &0 & \alpha_i \\
            0 & \cdots  &\alpha_i & 1 \\
            \vdots & \iddots &\iddots  & \vdots \\
        \alpha_i & 1 & \cdots & 0 
        \end{pmatrix}\,, \quad
    J_{V,i}^{(k)} = 
        \begin{pmatrix}
            0 & \cdots &0 & 1 \\
            0 & \cdots  &1 & 0 \\
            \vdots & \reflectbox{$\ddots$} &   & \vdots \\
        1 & 0 & \cdots & 0 
        \end{pmatrix} \,.
\end{equation}
The entries $\alpha_1 , \ldots, \alpha_r$ are the negatives of the distinct roots of $\det( Q_U + \lambda Q_V) = 0$ as a polynomial equation in $\lambda$, and $\sigma_i$ is a partition of the multiplicity of the $i^{\mathrm{th}}$ root.

Pencils associated to a given Segre symbol $\sigma$ determine a stratum $\Gr_\sigma$ within the space $\Gr(2,\IS^{n+1})$ of pencils of $(n+1)\times(n+1)$ symmetric matrices $\IS^{n+1}$.
The codimension of this stratum in $\Gr(2,\IS^{n+1})$ is computed in \cite{FevolaQuadrics} in terms of the conjugate partitions $\sigma_i^*$, $i=1,\ldots, r$, as
\begin{equation}
    \codim (\Gr_\sigma) = \sum_{i=1}^r \sum_{k=1}^{\sigma_i^{(1)}} \begin{pmatrix} \sigma_i^{*,(k)}+1 \\ 2 \end{pmatrix} - r \,.
\end{equation}
As we are interested in this construction fibered over a 2-dimensional space, we only need to consider Segre symbols with stratum $\Gr_\sigma$ of codimension 2 or lower.
We list these together with the geometry of the associated intersection, which can be found e.g. in \cite{HodgePedoe} or \cite[Example 3.1]{FevolaQuadrics}\footnote{In \cite{FevolaQuadrics}, the codimension in $\Gr(2,\IS^4)$ is the second of the three codimensions cited in the table of Example 3.1.}, in Table \ref{tab:pencilQuadrics}.

\begin{table}
    \centering
    \begin{tabular}{|c|c|c|}\hline
        Segre symbol & Codimension & Geometry of intersection of quadrics \\ \hline
         $[1,1,1,1]$ & 0 & irreducible quartic curve\\
         $[2,1,1]$ & 1 & nodal quartic curve\\
         $[3,1]$ & 2 & cuspidal quartic curve\\
         $[2,2]$ & 2 & twisted cubic and line intersecting in two points\\
         $[(1,1),1,1]$ & 2 & two conics intersecting in two points\\\hline
    \end{tabular}
    \caption{Strata of codimension at most 2 in $\Gr(2,\IS^4)$.}
    \label{tab:pencilQuadrics}
\end{table}

Turning now to the torus fibration $X^1$ over the base $\IP^2$, we will show that the last two entries in Table \ref{tab:pencilQuadrics} correspond to the fibers, in the  notation introduced in the paragraph above equation \eqref{eq:numberI2quartic}, over the points $S_\Delta^{(1)}$ and $S_\Delta^{(2)}$ of the base respectively. Recall that in~\eqref{eq:quaternaryQuadrics}, we have introduced two quaternary quadrics (i.e. quadrics in $\IP^3$)
\begin{equation}
    p_i = \myvec{x}^T A_i \myvec{x} \,, \quad i=1,2 \,,
\end{equation}
with the entries of the matrices $A_i$ homogeneous polynomials of the appropriate degree in the coordinates $(y_1:y_2:y_3)$ of the base.  $X^1$ is defined as the intersection of the two hypersurfaces $\{p_i=0\}$, $i=1,2$. 

Singularities of $X^1$ lie at points $p\in X^1$ at which the gradients $\nabla p_1$ and $\nabla p_2$ are proportional.
We first keep the base coordinates fixed (and omit them in our notation) to study singularities in the fiber.
Then
\begin{equation}
    (\myvec{x},\myvec{y}) \in X^1 \, \text{is a fiber singularity} \quad \Leftrightarrow \quad p_1(\myvec{x})=p_2(\myvec{x})=0 \,\wedge \,\exists (z_1:z_2) \in \IP^1 : (z_1 A_1 + z_2 A_2 )\myvec{x} = 0 \,.
\end{equation}
The point $(\myvec{x},\myvec{y}) \in X^1$ is hence a singular point of the fiber when the pencil of quadrics
\begin{equation}
    \Azz=z_1 A_1 + z_2 A_2
\end{equation} 
contains a member of rank less than 4, provided that an element in the kernel of this member lies on $X^1$. Given a two dimensional base of the elliptic fibration, the two cases $\corank \Azz =1$ and $\corank \Azz =2$ generically occur.
We shall consider them in turn.

\begin{itemize}
    \item Corank 1: Let $\myvec{x} \neq 0$ lie in (and hence span) the kernel of $\Azz$. The condition $(\myvec{x},\myvec{y}) \in X^1$ can be expressed in terms of the pencil of quadrics. By Jacobi's formula,
    \begin{equation}
        \partial_{z_1} \det \Azz = \tr \left(\adj(\Azz) \partial_{z_1} \Azz\right)= \tr \left(\adj(\Azz) A_1 \right)\,.
    \end{equation}
    The adjugate of a matrix $A$ satisfies $A\adj A = I\,\det A$, with $I$ representing the unit matrix. Each column of $\adj \Azz$ must thus lie in the kernel of $\Azz$, i.e. be proportional to $\myvec{x}$. 
    Therefore,
    \begin{equation}
        \adj \Azz = \myvec{x} \myvec{w}^T \,.
    \end{equation}
    We can further constrain $\myvec{w}$ by invoking that the adjugate of a symmetric matrix is symmetric:
    \begin{equation}
        \myvec{x} \myvec{w}^T = \myvec{w} \myvec{x}^T
    \end{equation}
    implies that $\myvec{w}$ is a multiple of $\myvec{x}$, $\myvec{w} = \gamma \myvec{x}$ for some $\gamma \neq 0$. Thus,
    \begin{equation}
        \partial_{z_1} \det \Azz = \gamma \,\tr \left(\myvec{x} \myvec{x}^T A_1 \right)= \gamma \,\myvec{x}^T A_1 \myvec{x} \,.
    \end{equation}
    But $A_{z_1,z_2} \myvec{x} = 0$ together with $p_1(x) = \myvec{x}^T A_1 \myvec{x} = 0$ implies $p_2(x)= \myvec{x}^T A_2 \myvec{x} = 0$. Reintroducing the base dependence for clarity, 
    we conclude that for $\myvec{x} \in \ker A_{z_1,z_2}(\myvec{y})$, $(\myvec{x},\myvec{y}) \in X^1$ exactly when the derivative of the determinant $\det \Azz(\myvec{y})$ with regard to a $\IP^1$ coordinate vanishes.
    
    Consulting Table \ref{tab:pencilQuadrics}, we see by explicitly writing out the canonical representatives of the quadrics associated to the respective Segre symbols that the pencils with $\codim \Gr_\sigma \le 2$ containing rank 1 members at ${(z_1:z_2) = (1:-\alpha_1)}$ have Segre symbol $[2,1,1]$, $[3,1]$, and $[2,2]$. From our discussion of double covers of quadrics in Section \ref{ss:X2}, we know that $[3,1]$ occurs precisely over the cusps of the discriminant divisor, while $[2,2]$ fibers arise over the points of $S_\Delta^{(1)}$. In particular, we read off from Table \ref{tab:pencilQuadrics} that the two components of $[2,2]$ fibers intersect the 4-section of $X^1$ in 3 and 1 point, respectively.
    
    \item Corank 2: Again by Jacobi's formula, the derivative of $\det \Azz$ now vanishes identically.\footnote{The adjugate of a matrix $A$ is the transpose of the cofactor matrix of $A$, i.e. the entries are first minors of $A$. These all vanish for $A$ of corank higher than 1.} Let 
    $\ker \Azz = \langle \myvec{x_1},\myvec{x_2} \rangle$. We can then find a linear combination of these vectors such that 
    \begin{equation}
        (\mu \myvec{x_1} + \lambda \myvec{x_2})^T A_1 (\mu \myvec{x_1} + \lambda \myvec{x_2}) = 0 \,.
    \end{equation}
    The vector $\mu \myvec{x_1} + \lambda \myvec{x_2}$ is then also isotropic with regard to $A_2$, hence lies on $X^1$.

    Table \ref{tab:pencilQuadrics} shows that the only pencils with $\codim \Gr_\sigma \le 2$ containing a rank 2 member at ${(z_1:z_2) = (1:-\alpha_1)}$ are the ones with Segre symbol $[(1,1),1,1]$. These fibers occur above the points $S_\Delta^{(2)}$. From Table \ref{tab:pencilQuadrics}, we can read off that the two components of $[(1,1),1,1]$ fibers intersect the 4-section of $X^1$ each in two points.
\end{itemize}

We finally argue that the threefold $X^1$ itself is smooth at the singular points of these fibers. For $(\myvec{x},\myvec{y}) \in X^1$ to be singular, we require the existence of a point $(z_1:z_2) \in \IP^1$ which in addition to
    \begin{equation}
        (z_1 A_1 + z_2 A_2) \myvec{x} = 0
    \end{equation}
    satisfies
    \begin{equation} \label{eq:X1sing}
        \myvec{x}^T (z_1 \partial_{y_i} A_1 + z_2 \partial_{y_i} A_2) \myvec{x} = 0 \,, \quad i = 1,2,3\,.
    \end{equation}
Let us again distinguish according to the corank of $z_1 A_1 + z_2 A_2$.
\begin{itemize}
    \item Corank 1: As above, we can show that the condition \eqref{eq:X1sing} is equivalent to
    \begin{equation} \label{eq:singBase}
        \partial_{y_i} \det \Azz = 0 \,, \quad i = 1,2,3\,.
    \end{equation}
    The condition for $X^1$ and $X^2$ to be singular hence coincide in this case. Having already imposed two conditions on the base coordinates to obtain a singular fiber, \eqref{eq:singBase} will generically not hold anywhere over a 2-dimensional base.
    \item Corank 2: With $\ker \Azz = \langle \myvec{x_1},\myvec{x_2} \rangle$ as above, we now need to impose two additional constraints on the linear combination $\mu \myvec{x_1} + \lambda \myvec{x_2}$, to satisfy \eqref{eq:X1sing}. Generically, no solutions will exist. Unlike the case of $X^2$, $X^1$ will hence generically be smooth over such points of the base as well.
\end{itemize}

\section{Jacobi forms and modular forms for congruence subgroups}
\subsection{The ring of weak Jacobi forms}
In this section, we will summarize some of the basic properties of (weak) Jacobi forms as presented in~\cite{Eichler1985}.
A Jacobi form of weight $w$ and index $m$ is a holomorphic function $\phi_{w,m}:\mathbb{H}\times\mathbb{C}\rightarrow\mathbb{C}$ that for each ${\tiny\left(\begin{array}{cc}a&b\\c&d\end{array}\right)}\in\SLtwoZ$ satisfies the modular transformation law
\begin{align}
    \phi_{w,m}\left(\frac{a\tau+b}{c\tau+d},\frac{z}{c\tau+d}\right)=(c\tau+d)^k{\rm e}^m\!\!\left[\frac{cz^2}{c\tau+d}\right]\phi_{w,m}(\tau,z)\,,
\end{align}
and for each $a,b\in\mathbb{Z}$ satisfies the elliptic transformation law
\begin{align}
    \phi_{w,m}(\tau,z+a\tau+b)={\rm e}^{-m}[a^2\tau+2az]\,\phi_{w,m}(\tau,z)\,.
\end{align}
This implies that $\phi_{w,m}(\tau,z)$ admits a Fourier expansion
\begin{align}
    \phi_{w,m}(\tau,z)=\sum\limits_{n,r}q^ny^r\,,\quad q=e^{2\pi{\rm i}\tau}\,,\quad y=e^{2\pi{\rm i}z}\,.
\end{align}
It is called a \textit{holomorphic Jacobi form} if $c(n,r)=0$ unless $4mn\ge r^2$ and it is called a \textit{weak Jacobi form} if $c(n,r)=0$ unless $n\ge 0$.

The space of weak Jacobi forms of even weight $w$ and index $m$ is generated as
\begin{align}
    J^{\text{weak}}_{w,m}=\bigoplus\limits_{k\ge 0}^m M_{w+2k}\left(\SLtwoZ\right)\phi_{-2,1}(\tau,z)^k\phi_{0,1}(\tau,z)^{m-k}\,,
\end{align}
where $M_{w+2k}\left(\SLtwoZ\right)$ is the ring of modular forms of weight $w+2k$.
The weak Jacobi forms $\phi_{-2,1}(\tau,z)$, $\phi_{0,1}(\tau,z)$ of index $1$ and respective weight $-2$ and $0$ are given by
\begin{align}
    \phi_{-2,1}(\tau,z)=-\frac{\vartheta_1(\tau,z)^2}{\eta(\tau)^6}\,,\quad \phi_{0,1}(\tau,z)=4\left(\frac{\vartheta_2(\tau,z)^2}{\vartheta_2(\tau,0)^2}+\frac{\vartheta_3(\tau,z)^2}{\vartheta_3(\tau,0)^2}+\frac{\vartheta_4(\tau,z)^2}{\vartheta_4(\tau,0)^2}\right)\,,
\end{align}
in terms of the Jacobi theta functions $\vartheta_1(\tau,z)=-\vartheta{\tiny\left[\begin{array}{c}\frac12\\\frac12\end{array}\right]}(\tau,z)$ and
\begin{align}
    \vartheta_2(\tau,z)=\vartheta{\tiny\left[\begin{array}{c}\frac12\\0\end{array}\right]}(\tau,z)\,,\quad \vartheta_3(\tau,z)=\vartheta{\tiny\left[\begin{array}{c}0\\0\end{array}\right]}(\tau,z)\,,\quad \vartheta_4(\tau,z)=\vartheta{\tiny\left[\begin{array}{c}0\\\frac12\end{array}\right]}(\tau,z)\,,
\end{align}
with
\begin{align}
    \vartheta{\tiny\left[\begin{array}{c}a\\b\end{array}\right]}(\tau,z)=\sum\limits_{n=-\infty}^\infty q^{\frac12(n+a)^2}y^{n+a}e^{2\pi{\rm i}b(n+a)}\,,\quad a,b\in\{0,1/2\}\,,\quad q=e^{2\pi{\rm i}\tau}\,,\quad y=e^{2\pi{\rm i}z}\,.
\end{align}

\subsection{Congruence subgroups of the modular group}
\label{sec:congruenceSubgroups}

The principal congruence subgroup $\Gamma(n)$ of level $n$ inside the modular group $\Gamma_1=\SLtwoZ$ is
\begin{align}
    \Gamma(n)=\left\{\,\left(\begin{array}{cc}a&b\\c&d\end{array}\right)\in\SLtwoZ\,\middle\vert\,b,c\equiv 0\text{ mod }n\,,\,\, a,d\equiv 1\text{ mod }n\,\right\}\,.
\end{align}
We denote the corresponding modular curve by $X(n)=\mathbb{H}/\Gamma(n)$.
More generally, any subgroup $\Gamma\subset\Gamma_1$ is called a congruence subgroup of level $n$ if $\Gamma(n)\subset\Gamma$.
Some important examples that will be relevant in this paper are
\begin{align}
    \begin{split}
        \Gamma_0(n)=&\left\{\,\left(\begin{array}{cc}a&b\\c&d\end{array}\right)\in\SLtwoZ\,\,\middle\vert\,\,c\equiv 0\text{ mod }n\,\right\}\,,\\
        \Gamma_1(n)=&\left\{\,\left(\begin{array}{cc}a&b\\c&d\end{array}\right)\in\SLtwoZ\,\,\middle\vert\,\,c\equiv 0\text{ mod }n\,,\,\, a,d\equiv 1\text{ mod }n\,\right\}\,,\\
        \Gamma^0(n)=&\left\{\,\left(\begin{array}{cc}a&b\\c&d\end{array}\right)\in\SLtwoZ\,\,\middle\vert\,\,b\equiv 0\text{ mod }n\,\right\}\,,\\
        \Gamma^1(n)=&\left\{\,\left(\begin{array}{cc}a&b\\c&d\end{array}\right)\in\SLtwoZ\,\,\middle\vert\,\,b\equiv 0\text{ mod }n\,,\,\, a,d\equiv 1\text{ mod }n\,\right\}\,.
    \end{split}
\end{align}
We denote the corresponding modular curves respectively by $X_0(n)=\mathbb{H}/\Gamma_0(n)$, $X_1(n)=\mathbb{H}/\Gamma_1(n)$, $X^0(n)=\mathbb{H}/\Gamma^0(n)$ and $X^1(n)=\mathbb{H}/\Gamma^1(n)$.

The so-called theta subgroup, defined by
\begin{align}
        \Lambda_2=&\left\{\,\left(\begin{array}{cc}
         a&b  \\
         c&d 
    \end{array}\right) \in \text{SL}(2,\mathbb{Z})\,\middle\vert\, ac\equiv 0\text{ mod }2\,,\,\, bd\equiv 0\text{ mod } 2\,\right\}\,,
\end{align}
also makes a brief appearance.

\subsection{Modular forms for congruence subgroups}
\label{app:modularJacobi}
The slash operator $\mslash{\gamma}\![w]$ for $w\in\mathbb{N}$ and $\gamma\in\text{SL}(2,\mathbb{Z})$ acts on a function $f:\mathbb{H}\rightarrow \mathbb{C}$ as
\begin{align}
    f\!\!\mslash{\gamma}\![w](\tau)=(c\tau+d)^{-w}f\left(\frac{a\tau+b}{c\tau+d}\right)\,.
\end{align}
A meromorphic function $\phi:\mathbb{H}\rightarrow \mathbb{C}$ is called a modular form of weight $w$ for a congruence subgroup $\Gamma\subset\Gamma_1$ if for each $\gamma\in\Gamma$ one has $\phi\vert_\gamma[w](\tau)=\phi(\tau)$ and for each $\gamma\in \Gamma_1$ the function $\phi\vert_\gamma[w](\tau)$ has a Fourier development of the form $\sum_{n\ge 0}c_nq^n$ in $q=\e[\tau]$.
Given a congruence subgroup $\Gamma\subset\Gamma_1$, we will denote the corresponding ring of modular forms by $M(\Gamma)$ and the vector space of weight $w$ modular forms by $M_w(\Gamma)$.
Given a weight $w$ modular form $\phi$, we will also write $\phi\vert_\gamma=\phi\vert_\gamma[w]$ and leave the weight implicit.

For $k\ge 2$, the Eisenstein series $E_{2k}(\tau)$ are $\Gamma_1$-modular forms of weight $2k$ and can be written as
\begin{align}
    E_{2k}(\tau)=1-\frac{4k}{B_{2k}}\sum\limits_{n\ge 1}\sigma_{2k-1}(n)q^n\,,
\end{align}
in terms of the Bernoulli numbers $B_{2k}$ and the divisor function $\sigma_{2k-1}(n)$. The quasi modular form $E_2(\tau)$ of weight 2 will also be relevant for us. It transforms as
\begin{align}  
    E_2\!\!\mslash{\gamma}\![2](\tau)=E_2(\tau)-2\pi{\rm i}\frac{c}{c\tau+d}\,.
\end{align}
One can check that for $n>1$, the combinations
\begin{align}
     E_{n,2}(\tau)=\frac{1}{n-1}\left(nE_2(n\tau)-E_2(\tau)\right)\,,
\end{align}
are modular forms for $\Gamma_1(n)$ of weight $2$.

A useful way to relate rings of modular forms is as follows.
Given a modular form $\widetilde{\phi}_w(\tau)$  of weight $w$ for a subgroup $\widetilde{\Gamma}\subset\SLtwoZ$ with $\widetilde{\Gamma}\subseteq\Gamma^1(n)$, let $m\in\mathbb{N}$ such that $m\vert n$ and define $\phi_w(\tau)=\widetilde{\phi}_w(m\tau)$.
Using
\begin{align}
    (c\tau+d)^{k}\widetilde{\phi}_w(\tau)=\widetilde{\phi}_w\left(\frac{a\tau+b}{c\tau+d}\right)\,,\quad  \left(\begin{array}{cc}a&b\\c&d\end{array}\right)\in\widetilde{\Gamma}\,,
\end{align}
and introducing $\tau'=\tau/m$, we obtain
\begin{align}
    (mc\tau'+d)^{k}\phi_w(\tau')=\phi_w\left(\frac{a\tau'+b/m}{mc\tau'+d}\right)\,,\quad  \left(\begin{array}{cc}a&b\\c&d\end{array}\right)\in\widetilde{\Gamma}\,.
\end{align}
Hence, $\phi_w(\tau)$ is a modular form for
\begin{align}
    \widetilde{\Gamma}[m]:=\left\{\,\left(\begin{array}{cc}a&b/m\\mc&d\end{array}\right)\,\,\middle\vert\,\,\left(\begin{array}{cc}a&b\\c&d\end{array}\right)\in\widetilde{\Gamma}\,\right\}\,.
\end{align}
This induces an isomorphism on the rings of modular forms
\begin{align}
    \sigma_m:\,M(\widetilde{\Gamma})\rightarrow M(\widetilde{\Gamma}[m])\,,\quad \sigma_m:\,\widetilde{\phi}(\tau)\mapsto \phi(\tau)=\widetilde{\phi}(m\tau)\,,\quad \sigma_m^{-1}:\,\phi(\tau)\mapsto\widetilde{\phi}(\tau)=\phi(\tau/m)\,.
    \label{eqn:sigmam}
\end{align}
The two cases that will be relevant to us are as follows:
\begin{align}
    \begin{split}
    \widetilde{\Gamma}=&\Gamma(n)\quad\Rightarrow\quad \widetilde{\Gamma}[n]=\Gamma_1(n^2)\,,\\
    \widetilde{\Gamma}=&\Gamma^1(n)\quad\Rightarrow\quad \widetilde{\Gamma}[n]=\Gamma_1(n)\,.
    \end{split}
\end{align}

\subsection{Special cusp forms for $\Gamma^\tseOne_\tseTwo$}
\label{sec:cuspForms}
Given a finite Abelian group $\Gamma^0$ and elements $\tseOne,\tseTwo\in\Gamma^0$, we now define certain cusp forms for the corresponding stabilizer group $\Gamma^\tseOne_\tseTwo$.
It will again be useful to work with additive notation for the elements $\tseOne,\tseTwo$.

We will first discuss the case $\Gamma^0=\IZ_{k}$.
We can write $\tseOne=r$ and $\tseTwo=s$, for some $k\in\mathbb{N}$ and $r,s\in\{0,\ldots,k-1\}$ and define
\begin{align}
	\begin{split}
		\Delta{\left[\begin{array}{c|cc}k&r&s\end{array}\right]}(\tau)=\exp\left(-2\pi{\rm i}\left(\frac{r}{k}\right)^2\tau\right)\phi_{-2,1}\left(\tau,\frac{r}{k}\tau+\frac{s}{k}\right)^{-1}\,.
	\end{split}
    \label{eqn:cuspForm1}
\end{align}
Let us show that this transforms like a modular form with a multiplier system for $\Gamma^\tseOne_\tseTwo$ of weight $2$.
First note that for every element $g={\tiny\left(\begin{array}{cc}a&b\\c&d\end{array}\right)}\in\Gamma^\tseOne_\tseTwo$, we have
\begin{align}
	\begin{split}
	\frac{1}{k}\left(ar+cs\right)\equiv \frac{r}{k}\text{ mod }1\,,\qquad
	\frac{1}{k}\left(br+ds\right)\equiv \frac{s}{k}\text{ mod }1\,.
	\end{split}
	\label{eqn:congruenciesCusp}
\end{align}
It will be useful to introduce
\begin{align}
	\sigma=\frac{ar+cs}{k}-\frac{r}{k}\in\mathbb{Z}\,,\qquad \rho=&\frac{r(a\tau+b)+s(c\tau+d)}{k}\,.
\end{align}
To unburden the notation, we also introduce $f(\tau)=\Delta{\left[\begin{array}{c|cc}k&r&s\end{array}\right]}(\tau)$.
We then calculate
{
\begin{align}
	\begin{split}
		&(c\tau+d)^{-2}f\left(\frac{a\tau+b}{c\tau+d}\right)\\
		=&(c\tau+d)^{-2}\exp\left(-2\pi{\rm i}\left(\frac{r}{k}\right)^2\frac{a\tau+b}{c\tau+d}\right)\phi_{-2,1}\left(\frac{a\tau+b}{c\tau+d},\frac{\rho}{c\tau+d}\right)^{-1}\\
		=&\exp\left(-2\pi{\rm i}\left[\frac{c\rho^2}{c\tau +d}+\left(\frac{r}{k}\right)^2\frac{a\tau+b}{c\tau+d}\right]\right)\phi_{-2,1}\left(\tau,\frac{r}{k}\tau+\frac{s}{k}+\sigma\tau\right)^{-1}\\
		=&\exp\left(-2\pi{\rm i}\left[\frac{c\rho^2}{c\tau +d}+\left(\frac{r}{k}\right)^2\left(\frac{a\tau+b}{c\tau+d}-\tau\right)-\sigma^2\tau-2\sigma\left(\frac{r}{k}\tau+\frac{s}{k}\right)\right]\right)f(\tau)\,,
	\end{split}
	\label{eqn:cuspTrans1}
\end{align}
}
where in the third line we have used~\eqref{eqn:congruenciesCusp}.
Making use of
\begin{align}
	c(a\tau+b)=ca\tau+ad-1=a(c\tau+d)-1\,,
\end{align}
we can rewrite
{
\begin{align}
	\begin{split}
		\frac{c\rho^2}{c\tau+d}=&\frac{1}{k^2}\left[r^2\frac{c(a\tau+b)^2}{c\tau+d}+s^2c(c\tau+d)+2src(a\tau+b)\right]\\
			=&\frac{1}{k^2}\left[-r^2\frac{a\tau+b}{c\tau+d}+\tau\left(ar+cs\right)^2+\left(ar+cs\right)\left(br+ds\right)-rs\right]\,.
	\end{split}
\end{align}
}
After combining this with~\eqref{eqn:cuspTrans1}, it is then easy to check that
\begin{align}
	f\left(\frac{a\tau+b}{c\tau+d}\right)=\lambda(g)(c\tau+d)^{2}f(\tau)\,,
\end{align}
with the phase $\lambda(g)$ being
\begin{align}
	\begin{split}
		\lambda(g)=&\exp\left(-2\pi{\rm i}\left(\left[a\frac{r}{k}+c\frac{s}{k}\right]\left[b\frac{r}{k}+(d-2)\frac{s}{k}\right]+\frac{rs}{k^2}\right)\right)\,.
	\end{split}
    \label{eqn:cuspMultiplier}
\end{align}

The previous discussion generalizes to groups
\begin{align}
    \Gamma'\simeq \mathbb{Z}_{k_1}\times\ldots\times \mathbb{Z}_{k_n}\,,
    \label{eqn:cuspFormGammaIso}
\end{align}
for any $n\in\mathbb{N}$ and $k_1,\ldots,k_n\in\mathbb{N}$.
Given two elements $\tseOne=(r_1,\ldots,r_n)$ and $\tseTwo=(s_1,\ldots,s_n)$, we can define
\begin{align}
	\begin{split}
		&\Delta{\tiny\left[\begin{array}{c|cc}k_1&r_1&s_1\\\vdots&\vdots&\vdots \\k_n&r_n&s_n\end{array}\right]}(\tau)\\
		=&\exp\left(-2\pi{\rm i}\left(\frac{r_1}{k_1}+\ldots +\frac{r_n}{k_n}\right)^2\tau\right)\phi_{-2,1}\left(\tau,\left(\frac{r_1}{k_1}+\ldots+\frac{r_n}{k_n}\right)\tau+\left(\frac{s_1}{k_1}+\ldots+\frac{s_n}{k_n}\right)\right)^{-1}\,.
	\end{split}
    \label{eqn:cuspFormGeneral}
\end{align}
This transforms like a cusp form for $\Gamma^\tseOne_\tseTwo$ of weight $2$ and a multiplier system given by $\lambda(g)$ in~\eqref{eqn:cuspMultiplier}, after replacing
\begin{align}
    \frac{r}{k}\rightarrow\frac{r_1}{k_1}+\ldots+\frac{r_n}{k_n}\,,\qquad \frac{s}{k}\rightarrow\frac{s_1}{k_1}+\ldots+\frac{s_n}{k_n}\,.
\end{align}

\subsection{Modular forms for $\Gamma(2)$}
The ring of $\Gamma(2)$ modular forms is generated by any two elements of $\{\theta_2(\tau)^4,\,\theta_3(\tau)^4,\theta_4(\tau)^4\}$, see e.g.~\cite{Doran:2013npa}, with the Jacobi forms being defined as $\theta_i(\tau)=\vartheta_i(\tau,0)$ for $i=2,3,4$, such that
\begin{align}
    \theta_2(\tau):=\sum\limits_{n=-\infty}^\infty q^{\frac12\left(n+\frac12\right)^2}\,,\quad \theta_3(\tau):=\sum\limits_{n=-\infty}^\infty q^{\frac12 n^2}\,,\quad \theta_4(\tau):=\sum\limits_{n=-\infty}^\infty(-1)^nq^{\frac12n^2}\,,
\end{align}
with $q=e^{2\pi i \tau}$.
We therefore introduce $x_i(\tau):=\theta_i(\tau)^4,\,i=2,3,4$, such that
\begin{align}
    \begin{split}
    x_2(\tau)=&16 \sqrt{q}+64 q^{3/2}+96 q^{5/2}+128 q^{7/2}+208 q^{9/2}+\mathcal{O}(q^{11/2})\,,\\
    x_3(\tau)=&1+8 \sqrt{q}+24 q+32 q^{3/2}+24 q^2+48 q^{5/2}+96 q^3+\mathcal{O}(q^{7/2})\,,\\
    x_4(\tau)=&1-8 \sqrt{q}+24 q-32 q^{3/2}+24 q^2-48 q^{5/2}+96 q^3+\mathcal{O}(q^{7/2})\,,
    \end{split}
\end{align}
and in particular $x_2(\tau)+x_4(\tau)=x_3(\tau)$.
We also define $\Delta_{i,j}(\tau):=\sqrt{x_i(\tau)x_j(\tau)},\,i<j$, with
\begin{align}
    \begin{split}
    \Delta_{2,3}(\tau)=&\frac{4}{q^{1/4}\phi_{-2,1}\left(\tau,\frac{\tau}{2}\right)}\,,\quad \Delta_{2,4}(\tau)=\frac{-4}{q^{1/4}\phi_{-2,1}\left(\tau,\frac{\tau+1}{2}\right)}\,,\quad \Delta_{3,4}(\tau)=\frac{-4}{\phi_{-2,1}(\tau,\frac12)}\,.
    \end{split}
\end{align}
The behavior of these forms under $\SLtwoZ$-transformations is summarized in Table~\ref{tab:MFGamma2trafo}.
\begin{table}[ht!]
\centering
\begin{align*}
\begin{array}{|c|c|c|c|c|c|c|}\hline
\gamma\in\text{SL}(2,\mathbb{Z})&x_2\vert_\gamma&x_3\vert_\gamma&x_4\vert_\gamma&\Delta_{2,3}\vert_\gamma&\Delta_{2,4}\vert_\gamma&\Delta_{3,4}\vert_\gamma\\\hline
S=\left(\begin{array}{cc}
0&-1\\1&0
\end{array}\right)&-x_4&-x_3&-x_2&-\Delta_{3,4}&-\Delta_{2,4}&-\Delta_{2,3}\\\hline
T=\left(\begin{array}{cc}
1&1\\0&1
\end{array}\right)&-x_2&x_4&x_3&i\Delta_{2,4}&i\Delta_{2,3}&\Delta_{3,4}\\\hline
U=\left(\begin{array}{cc}
1&0\\1&1
\end{array}\right)&x_3&x_2&-x_4&\Delta_{2,3}&i\Delta_{3,4}&i\Delta_{2,4}\\\hline
\end{array}
\end{align*}
\caption{The behavior of $\Gamma(2)$-modular forms under $\text{SL}(2,\mathbb{Z})$ transformations.}
\label{tab:MFGamma2trafo}
\end{table}

\subsection{Modular forms for $\Gamma_0(2)=\Gamma_1(2)$ and $\Gamma^1(2)$}
As is explained for example in Appendix~\cite[Appendix B.2.1]{Schimannek:2019ijf}, the ring of $\Gamma_1(2)$ modular forms is generated by $e_2(\tau):=E_{2,2}(\tau)$ and $e_4(\tau):=E_4(\tau)$.
On the other hand, using the isomorphism~\eqref{eqn:sigmam} and a convenient normalization, the ring of $\Gamma^1(2)$ modular forms is generated by
\begin{align}
   \tilde{e}_2(\tau):=-\frac12e_{2}(\tau/2)\,,\quad \tilde{e}_{4}(\tau):=\frac14\left[5e_{2}(\tau/2)^2-e_4(\tau/2)\right]=e_4(\tau)\,.
\end{align}
Note that since $\Gamma(n)$ is contained in both $\Gamma_1(n)$ and $\Gamma^1(n)$, we can express these modular forms in terms of the generators of the corresponding ring of modular forms.
For example, we have
\begin{align}
    \begin{split}
    e_2=x_3-\frac12x_2=\frac12(x_3+x_4)\,,\quad e_4=\tilde{e}_4=x_2^2-x_2x_3+x_3^2\,,\quad\tilde{e}_{2}=-\frac12(x_2+x_3)\,.
    \end{split}
\end{align}
On the other hand, a $\Gamma(2)$ modular form $\phi_n(\tau)$ of weight $n$ is $\Gamma_1(2)$-modular if $T:\,\phi_n\mapsto \phi_n$ and it is $\Gamma^1(2)$-modular if $U:\,\phi_n\mapsto \tau^n\phi_n$.
We also note that $\Delta_{3,4}^2$ and $\Delta_{2,3}^2$ are respectively cusp forms for $\Gamma_1(2)$ and $\Gamma^1(2)$ of weight $4$.

Using Table~\ref{tab:MFGamma2trafo}, we then obtain the transformations summarized in Table~\ref{tab:MFGamma12trafo}.
\begin{table}[ht!]
\centering
\begin{align*}
\begin{array}{|c|c|c|c|c|c|c|}\hline
\gamma\in\text{SL}(2,\mathbb{Z})&e_2\vert_\gamma&e_4\vert_\gamma&\Delta_{3,4}^2\vert_\gamma&\tilde{e}_2\vert_\gamma&\tilde{e}_4\vert_\gamma&\Delta_{2,3}^2\vert_\gamma\\\hline
S=\left(\begin{array}{cc}
0&-1\\1&0
\end{array}\right)&\tilde{e}_2&\tilde{e}_{4}&\Delta_{2,3}^2&e_2&e_4&\Delta_{3,4}^2\\\hline
T=\left(\begin{array}{cc}
1&1\\0&1
\end{array}\right)&e_2&e_4&\Delta_{3,4}^2&-(e_2+\tilde{e}_2)&\tilde{e}_4&-\Delta_{2,4}^2\\\hline
T^2=\left(\begin{array}{cc}
1&2\\0&1
\end{array}\right)&e_2&e_4&\Delta_{3,4}^2&\tilde{e}_2&\tilde{e}_4&\Delta_{2,3}^2\\\hline
U=\left(\begin{array}{cc}
1&0\\1&1
\end{array}\right)&-(e_2+\tilde{e}_2)&e_4&-\Delta_{2,4}^2&\tilde{e}_2&\tilde{e}_4&\Delta_{2,3}^2\\\hline
U^2=\left(\begin{array}{cc}
1&0\\2&1
\end{array}\right)&e_2&e_4&\Delta_{3,4}^2&\tilde{e}_2&\tilde{e}_4&\Delta_{2,3}^2\\\hline
\end{array}
\end{align*}
\caption{The behavior of $\Gamma_1(2)$ and $\Gamma^1(2)$-modular forms under $\text{SL}(2,\mathbb{Z})$ transformations.}
\label{tab:MFGamma12trafo}
\end{table}

\subsection{Modular forms for $\Gamma_1(4)$ and $\Gamma^1(4)$}
As reviewed e.g. in~\cite[Appendix B2.3]{Schimannek:2021pau}, the ring of $\Gamma_1(4)$ modular forms is
\begin{align}
    M\left(\Gamma_1(4)\right)=\langle E_{2,2},E_{4,1}\rangle\,,
\end{align}
in terms of the two generators
\begin{align}
    \begin{split}
    E_{2,2}(\tau)=&2E_2(2\tau)-E_2(\tau)=\frac{\eta(\tau)^8+32\eta(4\tau)^8}{\eta(2\tau)^4}=1+24q+24q^2+96q^3+\mathcal{O}(q^4)\,,\\
    \quad E_{4,1}(\tau)=&\frac{\eta(2\tau)^{10}}{\eta(\tau)^4\eta(4\tau)^4}=1+4q+4q^2+4q^3+\mathcal{O}(q^4)\,,
    \end{split}
\end{align}
of respective weight $2$ and $1$.
Using
\begin{align}
    E_{2,2}=\frac12\left(x_3+x_4\right)\,,\quad E_{4,1}=\frac12\left(\sqrt{x_3}+\sqrt{x_4}\right)\,,
\end{align}
we observe that equivalently,
\begin{align}
    M\left(\Gamma_1(4)\right)=\langle \epsilon_1,\epsilon_2\rangle\,,\quad \epsilon_1=\sqrt{x_3}+\sqrt{x_4}\,,\quad \epsilon_2=x_3+x_4\,.
\end{align}
We also have
\begin{align}
    E_{2,2}(\tau/4)=x_2(\tau)+6\sqrt{x_2(\tau)x_3(\tau)}+x_3(\tau)\,,\quad E_{4,1}(\tau/4)=\sqrt{x_2(\tau)}+\sqrt{x_3(\tau)}\,,
\end{align}
and using~\eqref{eqn:sigmam} with $\widetilde{\Gamma}=\Gamma^1(4)$ and $\widetilde{\Gamma}[4]=\Gamma_1(4)$, we find that
\begin{align}
M\left(\Gamma^1(4)\right)=\langle \tilde{\epsilon}_1,\tilde{\epsilon}_2\rangle\,,\quad \tilde{\epsilon}_1=\sqrt{x_2}+\sqrt{x_3}\,,\quad \tilde{\epsilon}_2=x_2+x_3\,.
\end{align}
Note that
\begin{align}
    S:\,\epsilon_1\mapsto i\tilde{\epsilon}_1\,,\quad \epsilon_2\mapsto -\tilde{\epsilon}_2\,.
\end{align}

\subsection{Modular forms for $\Gamma_1(4)\cap\Gamma(2)$}
\label{app:mformsG14capG2}
We will now determine the generators of the relevant subring of the ring of modular forms for ${\Gamma_1(4)\cap\Gamma(2)}$.
\begin{lemma}
      The group $\Gamma_1(4)\cap\Gamma(2)$ is generated by the elements
      \begin{align}
          g_1=\left(\begin{array}{cc}1&2\\0&1\end{array}\right)\,,\quad g_2=\left(\begin{array}{cc}1&0\\4&1\end{array}\right)\,,\quad g_3=\left(\begin{array}{cc}-3&2\\4&-3\end{array}\right)\,.
      \end{align}
      \label{lemma:Gamma14Gamma2gens}
\end{lemma}
\begin{proof}
    First, note that $\Gamma_1(4)=\langle h_1,h_2\rangle$, where 
    \begin{align}
        h_1=\left(\begin{array}{cc}1&1\\0&1\end{array}\right)\,,\quad h_2=\left(\begin{array}{cc}1&0\\4&1\end{array}\right)\,,
    \end{align}
    such that $g_1=h_1^2$, $g_2=h_2$ and $g_3=h_1^{-1}h_2h_1^{-1}$. Clearly, $\langle g_1, g_2, g_3 \rangle \subset \Gamma_1(4) \cap \Gamma(2)$. To see the converse, consider an element $h={\tiny\left(\begin{array}{cc}a&b\\c&d\end{array}\right)}\in \Gamma_1(4)$. We have
    \begin{align}
        h h_1^{\pm 1}=\left(\begin{array}{cc}
            a&b\pm a\\
            c&d\pm c
        \end{array}\right)\,,\quad         h_1^{\pm 1} h=\left(\begin{array}{cc}
            a\pm c&b\pm d\\
            c&d
        \end{array}\right)\,.
    \end{align}
    Since $a\equiv 1\text{ mod }4$ and  $d\equiv 1\text{ mod }4$, we find that $b \equiv 0\text{ mod }2$, and therefore $h\in \Gamma(2)$, if and only if $h$ can be written as
    \begin{align}
        h=h_1^{a_1}h_2^{b_1}h_1^{a_2}h_2^{b_2}\ldots h_1^{a_k}h_2^{b_k}\,,
    \end{align}
    for some $k\in\mathbb{N}$ and $a_i,b_i\in\mathbb{Z}$ with $a_1+\ldots+a_k\in 2\mathbb{Z}$.
    Note that
    \begin{align}
        h_1^{-1} h_2^a h_1^{-1}=(h_1^{-1} h_2 h_1)^ah_1^{-2}=(g_3g_1)^ag_1^{-1}\,.
    \end{align}
    If $a_1$ is odd, we can write
    \begin{align}
        h=g_1^{(a_1+1)/2}(g_3g_1)^{b_1}g_1^{-1}h'\,,\quad h'=h_1^{a_2+1}h_2^{b_2}h_1^{a_3}\ldots\,,
    \end{align}
    and if $a_1$ is even, we write
    \begin{align}
        h=g_1^{a_1/2}g_2^{b_1}h'\,,\quad h'=h_1^{a_2}h_2^{b_2}h_1^{a_3}\ldots\,.
    \end{align}
    In either case, we can write $h'=h_1^{a_1'}h_2^{b_1'}\ldots h_1^{a_{k'}'}h_2^{b_{k'}'}$ with $k'=k-1$ and $a_1'+\ldots+a_{k'}'\in2\mathbb{Z}$.
    Repeating this procedure until $k'=0$, we find that an element $h\in \Gamma_1(4)$ is also contained in $\Gamma(2)$ if and only if it can be written as a word in $g_1,g_2,g_3$.
\end{proof}

We now define the subspace of the space of $\Gamma_1(4)\cap\Gamma(2)$ weight $k$ modular forms
\begin{align}
    \Upsilon_k=\left\{\,\, \phi(\tau)\in M_k\left(\Gamma_1(4)\cap\Gamma(2)\right)\,\,\big\vert\,\, \phi\left(\tau+1\right)=(-1)^k\phi(\tau) \,\,\right\}\,,
    \label{eqn:UpsilonK}
\end{align}
and note that $\Upsilon=\bigoplus\limits_{k\ge 0}\Upsilon_k$ is a subring of $M\left(\Gamma_1(4)\cap\Gamma(2)\right)$.

\begin{lemma}
There is an isomorphism $\kappa_k:\,M_k\left(\Gamma^1(4)\right)\rightarrow \Upsilon_k$ that acts as
\begin{align}
    \kappa_k:\,\widetilde{\phi}(\tau)\mapsto \phi(\tau)=\tau^{-k}\widetilde{\phi}\left(\frac{2\tau-1}{\tau}\right)\,,\quad \kappa_k^{-1}:\,\phi(\tau)\mapsto \widetilde{\phi}(\tau)=(2-\tau)^{-k}\phi\left(\frac{1}{2-\tau}\right)\,.
    \label{eqn:kappamap}
\end{align}
\label{lemma:upsilon}
\end{lemma}
\begin{proof}
Let $\widetilde{\phi}(\tau)\in M_k\left(\Gamma^1(4)\right)$ and $\phi(\tau)$ be the image of $\widetilde{\phi}(\tau)$ under $\kappa_k$. We need to show that $\phi(\tau)\in \Upsilon_k$. Setting $\tau'=(2\tau-1)/\tau$, we verify that
\begin{align} \label{eq:checkT}
    \phi\left(\tau+1\right)=&(\tau+1)^{-k}\widetilde{\phi}\left(\frac{\tau'-4}{\tau'-3}\right)=(\tau+1)^{-k}(\tau'-3)^k\widetilde{\phi}(\tau')=(-1)^k\phi(\tau)\,.
\end{align}
To show that $\phi(\tau) \in M_k(\Gamma_1(4) \cap \Gamma(2))$, given Lemma~\ref{lemma:Gamma14Gamma2gens}, it suffices to verify the transformation properties under the action of $g_1,g_2$ and $g_3$. The correct transformation under $g_1$ is already implied by \eqref{eq:checkT}. For $g_2$ and $g_3$, we compute
\begin{align}
    \begin{split}
        (4\tau+1)^{-k}\phi\left(\frac{\tau}{4\tau+1}\right)=&\tau^{-k}\widetilde{\phi}\left(\tau'-4\right)=\phi(\tau)\,,\\
        (4\tau-3)^{-k}\phi\left(\frac{-3\tau+2}{4\tau-3}\right)=&(-3\tau+2)^{-k}\widetilde{\phi}\left(\frac{-7\tau'+4}{-2\tau'+1}\right)=\left(\frac{-2\tau'+1}{-3\tau+2}\right)^{k}\widetilde{\phi}(\tau')=\phi(\tau)\,.
    \end{split}
\end{align}
On the other hand, given $\phi(\tau)\in \Upsilon_k$ with image $\widetilde{\phi}(\tau)$ under $\kappa_k^{-1}$, we need to show that $\widetilde{\phi}(\tau) \in M_k(\Gamma^1(4))$. As $\Gamma^1(4) = \langle h_1^T, h_2^T \rangle$, the following computation completes the proof:
\begin{align}
    \begin{split}
        \widetilde{\phi}(\tau'+4)=&(-1)^k(\tau'+2)^{-k}\phi\left(\frac{\tau}{-4\tau+1}\right)=\widetilde{\phi}(\tau')\,,\\
        (\tau'+1)^{-k}\widetilde{\phi}\left(\frac{\tau'}{\tau'+1}\right)=&(\tau'+2)^{-k}\phi\left(\frac{\tau}{-4\tau+1}+1\right)=\widetilde{\phi}(\tau')\,.
    \end{split}
\end{align}
\end{proof}

We can now show the following:
\begin{lemma}
    The map~\eqref{eqn:kappamap} acts on the generators $\tilde{\epsilon}_1,\tilde{\epsilon}_2$ of $M\left(\Gamma^1(4)\right)$ as
\begin{align}
    \kappa_1(\tilde{\epsilon}_1)=i\hat{\epsilon}_1\,,\quad \kappa_2(\tilde{\epsilon}_2)=-\hat{\epsilon}_2\,.
\end{align}
    in terms of the modular forms
\begin{align}
    \begin{split}
    \hat{\epsilon}_1=\sqrt{x_3}-\sqrt{x_4}\,,\quad \hat{\epsilon}_2=x_3+x_4\,.
    \end{split}
\end{align}
    The ring $\Upsilon=\bigoplus\limits_{k\ge 0}\Upsilon_k$ is therefore generated by $\hat{\epsilon}_1,\hat{\epsilon}_2$.
\end{lemma}
\begin{proof}
In~\cite[Appendix B.2.3]{Schimannek:2021pau} it was calculated that
\begin{align}
    E_{4,1}\left(\frac{\tau}{2\tau+1}\right)=(2\tau+1)\sqrt{E_{2,2}(\tau)-E_{4,1}(\tau)^2}\,,\quad E_{2,2}\left(\frac{\tau}{2\tau+1}\right)=(2\tau+1)^2E_{2,2}(\tau)\,.
\end{align}
Recalling that $M\left(\Gamma^1(4)\right)$ is generated by $\widetilde{\phi}_1(\tau)=E_{4,1}(\tau/4)$ and $\widetilde{\phi}_2(\tau)=E_{2,2}(\tau/4)$, we define $\phi_i(\tau)=\kappa_i(\widetilde{\phi}_i)$, $i=1,2$ and calculate
\begin{align}
    \begin{split}
    \phi_1(\tau)=&\tau^{-1}E_{4,1}\left(\frac{2\tau-1}{4\tau}\right)=\tau^{-1}E_{4,1}\left(\frac{\tau-\frac12}{2(\tau-\frac12)+1}\right)\\
            =&2\sqrt{E_{2,2}\left(\tau-\frac12\right)-E_{4,1}\left(\tau-\frac12\right)^2}=i\left(\sqrt{x_3(\tau)}-\sqrt{x_4(\tau)}\right)\,,\\
    \phi_2(\tau)=&\tau^{-2}E_{4,2}\left(\frac{2\tau-1}{4\tau}\right)=\tau^{-2}E_{2,2}\left(\frac{\tau-\frac12}{2(\tau-\frac12)+1}\right)\\
    =&4E_{2,2}\left(\tau-\frac12\right)=6\sqrt{x_3(\tau)x_4(\tau)}-x_3(\tau)-x_4(\tau)\,.
    \end{split}
\end{align}
The claim then follows after expressing $E_{4,1}$ and $E_{2,2}$ in terms of $\tilde{\epsilon}_1,\tilde{\epsilon}_2$.
\end{proof}

\section{Intersection theory on the Chow ring}
\label{sec:Chow}
The natural setting for intersection theory on a variety $X$ (and  on more general schemes) is its Chow ring $\Ch_*(X)$. In this appendix, we formulate some of the concepts introduced in Section \ref{sec:genusonefibrations} in this context and provide further details which the reader might find useful when performing intersection theory on torus fibrations.
\begin{itemize}
    \item For a smooth fibration $\pi : X\rightarrow B$ over a smooth base $B$, the projection formula
        \begin{equation}
          \pi_*(A \cdot \pi^* b) = \pi_* A \cdot b \,,\quad A\in \Ch_*(X), \, b \in \Ch_*(B) \,,
       \end{equation}
    permits pushing forward intersection theory calculations on $X$ to calculations on $B$. As the proper pushforward of a point is a point, this in particular allows evaluating triple intersections in $X$ by pushing forward to the base $B$.
    \item To study the intersection theory of holomorphic sections and $m$-sections, let us first consider the more general setup where 
        \begin{equation}
            i: S \rightarrow X
        \end{equation}
        is an embedding of a smooth closed subvariety $S$ into $X$. When no confusion can arise, we will use the letter $S$ also to denote the class of $S$ and $i(S)$ in $\Ch_*(S)$, $\Ch_*(X)$ respectively (so that $i_* S = S)$. Then \cite[Example 8.1.1]{FultonIntersection}
        \begin{equation}
            i_* i^* A = S \cdot A \quad \mathrm{for} \quad A \in \Ch_*(X) \,.
        \end{equation}
        By adjunction,
        \begin{equation}
            K_S = (K_X + S)|_S = i^*(K_X + S)\,. 
        \end{equation}
        On a Calabi-Yau variety $X$, we thus have the relation
        \begin{equation} \label{eq:ssKs}
            i_* K_S = S \cdot S  \,.
        \end{equation}
    \item Let $i_* S = J_0$ be a holomorphic $m$-section. From the definitions,\footnote{For a proper map $\pi: X \rightarrow B$, $\pi_* J_0$ is given by the  multiple $\deg \pi|_{J_0}$ of the image, where $\deg \pi|_{J_{0}}:=0$ if $\dim \pi(J_0) < \dim J_0$, see e.g. \cite[Section 1.4]{FultonIntersection}.}
        \begin{equation}
            J_0 \cdot E = m \,, \quad \deg \pi|_{J_0} = m \,, \quad \pi_* J_0 = m B \,.
        \end{equation}
        Introduce the degree $m$ map $f = \pi \circ i : S \rightarrow B$, such that $f_* f^* = m \,\id_B$. Then, by the ramification formula \cite[Theorem 5.5]{MR0637060},
        \begin{equation} \label{eq:ram}
            K_S = f^* K_B + R_f \quad \Rightarrow \quad f_* K_S = m K_B + f_* R_f \,,
        \end{equation}
        where $R_f$ is the ramification divisor associated to $f$. Hence, upon acting on \eqref{eq:ssKs} by $\pi_*$ and comparing to \eqref{eq:ram}, we obtain
        \begin{equation}
            \pi_* \left( J_0 \cdot J_0 \right) = m K_B + f_* R_f \,.
        \end{equation}
        Also, by acting with $i_*$ on the LHS of \eqref{eq:ram},
        \begin{equation} \label{eq:iKs}
            i_* K_S = S \cdot \pi^* K_B + i_* R_f \,.
        \end{equation}
    \item We now specialize to $i_*S = J_0$ a holomorphic section. Then $R_f =0$, and by combining \eqref{eq:ram} and \eqref{eq:iKs},
        \begin{equation}
            J_0 \cdot J_0 = J_0 \cdot \pi^* K_B \quad \Rightarrow \quad J_0 \cdot J_0 \cdot J_0 = J_0 \cdot \pi^* \left( K_B \cdot K_B\right) \,,
        \end{equation}
        where we have used that as $\pi$ is flat, $\pi^*$ is functorial \cite[Section 1.7]{FultonIntersection}.
        Hence, 
        \begin{equation}
            \pi_* (J_0 \cdot J_0 \cdot J_0) = K_B \cdot K_B \,.
        \end{equation}
    \item For $J_0$ a holomorphic section, the relation
    \begin{equation} \label{eq:J0c2}
        J_0 \cdot c_2(TX) = c_2(B) - K_B \cdot K_B
    \end{equation}
    follows from the adjunction formula and the Whitney sum formula \cite{Bonetti:2011mw}.

    \item By functoriality of $\pi^*$,
    \begin{equation}
        J_\alpha \cdot J_\beta = \pi^*(\Jbase_\alpha \cdot \Jbase_\beta) = \Omega_{\alpha \beta}\, \pi^{-1}(p) = \Omega_{\alpha \beta}\, E \,.
    \end{equation}
    \item By definition of $m$-sections and vertical divisors, 
    \begin{align}
        \pi_* J_\alpha &= 0 \,, \quad \alpha =1, \ldots, b_2(B) \,,\\
        \pi_* J_0 & = m B \,, \\ \label{eq:piJ0J0}
        \pi_* \left(J_0 \cdot J_0\right) &= \sum_\alpha c^\alpha \Jbase_\alpha \quad \Leftrightarrow \quad c^\alpha = \Omega^{\alpha \beta} \pi_*(J_0 \cdot J_0) \cdot \Jbase_\beta \,.
    \end{align}
    When $J_0$ is a holomorphic section, we have merely introduced an expansion
    \begin{equation} \label{eq:expansionKB}
        K_B = \sum_\alpha c^\alpha j_\alpha 
    \end{equation}
    for the canonical divisor of the base in \eqref{eq:piJ0J0}. Note that in this case,
    \begin{align} \label{eq:cOc}
        c^\alpha \Omega_{\alpha \beta} c^\beta = K_B \cdot K_B \,.
    \end{align}

\end{itemize}

\addcontentsline{toc}{section}{References}
\bibliographystyle{utphys}
\bibliography{names}

\providecommand{\href}[2]{#2}\begingroup\raggedright\begin{thebibliography}{100}

\bibitem{Vafa1996}
C.~Vafa, ``{Evidence for F-theory},''
  \href{http://dx.doi.org/10.1016/0550-3213(96)00172-1}{{\em Nuclear Physics B}
  {\bfseries 469} no.~3, (June, 1996) 403–415}.
  \url{http://dx.doi.org/10.1016/0550-3213(96)00172-1}.

\bibitem{Morrison1996a}
D.~R. Morrison and C.~Vafa, ``{Compactifications of F-theory on Calabi-Yau
  threefolds. (I)},''
  \href{http://dx.doi.org/10.1016/0550-3213(96)00242-8}{{\em Nuclear Physics B}
  {\bfseries 473} no.~1–2, (Aug., 1996) 74–92}.
  \url{http://dx.doi.org/10.1016/0550-3213(96)00242-8}.

\bibitem{Morrison1996b}
D.~R. Morrison and C.~Vafa, ``{Compactifications of F-theory on Calabi-Yau
  threefolds (II)},''
  \href{http://dx.doi.org/10.1016/0550-3213(96)00369-0}{{\em Nuclear Physics B}
  {\bfseries 476} no.~3, (Sept., 1996) 437–469}.
  \url{http://dx.doi.org/10.1016/0550-3213(96)00369-0}.

\bibitem{deBoer:2001wca}
J.~de~Boer, R.~Dijkgraaf, K.~Hori, A.~Keurentjes, J.~Morgan, D.~R. Morrison,
  and S.~Sethi, ``{Triples, fluxes, and strings},''
  \href{http://dx.doi.org/10.4310/ATMP.2000.v4.n5.a1}{{\em Adv. Theor. Math.
  Phys.} {\bfseries 4} (2002) 995--1186},
  \href{http://arxiv.org/abs/hep-th/0103170}{{\ttfamily arXiv:hep-th/0103170}}.

\bibitem{Braun:2014oya}
V.~Braun and D.~R. Morrison, ``{F-theory on Genus-One Fibrations},''
  \href{http://dx.doi.org/10.1007/JHEP08(2014)132}{{\em JHEP} {\bfseries 08}
  (2014) 132}, \href{http://arxiv.org/abs/1401.7844}{{\ttfamily arXiv:1401.7844
  [hep-th]}}.

\bibitem{Morrison:2014era}
D.~R. Morrison and W.~Taylor, ``{Sections, multisections, and U(1) fields in
  F-theory},'' \href{http://arxiv.org/abs/1404.1527}{{\ttfamily arXiv:1404.1527
  [hep-th]}}.

\bibitem{Anderson:2014yva}
L.~B. Anderson, I.~n. Garc\'\i{}a-Etxebarria, T.~W. Grimm, and J.~Keitel,
  ``{Physics of F-theory compactifications without section},''
  \href{http://dx.doi.org/10.1007/JHEP12(2014)156}{{\em JHEP} {\bfseries 12}
  (2014) 156}, \href{http://arxiv.org/abs/1406.5180}{{\ttfamily arXiv:1406.5180
  [hep-th]}}.

\bibitem{Klevers:2014bqa}
D.~Klevers, D.~K. Mayorga~Pena, P.-K. Oehlmann, H.~Piragua, and J.~Reuter,
  ``{F-Theory on all Toric Hypersurface Fibrations and its Higgs Branches},''
  \href{http://dx.doi.org/10.1007/JHEP01(2015)142}{{\em JHEP} {\bfseries 01}
  (2015) 142}, \href{http://arxiv.org/abs/1408.4808}{{\ttfamily arXiv:1408.4808
  [hep-th]}}.

\bibitem{Mayrhofer:2014laa}
C.~Mayrhofer, E.~Palti, O.~Till, and T.~Weigand, ``{On Discrete Symmetries and
  Torsion Homology in F-Theory},''
  \href{http://dx.doi.org/10.1007/JHEP06(2015)029}{{\em JHEP} {\bfseries 06}
  (2015) 029}, \href{http://arxiv.org/abs/1410.7814}{{\ttfamily arXiv:1410.7814
  [hep-th]}}.

\bibitem{Cvetic:2015moa}
M.~Cveti\v{c}, R.~Donagi, D.~Klevers, H.~Piragua, and M.~Poretschkin,
  ``{F-theory vacua with $\mathbb Z_3$ gauge symmetry},''
  \href{http://dx.doi.org/10.1016/j.nuclphysb.2015.07.011}{{\em Nucl. Phys. B}
  {\bfseries 898} (2015) 736--750},
  \href{http://arxiv.org/abs/1502.06953}{{\ttfamily arXiv:1502.06953
  [hep-th]}}.

\bibitem{Oehlmann:2019ohh}
P.-K. Oehlmann and T.~Schimannek, ``{GV-Spectroscopy for F-theory on genus-one
  fibrations},'' \href{http://dx.doi.org/10.1007/JHEP09(2020)066}{{\em JHEP}
  {\bfseries 09} (2020) 066}, \href{http://arxiv.org/abs/1912.09493}{{\ttfamily
  arXiv:1912.09493 [hep-th]}}.

\bibitem{Knapp:2021vkm}
J.~Knapp, E.~Scheidegger, and T.~Schimannek, ``{On genus one fibered Calabi-Yau
  threefolds with 5-sections},''
  \href{http://arxiv.org/abs/2107.05647}{{\ttfamily arXiv:2107.05647
  [hep-th]}}.

\bibitem{Schimannek:2021pau}
T.~Schimannek, ``{Modular curves, the Tate-Shafarevich group and Gopakumar-Vafa
  invariants with discrete charges},''
  \href{http://dx.doi.org/10.1007/JHEP02(2022)007}{{\em JHEP} {\bfseries 02}
  (2022) 007}, \href{http://arxiv.org/abs/2108.09311}{{\ttfamily
  arXiv:2108.09311 [hep-th]}}.

\bibitem{Dierigl:2022zll}
M.~Dierigl, P.-K. Oehlmann, and T.~Schimannek, ``{The discrete Green-Schwarz
  mechanism in 6D F-Theory and Elliptic Genera of Non-Critical Strings},''
  \href{http://arxiv.org/abs/2212.04503}{{\ttfamily arXiv:2212.04503
  [hep-th]}}.

\bibitem{Katz:2022lyl}
S.~Katz, A.~Klemm, T.~Schimannek, and E.~Sharpe, ``{Topological Strings on
  Non-commutative Resolutions},''
  \href{http://dx.doi.org/10.1007/s00220-023-04896-2}{{\em Commun. Math. Phys.}
  {\bfseries 405} no.~3, (2024) 62},
  \href{http://arxiv.org/abs/2212.08655}{{\ttfamily arXiv:2212.08655
  [hep-th]}}.

\bibitem{Katz:2023zan}
S.~Katz and T.~Schimannek, ``{New non-commutative resolutions of determinantal
  Calabi-Yau threefolds from hybrid GLSM},''
  \href{http://arxiv.org/abs/2307.00047}{{\ttfamily arXiv:2307.00047
  [hep-th]}}.

\bibitem{Schimannek:2025cok}
T.~Schimannek, ``{In search of almost generic Calabi-Yau 3-folds},''
  \href{http://arxiv.org/abs/2504.06115}{{\ttfamily arXiv:2504.06115
  [hep-th]}}.

\bibitem{Pioline:wip}
B.~Pioline and T.~Schimannek, ``{Revisiting the Quantum Geometry of
  Torus-fibered Calabi-Yau Threefolds (work in progress)},''.

\bibitem{Dierigl:wip}
M.~Dierigl, P.-K. Oehlmann, and T.~Schimannek, ``{(work in progress)},''.

\bibitem{wipZ2Z2}
A.-K. Kashani-Poor and T.~Schimannek, ``{F-theory vacua with
  $\mathbb{Z}_2\times\mathbb{Z}_2$ gauge symmetry and symmetric determinantal
  double covers of $\mathbb{P}^1$-bundles on Fano surfaces},''.

\bibitem{Arras:2016evy}
P.~Arras, A.~Grassi, and T.~Weigand, ``{Terminal Singularities, Milnor Numbers,
  and Matter in F-theory},''
  \href{http://dx.doi.org/10.1016/j.geomphys.2017.09.001}{{\em J. Geom. Phys.}
  {\bfseries 123} (2018) 71--97},
  \href{http://arxiv.org/abs/1612.05646}{{\ttfamily arXiv:1612.05646
  [hep-th]}}.

\bibitem{Baume:2017hxm}
F.~Baume, M.~Cvetic, C.~Lawrie, and L.~Lin, ``{When rational sections become
  cyclic {\textemdash} Gauge enhancement in F-theory via Mordell-Weil
  torsion},'' \href{http://dx.doi.org/10.1007/JHEP03(2018)069}{{\em JHEP}
  {\bfseries 03} (2018) 069}, \href{http://arxiv.org/abs/1709.07453}{{\ttfamily
  arXiv:1709.07453 [hep-th]}}.

\bibitem{Klemm:1996hh}
A.~Klemm, P.~Mayr, and C.~Vafa, ``{BPS states of exceptional noncritical
  strings},'' \href{http://dx.doi.org/10.1016/S0920-5632(97)00422-2}{{\em Nucl.
  Phys. B Proc. Suppl.} {\bfseries 58} (1997) 177},
  \href{http://arxiv.org/abs/hep-th/9607139}{{\ttfamily arXiv:hep-th/9607139}}.

\bibitem{Haghighat:2013gba}
B.~Haghighat, A.~Iqbal, C.~Koz\c{c}az, G.~Lockhart, and C.~Vafa,
  ``{M-Strings},'' \href{http://dx.doi.org/10.1007/s00220-014-2139-1}{{\em
  Commun. Math. Phys.} {\bfseries 334} no.~2, (2015) 779--842},
  \href{http://arxiv.org/abs/1305.6322}{{\ttfamily arXiv:1305.6322 [hep-th]}}.

\bibitem{Haghighat:2014vxa}
B.~Haghighat, A.~Klemm, G.~Lockhart, and C.~Vafa, ``{Strings of Minimal 6d
  SCFTs},'' \href{http://dx.doi.org/10.1002/prop.201500014}{{\em Fortsch.
  Phys.} {\bfseries 63} (2015) 294--322},
  \href{http://arxiv.org/abs/1412.3152}{{\ttfamily arXiv:1412.3152 [hep-th]}}.

\bibitem{Candelas:1994hw}
P.~Candelas, A.~Font, S.~H. Katz, and D.~R. Morrison, ``{Mirror symmetry for
  two parameter models. 2.},''
  \href{http://dx.doi.org/10.1016/0550-3213(94)90155-4}{{\em Nucl. Phys. B}
  {\bfseries 429} (1994) 626--674},
  \href{http://arxiv.org/abs/hep-th/9403187}{{\ttfamily arXiv:hep-th/9403187}}.

\bibitem{Klemm:2012sx}
A.~Klemm, J.~Manschot, and T.~Wotschke, ``{Quantum geometry of elliptic
  Calabi-Yau manifolds},'' \href{http://arxiv.org/abs/1205.1795}{{\ttfamily
  arXiv:1205.1795 [hep-th]}}.

\bibitem{Alim:2012ss}
M.~Alim and E.~Scheidegger, ``{Topological Strings on Elliptic Fibrations},''
  \href{http://dx.doi.org/10.4310/CNTP.2014.v8.n4.a4}{{\em Commun. Num. Theor.
  Phys.} {\bfseries 08} (2014) 729--800},
  \href{http://arxiv.org/abs/1205.1784}{{\ttfamily arXiv:1205.1784 [hep-th]}}.

\bibitem{Huang:2015sta}
M.-x. Huang, S.~Katz, and A.~Klemm, ``{Topological String on elliptic CY
  3-folds and the ring of Jacobi forms},''
  \href{http://dx.doi.org/10.1007/JHEP10(2015)125}{{\em JHEP} {\bfseries 10}
  (2015) 125}, \href{http://arxiv.org/abs/1501.04891}{{\ttfamily
  arXiv:1501.04891 [hep-th]}}.

\bibitem{DelZotto:2016pvm}
M.~Del~Zotto and G.~Lockhart, ``{On Exceptional Instanton Strings},''
  \href{http://dx.doi.org/10.1007/JHEP09(2017)081}{{\em JHEP} {\bfseries 09}
  (2017) 081}, \href{http://arxiv.org/abs/1609.00310}{{\ttfamily
  arXiv:1609.00310 [hep-th]}}.

\bibitem{Gu:2017ccq}
J.~Gu, M.-x. Huang, A.-K. Kashani-Poor, and A.~Klemm, ``{Refined BPS invariants
  of 6d SCFTs from anomalies and modularity},''
  \href{http://dx.doi.org/10.1007/JHEP05(2017)130}{{\em JHEP} {\bfseries 05}
  (2017) 130}, \href{http://arxiv.org/abs/1701.00764}{{\ttfamily
  arXiv:1701.00764 [hep-th]}}.

\bibitem{DelZotto:2017mee}
M.~Del~Zotto, J.~Gu, M.-X. Huang, A.-K. Kashani-Poor, A.~Klemm, and
  G.~Lockhart, ``{Topological Strings on Singular Elliptic Calabi-Yau 3-folds
  and Minimal 6d SCFTs},''
  \href{http://dx.doi.org/10.1007/JHEP03(2018)156}{{\em JHEP} {\bfseries 03}
  (2018) 156}, \href{http://arxiv.org/abs/1712.07017}{{\ttfamily
  arXiv:1712.07017 [hep-th]}}.

\bibitem{DelZotto:2018tcj}
M.~Del~Zotto and G.~Lockhart, ``{Universal Features of BPS Strings in
  Six-dimensional SCFTs},''
  \href{http://dx.doi.org/10.1007/JHEP08(2018)173}{{\em JHEP} {\bfseries 08}
  (2018) 173}, \href{http://arxiv.org/abs/1804.09694}{{\ttfamily
  arXiv:1804.09694 [hep-th]}}.

\bibitem{Schimannek:2019ijf}
T.~Schimannek, ``{Modularity from Monodromy},''
  \href{http://dx.doi.org/10.1007/JHEP05(2019)024}{{\em JHEP} {\bfseries 05}
  (2019) 024}, \href{http://arxiv.org/abs/1902.08215}{{\ttfamily
  arXiv:1902.08215 [hep-th]}}.

\bibitem{Cota:2019cjx}
C.~F. Cota, A.~Klemm, and T.~Schimannek, ``{Topological strings on genus one
  fibered Calabi-Yau 3-folds and string dualities},''
  \href{http://dx.doi.org/10.1007/JHEP11(2019)170}{{\em JHEP} {\bfseries 11}
  (2019) 170}, \href{http://arxiv.org/abs/1910.01988}{{\ttfamily
  arXiv:1910.01988 [hep-th]}}.

\bibitem{Duan:2020imo}
Z.~Duan, D.~J. Duque, and A.-K. Kashani-Poor, ``{Weyl invariant Jacobi forms
  along Higgsing trees},''
  \href{http://dx.doi.org/10.1007/JHEP04(2021)224}{{\em JHEP} {\bfseries 04}
  (2021) 224}, \href{http://arxiv.org/abs/2012.10427}{{\ttfamily
  arXiv:2012.10427 [hep-th]}}.

\bibitem{Duque:2022tub}
D.~J. Duque and A.-K. Kashani-Poor, ``{Affine characters at negative level and
  elliptic genera of non-critical strings},''
  \href{http://dx.doi.org/10.1007/JHEP07(2023)208}{{\em JHEP} {\bfseries 07}
  (2023) 208}, \href{http://arxiv.org/abs/2211.14601}{{\ttfamily
  arXiv:2211.14601 [hep-th]}}.

\bibitem{caldararuThesis}
A.~Căldăraru, {\em Derived categories of twisted sheaves on Calabi-Yau
  manifolds}.
\newblock Phd thesis, Cornell University, 2000.
\newblock
  \url{https://people.math.wisc.edu/~caldararu/publications/ThesisSingleSpaced.pdf}.

\bibitem{Caldararu2002}
A.~Caldararu, ``Derived categories of twisted sheaves on elliptic threefolds,''
  \href{http://dx.doi.org/10.1515/crll.2002.022}{{\em Journal f\"{u}r die reine
  und angewandte Mathematik (Crelles Journal)} {\bfseries 2002} no.~544, (Jan.,
  2002) 161--179}, \href{http://arxiv.org/abs/math/0012083}{{\ttfamily
  arXiv:math/0012083}}. \url{https://doi.org/10.1515/crll.2002.022}.

\bibitem{carthy}
K.~Conrad, ``Characters of finite abelian groups.''
\newblock \url{https://kconrad.math.uconn.edu/blurbs/grouptheory/charthy.pdf}.

\bibitem{bookSmithNormalForm}
C.~Norman, \href{http://dx.doi.org/10.1007/978-1-4471-2730-7}{{\em Finitely
  generated {Abelian} groups and similarity of matrices over a field.}}
\newblock Springer Undergrad. Math. Ser. Berlin: Springer, 2012.

\bibitem{Anderson:2023wkr}
L.~B. Anderson, J.~Gray, and P.-K. Oehlmann, ``{Twisted Fibrations in
  M/F-theory},'' \href{http://dx.doi.org/10.1007/JHEP01(2024)017}{{\em JHEP}
  {\bfseries 01} (2024) 017}, \href{http://arxiv.org/abs/2308.07364}{{\ttfamily
  arXiv:2308.07364 [hep-th]}}.

\bibitem{Anderson:2023tfy}
L.~B. Anderson, J.~Gray, and P.-K. Oehlmann, ``{Calabi-Yau Genus-One Fibrations
  and Twisted Dimensional Reductions of F-theory},'' in {\em {Nankai Symposium
  on Mathematical Dialogues}: {In celebration of S.S.Chern's 110th
  anniversary}}.
\newblock 8, 2023.
\newblock \href{http://arxiv.org/abs/2308.12826}{{\ttfamily arXiv:2308.12826
  [hep-th]}}.

\bibitem{Ahmed:2024wve}
H.~Ahmed, P.-K. Oehlmann, and F.~Ruehle, ``{A Twist on Heterotic Little String
  Duality},'' \href{http://arxiv.org/abs/2411.05313}{{\ttfamily
  arXiv:2411.05313 [hep-th]}}.

\bibitem{aglecture}
A.~Sutherland, ``{Introduction to Arithmetic Geometry, Lecture 26}.''
\newblock
  \url{https://ocw.mit.edu/courses/18-782-introduction-to-arithmetic-geometry-fall-2013}.

\bibitem{dg92}
I.~Dolgachev and M.~Gross, ``Elliptic threefolds. {I}. {O}gg-{S}hafarevich
  theory,'' {\em J. Algebraic Geom.} {\bfseries 3} no.~1, (1994) 39--80.

\bibitem{Kloosterman}
R.~Kloosterman, ``The average rank of elliptic {$n$}-folds,''
  \href{http://dx.doi.org/10.1512/iumj.2012.61.4540}{{\em Indiana Univ. Math.
  J.} {\bfseries 61} no.~1, (2012) 131--146}.
  \url{https://doi.org/10.1512/iumj.2012.61.4540}.

\bibitem{KawamataKodairaDimension}
Y.~Kawamata, ``Minimal models and the {K}odaira dimension of algebraic fiber
  spaces,'' \href{http://dx.doi.org/10.1515/crll.1985.363.1}{{\em J. Reine
  Angew. Math.} {\bfseries 363} (1985) 1--46}.
  \url{https://doi.org/10.1515/crll.1985.363.1}.

\bibitem{Morrison:2012np}
D.~R. Morrison and W.~Taylor, ``{Classifying bases for 6D F-theory models},''
  \href{http://dx.doi.org/10.2478/s11534-012-0065-4}{{\em Central Eur. J.
  Phys.} {\bfseries 10} (2012) 1072--1088},
  \href{http://arxiv.org/abs/1201.1943}{{\ttfamily arXiv:1201.1943 [hep-th]}}.

\bibitem{Derenthal2008}
U.~Derenthal, M.~Joyce, and Z.~Teitler, ``The nef cone volume of generalized
  del pezzo surfaces,'' \href{http://dx.doi.org/10.2140/ant.2008.2.157}{{\em
  Algebra \& Number Theory} {\bfseries 2} no.~2, (Mar., 2008) 157–182}.
  \url{http://dx.doi.org/10.2140/ant.2008.2.157}.

\bibitem{KodairaII}
K.~Kodaira, ``On the structure of compact complex analytic surfaces. {II},''
  \href{http://dx.doi.org/10.2307/2373150}{{\em Amer. J. Math.} {\bfseries 88}
  (1966) 682--721}. \url{https://doi.org/10.2307/2373150}.

\bibitem{Grassi:2001xu}
A.~Grassi, Z.~Guralnik, and B.~A. Ovrut, ``{Knots, braids and BPS states in M
  theory},'' \href{http://dx.doi.org/10.1088/1126-6708/2002/06/023}{{\em JHEP}
  {\bfseries 06} (2002) 023},
  \href{http://arxiv.org/abs/hep-th/0110036}{{\ttfamily arXiv:hep-th/0110036}}.

\bibitem{Kollar:2012pv}
J.~Kollar, ``{Deformations of elliptic Calabi-Yau manifolds},''
  \href{http://arxiv.org/abs/1206.5721}{{\ttfamily arXiv:1206.5721 [math.AG]}}.

\bibitem{Grassi:2011hq}
A.~Grassi and D.~R. Morrison, ``{Anomalies and the Euler characteristic of
  elliptic Calabi-Yau threefolds},''
  \href{http://dx.doi.org/10.4310/CNTP.2012.v6.n1.a2}{{\em Commun. Num. Theor.
  Phys.} {\bfseries 6} (2012) 51--127},
  \href{http://arxiv.org/abs/1109.0042}{{\ttfamily arXiv:1109.0042 [hep-th]}}.

\bibitem{BarthComplexSurfaces}
W.~P. Barth, K.~Hulek, C.~A.~M. Peters, and A.~Van~de Ven,
  \href{http://dx.doi.org/10.1007/978-3-642-57739-0}{{\em Compact complex
  surfaces}}, vol.~4 of {\em Ergebnisse der Mathematik und ihrer Grenzgebiete.
  3. Folge. A Series of Modern Surveys in Mathematics}.
\newblock Springer-Verlag, Berlin, second~ed., 2004.
\newblock \url{https://doi.org/10.1007/978-3-642-57739-0}.

\bibitem{ViehwegSubadditivity}
E.~Viehweg, ``Canonical divisors and the additivity of the {K}odaira dimension
  for morphisms of relative dimension one,'' {\em Compositio Math.} {\bfseries
  35} no.~2, (1977) 197--223.
  \url{http://www.numdam.org/item?id=CM_1977__35_2_197_0}.

\bibitem{SHIODA1972}
T.~Shioda, ``On elliptic modular surfaces,''
  \href{http://dx.doi.org/10.2969/jmsj/02410020}{{\em Journal of the
  Mathematical Society of Japan} {\bfseries 24} no.~1, (Jan., 1972) }.
  \url{http://dx.doi.org/10.2969/jmsj/02410020}.

\bibitem{tate1965algebraic}
J.~T. Tate, ``Algebraic cycles and poles of zeta functions,'' in {\em
  Arithmetical Algebraic Geometry (Proc. Conf. Purdue Univ., 1963)},
  pp.~93--110, Harper \& Row.
\newblock 1965.

\bibitem{SB_1964-1966__9__415_0}
J.~Tate, ``On the conjectures of {Birch} and {Swinnerton-Dyer} and a geometric
  analog,'' in {\em S\'eminaire Bourbaki : ann\'ees 1964/65 1965/66, expos\'es
  277-312}, no.~9 in S\'eminaire Bourbaki, pp.~415--440.
\newblock Soci\'et\'e math\'ematique de France, 1966.
\newblock \url{http://www.numdam.org/item/SB_1964-1966__9__415_0/}.
\newblock talk:306.

\bibitem{Wazir2004}
R.~Wazir, ``Arithmetic on elliptic threefolds,''
  \href{http://dx.doi.org/10.1112/s0010437x03000381}{{\em Compositio
  Mathematica} {\bfseries 140} no.~03, (May, 2004) 567–580}.
  \url{http://dx.doi.org/10.1112/S0010437X03000381}.

\bibitem{Grimm:2013oga}
T.~W. Grimm, A.~Kapfer, and J.~Keitel, ``{Effective action of 6D F-Theory with
  U(1) factors: Rational sections make Chern-Simons terms jump},''
  \href{http://dx.doi.org/10.1007/JHEP07(2013)115}{{\em JHEP} {\bfseries 07}
  (2013) 115}, \href{http://arxiv.org/abs/1305.1929}{{\ttfamily arXiv:1305.1929
  [hep-th]}}.

\bibitem{Wilson1989}
P.~M.~H. Wilson, ``Calabi-yau manifolds with large picard number,''
  \href{http://dx.doi.org/10.1007/bf01388848}{{\em Inventiones Mathematicae}
  {\bfseries 98} no.~1, (Feb., 1989) 139–155}.
  \url{http://dx.doi.org/10.1007/BF01388848}.

\bibitem{OGUISO1993}
K.~Oguiso, ``On algebraic fiber space structures on a calabi-yau 3-fold,''
  \href{http://dx.doi.org/10.1142/s0129167x93000248}{{\em International Journal
  of Mathematics} {\bfseries 04} no.~03, (June, 1993) 439–465}.
  \url{http://dx.doi.org/10.1142/S0129167X93000248}.

\bibitem{Wilson1994}
P.~M.~H. Wilson, ``The existence of elliptic fibre space structures on
  calabi-yau threefolds,'' \href{http://dx.doi.org/10.1007/bf01450510}{{\em
  Mathematische Annalen} {\bfseries 300} no.~1, (Sept., 1994) 693–703}.
  \url{http://dx.doi.org/10.1007/BF01450510}.

\bibitem{Donagi:2003av}
R.~Donagi and T.~Pantev, ``{Torus fibrations, gerbes, and duality},''
  \href{http://arxiv.org/abs/math/0306213}{{\ttfamily arXiv:math/0306213}}.

\bibitem{JongBrauer}
A.~J. de~Jong, ``A result of gabber,''.
  \url{https://www.math.columbia.edu/~dejong/papers/2-gabber.pdf}.

\bibitem{Schrer2005}
S.~Schr\"{o}er, ``Topological methods for complex-analytic brauer groups,''
  \href{http://dx.doi.org/10.1016/j.top.2005.02.005}{{\em Topology} {\bfseries
  44} no.~5, (Sept., 2005) 875–894}.
  \url{http://dx.doi.org/10.1016/j.top.2005.02.005}.

\bibitem{Weigand:2018rez}
T.~Weigand, ``{F-theory},'' {\em PoS} {\bfseries TASI2017} (2018) 016,
  \href{http://arxiv.org/abs/1806.01854}{{\ttfamily arXiv:1806.01854
  [hep-th]}}.

\bibitem{Bhardwaj:2015oru}
L.~Bhardwaj, M.~Del~Zotto, J.~J. Heckman, D.~R. Morrison, T.~Rudelius, and
  C.~Vafa, ``{F-theory and the Classification of Little Strings},''
  \href{http://dx.doi.org/10.1103/PhysRevD.93.086002}{{\em Phys. Rev. D}
  {\bfseries 93} no.~8, (2016) 086002},
  \href{http://arxiv.org/abs/1511.05565}{{\ttfamily arXiv:1511.05565
  [hep-th]}}. [Erratum: Phys.Rev.D 100, 029901 (2019)].

\bibitem{Green:1984bx}
M.~B. Green, J.~H. Schwarz, and P.~C. West, ``{Anomaly Free Chiral Theories in
  Six-Dimensions},'' \href{http://dx.doi.org/10.1016/0550-3213(85)90222-6}{{\em
  Nucl. Phys. B} {\bfseries 254} (1985) 327--348}.

\bibitem{GREEN1984117}
M.~B. Green and J.~H. Schwarz, ``Anomaly cancellations in supersymmetric d = 10
  gauge theory and superstring theory,''
  \href{http://dx.doi.org/https://doi.org/10.1016/0370-2693(84)91565-X}{{\em
  Physics Letters B} {\bfseries 149} no.~1, (1984) 117--122}.
  \url{https://www.sciencedirect.com/science/article/pii/037026938491565X}.

\bibitem{Sagnotti:1992qw}
A.~Sagnotti, ``{A Note on the Green-Schwarz mechanism in open string
  theories},'' \href{http://dx.doi.org/10.1016/0370-2693(92)90682-T}{{\em Phys.
  Lett. B} {\bfseries 294} (1992) 196--203},
  \href{http://arxiv.org/abs/hep-th/9210127}{{\ttfamily arXiv:hep-th/9210127}}.

\bibitem{Monnier:2018nfs}
S.~Monnier and G.~W. Moore, ``{Remarks on the Green{\textendash}Schwarz Terms
  of Six-Dimensional Supergravity Theories},''
  \href{http://dx.doi.org/10.1007/s00220-019-03341-7}{{\em Commun. Math. Phys.}
  {\bfseries 372} no.~3, (2019) 963--1025},
  \href{http://arxiv.org/abs/1808.01334}{{\ttfamily arXiv:1808.01334
  [hep-th]}}.

\bibitem{Dierigl:2025rfn}
M.~Dierigl and M.~Tartaglia, ``{(Quadratically) Refined Discrete Anomaly
  Cancellation},'' \href{http://arxiv.org/abs/2504.02934}{{\ttfamily
  arXiv:2504.02934 [hep-th]}}.

\bibitem{Andrianopoli:1996cm}
L.~Andrianopoli, M.~Bertolini, A.~Ceresole, R.~D'Auria, S.~Ferrara, P.~Fre, and
  T.~Magri, ``{N=2 supergravity and N=2 superYang-Mills theory on general
  scalar manifolds: Symplectic covariance, gaugings and the momentum map},''
  \href{http://dx.doi.org/10.1016/S0393-0440(97)00002-8}{{\em J. Geom. Phys.}
  {\bfseries 23} (1997) 111--189},
  \href{http://arxiv.org/abs/hep-th/9605032}{{\ttfamily arXiv:hep-th/9605032}}.

\bibitem{Bonetti:2011mw}
F.~Bonetti and T.~W. Grimm, ``{Six-dimensional (1,0) effective action of
  F-theory via M-theory on Calabi-Yau threefolds},''
  \href{http://dx.doi.org/10.1007/JHEP05(2012)019}{{\em JHEP} {\bfseries 05}
  (2012) 019}, \href{http://arxiv.org/abs/1112.1082}{{\ttfamily arXiv:1112.1082
  [hep-th]}}.

\bibitem{Bonetti:2013ela}
F.~Bonetti, T.~W. Grimm, and S.~Hohenegger, ``{One-loop Chern-Simons terms in
  five dimensions},'' \href{http://dx.doi.org/10.1007/JHEP07(2013)043}{{\em
  JHEP} {\bfseries 07} (2013) 043},
  \href{http://arxiv.org/abs/1302.2918}{{\ttfamily arXiv:1302.2918 [hep-th]}}.

\bibitem{Bhardwaj:2019fzv}
L.~Bhardwaj, P.~Jefferson, H.-C. Kim, H.-C. Tarazi, and C.~Vafa, ``{Twisted
  Circle Compactifications of 6d SCFTs},''
  \href{http://dx.doi.org/10.1007/JHEP12(2020)151}{{\em JHEP} {\bfseries 12}
  (2020) 151}, \href{http://arxiv.org/abs/1909.11666}{{\ttfamily
  arXiv:1909.11666 [hep-th]}}.

\bibitem{Ferrara:1996wv}
S.~Ferrara, R.~Minasian, and A.~Sagnotti, ``{Low-energy analysis of M and F
  theories on Calabi-Yau threefolds},''
  \href{http://dx.doi.org/10.1016/0550-3213(96)00268-4}{{\em Nucl. Phys. B}
  {\bfseries 474} (1996) 323--342},
  \href{http://arxiv.org/abs/hep-th/9604097}{{\ttfamily arXiv:hep-th/9604097}}.

\bibitem{Cadavid:1995bk}
A.~C. Cadavid, A.~Ceresole, R.~D'Auria, and S.~Ferrara, ``{Eleven-dimensional
  supergravity compactified on Calabi-Yau threefolds},''
  \href{http://dx.doi.org/10.1016/0370-2693(95)00891-N}{{\em Phys. Lett. B}
  {\bfseries 357} (1995) 76--80},
  \href{http://arxiv.org/abs/hep-th/9506144}{{\ttfamily arXiv:hep-th/9506144}}.

\bibitem{Antoniadis:1997eg}
I.~Antoniadis, S.~Ferrara, R.~Minasian, and K.~S. Narain, ``{R**4 couplings in
  M and type II theories on Calabi-Yau spaces},''
  \href{http://dx.doi.org/10.1016/S0550-3213(97)00572-5}{{\em Nucl. Phys. B}
  {\bfseries 507} (1997) 571--588},
  \href{http://arxiv.org/abs/hep-th/9707013}{{\ttfamily arXiv:hep-th/9707013}}.

\bibitem{Kawamata1988}
Y.~Kawamata, ``Crepant blowing-up of 3-dimensional canonical singularities and
  its application to degenerations of surfaces,''
  \href{http://dx.doi.org/10.2307/1971417}{{\em The Annals of Mathematics}
  {\bfseries 127} no.~1, (Jan., 1988) 93}.
  \url{http://dx.doi.org/10.2307/1971417}.

\bibitem{ReidMinMod}
M.~Reid, \href{http://dx.doi.org/10.2969/aspm/00110131}{``Minimal models of
  canonical {$3$}-folds,''} in {\em Algebraic varieties and analytic varieties
  ({T}okyo, 1981)}, vol.~1 of {\em Adv. Stud. Pure Math.}, pp.~131--180.
\newblock North-Holland, Amsterdam, 1983.
\newblock \url{https://doi.org/10.2969/aspm/00110131}.

\bibitem{Grassi:2018rva}
A.~Grassi and T.~Weigand, ``{On topological invariants of algebraic threefolds
  with ($\mathbb Q$-factorial) singularities},''
  \href{http://arxiv.org/abs/1804.02424}{{\ttfamily arXiv:1804.02424
  [math.AG]}}.

\bibitem{Lawrie:2012gg}
C.~Lawrie and S.~Sch\"afer-Nameki, ``{The Tate Form on Steroids: Resolution and
  Higher Codimension Fibers},''
  \href{http://dx.doi.org/10.1007/JHEP04(2013)061}{{\em JHEP} {\bfseries 04}
  (2013) 061}, \href{http://arxiv.org/abs/1212.2949}{{\ttfamily arXiv:1212.2949
  [hep-th]}}.

\bibitem{Braun:2013nqa}
V.~Braun, T.~W. Grimm, and J.~Keitel, ``{Geometric Engineering in Toric
  F-Theory and GUTs with U(1) Gauge Factors},''
  \href{http://dx.doi.org/10.1007/JHEP12(2013)069}{{\em JHEP} {\bfseries 12}
  (2013) 069}, \href{http://arxiv.org/abs/1306.0577}{{\ttfamily arXiv:1306.0577
  [hep-th]}}.

\bibitem{Borchmann:2013hta}
J.~Borchmann, C.~Mayrhofer, E.~Palti, and T.~Weigand, ``{SU(5) Tops with
  Multiple U(1)s in F-theory},''
  \href{http://dx.doi.org/10.1016/j.nuclphysb.2014.02.006}{{\em Nucl. Phys. B}
  {\bfseries 882} (2014) 1--69},
  \href{http://arxiv.org/abs/1307.2902}{{\ttfamily arXiv:1307.2902 [hep-th]}}.

\bibitem{Buchmuller:2017wpe}
W.~Buchmuller, M.~Dierigl, P.-K. Oehlmann, and F.~Ruehle, ``{The Toric SO(10)
  F-Theory Landscape},'' \href{http://dx.doi.org/10.1007/JHEP12(2017)035}{{\em
  JHEP} {\bfseries 12} (2017) 035},
  \href{http://arxiv.org/abs/1709.06609}{{\ttfamily arXiv:1709.06609
  [hep-th]}}.

\bibitem{Apruzzi:2018nre}
F.~Apruzzi, L.~Lin, and C.~Mayrhofer, ``{Phases of 5d SCFTs from M-/F-theory on
  Non-Flat Fibrations},'' \href{http://dx.doi.org/10.1007/JHEP05(2019)187}{{\em
  JHEP} {\bfseries 05} (2019) 187},
  \href{http://arxiv.org/abs/1811.12400}{{\ttfamily arXiv:1811.12400
  [hep-th]}}.

\bibitem{Dierigl:2018nlv}
M.~Dierigl, P.-K. Oehlmann, and F.~Ruehle, ``{Global Tensor-Matter Transitions
  in F-Theory},'' \href{http://dx.doi.org/10.1002/prop.201800037}{{\em Fortsch.
  Phys.} {\bfseries 66} no.~7, (2018) 1800037},
  \href{http://arxiv.org/abs/1804.07386}{{\ttfamily arXiv:1804.07386
  [hep-th]}}.

\bibitem{Apruzzi:2019opn}
F.~Apruzzi, C.~Lawrie, L.~Lin, S.~Sch\"afer-Nameki, and Y.-N. Wang, ``{Fibers
  add Flavor, Part I: Classification of 5d SCFTs, Flavor Symmetries and BPS
  States},'' \href{http://dx.doi.org/10.1007/JHEP11(2019)068}{{\em JHEP}
  {\bfseries 11} (2019) 068}, \href{http://arxiv.org/abs/1907.05404}{{\ttfamily
  arXiv:1907.05404 [hep-th]}}.

\bibitem{Apruzzi:2019enx}
F.~Apruzzi, C.~Lawrie, L.~Lin, S.~Sch\"afer-Nameki, and Y.-N. Wang, ``{Fibers
  add Flavor, Part II: 5d SCFTs, Gauge Theories, and Dualities},''
  \href{http://dx.doi.org/10.1007/JHEP03(2020)052}{{\em JHEP} {\bfseries 03}
  (2020) 052}, \href{http://arxiv.org/abs/1909.09128}{{\ttfamily
  arXiv:1909.09128 [hep-th]}}.

\bibitem{Anderson:2018heq}
L.~B. Anderson, A.~Grassi, J.~Gray, and P.-K. Oehlmann, ``{F-theory on Quotient
  Threefolds with (2,0) Discrete Superconformal Matter},''
  \href{http://dx.doi.org/10.1007/JHEP06(2018)098}{{\em JHEP} {\bfseries 06}
  (2018) 098}, \href{http://arxiv.org/abs/1801.08658}{{\ttfamily
  arXiv:1801.08658 [hep-th]}}.

\bibitem{Anderson:2019kmx}
L.~B. Anderson, J.~Gray, and P.-K. Oehlmann, ``{F-Theory on Quotients of
  Elliptic Calabi-Yau Threefolds},''
  \href{http://dx.doi.org/10.1007/JHEP12(2019)131}{{\em JHEP} {\bfseries 12}
  (2019) 131}, \href{http://arxiv.org/abs/1906.11955}{{\ttfamily
  arXiv:1906.11955 [hep-th]}}.

\bibitem{Kohl:2021rxy}
F.~B. Kohl, M.~Larfors, and P.-K. Oehlmann, ``{F-theory on 6D symmetric
  toroidal orbifolds},'' \href{http://dx.doi.org/10.1007/JHEP05(2022)064}{{\em
  JHEP} {\bfseries 05} (2022) 064},
  \href{http://arxiv.org/abs/2111.07998}{{\ttfamily arXiv:2111.07998
  [hep-th]}}.

\bibitem{Kim:2019dqn}
H.-C. Kim, S.-S. Kim, and K.~Lee, ``{Higgsing and twisting of 6d D$_{N}$ gauge
  theories},'' \href{http://dx.doi.org/10.1007/JHEP10(2020)014}{{\em JHEP}
  {\bfseries 10} (2020) 014}, \href{http://arxiv.org/abs/1908.04704}{{\ttfamily
  arXiv:1908.04704 [hep-th]}}.

\bibitem{Braun:2021lzt}
A.~P. Braun, J.~Chen, B.~Haghighat, M.~Sperling, and S.~Yang, ``{Fibre-base
  duality of 5d KK theories},''
  \href{http://dx.doi.org/10.1007/JHEP05(2021)200}{{\em JHEP} {\bfseries 05}
  (2021) 200}, \href{http://arxiv.org/abs/2103.06066}{{\ttfamily
  arXiv:2103.06066 [hep-th]}}.

\bibitem{Kim:2021cua}
H.-C. Kim, M.~Kim, and S.-S. Kim, ``{Topological vertex for 6d SCFTs with
  $\mathbb{Z}_2$-twist},''
  \href{http://dx.doi.org/10.1007/JHEP03(2021)132}{{\em JHEP} {\bfseries 03}
  (2021) 132}, \href{http://arxiv.org/abs/2101.01030}{{\ttfamily
  arXiv:2101.01030 [hep-th]}}.

\bibitem{Bhardwaj:2022ekc}
L.~Bhardwaj, ``{Discovering T-dualities of little string theories},''
  \href{http://dx.doi.org/10.1007/JHEP02(2024)046}{{\em JHEP} {\bfseries 02}
  (2024) 046}, \href{http://arxiv.org/abs/2209.10548}{{\ttfamily
  arXiv:2209.10548 [hep-th]}}.

\bibitem{Lee:2022uiq}
K.~Lee, K.~Sun, and X.~Wang, ``{Twisted elliptic genera},''
  \href{http://dx.doi.org/10.1007/JHEP04(2024)035}{{\em JHEP} {\bfseries 04}
  (2024) 035}, \href{http://arxiv.org/abs/2212.07341}{{\ttfamily
  arXiv:2212.07341 [hep-th]}}.

\bibitem{Dong:1997ea}
C.-y. Dong, H.-s. Li, and G.~Mason, ``{Modular invariance of trace functions in
  orbifold theory},'' \href{http://dx.doi.org/10.1007/s002200000242}{{\em
  Commun. Math. Phys.} {\bfseries 214} (2000) 1--56},
  \href{http://arxiv.org/abs/q-alg/9703016}{{\ttfamily arXiv:q-alg/9703016}}.

\bibitem{ZUBER1986127}
J.-B. Zuber, ``Discrete symmetries of conformal theories,''
  \href{http://dx.doi.org/https://doi.org/10.1016/0370-2693(86)90936-6}{{\em
  Physics Letters B} {\bfseries 176} no.~1, (1986) 127--129}.
  \url{https://www.sciencedirect.com/science/article/pii/0370269386909366}.

\bibitem{Vafa:1986wx}
C.~Vafa, ``{Modular Invariance and Discrete Torsion on Orbifolds},''
  \href{http://dx.doi.org/10.1016/0550-3213(86)90379-2}{{\em Nucl. Phys. B}
  {\bfseries 273} (1986) 592--606}.

\bibitem{Vafa:1994rv}
C.~Vafa and E.~Witten, ``{On orbifolds with discrete torsion},''
  \href{http://dx.doi.org/10.1016/0393-0440(94)00048-9}{{\em J. Geom. Phys.}
  {\bfseries 15} (1995) 189--214},
  \href{http://arxiv.org/abs/hep-th/9409188}{{\ttfamily arXiv:hep-th/9409188}}.

\bibitem{Aspinwall:1995rb}
P.~S. Aspinwall, D.~R. Morrison, and M.~Gross, ``{Stable singularities in
  string theory},'' \href{http://dx.doi.org/10.1007/BF02104911}{{\em Commun.
  Math. Phys.} {\bfseries 178} (1996) 115--134},
  \href{http://arxiv.org/abs/hep-th/9503208}{{\ttfamily arXiv:hep-th/9503208}}.

\bibitem{GabberDJ}
A.~de~Jong, ``A result of gabber,''.
  \url{https://www.math.columbia.edu/~dejong/papers/2-gabber.pdf}.

\bibitem{LazarsfeldPositivityI}
R.~Lazarsfeld, \href{http://dx.doi.org/10.1007/978-3-642-18808-4}{{\em
  Positivity in algebraic geometry. {I}}}, vol.~48 of {\em Ergebnisse der
  Mathematik und ihrer Grenzgebiete. 3. Folge. A Series of Modern Surveys in
  Mathematics}.
\newblock Springer-Verlag, Berlin, 2004.
\newblock \url{https://doi.org/10.1007/978-3-642-18808-4}.
\newblock Classical setting: line bundles and linear series.

\bibitem{Witten:1998cd}
E.~Witten, ``{D-branes and K-theory},''
  \href{http://dx.doi.org/10.1088/1126-6708/1998/12/019}{{\em JHEP} {\bfseries
  12} (1998) 019}, \href{http://arxiv.org/abs/hep-th/9810188}{{\ttfamily
  arXiv:hep-th/9810188}}.

\bibitem{Sharpe:1999qz}
E.~R. Sharpe, ``{D-branes, derived categories, and Grothendieck groups},''
  \href{http://dx.doi.org/10.1016/S0550-3213(99)00535-0}{{\em Nucl. Phys. B}
  {\bfseries 561} (1999) 433--450},
  \href{http://arxiv.org/abs/hep-th/9902116}{{\ttfamily arXiv:hep-th/9902116}}.

\bibitem{Douglas:2000gi}
M.~R. Douglas, ``{D-branes, categories and N=1 supersymmetry},''
  \href{http://dx.doi.org/10.1063/1.1374448}{{\em J. Math. Phys.} {\bfseries
  42} (2001) 2818--2843}, \href{http://arxiv.org/abs/hep-th/0011017}{{\ttfamily
  arXiv:hep-th/0011017}}.

\bibitem{Aspinwall:2004jr}
P.~S. Aspinwall,
  \href{http://dx.doi.org/10.1142/9789812775108_0001}{``{D-branes on Calabi-Yau
  manifolds},''} in {\em {Theoretical Advanced Study Institute in Elementary
  Particle Physics (TASI 2003): Recent Trends in String Theory}}, pp.~1--152.
\newblock 3, 2004.
\newblock \href{http://arxiv.org/abs/hep-th/0403166}{{\ttfamily
  arXiv:hep-th/0403166}}.

\bibitem{Kapustin:2000aa}
A.~Kapustin and D.~Orlov, ``{Vertex algebras, mirror symmetry, and D-branes:
  The Case of complex tori},''
  \href{http://dx.doi.org/10.1007/s00220-002-0755-7}{{\em Commun. Math. Phys.}
  {\bfseries 233} (2003) 79--136},
  \href{http://arxiv.org/abs/hep-th/0010293}{{\ttfamily arXiv:hep-th/0010293}}.

\bibitem{Herbst:2008jq}
M.~Herbst, K.~Hori, and D.~Page, ``{Phases Of N=2 Theories In 1+1 Dimensions
  With Boundary},'' \href{http://arxiv.org/abs/0803.2045}{{\ttfamily
  arXiv:0803.2045 [hep-th]}}.

\bibitem{An2001}
S.~Y. An, S.~Y. Kim, D.~C. Marshall, S.~H. Marshall, W.~G. McCallum, and A.~R.
  Perlis, ``Jacobians of genus one curves,''
  \href{http://dx.doi.org/10.1006/jnth.2000.2632}{{\em Journal of Number
  Theory} {\bfseries 90} no.~2, (Oct., 2001) 304–315}.
  \url{http://dx.doi.org/10.1006/jnth.2000.2632}.

\bibitem{MirandaSmooth}
R.~Miranda, ``Smooth models for elliptic threefolds,'' in {\em The birational
  geometry of degenerations ({C}ambridge, {M}ass., 1981)}, vol.~29 of {\em
  Progr. Math.}, pp.~85--133.
\newblock Birkh\"auser, Boston, MA, 1983.

\bibitem{Weil1983}
A.~Weil, {\em Euler and the Jacobians of Elliptic Curves},
  \href{http://dx.doi.org/10.1007/978-1-4757-9284-3_15}{p.~353–359}.
\newblock Birkh\"{a}user Boston, 1983.
\newblock \url{http://dx.doi.org/10.1007/978-1-4757-9284-3_15}.

\bibitem{Addington2020}
N.~Addington and A.~Wray, ``Twisted fourier–mukai partners of enriques
  surfaces,'' \href{http://dx.doi.org/10.1007/s00209-020-02555-z}{{\em
  Mathematische Zeitschrift} {\bfseries 297} no.~3–4, (July, 2020)
  1239–1247}. \url{http://dx.doi.org/10.1007/s00209-020-02555-z}.

\bibitem{Moulinos2019}
T.~Moulinos, ``Derived azumaya algebras and twisted k-theory,''
  \href{http://dx.doi.org/10.1016/j.aim.2019.04.045}{{\em Advances in
  Mathematics} {\bfseries 351} (July, 2019) 761–803}.
  \url{http://dx.doi.org/10.1016/j.aim.2019.04.045}.

\bibitem{Calabrese2015}
J.~R. Calabrese and R.~P. Thomas, ``Derived equivalent calabi–yau threefolds
  from cubic fourfolds,''
  \href{http://dx.doi.org/10.1007/s00208-015-1260-6}{{\em Mathematische
  Annalen} {\bfseries 365} no.~1–2, (Aug., 2015) 155–172}.
  \url{http://dx.doi.org/10.1007/s00208-015-1260-6}.

\bibitem{Kuznetsov2013}
A.~Kuznetsov, ``Scheme of lines on a family of 2-dimensional quadrics: geometry
  and derived category,''
  \href{http://dx.doi.org/10.1007/s00209-013-1217-y}{{\em Mathematische
  Zeitschrift} {\bfseries 276} no.~3–4, (Sept., 2013) 655–672}.
  \url{http://dx.doi.org/10.1007/s00209-013-1217-y}.

\bibitem{HarrisTu84}
J.~Harris and L.~W. Tu, ``On symmetric and skew-symmetric determinantal
  varieties,'' \href{http://dx.doi.org/10.1016/0040-9383(84)90026-0}{{\em
  Topology} {\bfseries 23} no.~1, (1984) 71--84}.
  \url{https://doi.org/10.1016/0040-9383(84)90026-0}.

\bibitem{Gross:1993fd}
M.~Gross, ``{A Finiteness theorem for elliptic Calabi-Yau threefolds},''
  \href{http://arxiv.org/abs/alg-geom/9305002}{{\ttfamily
  arXiv:alg-geom/9305002}}.

\bibitem{Filipazzi:2021dcw}
S.~Filipazzi, C.~D. Hacon, and R.~Svaldi, ``{Boundedness of elliptic Calabi-Yau
  threefolds},'' \href{http://dx.doi.org/10.4171/JEMS/1467}{{\em J. Eur. Math.
  Soc.} {\bfseries 27} (2025) 3583--3650},
  \href{http://arxiv.org/abs/2112.01352}{{\ttfamily arXiv:2112.01352
  [math.AG]}}.

\bibitem{Braun:2014qka}
V.~Braun, T.~W. Grimm, and J.~Keitel, ``{Complete Intersection Fibers in
  F-Theory},'' \href{http://dx.doi.org/10.1007/JHEP03(2015)125}{{\em JHEP}
  {\bfseries 03} (2015) 125}, \href{http://arxiv.org/abs/1411.2615}{{\ttfamily
  arXiv:1411.2615 [hep-th]}}.

\bibitem{Ooguri:2006in}
H.~Ooguri and C.~Vafa, ``{On the Geometry of the String Landscape and the
  Swampland},'' \href{http://dx.doi.org/10.1016/j.nuclphysb.2006.10.033}{{\em
  Nucl. Phys. B} {\bfseries 766} (2007) 21--33},
  \href{http://arxiv.org/abs/hep-th/0605264}{{\ttfamily arXiv:hep-th/0605264}}.

\bibitem{Kumar2009}
V.~Kumar and W.~Taylor, ``A bound on 6d $\mathcal{N}=1$ supergravities,''
  \href{http://dx.doi.org/10.1088/1126-6708/2009/12/050}{{\em Journal of High
  Energy Physics} {\bfseries 2009} no.~12, (Dec., 2009) 050–050}.
  \url{http://dx.doi.org/10.1088/1126-6708/2009/12/050}.

\bibitem{Kumar2010}
V.~Kumar, D.~R. Morrison, and W.~Taylor, ``Global aspects of the space of 6d $
  \mathcal{N} = 1 $ supergravities,''
  \href{http://dx.doi.org/10.1007/jhep11(2010)118}{{\em Journal of High Energy
  Physics} {\bfseries 2010} no.~11, (Nov., 2010) }.
  \url{http://dx.doi.org/10.1007/JHEP11(2010)118}.

\bibitem{Kumar2011}
V.~Kumar, D.~S. Park, and W.~Taylor, ``6d supergravity without tensor
  multiplets,'' \href{http://dx.doi.org/10.1007/jhep04(2011)080}{{\em Journal
  of High Energy Physics} {\bfseries 2011} no.~4, (Apr., 2011) }.
  \url{http://dx.doi.org/10.1007/JHEP04(2011)080}.

\bibitem{Seiberg2011}
N.~Seiberg and W.~Taylor, ``Charge lattices and consistency of 6d
  supergravity,'' \href{http://dx.doi.org/10.1007/jhep06(2011)001}{{\em Journal
  of High Energy Physics} {\bfseries 2011} no.~6, (June, 2011) }.
  \url{http://dx.doi.org/10.1007/JHEP06(2011)001}.

\bibitem{Kim2019}
H.-C. Kim, G.~Shiu, and C.~Vafa, ``Branes and the swampland,''
  \href{http://dx.doi.org/10.1103/physrevd.100.066006}{{\em Physical Review D}
  {\bfseries 100} no.~6, (Sept., 2019) }.
  \url{http://dx.doi.org/10.1103/PhysRevD.100.066006}.

\bibitem{Tarazi:2021duw}
H.-C. Tarazi and C.~Vafa, ``{On The Finiteness of 6d Supergravity Landscape},''
  \href{http://arxiv.org/abs/2106.10839}{{\ttfamily arXiv:2106.10839
  [hep-th]}}.

\bibitem{Grassi:2023aks}
A.~Grassi, ``{Spectrum bounds in geometry},''
  \href{http://arxiv.org/abs/2304.07819}{{\ttfamily arXiv:2304.07819
  [math.AG]}}.

\bibitem{Hamada2024a}
Y.~Hamada and G.~J. Loges, ``Enumerating 6d supergravities with $t\le 1$,''
  \href{http://dx.doi.org/10.1007/jhep12(2024)167}{{\em Journal of High Energy
  Physics} {\bfseries 2024} no.~12, (Dec., 2024) }.
  \url{http://dx.doi.org/10.1007/JHEP12(2024)167}.

\bibitem{Hamada2024b}
Y.~Hamada and G.~J. Loges, ``Towards a complete classification of 6d
  supergravities,'' \href{http://dx.doi.org/10.1007/jhep02(2024)095}{{\em
  Journal of High Energy Physics} {\bfseries 2024} no.~2, (Feb., 2024) }.
  \url{http://dx.doi.org/10.1007/JHEP02(2024)095}.

\bibitem{Kim:2024hxe}
H.-C. Kim, C.~Vafa, and K.~Xu, ``{Finite Landscape of 6d N=(1,0)
  Supergravity},'' \href{http://arxiv.org/abs/2411.19155}{{\ttfamily
  arXiv:2411.19155 [hep-th]}}.

\bibitem{Hamada:2025vga}
Y.~Hamada and G.~J. Loges, ``{A finite 6d supergravity landscape from
  anomalies},'' \href{http://arxiv.org/abs/2507.20949}{{\ttfamily
  arXiv:2507.20949 [hep-th]}}.

\bibitem{Birkar:2025rcg}
C.~Birkar and S.-J. Lee, ``{Explicit Bounds on the Spectrum of 6d N=(1,0)
  Supergravity},'' \href{http://arxiv.org/abs/2507.06295}{{\ttfamily
  arXiv:2507.06295 [hep-th]}}.

\bibitem{Birkar:2025gvs}
C.~Birkar and S.-J. Lee, ``{A Picard rank bound for base surfaces of elliptic
  Calabi-Yau 3-folds},'' \href{http://arxiv.org/abs/2507.06317}{{\ttfamily
  arXiv:2507.06317 [hep-th]}}.

\bibitem{Cox1999-nj}
D.~A. Cox and S.~Katz, {\em Mirror symmetry and algebraic geometry}.
\newblock Mathematical Surveys and Monographs. American Mathematical Society,
  Providence, RI, Mar., 1999.

\bibitem{Davies:2009ub}
R.~Davies, ``{Quotients of the conifold in compact Calabi-Yau threefolds, and
  new topological transitions},''
  \href{http://dx.doi.org/10.4310/ATMP.2010.v14.n3.a6}{{\em Adv. Theor. Math.
  Phys.} {\bfseries 14} no.~3, (2010) 965--990},
  \href{http://arxiv.org/abs/0911.0708}{{\ttfamily arXiv:0911.0708 [hep-th]}}.

\bibitem{Beauville1983}
A.~Beauville, ``Variétés k\"{a}hleriennes dont la première classe de chern
  est nulle,'' \href{http://dx.doi.org/10.4310/jdg/1214438181}{{\em Journal of
  Differential Geometry} {\bfseries 18} no.~4, (Jan., 1983) }.
  \url{http://dx.doi.org/10.4310/jdg/1214438181}.

\bibitem{Candelas:1989ug}
P.~Candelas, P.~S. Green, and T.~Hubsch, ``{Rolling Among Calabi-Yau Vacua},''
  \href{http://dx.doi.org/10.1016/0550-3213(90)90302-T}{{\em Nucl. Phys. B}
  {\bfseries 330} (1990) 49}.

\bibitem{Doran:2024kcb}
C.~Doran, B.~Pioline, and T.~Schimannek, ``{Enumerative geometry and modularity
  in two-modulus K3-fibered Calabi-Yau threefolds},''
  \href{http://arxiv.org/abs/2408.02994}{{\ttfamily arXiv:2408.02994
  [hep-th]}}.

\bibitem{WernerThesis}
J.~Werner, {\em Kleine Aufl{\"o}sungen spezieller dreidimensionaler
  Variet{\"a}ten}.
\newblock Phd thesis, University of Bonn, 1987.

\bibitem{werner2022smallresolutionsspecialthreedimensional}
J.~Werner, S.~Venter, and N.~Addington, ``Small resolutions of special
  three-dimensional varieties,'' 2022.
\newblock \url{https://arxiv.org/abs/2208.01383}.

\bibitem{Clemens1983}
C.~Clemens, ``Double solids,''
  \href{http://dx.doi.org/10.1016/0001-8708(83)90025-7}{{\em Advances in
  Mathematics} {\bfseries 47} no.~2, (Feb., 1983) 107–230}.
  \url{http://dx.doi.org/10.1016/0001-8708(83)90025-7}.

\bibitem{Matsuki2002}
K.~Matsuki, \href{http://dx.doi.org/10.1007/978-1-4757-5602-9}{{\em
  Introduction to the Mori Program}}.
\newblock Springer New York, 2002.
\newblock \url{http://dx.doi.org/10.1007/978-1-4757-5602-9}.

\bibitem{Namikawa1995}
Y.~Namikawa and J.~H.~M. Steenbrink, ``Global smoothing of calabi-yau
  threefolds,'' \href{http://dx.doi.org/10.1007/bf01231450}{{\em Inventiones
  Mathematicae} {\bfseries 122} no.~1, (Dec., 1995) 403–419}.
  \url{http://dx.doi.org/10.1007/BF01231450}.

\bibitem{DolgachevBook}
I.~V. Dolgachev, \href{http://dx.doi.org/10.1017/CBO9781139084437}{{\em
  Classical algebraic geometry}}.
\newblock Cambridge University Press, Cambridge, 2012.
\newblock \url{https://doi.org/10.1017/CBO9781139084437}.
\newblock A modern view.

\bibitem{FevolaQuadrics}
C.~Fevola, Y.~Mandelshtam, and B.~Sturmfels, ``Pencils of quadrics: old and
  new,'' \href{http://dx.doi.org/10.4418/2021.76.2.2}{{\em Matematiche
  (Catania)} {\bfseries 76} no.~2, (2021) 319--335}.
  \url{https://doi.org/10.4418/2021.76.2.2}.

\bibitem{HodgePedoe}
W.~V.~D. Hodge and D.~Pedoe,
  \href{http://dx.doi.org/10.1017/CBO9780511623899}{{\em Methods of algebraic
  geometry. {V}ol. {II}}}.
\newblock Cambridge Mathematical Library. Cambridge University Press,
  Cambridge, 1994.
\newblock \url{https://doi.org/10.1017/CBO9780511623899}.
\newblock Book III: General theory of algebraic varieties in projective space,
  Book IV: Quadrics and Grassmann varieties, Reprint of the 1952 original.

\bibitem{Eichler1985}
M.~Eichler and D.~Zagier,
  \href{http://dx.doi.org/10.1007/978-1-4684-9162-3}{{\em The Theory of Jacobi
  Forms}}.
\newblock Birkh\"{a}user Boston, 1985.
\newblock \url{http://dx.doi.org/10.1007/978-1-4684-9162-3}.

\bibitem{Doran:2013npa}
C.~F. Doran, T.~Gannon, H.~Movasati, and K.~M. Shokri, ``{Automorphic forms for
  triangle groups},'' \href{http://dx.doi.org/10.4310/CNTP.2013.v7.n4.a4}{{\em
  Commun. Num. Theor Phys.} {\bfseries 07} (2013) 689--737},
  \href{http://arxiv.org/abs/1307.4372}{{\ttfamily arXiv:1307.4372 [math.NT]}}.

\bibitem{FultonIntersection}
W.~Fulton, \href{http://dx.doi.org/10.1007/978-1-4612-1700-8}{{\em Intersection
  theory}}, vol.~2 of {\em Ergebnisse der Mathematik und ihrer Grenzgebiete. 3.
  Folge. A Series of Modern Surveys in Mathematics}.
\newblock Springer-Verlag, Berlin, second~ed., 1998.
\newblock \url{https://doi.org/10.1007/978-1-4612-1700-8}.

\bibitem{MR0637060}
S.~Iitaka, {\em Algebraic geometry}, vol.~24 of {\em North-Holland Mathematical
  Library}.
\newblock Springer-Verlag, New York-Berlin, 1982.
\newblock An introduction to birational geometry of algebraic varieties,
  Graduate Texts in Mathematics, 76.

\end{thebibliography}\endgroup

\end{document}